\documentclass[letterpaper,11pt]{article}

\usepackage[letterpaper,margin=1in]{geometry}
\usepackage[utf8]{inputenc}
\usepackage{times}

\usepackage{enumerate}
\usepackage{enumitem}
\setlist[enumerate,1]{label=(\roman*), leftmargin=2.2em, itemsep=0pt,topsep=0.3em}
\setlist[enumerate,2]{label=(\alph*),nosep}
\setlist[itemize]{itemsep=0pt,topsep=0.3em}

\usepackage{subcaption}

\usepackage[table, dvipsnames]{xcolor}

\definecolor{linkcol}{rgb}{0,0,0.38}
\definecolor{citecol}{rgb}{0,0.2,0}
\definecolor{urlcol}{rgb}{0.1,0.35,0}

\usepackage[pdfpagelabels,bookmarks=false,hyperfootnotes=false]{hyperref}
\hypersetup{colorlinks, linkcolor=linkcol, citecolor=citecol, urlcolor=urlcol}

\usepackage{url}

\usepackage{amsthm}
\usepackage{amsmath}
\usepackage{amssymb}
\usepackage{thmtools}
\usepackage{thm-restate}
\usepackage{mathtools}
\usepackage{bbm}
\usepackage{xfrac}

\usepackage[nameinlink, capitalise]{cleveref}
\Crefname{appsec}{Appendix}{Appendices}

\newtheorem{theorem}{Theorem}[section]
\newtheorem{lemma}[theorem]{Lemma}
\newtheorem{proposition}[theorem]{Proposition}
\newtheorem{corollary}[theorem]{Corollary}

\theoremstyle{definition}
\newtheorem{definition}[theorem]{Definition}

\usepackage{aligned-overset}

\usepackage[
	backend=biber,
	style=alphabetic,
	citestyle=alphabetic,
	maxalphanames=4,
	maxcitenames=99,
	mincitenames=98,
	maxbibnames=99,
	giveninits=true,
]{biblatex}

\usepackage{xpatch}

\makeatletter
\patchcmd\blx@bblinput{\blx@blxinit}
                      {\blx@blxinit
                      }{}{\fail}
\makeatother

\usepackage{tikz}
\usetikzlibrary{calc}
\usetikzlibrary{math}
\usetikzlibrary{shapes.geometric}
\usetikzlibrary{arrows}
\usetikzlibrary{decorations.pathreplacing, decorations.pathmorphing}

\usepackage{float}
\usepackage{graphicx, dsfont}
\usepackage{xspace}
\usepackage[linesnumbered, ruled, vlined]{algorithm2e}
\usepackage[margin=10pt,font=small,labelfont=bf,skip=-5pt]{caption}

\usepackage[textsize=footnotesize, color=orange!40]{todonotes}
\newcommand\bN{\ensuremath{\mathbb{N}}}
\newcommand\bZ{\ensuremath{\mathbb{Z}}}

\newcommand\bR{\ensuremath{\mathbb{R}}}
\newcommand\bE{\ensuremath{\mathbb{E}}}
\newcommand\bP{\ensuremath{\mathbb{P}}}

\newcommand\cB{\ensuremath{\mathcal{B}}}

\newcommand\cH{\ensuremath{\mathcal{H}}}
\newcommand\cI{\ensuremath{\mathcal{I}}}
\newcommand\cJ{\ensuremath{\mathcal{J}}}
\newcommand\cF{\ensuremath{\mathcal{F}}}

\newcommand\cS{\ensuremath{\mathcal{S}}}
\newcommand\cU{\ensuremath{\mathcal{U}}}
\newcommand\cP{\ensuremath{\mathcal{P}}}
\newcommand\cW{\ensuremath{\mathcal{W}}}

\DeclareRobustCommand{\OPT}{\ensuremath{\mathrm{OPT}}}
\DeclarePairedDelimiter\ceil{\lceil}{\rceil}
\DeclarePairedDelimiter\floor{\lfloor}{\rfloor}
\renewcommand{\epsilon}{\varepsilon}

\def\cupp{\stackrel{.}{\cup}}

\DeclareMathOperator*{\diam}{diam}
\DeclareMathOperator*{\poly}{poly}
\DeclareMathOperator*{\level}{level}
\DeclareMathOperator*{\height}{height}
\DeclareMathOperator*{\width}{width}

\newcommand\scconstant{\ensuremath{\eta}}

\newcommand\adaptivitygap{\ensuremath{\alpha_{\mathrm{adapt}}}}

\DeclareMathOperator*{\argmin}{arg\,min}

\newcommand*{\hop}{\ensuremath{\mathrm{Hop}}}

\newcommand\expectation[1]{\ensuremath{\mathbb{E}\left[ #1 \right]}}
\newcommand\probability[1]{\ensuremath{\mathbb{P} \left[ #1 \right]}}
\newcommand\probabilitycustom[2]{\ensuremath{\mathbb{P}_{ #1 } \left[ #2 \right]}}
\newcommand\condexp[3]{\ensuremath{\mathbb{E}_{ #1 }\left[ #2 \;\middle|\; #3 \right]}}

\newcommand{\ATSP}{\textsc{ATSP}\xspace}
\newcommand{\APrioriATSP}{Asymmetric A Priori TSP\xspace}
\newcommand{\HopATSP}{Hop-ATSP\xspace}

\newcommand\parset[1]{\ensuremath{\left\{ #1 \right\}}}
\newcommand\expectationAP[1]{\ensuremath{\mathbb{E}_{A \sim p}\left[ #1 \right]}}
\newcommand\expectationcustom[2]{\ensuremath{\mathbb{E}_{ #1 }\left[ #2 \right]}}

\newcommand\congestion{\ensuremath{{\mathrm{cong}}}}

\newcommand{\addparskip}{\medskip}

\newcommand\closedint[1]{\ensuremath{\left[#1\right]_{\prec}}}
\newcommand\openint[1]{\ensuremath{\left(#1\right)_{\prec}}}
\newcommand\coveringsum{\ensuremath{\oplus}}
\newcommand\cost{\ensuremath{\mathrm{cost}}}
\newcommand\intcost{\ensuremath{\mathrm{cost}_{\mathrm{int}}}}
\newcommand\extcost{\ensuremath{\mathrm{cost}_{\mathrm{ext}}}}

\newcommand\numnewcovered{\ensuremath{\Lambda}}
\newcommand\singletonprofile{\ensuremath{\Psi}} 
\title{Approximating Asymmetric A Priori TSP\\ beyond the Adaptivity Gap}

\author{
Manuel Christalla\thanks{
Department of Computer Science, ETH Zurich, Switzerland.
Email: \href{mailto:mchristalla@ethz.ch}{mchristalla@ethz.ch}. \newline
Part of this work was done while the author was at University of Bonn.
}
\and
Luise Puhlmann\thanks{
Research Institute for Discrete Mathematics and Hausdorff Center for Mathematics, University of Bonn, Germany.
\newline
Email: \href{mailto:puhlmann@dm.uni-bonn.de}{puhlmann@dm.uni-bonn.de}.
}
\and
Vera Traub\thanks{
Department of Computer Science, ETH Zurich, Switzerland.
Email: \href{mailto:vtraub@ethz.ch}{vtraub@ethz.ch}.\newline
Part of this work was done while the author was at University of Bonn.
}
}
\date{}

\begin{document}

\maketitle

\begin{abstract}
In Asymmetric A Priori TSP (with independent activation probabilities) we are given an instance of the Asymmetric Traveling Salesman Problem together with an activation probability for each vertex.
The task is to compute a tour that minimizes the expected length after short-cutting to the randomly sampled set of active vertices.

We prove a polynomial lower bound on the adaptivity gap for Asymmetric A Priori TSP.   
Moreover, we show that a poly-logarithmic approximation ratio, and hence an approximation ratio below the adaptivity gap, can be achieved by a randomized algorithm with quasi-polynomial running time.

To achieve this, we provide a series of polynomial-time reductions.
First we reduce to a novel generalization of the Asymmetric Traveling Salesman Problem, called Hop-ATSP.
Next, we use directed low-diameter decompositions to obtain structured instances, for which we then provide a reduction to a covering problem. Eventually, we obtain a polynomial-time reduction of Asymmetric A Priori TSP to a problem of finding a path in an acyclic digraph minimizing a particular objective function, for which we give an $O(\log n)$-approximation algorithm in quasi-polynomial time.
\end{abstract}

\thispagestyle{empty}
\newpage

\setcounter{page}{1}

\section{Introduction}

The traveling salesman problem (TSP) is among the most famous and well-studied problems in combinatorial optimization and theoretical computer science.
An instance of the Asymmetric TSP consists of a finite set $V$ of vertices and distances/costs $c: V \times V \to \bE_{\geq 0}$ satisfying the triangle inequality, i.e., $c(x,z) \leq c(x,y) + c(y,z)$ for all $x,y,z \in V$.
The task is to compute a shortest tour on $V$, that is a cycle with vertex set $V$ minimizing the sum of the edge costs.
Symmetric TSP is the special case where $c(x,y)=c(y,x)$ for all $x,y\in V$.

While it is easy to achieve constant-factor approximations for the Symmetric TSP \cite{rosenkrantz1977analysis} (improved in \cite{christofides,serdjukov,karlin2024slightly}), for a long time only logarithmic approximation factors were known for the Asymmetric TSP \cite{frieze1982worst,blaser2008new,kaplan2005approximation, feige2007improved} (see also \cite{traub25}).
The first constant-factor approximations for the more general Asymmetric TSP have been discovered relatively recently \cite{svensson2020constant} (improved in \cite{traub2022improved,traub25}).

In this work, we focus on a stochastic version of the Asymmetric TSP, where each vertex $v\in V$ has an activation probability $p(v)$.
We write $A \sim p$ to indicate that we sample a random set $A\subseteq V$ of active vertices from $p$, i.e., every vertex belongs to $A$ with probability $p(v)$ and the sampling of different vertices is independent.
For a tour $T$ on $V$, the tour $T[A]$ on $A$ results from $T$ by short-cutting to the set $A$ of active vertices, i.e., skipping all visits of vertices outside of $A$.
See \Cref{fig:example_shortcutting}.

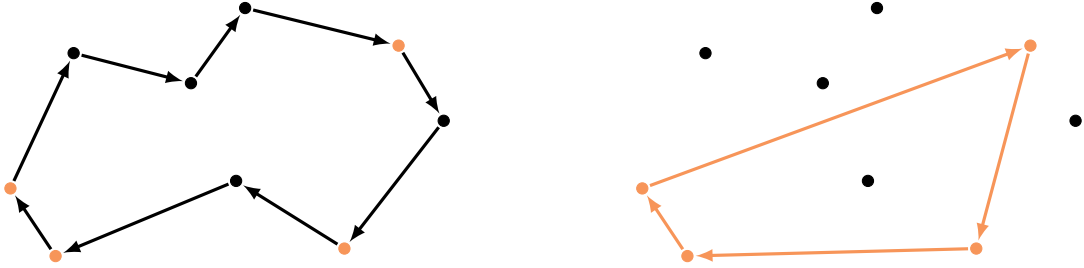
\begin{figure}
\begin{center}
	\begin{tikzpicture}[
xscale=1.2,
vertex/.style={circle,draw=black, fill,inner sep=1.5pt, outer sep=1pt},
edge/.style={-latex,very thick},
box/.style={very thick, fill=Gray, fill opacity=0.2, rounded corners=8pt},
]

\definecolor{color1}{named}{Peach} 

\begin{scope}[every node/.style={vertex}]
\node (v1) at (0,0) {};
\node[color=color1] (v2) at (-2,-1) {};
\node[color=color1] (v3) at (-2.5,-0.1) {};
\node (v4) at (-1.8,1.7) {};
\node (v5) at (-0.5,1.3) {};
\node (v6) at (0.1,2.3) {};
\node[color=color1] (v7) at (1.8,1.8) {};
\node (v8) at (2.3,0.8) {};
\node[color=color1] (v9) at (1.2,-0.9) {};
\foreach \i [evaluate=\i as \next using int(\i+1)] in {1,...,8}{
	\draw[edge] (v\i) to (v\next);
}
\draw[edge] (v9) to (v1);
\end{scope}

\begin{scope}[every node/.style={vertex}, xshift=7cm]
	\node (v1) at (0,0) {};
	\node[color=color1] (v2) at (-2,-1) {};
	\node[color=color1] (v3) at (-2.5,-0.1) {};
	\node (v4) at (-1.8,1.7) {};
	\node (v5) at (-0.5,1.3) {};
	\node (v6) at (0.1,2.3) {};
	\node[color=color1] (v7) at (1.8,1.8) {};
	\node (v8) at (2.3,0.8) {};
	\node (v9)[color=color1] at (1.2,-0.9) {};
	\draw[edge, color=color1] (v9) to (v2);
	\draw[edge, color=color1] (v2) to (v3);
	\draw[edge, color=color1] (v3) to (v7);
	\draw[edge, color=color1] (v7) to (v9);
\end{scope}
\end{tikzpicture} \end{center}
\caption{Left: a tour $T$ through a set of vertices; right: the tour $T[A]$ resulting from cutting $T$ short to the set $A$ of orange vertices.\label{fig:example_shortcutting}
}
\end{figure}

In Asymmetric A Priori TSP (with independent activation) we are given a set $V$ of vertices, an activation probability $p(v) \in [0,1]$ for each vertex $v\in V$, and a distance/cost function $c : V\times V \to \bR_{\geq 0}$ satisfying the triangle inequality.
The task is to compute a tour~$T$ on $V$ minimizing the expected cost
\begin{equation}\label{eq:objective}
\bE_{A\sim p} \Big[ c\big(T[A]\big)  \Big]
\end{equation}
after short-cutting $T$ to the random set $A$ of active vertices.

The symmetric special case, where $c(x,y) = c(y,x)$ for all $x,y \in V$, is well-studied.
Shmoys and Talwar \cite{shmoys_talwar08} gave a randomized 4-approximation algorithm and a deterministic 8-approximation algorithm.
Independently, also Garg, Gupta, Leonardi, and Sankowski \cite{garg942stochastic} gave a randomized constant-factor approximation for the Symmetric A Priori TSP.
Using better approximation algorithms for TSP as a subroutine, one can improve the randomized algorithm by Shmoys and Talwar to achieve an approximation factor slightly below $3.5$ (see \cite{van_ee_sitters_2018,blauth_et_al}).
Van Zuylen \cite{van2011deterministic} gave a deterministic $6.5$-approximation algorithm.
Recently, Blauth, Neuwohner, Puhlmann, and Vygen \cite{blauth_et_al} improved the approximation factors further, achieving a randomized $3.1$-approximation and a deterministic $5.9$-approximation.

The approximation factor compares the objective value~\eqref{eq:objective} of the computed tour to the optimal objective value, i.e.,
to  the expected cost 
\[
 \OPT(V,c,p) \ \coloneqq\ \min_{T\text{ tour on }V} \bE_{A\sim p} \Big[ c\big(T[A]\big)  \Big]
\]
of an optimal a priori tour after short-cutting.
Many of the above-mentioned algorithms have been analyzed with respect to the expected length of the best a posteriori tour, i.e., with respect to $\bE_{A \sim p}\big[ \rm{TSP}(A,c) \big]$, where $\rm{TSP}(A,c)$ denotes the length of a shortest tour on the vertex set $A$.
This is the expected length of the in hindsight optimal tour, i.e., the best tour one could achieve when knowing the set $A$ upfront instead of knowing merely the activation probabilities.
It provides a lower bound on the optimal value $\OPT(V,c,p)$.
With such an analysis it is impossible to achieve approximation factors below the \emph{adaptivity gap}
\[
\sup \left\{ \frac{ \OPT(V,c,p)}{\bE_{A \sim p}\big[ \rm{TSP}(A,c) \big] } : (V,c,p)\text{ instance of A Priori TSP}\right\}.
\]
Blauth, Neuwohner, Puhlmann, and Vygen \cite{blauth_et_al} analyze their algorithm with respect to the cost $\OPT(V,c,p)$ of the optimal a priori tour and not with respect to the optimal a posteriori tour.
However, they do not achieve an approximation factor below the adaptivity gap of the Symmetric A Priori TSP, because \cite{shmoys_talwar08} implies that the adaptivity gap of Symmetric A Priori TSP is at most $3$.

In contrast to the symmetric special case, to the best of our knowledge, no nontrivial approximation algorithms for the Asymmetric A Priori TSP are known.

\subsection{Our Contribution}\label{sec:our_contribution}

We provide the first nontrivial approximation algorithms for the Asymmetric A Priori TSP.
Observe that any cycle on $V$ is an $n$-approximation.
We show that one can obtain an $O(\sqrt{n})$-approximation, even when comparing to the optimal a posteriori tour.
This also implies an upper bound on the adaptivity gap.

\begin{restatable}{theorem}{SimpleApprox}\label{thm:simple_approx}
There is a deterministic polynomial-time algorithm that for any instance $(V,c,p)$ of Asymmetric A Priori TSP, computes a tour $T$ with
\[
 \bE_{A\sim p} \Big[ c\big(T[A]\big)  \Big] \ \leq\ O\big(\sqrt{n}\big) \cdot \bE_{A \sim p}\Big[ \rm{TSP}(A,c) \Big],
\]
where $n=|V|$ denotes the number of vertices.
In particular, the adaptivity gap of the Asymmetric A Priori TSP is $O(\sqrt{n})$.
\end{restatable}

Moreover, we prove a polynomial lower bound on the adaptivity gap of Asymmetric A Priori TSP.

\begin{restatable}{theorem}{LowerBoundAdaptivity}\label{thm:lower_bound}
The adaptivity gap of the Asymmetric A Priori TSP is $\Omega(n^{\sfrac{1}{4}} \cdot \log^{-1} n)$, where $n$ denotes the number of vertices.
\end{restatable}

This shows that one cannot achieve subpolynomial approximation factors when comparing solely to the optimal a posteriori tour (even when allowing exponential running time).
Our main contribution is a randomized algorithm with a poly-logarithmic approximation factor and  quasi-polynomial running time:
\begin{theorem}\label{thm:main}
There is a randomized algorithm with running time $n^{O(\log n)}$ that given an instance $(V,c,p)$ of Asymmetric A Priori TSP, computes a tour $T$ such that in expectation 
\[
  \bE_{A\sim p} \Big[ c\big(T[A]\big)  \Big] \ \leq\ O(\log^{8}n) \cdot \OPT(V,c,p),
\]
where $n=|V|$ denotes the number of vertices.
\end{theorem}
Together with \Cref{thm:lower_bound}, this shows that we can achieve an approximation factor below the adaptivity gap in quasi-polynomial time.
We remark that we did not optimize the exponent of the $\log n$ in the approximation ratio.

In order to prove \Cref{thm:main}, we provide a series of polynomial-time reductions, leading to a covering problem, and eventually to a problem of finding a path in an acyclic graph minimizing a particular objective function.
We emphasize that our algorithm for (approximately) solving the latter problem is the only part that is not polynomial-time, but requires quasi-polynomial running time.
Our proofs even imply that the covering problem to which we reduce is equivalent to Asymmetric A Priori TSP up to poly-logarithmic factors in the approximation guarantee and polynomial factors in the running time (see \Cref{sec:conclusion}).
\addparskip

Although simple $O(\log n)$-approximation algorithms for ATSP are known (see e.g. \cite{frieze1982worst,blaser2008new}), none of them seems to extend to the Asymmetric A Priori TSP.
We therefore do not directly generalize any of these algorithms, but develop a different approach.

We introduce a natural generalization of ATSP, where instead of paying the distance of each vertex to its direct successor on the tour, we pay the distance of each vertex to the $k$ vertices succeeding it.
Here, $k$ is part of the input and could be large (e.g. $k\approx\sqrt{n})$.
We call this problem Hop-ATSP and provide a reduction from Asymmetric A Priori TSP to Hop-ATSP.

In order to find a poly-logarithmic approximation for Hop-ATSP, we use directed low-diameter decompositions \cite{bernstein2022negative,bringmann2025near,li2025simpler}, which played an important role in recent breakthroughs in the context of fast graph algorithms, including \cite{bernstein2022negative,bringmann2023negative}.
Directed low-diameter decompositions allow us to assume that our instances of Hop-ATSP have a very particular structure.
Exploiting this, we then provide a reduction of Hop-ATSP to a covering problem, which is the main technical challenge in our series of reductions.

In order to solve this covering problem via the set cover greedy algorithm (or other well-known set cover algorithms), we then need an oracle for finding a ``best'' set to pick next.
This problem turns out to be a problem of finding a path in a directed acyclic graph minimizing a particular objective function.
While directly minimizing this objective function in a dynamic programming approach leads to an exponential running time,
we show that a suitable approximation of it can be minimized in quasi-polynomial time.

\subsection{Further related work}

For the Symmetric A Priori TSP, \cite{amar2017exact}  proposed an exact algorithm, and \cite{jaillet1985probabilistic,jaillet1988priori,bertsimas1990priori} investigated the behavior of random instances.
Symmetric A Priori TSP has also been studied in a setting where the activation of different vertices is not necessarily independent \cite{schalekamp2008algorithms,van2018priori,gorodezky2010improved}.
The variant of the problem where the set of active vertices is chosen adversarially instead of randomly is known as Universal TSP \cite{jia2005universal,schalekamp2008algorithms, gorodezky2010improved,bhalgat2011optimal,gupta06,hajiaghayi06}.

Besides TSP, many other problems have been studied in an a priori setting, including for example Vehicle Routing problems, Set Cover, and Steiner Tree \cite{adamczyk2017optimum,garg942stochastic,feuerstein2014scheduling,van2018priori,fernstrom2020constant,navidi2020approximation, grandoni2013set,gupta2024set, hajiaghayi2005oblivious}.

\subsection{Structure of the paper}

In \Cref{sec:overview} we expand on the techniques we use to prove our main result, \Cref{thm:main}.
Next, in \Cref{sec:adaptivity_gap}, we prove our lower bound on the adaptivity gap of Asymmetric A Priori TSP (\Cref{thm:lower_bound}).
In \Cref{sec:reducing_hop_atsp} we explain our reduction to Hop-ATSP and prove \Cref{thm:simple_approx}.
\Cref{sec:hierarchically_ordered,sec:reducing_to_covering} contain our reduction  of Asymmetric A Priori TSP to a covering problem.
 \Cref{sec:monotone_paths} shows how we can solve the covering problem approximately in quasi-polynomial time, and \Cref{sec:conclusion} contains some further discussion.

\section{Outline of the Proof of \Cref{thm:main}}\label{sec:overview}

In this section we provide an overview of the proof of our main result (\Cref{thm:main}).
As a first step, we reduce to a natural generalization of ATSP which we call Hop-ATSP (\Cref{sec:overview_hop_atsp}).
Using directed low diameter decompositions, we then show that we may assume that our instances of Hop-ATSP have a very particular structure (\Cref{sec:overview_hierarchy}).
Next, we show that for these structured instances, which we call hierarchically ordered instances, the problem of Hop-ATSP can be reduced to a covering problem (\Cref{sec:overview_covering}).
In order to solve the covering problem, e.g.\ via the set cover greedy algorithm, we then need an approximation algorithm for finding a path in an acyclic digraph minimizing a particular objective function.
In \Cref{sec:overview_dp} we describe how to achieve this by a dynamic program in quasi-polynomial time.

\subsection{Reducing to Hop-ATSP}\label{sec:overview_hop_atsp}

Recall that in Asymmetric TSP, the task is to compute a tour $T$, i.e.\ a Hamiltonian cycle on $V$, that minimizes the cost $c(T)$ of the walk.
By the triangle inequality, every closed walk visiting all vertices in $V$ can be short-cutted to a Hamiltonian cycle on $V$ without increasing the cost.
Thus, we will also allow such walks as solutions for ATSP and we will also call them \emph{tours}.

If a tour $T$ visits the vertices $v_1, v_2, \dots, v_{r+1} = v_1$ in this order, then $c(T) \coloneqq \sum_{i=1}^{r} c(v_i, v_{i+1})$.
We consider a generalization of ATSP, where we are given a hop-distance $k\in \bN$ and the task is to compute a tour $T$ minimizing the $k$-hop cost
\begin{equation}
\label{eq:khop_objective}
c^{(k)} (T) \coloneqq \sum_{\Delta=1}^k \sum_{i=1}^{r} c(v_i, v_{i+\Delta}),
\end{equation}
where $v_{r+l} \coloneqq v_l$ for all $l\in \bN$.
In other words, instead of paying only for the distance of every vertex to its successor on the tour, we pay the distance of every vertex to the $k$ vertices succeeding it on the tour.

\begin{restatable}[Hop-ATSP]{definition}{defHopATSP}\label{def:hop_atsp}
An instance of Hop-ATSP consists of
\begin{itemize}
    \item a finite vertex set $V$,
    \item a cost function $c: V \times V \to \bR_{\geq 0}$ satisfying the triangle inequality, and
    \item a hop-distance $k\in \bN$ with $k < |V|$.
\end{itemize}
The task in Hop-ATSP is to compute a tour $T$ on $V$ minimizing $c^{(k)}(T)$.
\end{restatable}
We call an instance $\cI=(V,c,k)$ of Hop-ATSP \emph{well-scaled} if $c$ takes only values in $\{0,1,\dots, 2n^3\}$ and $\OPT(V,c,k) \geq  n^2$, where  $n := |V|$.
In \Cref{sec:reducing_hop_atsp}, we provide a reduction of Asymmetric A Priori TSP to well-scaled instances of Hop-ATSP.

\begin{restatable}{theorem}{thmHopATSP}\label{thm:reducing_hop_atsp}
	Let $\alpha : \bN \to \bR_{\geq 1}$ and $t: \bN \to \bN$ monotonically increasing.
	Suppose we have an algorithm that computes for every well-scaled instance $\cI=(V,c,k)$ of Hop-ATSP in time $t(n)$ a tour $T$ with $c^{(k)}(T) \leq \alpha(n) \cdot \OPT(\cI)$, where $n=|V|$.
	
	Then there is an algorithm that computes for every instance $\cJ=(V,c, p)$ of the Asymmetric A Priori TSP in time $n\cdot t(\poly(n)) + \poly(n)$ a tour $T$ with 
	\[
	\bE_{A\sim p}\big[ c(T[A]) \big] \ \le\ O(\alpha(5n^2)) \cdot \OPT(\cJ),
	\]
	where $n=|V|$ and $\poly(n)$ denotes some polynomial in $n$.
\end{restatable}
In particular, if $\alpha(n) = \log^q n$ for some constant $q$, then we achieve an approximation ratio of $O(\log^q n)$ for the Asymmetric A~Priori TSP.
We remark that throughout the paper, we did not make any efforts to optimize constants.

In order to prove \Cref{thm:reducing_hop_atsp}, we first proceed analogously to the symmetric case \cite{blauth_et_al} to show that we may assume that all vertices $v$ have the same activation probability $p(v) = \delta$.
Then we observe that for this special case of uniform activation probabilities, the objective function of Asymmetric A~Priori TSP,
\[
\expectationcustom{A \sim p}{ c\left(T[A] \right)} = \sum_{\Delta=1}^{r - 1} \sum_{i=1}^r \delta^2(1-\delta)^{\Delta - 1} \cdot c(v_i,v_{i+\Delta}),
\]
and the objective function of Hop-ATSP \eqref{eq:khop_objective} scaled by a factor $\frac{1}{k^2}$ differ by at most a constant factor when choosing $k = \lfloor \frac{1}{\delta} \rfloor$.
For details we refer to \Cref{sec:reducing_hop_atsp}.

\subsection{Reducing to Hierarchically Ordered Instances}\label{sec:overview_hierarchy}

We now focus on approximation algorithms for Hop-ATSP.
We show that, at the cost of a poly-logarithmic factor in the approximation ratio, we may assume that the instance has a very particular structure and we will call such structured instances \emph{hierarchically ordered}.

In a hierarchically ordered instance we have a total order $\prec$ on the vertices.
Moreover, the instance has some depth $L \in \bN$ and we have a partition $\cH_{\ell}$ of $V$  for each level $\ell \in [L]$, where
\begin{itemize}
\item $\cH_1 = \{ V\}$ is the trivial partition and $\cH_L = \{ \{v\} : v\in V\}$ is the partition of $V$ into singletons,
\item every partition $\cH_{\ell}$  (with $\ell \geq 2$) is a refinement of the partition $\cH_{\ell -1}$, i.e.\,for every set $H \in \cH_{\ell}$ there exists a set $H' \in \cH_{\ell -1}$ with $H \subseteq H'$, and
\item every set $H\in \cH_{\ell}$ for some $\ell \in [L]$ is  a set of vertices appearing consecutively in the total order $\prec$.
\end{itemize}
See \Cref{fig:hierarchical_partition} for an illustration.

\begin{figure}
\begin{center}
\begin{tikzpicture}[
xscale=0.6,
vertex/.style={circle,draw=black, fill,inner sep=1.5pt, outer sep=1pt},
edge/.style={-latex,very thick},
box/.style={very thick, fill=Gray, fill opacity=0.2, rounded corners=8pt},
]

\definecolor{color1}{named}{Peach}  
\definecolor{color2}{named}{ForestGreen}  
\definecolor{color3}{named}{Violet}  
\definecolor{color4}{named}{ProcessBlue} 

\begin{scope}[every node/.style={vertex}]
\foreach \i in {1,...,20} {
\node (v\i) at ({(\i-1)*1.1} ,0) {};
}
\end{scope}

\begin{scope}[box/.append style={draw=color2, fill opacity=0.1}]   
\def\slack{0.6}
\draw[box] ($(v1) - ( \slack,\slack)$) rectangle ($(v20)+(\slack,\slack)$);
\end{scope}

\begin{scope}[box/.append style={draw=color1, fill opacity=0.1}]   
\def\slack{0.475}
\draw[box] ($(v1) - ( \slack,\slack)$) rectangle ($(v6)+(\slack,\slack)$);
\draw[box] ($(v7) - ( \slack,\slack)$) rectangle ($(v12)+(\slack,\slack)$);
\draw[box] ($(v13) - ( \slack,\slack)$) rectangle ($(v20)+(\slack,\slack)$);
\end{scope}

\begin{scope}[box/.append style={draw=color3, fill opacity=0.15}]   
\def\slack{0.35}
\draw[box] ($(v1) - ( \slack,\slack)$) rectangle ($(v4)+(\slack,\slack)$);
\draw[box] ($(v5) - ( \slack,\slack)$) rectangle ($(v6)+(\slack,\slack)$);
\draw[box] ($(v7) - ( \slack,\slack)$) rectangle ($(v12)+(\slack,\slack)$);
\draw[box] ($(v13) - ( \slack,\slack)$) rectangle ($(v15)+(\slack,\slack)$);
\draw[box] ($(v16) - ( \slack,\slack)$) rectangle ($(v17)+(\slack,\slack)$);
\draw[box] ($(v18) - ( \slack,\slack)$) rectangle ($(v20)+(\slack,\slack)$);
\end{scope}

\draw[edge, color2] (v10) to[bend left=20] node[midway, above] {$\level(e_1) = 1$} (v18);
\draw[edge, color1] (v5) to[bend left=56]  node[midway, below] {$\level(e_2) = 2$} (v2);
\draw[edge, color3] (v8) to[bend right=40] node[midway, below] {$\level(e_3) = 3$} (v12);

\node[color=color2] at (22.5,0.6) {$\cH_1$};
\node[color=color1] at (22.5,0) {$\cH_2$};
\node[color=color3] at (22.5,-0.6) {$\cH_3$};
\end{tikzpicture}
 \end{center}
\caption{\label{fig:hierarchical_partition}
 An example of a hierarchical partition $\cH = \left( \cH_{\ell} \right)_{\ell=1}^{4}$, where $\cH_4 = \{\{v\} : v \in V\}$ is not shown.
 All vertices are drawn from left to right according to the total order $\prec$.
 The edges $e_1$ and $e_3$ are forward edges, while $e_2$ is a backward edge.
}
\end{figure}
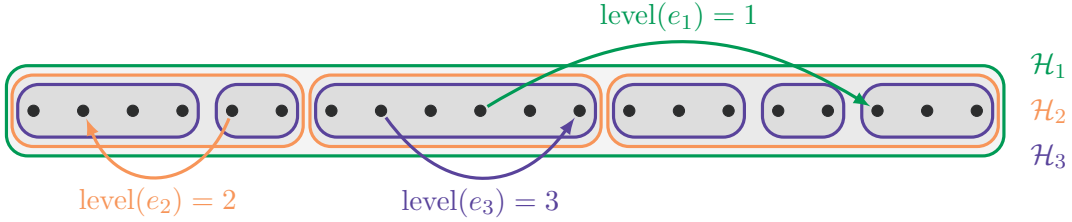

For an edge $e=(v,w)$ we define $\level(e)$ to be the maximal $\ell \in [L]$ such that both endpoints of $e$ belong to the same element of the partition $\cH_{\ell}$.
We call $e$ a \emph{forward edge} if $v \prec w$ and a \emph{backward edge} otherwise.
In a hierarchically ordered instance, we have a maximum edge cost $D_{\ell}$ for each level $\ell \in [L]$, with the property that $D_{\ell} \geq 2\cdot D_{\ell +1}$ for each $\ell\in[L-1]$.
The cost function $c$ satisfies
\begin{itemize}
\item $c(e) \leq D_{\level(e)}$  if $e$ is a forward edge, and
\item $c(e) = D_{\level(e)}$  if $e$ is a backward edge.
\end{itemize}

We will denote a hierarchically ordered instance as a tuple $(V, L, \cH, D, \prec, c,k)$ with the hierarchical partition $\cH=\big(\cH_{\ell}\big)_{\ell=1}^L$ and the maximum edge costs $D=\big(D_{\ell})_{\ell=1}^{L}$.
We provide the following reduction of general instances of Hop-ATSP to hierarchically ordered instances:

\begin{restatable}{theorem}{thmHierarchy}\label{thm:hierarchically_ordered}
Given a well-scaled instance $\cI =(V,c,k)$ of Hop-ATSP with $n=|V|$ vertices, we can in polynomial time sample a hierarchically ordered instance $\cJ =(V, L, \cH, D, \prec, \tilde{c}, k)$ such that
\begin{itemize}
 \item $\tilde{c}(v,w) \geq c(v,w)$ for all $v,w \in V$,
\item $\mathbb{E}[\OPT(\cJ)] \leq L \cdot O(\log n \cdot \log \log n) \cdot \OPT(\cI)$, and
\item $L = O(\log n)$.
\end{itemize}
\end{restatable}

\Cref{thm:hierarchically_ordered} implies that for any tour $T$ satisfying $\tilde{c}^{(k)}(T) \leq \alpha \cdot \OPT(\cJ)$ for some approximation ratio $\alpha \geq 1$, we have $c^{(k)}(T) \leq  L \cdot O(\alpha  \cdot \log n \cdot \log \log n )\cdot \OPT(\cI)$ in expectation.
Therefore, in order to find a poly-logarithmic approximation for Hop-ATSP, it suffices to consider hierarchically ordered instances with $L=O(\log n)$.

In order to prove \Cref{thm:hierarchically_ordered}, we repeatedly apply directed low-diameter decompositions.
See \Cref{sec:hierarchically_ordered} for details.

\subsection{Reducing to a Covering Problem}\label{sec:overview_covering}

Next, we explain how we reduce hierarchically ordered instances of Hop-ATSP to a covering problem.
To explain the main ideas, we first focus on the special case $L=2$.
In this case, we have only two partitions $\cH_1 = \{V\}$ and $\cH_2 =\{ \{v\} : v\in V\}$.
Every forward edge has cost at most $D_1$ and every backward edge has cost exactly $D_1$.
We provide a reduction to a covering problem, where we have to cover the vertices in~$V$ by \emph{monotone paths}.
\begin{definition}[Monotone path]
A path $P$ is called \emph{monotone} if it contains only forward edges.
\end{definition}

To define our covering problem, we extend the definition of the $k$-hop costs $c^{(k)}$ from tours to paths. 
For a path $P$ with vertices $v_1,\dots, v_r$ visited in this order, we define
\[
c^{(k)}(P) = \sum_{\Delta=1}^{k} \sum_{i=1}^{r}  c(v_i, v_{i+\Delta}),
\]
where we define $c(v_i,v_l) \coloneqq 0$ for $l > r$.
Then we consider the covering problem of finding a set $\cP$ of monotone paths with $V \subseteq \bigcup_{P\in \cP} V(P)$  minimizing the weight
\[
w(\cP)\ \coloneqq\ \sum_{P\in \cP} w(P),
\]
where $w(P)\coloneqq c^{(k)}(P) + k^2 \cdot D_1$.
The definition of the weight $w(P)$ is chosen to ensure the following key properties:
\begin{enumerate}
\item \label{item:set_cover_L2_easy_direction} 
any solution $\cP$ to this covering problem can be converted into a tour $T$ with  $c^{(k)}(T) \leq w(\cP)$, and
\item  \label{item:set_cover_L2_difficult_direction} 
any tour $T$ can be converted into a feasible covering $\cP$ satisfying $w(\cP) \leq \gamma \cdot c^{(k)}(T)$, where $\gamma \geq 1$ is a constant independent of the Hop-ATSP instance.
\end{enumerate}
In order to prove~\ref{item:set_cover_L2_easy_direction}, we construct a tour $T$ by concatenating the paths from $\cP$ in an arbitrary order.
Observe that whenever we concatenate two paths $P_1$ and $P_2$ to a new path $P$, we have 
\[
c^{(k)}(P) \leq c^{(k)}(P_1) + c^{(k)}(P_2) + k^2 \cdot D_1,
\] 
because $c(v,w) \leq D_1$ for any $v,w \in V$.
See \Cref{fig:outline_cost_merging_paths}.
This implies $w(P) \leq w(P_1) + w(P_2)$. 
Iteratively applying this observation, we obtain a single path $\hat P$ visiting all vertices with $w(\hat P) \leq \sum_{P\in \cP} w(P)$.
If we turn $\hat P$ into a tour $T$ by returning to the first vertex of the path $\hat P$ once we reached its last vertex, we get  $c^{(k)}(T) \leq c^{(k)}(\hat P) + k^2 \cdot D_1 = w(\hat P)$. 
This shows \ref{item:set_cover_L2_easy_direction}.

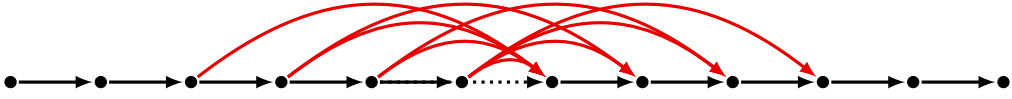
\begin{figure}
\begin{center}

\begin{tikzpicture}[
xscale=1.2,
yscale=2,
vertex/.style={circle,draw=black, fill,inner sep=1.5pt, outer sep=1pt},
edge/.style={-latex,very thick}
]

\begin{scope}[every node/.style={vertex}]
\foreach \i in {1,...,6}{
   \node (v\i) at (\i,0) {};
   \pgfmathtruncatemacro{\x}{\i+6}  
   \node (w\i) at (\x,0) {};
}
\end{scope}

\foreach \i in {1,...,5}{
   \pgfmathtruncatemacro{\next}{\i+1}  
   \draw[edge] (v\i) to (v\next);
   \draw[edge] (w\i) to (w\next);
}
\draw[edge, dotted] (v5) to (w1);

\foreach \i/\j in {3/1,4/1,4/2,5/1,5/3,6/1,6/2,6/3,6/4}{
   \draw[edge, bend left=25, red!90!black] (v\i) to (w\j);
}

\end{tikzpicture} \end{center}
\caption{
If we concatenate the two paths $P_1$ and $P_2$ (black, solid) by adding the dotted edge to obtain a path $P$, the $k$-hop cost $c^{(k)}(P)$ of the new path $P$ is the sum of $c^{(k)}(P_1) + c^{(k)}(P_2)$ and the cost $c(e)$ of all edges $e$ shown in red, where in this example $k=4$.
}
\label{fig:outline_cost_merging_paths}
\end{figure}

The more difficult property to prove is \ref{item:set_cover_L2_difficult_direction}.
We start by some simple observations.
Using $k < |V|$, one can show $c^{(k)}(T) \geq \Omega(k^2) \cdot D_1$ for any tour $T$ (see \Cref{lem:first_level_fixed_cost_lb}).
Moreover, we may assume that the tour $T$ visits every vertex exactly once (see \Cref{lem:khop_skipping_vertices}).
Then we turn the tour $T$ into a path $P$ by removing an arbitrary edge and observe that $w(P) \leq c^{(k)}(T) + k^2 \cdot D_1 \leq O(1) \cdot c^{(k)} (T)$.

The key part of the proof of \ref{item:set_cover_L2_difficult_direction} is to show that we can transform $P$ into a collection of monotone paths covering all vertices without increasing the total weight by more than a constant factor.
First, we will ensure that we have at least $k$ vertices between any two backward edges of $P$.
To this end, we partition the vertices of $P$ into intervals of consecutive vertices containing $k$ vertices each (where the last interval may contain fewer vertices).
Then we change the order of the vertices in each interval by sorting them according to the total order $\prec$.
Using that $c$ satisfies the triangle inequality, one can show that this increases $c^{(k)}(P)$ by at most a constant factor.
After this reordering, backward edges can only appear at the boundaries of the sorted intervals and thus we indeed have at least $k$ vertices between any two backward edges of $P$.

Next, we explain how we get rid of the remaining backward edges.
We will maintain a collection $\cP$ of paths, where in each path we have at least $k$ vertices between any two backward edges.
Now consider a backward edge $e$ contained in a path $P \in \cP$.
For simplicity, assume $k$ to be even and let $v_1, \dots, v_{\sfrac{k}{2}} $ be the $\sfrac{k}{2}$ vertices that $P$ visits before $e$ and let $w_1,\dots , w_{\sfrac{k}{2}}$ be the $\sfrac{k}{2}$ vertices that $P$ visits after $e$ in this order.\footnote{We assume here for simplicity, that $P$ indeed has $\frac{k}{2}$ vertices after the backward edge.}
See \Cref{fig:outline_case distinction_backward_edge}.
When considering the backward edge $e$, we will make sure to charge any increase of $w(\cP)$ to the $k$-hop cost $c^{(k)}(P^e)$ of the subpath $P^{e}$ of $P$ with vertices $V^e = \{v_1, \dots, v_{\sfrac{k}{2}}, w_1, \dots, w_{\sfrac{k}{2}}\}$.
This ensures that the weight increase for removing different backward edges is charged against disjoint subpaths of $P$.
Note that because $P^{e}$ has only $k$ vertices, $c^{(k)}(P^e)$ is the sum of all costs $c(x,y)$ where $x \in V^e$ is visited before $y\in V^e$ by the path $P^e$.

\begin{figure}
\begin{center}

\begin{tikzpicture}[
xscale=0.7,
vertex/.style={circle,draw=black, fill,inner sep=1.5pt, outer sep=1pt},
edge/.style={-latex,very thick}
]

\begin{scope}[every node/.style={vertex}]
\node (v1) at (6,1) {};
\node (v2) at (7,1) {};
\node (v3) at (8,1) {};
\node (v4) at (9,1) {};
\node (v5) at (10,1) {};
\node (w1) at (1,0) {};
\node (w2) at (2,0) {};
\node (w3) at (3,0) {};
\node (w4) at (4,0) {};
\node (w5) at (5,0) {};
\end{scope}

\foreach \i in {1,...,5} {
    \node[above] at (v\i) {$v_{\i}$};
    \node[below] at (w\i) {$w_{\i}$};
}
\foreach \i in {1,...,4}{
   \pgfmathtruncatemacro{\next}{\i+1}  
   \draw[edge] (v\i) to (v\next);
   \draw[edge] (w\i) to (w\next);
}
\draw[edge, densely dashed, red!90!black, out=-150,in=30, looseness=0.5] (v5) to (w1);
\node[red!90!black] (e) at (8.5,0.4) {$e$};

\begin{scope}[shift={(12,0)}]
\begin{scope}[every node/.style={vertex}]
\node (v1) at (3,1) {};
\node (v2) at (5,1) {};
\node (v3) at (8,1) {};
\node (v4) at (9,1) {};
\node (v5) at (10,1) {};
\node (w1) at (1,0) {};
\node (w2) at (2,0) {};
\node (w3) at (4,0) {};
\node (w4) at (6,0) {};
\node (w5) at (7,0) {};
\end{scope}

\foreach \i in {1,...,5} {
    \node[above] at (v\i) {$v_{\i}$};
    \node[below] at (w\i) {$w_{\i}$};
}
\foreach \i in {1,...,4}{
   \pgfmathtruncatemacro{\next}{\i+1}  
   \draw[edge] (v\i) to (v\next);
   \draw[edge] (w\i) to (w\next);
}
\draw[edge, densely dashed, red!90!black, out=-150,in=30, looseness=0.5] (v5) to (w1);
\node[red!90!black] (e) at (8.5,0.4) {$e$};
\end{scope}

\end{tikzpicture} \end{center}
\caption{Illustration of the path $P^e$ for a backward edge $e$. 
Vertices are drawn from left to right according to the total order $\prec$.
We distinguish the two cases $w_{\sfrac{k}{2}} \prec v_1$ (left) and $v_1 \prec w_{\sfrac{k}{2}}$ (right).}
\label{fig:outline_case distinction_backward_edge}
\end{figure}
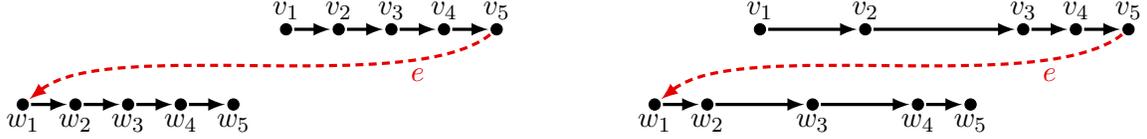

Observe that $v_1 \prec v_2 \prec \dots v_{\sfrac{k}{2}}$ and $w_1 \prec \dots \prec w_{\sfrac{k}{2}}$.
We distinguish two cases.
If $w_{\sfrac{k}{2}} \prec v_{1}$, then we split $P$ at the backward edge $e$ into two paths.
Because in this case every edge $v_i \prec w_j$ with $i,j \in \big[\frac{k}{2}\big]$ is a backward edge, $c^{(k)}(P^{e})$ is at least $\big(\frac{k}{2}\big)^2 \cdot D_1$.
Hence, we can indeed charge the weight increase of $k^2 \cdot D_1$ caused by the splitting of $P$ into two paths against  $c^{(k)}(P^{e})$.

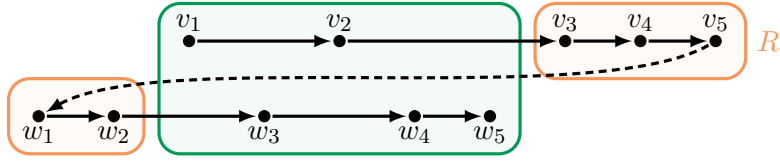
\begin{figure}[t]
\begin{center}

\begin{tikzpicture}[
vertex/.style={circle,draw=black, fill,inner sep=1.5pt, outer sep=1pt},
edge/.style={-latex,very thick}
]

\begin{scope}[every node/.style={vertex}]
\node (v1) at (3,1) {};
\node (v2) at (5,1) {};
\node (v3) at (8,1) {};
\node (v4) at (9,1) {};
\node (v5) at (10,1) {};
\node (w1) at (1,0) {};
\node (w2) at (2,0) {};
\node (w3) at (4,0) {};
\node (w4) at (6,0) {};
\node (w5) at (7,0) {};
\end{scope}

\def\xslack{0.4}	
\def\yslack{0.5}	
\draw[ForestGreen, very thick, fill=ForestGreen, fill opacity=0.05, rounded corners=8pt] ($(3,0) - (\xslack, \yslack)$) rectangle ($(7,1) + (\xslack, \yslack)$);
\draw[Peach, very thick, fill=Peach, fill opacity=0.05, rounded corners=8pt] ($(8,1) - (\xslack, \yslack)$) rectangle ($(10,1) + (\xslack, \yslack)$);
\draw[Peach, very thick, fill=Peach, fill opacity=0.05, rounded corners=8pt] ($(1,0) - (\xslack, \yslack)$) rectangle ($(2,0) + (\xslack, \yslack)$);

\node[Peach] (R) at (10.7,1) {$R$};

\foreach \i in {1,...,5} {
    \node[above] at (v\i) {$v_{\i}$};
    \node[below] at (w\i) {$w_{\i}$};
}
\foreach \i in {1,...,4}{
   \pgfmathtruncatemacro{\next}{\i+1}  
   \draw[edge] (v\i) to (v\next);
   \draw[edge] (w\i) to (w\next);
}
\draw[edge, densely dashed, out=-150,in=30, looseness=0.5] (v5) to (w1);

\end{tikzpicture} \end{center}
\caption{Illustration of the case $v_1 \prec w_{\sfrac{k}{2}}$, where the vertices are drawn from left to right according to the total order $\prec$.
We remove the vertices in $R$ (orange) and reorder the remaining vertices of $P^e$ (shown in green) according to $\prec$, i.e., in this example in the order $v_1, w_3, v_2, w_4, w_5$.
}
\label{fig:outline_second_case_backward_edge}
\end{figure}

Now consider the remaining case, where we have $v_1 \prec w_{\sfrac{k}{2}}$, illustrated in \Cref{fig:outline_second_case_backward_edge}.
Then we remove all vertices $v_i$ with $i\in \big[ \frac{k}{2} \big]$ and $w_{\sfrac{k}{2}} \prec v_i$ from the path and put them into a remainder set $R$, which is initially empty.
Similarly, we remove all vertices $w_i$ with $w_i \prec v_1$ from the path and add them to $R$.
Note that $v_1$ and $w_{\sfrac{k}{2}}$ were not removed.
Now we sort all vertices from $V^e$ that were not removed according to the total order $\prec$.
Because these vertices $u\in V^e$ that were not removed satisfy $v_1 \preceq u \preceq w_{\sfrac{k}{2}}$, this reordering does not introduce any new backward edges.

We will charge the increase of $w(\cP) + |R| \cdot 2k \cdot D_1$ against $c^{(k)}(P^{e})$ and will later show that we can re-insert every vertex from $R$ into some path while increasing the cost of the path by at most $2k \cdot D_1$.
To see that we can indeed charge the increase of  $w(\cP) + |R| \cdot 2k \cdot D_1$ against $c^{(k)}(P^{e})$, we observe that every vertex in $R$ has at least $\frac{k}{2}$ incident edges of cost $D_1$ that contribute to $c^{(k)}(P^{e})$.
Moreover, using that $c$ satisfies the triangle inequality, one can show that removing the vertices in $R$ does not increase the $k$-hop costs $c^{(k)}(P)$ of our path and the reordering of the vertices in $V^e\setminus R$ can increase the cost by at most some constant factor times $c^{(k)}(P^{e})$.

It remains to show how to deal with the removed vertices in $R$.
After having considered each backward edge $e$, all paths in $\cP$ are monotone.
We insert each vertex from $R$ into an arbitrary path $P\in \cP$, maintaining monotonicity.
This insertion increases $c^{(k)}(P)$ by at most $2k \cdot D_1$.
Doing this for all vertices in $R$, we obtain a collection $\cP'$ of monotone paths covering all vertices with $w(\cP' )\leq w(\cP) + 2k\cdot D_1 \cdot |R|$.
This shows \ref{item:set_cover_L2_difficult_direction}.

Recall that so far we only considered the special case $L=2$.
In the general case, we consider a more involved covering problem, taking into account the hierarchical structure of the instance.
We will again aim at covering $V$ by monotone paths, but now we will in addition assign a level to every such path we include in our covering.

\begin{definition}[Path-level pair]
A \emph{path-level pair} is a pair $(P,\ell)$ where $P$ is a monotone path, $\ell\in [L]$, and $V(P) \subseteq H$ for some $H\in \cH_{\ell}$.
\end{definition}

A feasible solution to our covering problem will now consist of such path-level pairs.
In order to be able to turn a solution to our covering problem into a solution to Hop-ATSP without increasing the objective value significantly (analogous to~\ref{item:set_cover_L2_easy_direction} in the case $L=2$), we define our covering problem as follows.

\begin{restatable}[Path Covering]{problem}{CoveringProblem}\label{problem:covering}
Given a hierarchically ordered instance $\cI =(V, L, \cH, D, \prec, c, k)$ of Hop-ATSP, 
find a set $\cP$ of path-level pairs $(P,\ell)$ such that
\begin{enumerate}[label=(\alph*)]
\item\label{item:pseudo_cover_condition} for every vertex $v\in V$ there is some pair $(P,\ell)\in \cP$ with $v\in V(P)$, and
\item\label{item:overview_extra_covering_condition}
 for every level $\ell \in \{2,3,\dots, L\}$ and every set $H\in \cH_{\ell}$, one of the following applies:
\begin{enumerate}[label=(\roman*)]
\item\label{item:direct_covering} for every vertex $v\in H$ there is a pair $(P,j)\in \cP$ with $j < \ell$ and $v\in V(P)$, or
\item\label{item:responsibility} there is a pair $(P,j)\in \cP$ with $j < \ell$ and  $|V(P)\cap H|\geq k$,
\end{enumerate}
\end{enumerate}
while minimizing the weight
\[
 w(\cP) \coloneqq \sum_{(P,\ell) \in \cP} w(P,\ell) = \sum_{(P,\ell) \in \cP} \Big( c^{(k)}(P) + k^2 \cdot D_{\ell} \Big).
\]
\end{restatable}

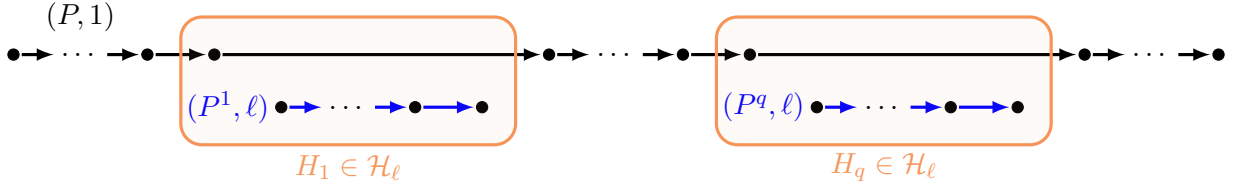
\begin{figure}
\begin{center}

\begin{tikzpicture}[
xscale=0.89,
vertex/.style={circle,draw=black, fill,inner sep=1.5pt, outer sep=1pt},
edge/.style={-latex,very thick}
]
\def\xslack{0.5}	
\def\yslack{0.5}	

\node at (1,0.5) {$(P, 1)$};
\begin{scope}[every node/.style={vertex}]
\node (v0) at (0,0) {};
\node (v2) at (2,0) {};
\node (v3) at (3,0) {};
\node (v8) at (8,0) {};
\node (v10) at (10,0) {};
\node (v11) at (11,0) {};
\node (v16) at (16,0) {};
\node (v18) at (18,0) {};
\end{scope}
\draw[edge] (v0) to (0.6,0);
\node at (1,0) {$\dots$};
\draw[edge] (1.4,0) to (v2);
\draw[edge] (v2) to (v3);
\draw[edge] (v3) to (v8);
\draw[edge] (v8) to (8.6,0);
\node at (9,0) {$\dots$};
\draw[edge] (9.4,0) to (v10);
\draw[edge] (v10) to (v11);
\draw[edge] (v11) to (v16);
\draw[edge] (v16) to (16.6,0);
\node at (17,0) {$\dots$};
\draw[edge] (17.4,0) to (v18);

\begin{scope}[shift={(2,-0.7)}]
\begin{scope}[every node/.style={vertex}]
\node (u2) at (2,0) {};
\node (u4) at (4,0) {};
\node (u5) at (5,0) {};
\end{scope}
\draw[edge,blue] (u2) to (2.6,0);
\node at (3,0) {$\dots$};
\draw[edge,blue] (3.4,0) to (u4);
\draw[edge,blue] (u4) to (u5);
\node[blue] at (1.2,0) {$(P^1, \ell)$};
\draw[Peach, very thick, fill=Peach, fill opacity=0.05, rounded corners=8pt] ($(1,0) - (\xslack, \yslack)$) rectangle ($(5,0.7) + (\xslack, \yslack)$);
\node[Peach] at (3,-0.8) {$H_1\in \cH_{\ell}$};
\end{scope}

\begin{scope}[shift={(10,-0.7)}]
\begin{scope}[every node/.style={vertex}]
\node (u2) at (2,0) {};
\node (u4) at (4,0) {};
\node (u5) at (5,0) {};
\end{scope}
\draw[edge,blue] (u2) to (2.6,0);
\node at (3,0) {$\dots$};
\draw[edge,blue] (3.4,0) to (u4);
\draw[edge,blue] (u4) to (u5);
\node[blue] at (1.2,0) {$(P^q, \ell)$};
\draw[Peach, very thick, fill=Peach, fill opacity=0.05, rounded corners=8pt] ($(1,0) - (\xslack, \yslack)$) rectangle ($(5,0.7) + (\xslack, \yslack)$);
\node[Peach] at (3,-0.8) {$H_q\in \cH_{\ell}$};
\end{scope}

\end{tikzpicture} \end{center}
\caption{\label{fig:example_covering_needs_complicated_condition}
Example illustrating why we need condition \ref{item:overview_extra_covering_condition} in \Cref{problem:covering}.
Suppose that all edges in $\delta(H_i)$ with $i\in[q]$ have cost $D_1$, and all other forward edges have cost $0$.
Consider the path-level pairs $(P,1)$ and $(P^1, \ell), \dots, (P^q,\ell)$ shown in the picture, where each path $P^i$ has $k$ vertices. The path $P$ visits exactly one vertex from each set $H_i$ and it visits $k$ vertices in between any two such vertices.
Then $c^{(k)}(P)=2qk \cdot D_1$ and $c^{(k)}(P^i)=0$ for $i\in[q]$, and thus $w(P,1) + \sum_{i=1}^{q} w(P^i,\ell) = k^2 \cdot D_1 + 2qk \cdot D_1 + q  k^2 \cdot D_{\ell}$.
However, any tour $T$ covering all vertices has $k$-hop cost $c^{(k)}(T) \geq q \cdot k^2 \cdot D_1$, which is much larger than $w(P,1) + \sum_{i=1}^{q} w(P^i,\ell)$ in case $D_1$ is much larger than $D_{\ell}$ and $q, k$ are sufficiently large.
}
\end{figure}

 We highlight that property~\ref{item:overview_extra_covering_condition} is necessary in order to be able to transform $\cP$ into a tour $T$ with $c^{(k)}(T) = O(1) \cdot w(\cP)$, as the example in \Cref{fig:example_covering_needs_complicated_condition} shows.
 
In \Cref{sec:reducing_to_covering} we show that any $\alpha$-approximation algorithm for \Cref{problem:covering} leads to an $O(L^3 \cdot \alpha)$-approximation algorithm for hierarchically ordered instances of Hop-ATSP, with a polynomial overhead in the running time.
The main ideas we use to prove this are similar to the special case $L=2$, but we need a much more careful recursive argument.
In particular, we cannot afford to increase the objective value $L$~times by a constant factor $\gamma$ when transforming a tour into a solution to \Cref{problem:covering} and some care is needed to avoid this.
We refer to \Cref{sec:reducing_to_covering} for details.

\subsection{Solving the Covering Problem via an Algorithm for Monotone Paths}\label{sec:overview_dp}

It remains to provide an approximation algorithm for our covering problem.
For simplicity, in this outline we again focus on the special case $L=2$.
Recall that in our covering problem, the task is to find a set $\cP$ of monotone paths with $V \subseteq \bigcup_{P\in \cP} V(P)$  minimizing the weight $w(\cP) = \sum_{P\in \cP} w(P)$, where $w(P)= c^{(k)}(P) + k^2 \cdot D_1$.
We will use the greedy algorithm for set cover.
It provides an $O(\alpha \cdot \log n)$-approximation for our covering problem, assuming that we have an $\alpha$-approximation algorithm for the following problem of finding a monotone path minimizing a particular objective function.
Given a hierarchically ordered instance of Hop-ATSP (with $L=2$) and a set $S\subsetneq V$ of already covered vertices,
the task is to find a monotone path $P$ minimizing
\begin{equation}\label{eq:outline_actual_objective_function}
\frac{w(P)}{|V(P)\setminus S|} \ =\ \frac{c^{(k)}(P) + k^2 \cdot D_1}{|V(P)\setminus S|},
\end{equation}
which we define to be $\infty$ for $V(P) \setminus S = \emptyset$.

In order to find such a good path, we will use a dynamic program.
A key challenge here is that a naive realization of this approach would lead to a running time of $\Omega(n^k)$.
In order to avoid this, we will not directly try to minimize the objective function~\eqref{eq:outline_actual_objective_function}, but instead consider some approximation of it.
More precisely, for each offset $q\in \{0,1,\dots,k\}$ and every monotone path $P$, we define a weight bound
$w^q(P) \geq w(P)$.
These weight bounds will have the following key properties:
\begin{enumerate}
\item\label{item:overview_weight_bound_upper_bound} for every monotone path $P$ and every offset $q\in \{0,1,\dots,k\}$, we have $w^q(P) \geq w(P)$,
\item\label{item:overview_property_weight_bound_good_for_some_offset}  for every monotone path $P$, there exists an offset $q\in\{0,1,\dots,k\}$ such that $w^q(P) = O(\log(k)) \cdot w(P)$, and
\item\label{item:overview_property_weight_bound_efficient}  we can find a monotone path $P$ and an offset $q$ minimizing $\frac{w^q(P)}{|V(P)\setminus S|}$ in quasi-polynomial time.
\end{enumerate}
From these properties, it follows immediately that we have a quasi-polynomial-time $O(\log k)$-approximation algorithm for finding a monotone path minimizing~\eqref{eq:outline_actual_objective_function}, and thus a quasi-polynomial-time $O(\log k \cdot \log n)$-approximation algorithm for our covering problem.

Let us now explain how we define the weight bound $w^q(P)$.
See \Cref{fig:weight_bound} for an illustration.
For simplicity of the explanation, we will assume $k= {2^{\Gamma} -1}$ for some $\Gamma\in \bN$.
We number the vertices of $P$ as $v_0,\dots, v_r$ and assign a height to every vertex $v_i\in V(P)$.
We define $\height(i)$ to be the maximum $h \in \{0,1\dots, \Gamma \}$ such that $i$ is divisible by $2^h$.
For each vertex $v_i \in V(P)$ the distance from the $2^{\height(i+q)}-1$ vertices preceding $v$ on $P$ and the distance to the $2^{\height(i+q)}-1$ vertices succeeding $v$ on $P$ will contribute to the weight bound $w^q(P)$.
More precisely,
\[
w^q(P) \ \coloneqq\ k^2 \cdot D_{1} + \sum_{i=0}^r  \ \sum_{\Delta = 1}^{2^{\height(i+q)}-1} 2^{\height(i+q)} \cdot \Big( c(v_{i-\Delta}, v_i) + c(v_i, v_{i+\Delta}) \Big)
\]

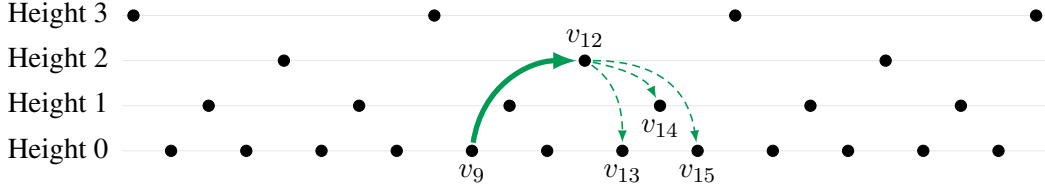
\begin{figure}
\begin{center}
\begin{tikzpicture}[
vertex/.style={circle,draw=black, fill,inner sep=1.5pt, outer sep=1pt},
edge/.style={-latex},
xscale=0.5,
yscale=0.6
]
\definecolor{color1}{named}{ForestGreen}  
\definecolor{color2}{named}{Violet}

\def\slack{0.3}
\begin{scope}
\foreach \h in {0,1,2,3} {
  \pgfmathtruncatemacro{\hc}{\h}
  \ifcase\hc
    \def\hlabel{Height 0}
  \or
    \def\hlabel{Height 1}
  \or
    \def\hlabel{Height 2}
  \or
    \def\hlabel{Height 3}
  \fi
  \draw[thin, opacity=0.1] (0-\slack, \h) -- (24+\slack, \h);
  \node[left] at (0-\slack-0.1, \h) {{\hlabel}};
}

\foreach \i in {0,...,24} {
\pgfmathsetmacro{\height}{
    ifthenelse(mod(\i,8)==0,3,
    ifthenelse(mod(\i,4)==0,2,
      ifthenelse(mod(\i,2)==0,1,0)
  )
  )
}
\node[vertex] (v\i) at ({(\i)*1},\height) {};
}
\end{scope}
\draw[edge, out=80, in=180, looseness=1, line width=2pt, color1] (v9) to (v12);

\draw[edge,out=-30, in=90, semithick, densely dashed, color1] (v12) to (v13);
\draw[edge,out=-10, in=120, looseness=0.8, semithick, densely dashed, color1] (v12) to (v14);
\draw[edge,out=0, in=100, looseness=1.1, semithick, densely dashed, color1] (v12) to (v15);
\node[above=1.5pt] at (v12) {$v_{12}$};
\node[below=1.5pt] at (v9) {$v_{9}$};
\node[below=1.5pt] at (v13) {$v_{13}$};
\node[below=1.5pt] at (v14) {$v_{14}$};
\node[below=1.5pt] at (v15) {$v_{15}$};
\end{tikzpicture}
 \end{center}
\caption{\label{fig:weight_bound}
An example illustrating the weight bound $w^q(P)$.
In this example, we have $q=0$, $\Gamma = 3$, and $k=2^{3} - 1 = 7$.
The vertices $v_0, \dots, v_{24}$ are drawn from left to right, and each vertex $v_i$ is vertically positioned according to $\height(i)$.
Suppose that $P$ visits these vertices in this order.
To prove that the weight bound is an upper bound on $w(P)$, we use the triangle inequality,
bounding the cost $c(v_i,v_j)$ of an edge contributing to $c^{(k)}(P)$ by $c(v_i,v_m) + c(v_m,v_j)$, where we choose $m\in \{i,\dots,j\}$ maximizing $\height(m)$.
For example,the edges $(v_9,v_j)$ for $j \in \{12,\dots,15\}$ contribute to $c^{(k)}(P)$.
Our weight bound $w^0(P)$ bounds the contribution of these edges to $c^{(k)}(P)$ via the triangle inequality: $c(v_9,v_j) \leq c(v_9,v_{12}) + c(v_{12},v_{j})$ for $j \in \{12,\dots,15\}$.
Note that $c(v_9,v_{12})$ indeed contributes to $w^0(P)$ with a factor of $2^{\height(12)} = 4$.
}
\end{figure}

Using that $c$ satisfies the triangle inequality, one can show \ref{item:overview_weight_bound_upper_bound}.
In order to prove~\ref{item:overview_property_weight_bound_good_for_some_offset}, we show that choosing $q\in\{0,1,\dots,k\}$ uniformly at random leads to $\bE \big[ w^q(P) \big] = O(\Gamma) \cdot w(P) = O(\log(k))\cdot w(P)$.
To prove~\ref{item:overview_property_weight_bound_efficient}, we use a dynamic program.
We highlight that this dynamic program is the only part of our algorithm that does not have a polynomial runtime.
For details of how we prove these properties and how we extend these arguments to the general case of arbitrary $L$, we refer to \Cref{sec:monotone_paths}.

\section{Lower Bound on the Adaptivity Gap}\label{sec:adaptivity_gap}

In this section we prove \Cref{thm:lower_bound}, i.e., we prove a polynomial lower bound on the adaptivity gap of the Asymmetric A Priori TSP:

\begin{definition}[Adaptivity gap]
\label{def:adaptivity_gap}
The \emph{adaptivity gap} of an Asymmetric A Priori TSP instance $\mathcal{I} = (V,c,p)$ with $\expectationAP{{\rm TSP}(A,c)} > 0$ is
\begin{equation}
\label{eq:adaptivity_gap}
\frac{\OPT(\mathcal{I})}{\expectationAP{{\rm TSP}(A,c)}}.
\end{equation} 
For $n \in \bN$ we define $\adaptivitygap(n)$ to be the supremum of this ratio \eqref{eq:adaptivity_gap} taken over all instances $(V,c,p)$ with $\expectationAP{{\rm TSP}(A,c)} > 0$ and $|V| = n$.
\end{definition}
We  prove that the adaptivity gap $\adaptivitygap(n)$ is at least $\Omega(n^{\sfrac{1}{4}} \cdot \log^{-1} n)$.
To this end, we construct a family of Asymmetric A Priori TSP instances, which we call the \emph{asymmetric grid instances}.

\subsection{The Asymmetric Grid Instances}

The asymmetric grid instances arise as the metric closure of a particular family of directed graphs with unit weight edges, whose construction we present now.
See \Cref{fig:asymm_grid_graph} for an illustration.
We define the graph $G_k = \left( V_k, E_k \right)$ for $k \in \bZ_{\geq 2}$ as follows:
$V_k$ consists of $k^2$ distinct vertices $\parset{v_{i,j} : i,j \in \parset{1, \dots , k}}$.
As the set of edges we choose
\begin{equation*}
 E_k\ \coloneqq\ \Big\{ (v_{i,j} ,\ v_{i^\prime,j^\prime})\  :\  j + 1 = j^\prime \text{ or } (j = k \text{ and } j^\prime= 1) \Big\}
\end{equation*}
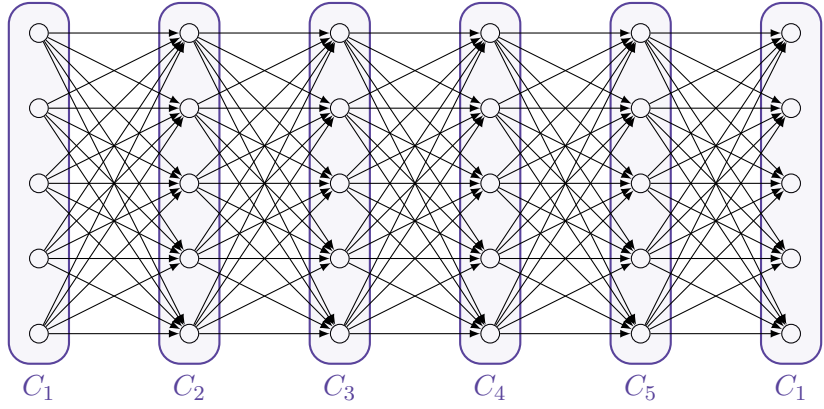
\begin{figure}
\begin{center}
\begin{tikzpicture}[
vertex/.style={circle,draw=black,inner sep=2.5pt},
edge/.style={-latex},
box/.style={Violet, thick, fill=Violet, fill opacity=0.05, rounded corners=8pt},
scale=1
]
\def\k{5}

\begin{scope}[every node/.style={vertex}]

\foreach \i in {1,...,6} {
\foreach \j in {1,...,\k}{
\node (x\i y\j) at (2*\i, \j) {};
}
}      
\end{scope}

\def\slack{0.4}

\foreach \i in {1,...,\k} {
\draw[box] (2*\i - \slack,1 - \slack) rectangle (2*\i + \slack,5 + \slack);
\node at (x\i y1)[below=12pt] {\textcolor{Violet}{$C_{\i}$}};
}

\draw[box] (2*6 - \slack,1 - \slack) rectangle (2*6 + \slack,5 + \slack);
\node at (x6y1)[below=12pt] {\textcolor{Violet}{$C_{1}$}};

\foreach \i in {1,...,\k} {
\foreach \jone in {1,...,\k} {
\foreach \jtwo in {1,...,\k} {

\pgfmathtruncatemacro{\itwo}{\i + 1}
\draw[edge] (x\i y\jone) -- (x\itwo y\jtwo);

}
}
}
\end{tikzpicture}
 \end{center}
\caption{
The graph $G_k$ for $k=5$. Each column $C_j$ contains $k$ distinct vertices. 
Between consecutive columns $C_j$ and $C_{j+1}$, there is a complete bipartite graph with edges directed towards $C_{j+1}$. 
Note that in this figure, the vertices of $C_1$ are depicted twice for clarity. 
To form the actual graph $G_5$, the two representations of each vertex in $C_1$ should be identified as a single vertex.
}
\label{fig:asymm_grid_graph}
\end{figure}

Let $c_k: V_k \times V_k \to \bZ_{\geq 0}$ denote the shortest-path metric in $G_k$ with unit-weight edges.
Then, we obtain an instance $(V_k, c_k)$ of ATSP.
We introduce some notation that makes it easier for us to reference certain subsets of $V_k$.
A subset of the form $C_j = \parset{v_{i,j} : i \in [k]}$ is called a \emph{column} for $j \in [j]$.
Analogously,
subsets $R_i = \parset{v_{i,j} : j \in [k]}$ are called \emph{rows}.
To obtain instances of Asymmetric A Priori TSP, we define activation probabilities $p_k: V_k \to (0,1]$ by $p_k(v) \coloneqq 1/k$ for all $v \in V_k$.
\addparskip

We will prove that for $k\geq 5$, the expected cost of an optimal a priori tour satisfies
\begin{equation}\label{eq:adaptivity_lb_lb}
\OPT\left(V_k, c_k, p_k \right) \geq \frac{k^{3/2}}{2^9e^4},
\end{equation}
and the expected cost of an optimal a posteriori tour satisfies
\begin{equation}\label{eq:adaptivity_lb_ub}
\expectationcustom{A \sim p_k}{{\rm TSP}(A,c_k)} \leq 5 \cdot k \log k.
\end{equation}

\subsection{Lower Bound for an Optimal A Priori Tour}

We first prove~\eqref{eq:adaptivity_lb_lb}, i.e., we prove a lower bound for the optimal a priori tour.
Let $k \geq 5$ and let 
$T^\ast$ be a Hamiltonian cycle on $V_k$ and let us number the vertices as $x_1 , \dots, x_{k^2}$ according to the order in which they are visited by $T^\ast$.
We write $E^{(k)} \coloneqq V_k \times V_k$ for the edge set of the complete directed graph on $V_k$, and $\delta^+(X) \coloneqq X \times V_k$ and $\delta^-(X) \coloneqq V_k \times X$ for any $X \subseteq V_k$.

We have
\begin{equation}\label{eq:grid_expression_we_want_to_lower_bound}
\expectationcustom{A \sim p_k}{c_k(T^\ast[A]) } = \sum_{e \in E^{(k)}} \probabilitycustom{A \sim p_k}{e \in E(T^\ast[A])} \cdot c_k(e),
\end{equation}
where, for $e=(x_i, x_j) \in E^{(k)}$,
\begin{equation*}
\probabilitycustom{A \sim p_k}{e \in E(T^\ast[A])} = \begin{cases} (1/k)^2 \cdot (1 - 1/k)^{j-i - 1} \quad&\text{if $i < j$,} \\ (1/k)^2 \cdot (1 - 1/k)^{k^2 + j-i - 1} \quad&\text{if $i > j$.} \end{cases}
\end{equation*}
We analyze the contribution of $\delta^+ (C_j) $ to $\expectationcustom{A \sim p_k}{c_k(T^\ast[A])}$, that is,
\begin{equation*}
\sum_{e \in \delta^+(C_j)} \probabilitycustom{A \sim p_k}{e \in E(T^\ast[A])} \cdot c_k(e)
\end{equation*}
separately for all $j \in [k]$. 

Observe that $k \geq 5$ implies $4k < k^2$.
Let $j \in [k]$ be fixed and define 
\[
S(x_i)\ \coloneqq\ \big\{ x_{\ell} : \ell = i + 1, \dots , i + 4k \big\}
\]
 for all $i \in \parset{1, \dots , k^2}$.
In other words, the set $S(x_i)$ contains the $4k$ vertices following $x_i$ in the tour $T^\ast$.
For $x_{\ell} \in S(x_i)$, we have
\begin{equation}
\label{eq:adaptivity_lb_lb_eq1}
\begin{split}
\probabilitycustom{A \sim p_k}{(x_i,x_{\ell}) \in E(T^\ast[A])} &\geq (1/k)^2 \cdot (1-1/k)^{4k} \\
&\geq (1/k)^2 \cdot \frac{1}{(2e)^4}.
\end{split}
\end{equation}
In our proof, we will focus on lower bounding the contribution of edges $(x_i,x_{\ell})$ with $x_{\ell}\in S(x_i)$ to \eqref{eq:grid_expression_we_want_to_lower_bound}.

\begin{definition}[Saturated interval]
We say that $I \subseteq V_k$ is an \emph{interval} of $T^\ast$ if it is the vertex set of a subpath of $T^\ast$.
Moreover, we call $I$ \emph{saturated} if it satisfies both of the following conditions:
\begin{itemize}
\item
it contains at most $4k$ vertices, and

\item
it contains at least $\sqrt{k}$ vertices from $C_j$.
\end{itemize}
\end{definition}

For an interval $I$ of $T^\ast$, we denote by $E_j[I] \coloneqq \delta^+(C_j) \cap (I \times I)$ the set of edges in $\delta^+(C_j)$ with both endpoints in $I$. 
We prove a lower bound on the contribution of these edges to the overall cost $\expectationcustom{A \sim p_k}{c_k(T^\ast[A])}$ in case the interval $I$ is saturated.
We use that each of theses edges has cost $k$ with respect to the cost function $c_k$ (because both endpoints are in the same column), but the edges also have a significant probability of being part of $T^\ast[A]$ (because both endpoints are in the interval $I$ with at most $4k$ vertices).

\begin{lemma}\label{lem:adaptivity_lb_lb_eq2}
Let $I$ be a saturated interval of $T^\ast$. Then
\begin{equation*}
\sum_{e \in E_j[I]} \probabilitycustom{A \sim p_k}{e \in E(T^\ast[A])} \cdot c_k(e) \ \geq\ \left( |I \cap C_j| \right)^2 \cdot \frac{1}{2^7e^4k}.
\end{equation*}
\end{lemma}
\begin{proof}
Let $x_i,x_{\ell} \in I \cap C_j$ with $x_i \neq x_{\ell}$.
Because $x_i, x_{\ell} \in I$, we have $x_i \in S(x_{\ell})$ or $x_{\ell} \in S(x_i)$.
Without loss of generality, assume that $x_{\ell} \in S(x_i)$.
Since $x_i \neq x_{\ell}$ and $x_i,x_{\ell} \in C_j$, we have $c_k(x_i, x_{\ell}) = k$.
There are $\binom{|I \cap C_j|}{2}$ such pairs $x_i,x_{\ell´}$.
Thus, we have
\begin{equation*}
\begin{split}
\sum_{e \in E_j[I]} \probabilitycustom{A \sim p_k}{e \in E(T^\ast[A])} \cdot c_k(e) \overset{\eqref{eq:adaptivity_lb_lb_eq1}}&{\geq} \sum_{e \in E_j[I]} (1/k)^2 \cdot \frac{1}{(2e)^4} \cdot c_k (e) \\ 
& \geq \binom{|I \cap C_j|}{2} \cdot (1/k)^2  \cdot \frac{1}{(2e)^4} \cdot k \\
&\geq \left( |I \cap C_j| -1 \right)^2 \cdot \frac{1}{2^5e^4k} \\
&\geq \left( |I \cap C_j| \right)^2 \cdot \frac{1}{2^7e^4k},
\end{split}
\end{equation*}
using that $|I \cap C_j| \geq \sqrt{k} \geq 2$ for the last inequality.
\end{proof}

\begin{lemma}
\label{lem:adaptivity_lb_lb}
For $k \geq 5$, 
\begin{equation*}
\OPT\left(V_k, c_k, p_k \right) \geq \frac{k^{3/2}}{2^9e^4}.
\end{equation*}
\end{lemma}
\begin{proof}
Let $I_1, \dots, I_m$ be a family of pairwise disjoint saturated intervals of $T^\ast$ that maximizes $\sum_{i = 1}^m |I_i|$.
We say that $v \in V_k$ is \emph{marked} if $v \in \bigcup_{i=1}^m I_i$.
Otherwise $v$ is \emph{unmarked}.
Observe that $m \leq \sqrt{k}$ since each interval contains at least $\sqrt{k}$ vertices from $C_j$ and $|C_j| = k$.
We distinguish the following two cases:

\paragraph*{Case 1:}
Suppose that at least $k/2$ vertices of $C_j$ are marked.
In other words, we have $\sum_{i = 1}^m |I_i \cap C_j| \geq k/2$.
Using the Arithmetic-Quadratic-Mean inequality and the fact that the $I_i$ are pairwise disjoint, we obtain
\begin{equation*}
\begin{split}
\sum_{e \in \delta^+(C_j) } \probabilitycustom{A \sim p_k}{e \in E(T^\ast[A])} \cdot c_k(e) &\geq
\sum_{i = 1}^m \sum_{e \in E_j[I_i]} \probabilitycustom{A \sim p_k}{e \in E(T^\ast[A])} \cdot c_k(e) \\
\overset{Lem.~\ref{lem:adaptivity_lb_lb_eq2}}&{\geq} \frac{1}{2^7e^4k} \cdot \sum_{i = 1}^m \left( |I_i \cap C_j| \right)^2 \\
&\geq \frac{1}{2^7e^4k} \cdot \frac{\left( \sum_{i = 1}^m |I_i \cap C_j|\right)^2}{m} \\
&\geq \frac{1}{2^9e^4k} \cdot \frac{k^2}{\sqrt{k}} \\
&= \frac{\sqrt{k}}{2^9e^4}.
\end{split}
\end{equation*}

\paragraph*{Case 2:}
Suppose that at least $k/2$ vertices of $C_j$ are unmarked.
We define the \emph{congestion} of a vertex $x_{\ell}$ as  
\[
\congestion (x_{\ell}) \coloneqq |\parset{x_i \in C_j : \text{$x_i$ unmarked and $x_{\ell} \in S(x_i) $}}|.
\]
We claim that  our choice of $I_1, \dots, I_m$ implies $\congestion(x_{\ell}) \leq \sqrt{k}$ for all $\ell \in [k^2]$. 

Suppose for the sake of deriving a contradiction that $\congestion(x_{\ell}) > \sqrt{k}$.
Without loss of generality, after renumbering we can assume that $\ell = k^2$, i.e, $x_{\ell}$ is the last vertex visited by the tour $T^\ast$.
Let 
\begin{align*}
i_1\ \coloneqq&\  \min \parset{i \in [k^2] : \text{$x_{i} \in C_j$ unmarked and $x_{\ell} \in S(x_{i})$}} \\
i_2\ \coloneqq&\ \max \parset{i \in [k^2] : \text{$x_{i} \in C_j$ unmarked and $x_{\ell} \in S(x_{i})$}}.
\end{align*}
The assumption that $\congestion(x_{\ell}) > \sqrt{k}$ implies that the interval $I^\ast$ given by the subpath $x_{i_1}, x_{i_1 + 1} , \dots , x_{i_2} $ of $T^\ast$ is saturated ($|I^\ast| \leq 4k$ since $x_{i_2} \in S( x_{i_1})$).
Every interval $I_z$ for some $z\in [m]$ satisfies $I_z \cap I^\ast = \emptyset$ or $I_i \subseteq I^\ast$, because $x_{i_1}$ and $x_{i_2}$ are unmarked.
Hence, removing all saturated intervals of $I_1, \dots , I_m$ that intersect $I^\ast$ and adding the interval $I^\ast$ instead still yields a family of pairwise disjoint saturated intervals.
This contradicts the choice of our family maximizing $\sum_{j = 1}^m |I_j|$ since $I^\ast$ contains unmarked vertices.
We conclude that, indeed,  $\congestion(x_\ell) \leq \sqrt{k}$ for all $\ell \in [k^2]$.

We consider the set
 \[
 Z \coloneqq \parset{ v \in V_k : c_k(x,v) \geq \sqrt{k} - 1 \text{ for all } x \in C_j}
 \]
 of vertices that are expensive to reach from column $C_j$.
Note that 
\begin{equation*}
|V_k \setminus Z| \leq \left|\bigcup_{i = 0}^{\floor*{\sqrt{k}} - 1} C_{j + i} \right| \leq \sqrt{k} \cdot k,
\end{equation*}
where $C_{z} \coloneqq C_{z - k}$ for $z > k$.
We consider the edges from an unmarked vertex $x_i$ in column $C_j$ to a vertex in~$S(x_i)$:
\begin{equation*}
\widetilde{E} \coloneqq \bigcup_{x_i \in C_j} \big\{(x_i,x_{\ell}) : \text{$x_i$ unmarked and $x_{\ell} \in S(x_i)$}\big\} \subseteq \delta^+(C_j).
\end{equation*}
By definition of $S(x_i)$ and our assumption that at least $k/2$ vertices of $C_j$ are unmarked, we have $|\widetilde{E}| \geq 4k \cdot k/2 = 2k^2$.
Because all vertices have congestion at most $\sqrt{k}$, at most $|V_k \setminus Z| \cdot \sqrt{k} \leq k^2$ edges of $\widetilde{E}$ end in a vertex from $V_k \setminus Z$.
This shows that $|\widetilde{E} \cap \delta^- (Z)| \geq k^2$.

We conclude
\begin{equation*}
\begin{split}
\sum_{e \in \delta^+(C_j)} \probabilitycustom{A \sim p_k}{e \in E(T^\ast[A])} \cdot c_k(e) &\geq \sum_{e \in \widetilde{E}} \probabilitycustom{A \sim p_k}{e \in E(T^\ast[A])} \cdot c_k (e) \\
\overset{\eqref{eq:adaptivity_lb_lb_eq1}}&{\geq}
 \sum_{e \in \widetilde{E}} (1/k)^2 \cdot \frac{1}{(2e)^4} \cdot c_k(e) \\
&\geq \sum_{e \in \widetilde{E} \cap \delta^- (Z)} (1/k)^2 \cdot \frac{1}{(2e)^4} \cdot c_k(e) \\
&\geq \sum_{e \in \widetilde{E} \cap \delta^- (Z)} (1/k)^2 \cdot \frac{1}{(2e)^4} \cdot (\sqrt{k}-1)\\
&\geq \sum_{e \in \widetilde{E} \cap \delta^- (Z)} (1/k)^2 \cdot \frac{1}{2^5e^4} \cdot \sqrt{k}\\
&\geq k^2 \cdot (1/k)^2 \cdot \frac{1}{2^5e^4} \cdot \sqrt{k} \\
&= \frac{\sqrt{k}}{2^5e^4}.
\end{split}
\end{equation*}

\addparskip

We have shown that in both cases
\begin{equation*}
\sum_{e \in \delta^+(C_j)} \probabilitycustom{A \sim p_k}{e \in E(T^\ast[A])} \cdot c_k(e) \geq \frac{\sqrt{k}}{2^9e^4}.
\end{equation*}
This completes our proof, because
\begin{equation*}
\begin{split}
\expectationcustom{A \sim p_k}{c_k(T^\ast[A]) } 
&= \sum_{e \in E^{(k)}} \probabilitycustom{A \sim p_k}{e \in E(T^\ast[A])} \cdot c_k(e) \\
&= \sum_{j=1}^k \sum_{e \in \delta^+(C_j)} \probabilitycustom{A \sim p_k}{e \in E(T^\ast[A])} \cdot c_k(e) \\
&\geq \frac{k^{3/2}}{2^9e^4}. 
\end{split}
\end{equation*}
\end{proof}

The tour shown in \Cref{fig:asymp_opt_apriori_for_grid} shows that the result of \Cref{lem:adaptivity_lb_lb} is asymptotically tight.
For a detailed proof of the tightness, see \Cref{appendix:grid}.
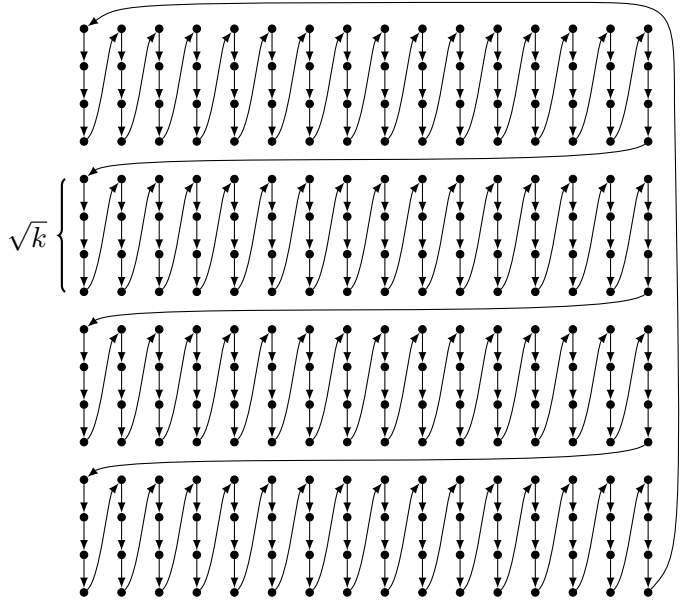
\begin{figure}
\begin{center}
\begin{tikzpicture}[
vertex/.style={circle,draw=black, fill, inner sep=1pt},
edge/.style={-latex},
box/.style={Peach, very thick, fill=Violet, fill opacity=0.05, rounded corners=8pt},
scale=0.5
]
\def\k{16}
\def\sqrtk{4}
\pgfmathsetmacro{\kkm}{16*16 - 1}

\begin{scope}[every node/.style={vertex}]

\foreach \i in {1,...,\k} {
\foreach \j in {1,...,\k}{
\node (x\i y\j) at (\j,16 - \i + 1) {};
}
}

\end{scope}

\newcommand{\getindex}[1]{
\pgfmathsetmacro{\num}{#1}

\pgfmathsetmacro{\t}{mod(floor((\num - 1) / \sqrtk), \k) + 1}
\pgfmathsetmacro{\l}{mod(\num - 1, \sqrtk) + 1}
\pgfmathsetmacro{\j}{floor((\num - 1) / (\k * \sqrtk))}

\pgfmathsetmacro{\first}{\j * \sqrtk + \l}
\pgfmathsetmacro{\second}{\t}

\pgfmathtruncatemacro{\firstindex}{\first}
\pgfmathtruncatemacro{\secondindex}{\second}
}

\foreach \i in {1,...,\kkm} {
\ifnum\i=64\else        
\ifnum\i=128\else        
\ifnum\i=192\else        

\getindex{\i}
\pgfmathsetmacro{\prevx}{\firstindex}
\pgfmathsetmacro{\prevy}{\secondindex}

\getindex{\i +1}
\pgfmathsetmacro{\nextx}{\firstindex}
\pgfmathsetmacro{\nexty}{\secondindex}

\ifnum\numexpr\i-4*(\i/4)\relax=0
\draw[edge, out=40, in=-130, looseness=0.4] (x\prevx y\prevy) to (x\nextx y\nexty);
\else
\draw[edge] (x\prevx y\prevy) to (x\nextx y\nexty);
\fi\fi\fi\fi
}

\draw[edge] (x4y16) .. controls (15,12.2) and (2,12.8) .. (x5y1);  
\draw[edge] (x8y16) .. controls (15,8.2) and (2,8.8) .. (x9y1);  
\draw[edge] (x12y16) .. controls (15,4.2) and (2,4.8) .. (x13y1);  

\draw[] (x16y16) .. controls (16.9,2) .. (16.7,14);
\draw[] (16.7,14) .. controls (16.7,16.7) .. (14,16.7);
\draw[edge] (14,16.7) .. controls (2,16.7) .. (x1y1);

\draw[decorate, decoration={brace}, thick] (0.5,9) -- (0.5,12) node[midway,left] {$\sqrt{k}$ \ };

\def\xslack{0.4}	
\def\yslack{0.5}	

\end{tikzpicture}
 \end{center}
\vspace*{2mm}
\caption{
An example for $k=16$.
An asymptotically optimal a priori tour $T$ for $(V_k,c_k,p_k)$ is obtained by grouping the rows into blocks of size $\sqrt{k}$ and traversing each block in a column-wise manner.
}
\label{fig:asymp_opt_apriori_for_grid}
\end{figure}

\subsection{Upper Bound on the Expected Cost of an Optimal A Posteriori Tour}

We now prove~\eqref{eq:adaptivity_lb_ub}.
We can easily characterize the cost ${\rm TSP}(A, c_k)$ of an optimal tour on any set $A \subseteq V_k$:

\begin{lemma}
\label{lem:OPTATSP_for_grid}
For any $A \subseteq V_k$, we have 
\begin{equation*}
{\rm TSP}(A, c_k) = k \cdot \max \parset{|A \cap C_j| : j=1, \dots , k}.
\end{equation*}
\end{lemma}
\begin{proof}
Because any two vertices in the same column have distance $k$ from each other (in either direction), we have
${\rm TSP}(A\cap C_j, c_k)  \geq k \cdot |C_j|$ for each $j\in [k]$.
Thus, ${\rm TSP}(A, c_k) \geq {\rm TSP}(A\cap C_j, c_k) \geq k \cdot |C_j|$, 
implying that ${\rm TSP}(A, c_k)$ is at least as large as claimed.

For the upper bound, suppose that at most $t$ vertices per column are active.
After relabeling all vertices, we can assume that without loss of generality all vertices of $A$ are contained in the first $t$ rows.
We can visit all vertices in these rows by a tour of length $t\cdot k$ by traversing one row after the other, always visiting first the vertex in $C_1$, then in $C_2$, and so on until $C_k$.
\end{proof}

Using \Cref{lem:OPTATSP_for_grid}, we obtain the following upper bound on the expected cost of an optimal a posteriori tour:

\begin{lemma}
\label{lem:adaptivity_lb_ub}
For $k \geq 3$, we have 
\begin{equation*}
\expectationcustom{A \sim p_k}{{\rm TSP}(A,c_k)} \leq 5 \cdot k \log k.
\end{equation*}
\end{lemma}
\begin{proof}
Let $X_{i,j}$ be the random indicator variable for the event that $v \in A$ (i.e., $X_{i,j}$ is a Bernoulli random variable with parameter $1/k$).
We denote by $\widehat{X}_j \coloneqq \sum_{i = 1}^k X_{i,j}$ the number of active vertices for each column $j$.
Note that $\expectationcustom{A \sim p_k}{\widehat{X}_j} = 1$.
By a Chernoff-bound,
\begin{equation*}
\begin{split}
\probabilitycustom{A \sim p_k}{\widehat{X}_j \geq 1+ (2\log k - 1)} 
\leq e^{- \frac{(2\log k - 1)}{3}}
\leq \frac{e}{k^2}.
\end{split}
\end{equation*}
We denote $Y \coloneqq\displaystyle \max_{j = 1, \dots, k} \widehat{X}_j$.
Then, by a union bound, we obtain $\probabilitycustom{A \sim p_k}{Y  \geq 2\log k} \leq e/k$.
Thus, we conclude
\begin{equation*}
\begin{split} 
\expectationcustom{A \sim p_k}{Y} &= \probabilitycustom{A \sim p_k}{Y  \geq 2\log k} \cdot \expectationcustom{A \sim p_k}{Y \mid Y  \geq 2\log k } \\ & \quad + \probabilitycustom{A \sim p_k}{Y  < 2\log k} \cdot \expectationcustom{A \sim p_k}{Y \mid Y  < 2\log k } \\ 
& \leq \frac{e}{k} \cdot k + 1 \cdot (2 \log k) \\
&\leq 5 \log k.
\end{split}
\end{equation*}
Finally, by \Cref{lem:OPTATSP_for_grid}, we have
\begin{equation*}
\expectationcustom{A \sim p_k}{{\rm TSP}(A,c_k)} \leq \expectationcustom{A \sim p_k}{k \cdot Y} \leq 5k\log k.
\qedhere
\end{equation*}
\end{proof}

\subsection{Lower Bound on the Adaptivity Gap}

We can now conclude the claimed lower bound for $\adaptivitygap(n)$:

\LowerBoundAdaptivity*
\begin{proof}
Let $n \in \bN, n \geq 25$.
For $k^\prime \coloneqq \floor*{\sqrt{n}}$ and $n^\prime \coloneqq(k^\prime)^2$ we have $n \geq n^\prime$.
Observe that $\adaptivitygap$ is monotonically increasing since we can simply add vertices with activation probability $0$, which changes neither the numerator nor the denominator of the ratio in~\eqref{eq:adaptivity_gap}.\footnote{If one does not allow vertices with activation probability zero, it is still possible to add vertices with arbitrarily small activation probability as the adaptivity gap is defined as a supremum.}
Therefore, it suffices to show that $\adaptivitygap(n^\prime) \geq c \cdot \frac{n^{1/4}}{\log n}$ for some constant $c \in \bR$.

Note that $|V_k| = (k^\prime)^2 = n^\prime$ and $k^\prime \geq 5$.
Applying \Cref{lem:adaptivity_lb_lb} and \Cref{lem:adaptivity_lb_ub} to $(V_{k^\prime}, c_{k^\prime}, p_{k^\prime})$
we obtain
\begin{equation*}
\adaptivitygap(n^\prime) \geq \frac{\OPT(V_{k^\prime}, c_{k^\prime}, p_{k^\prime})}{\expectationcustom{A \sim p_{k^\prime}}{{\rm TSP}(A,c_{k^\prime})}}
\geq c \cdot \frac{k^\prime \cdot \sqrt{k^\prime}}{k^\prime \log k^\prime} = c \cdot \frac{\sqrt{k^\prime} }{\log k^\prime}
\end{equation*}
for $c \coloneqq 1/(5\cdot 2^9e^4)$.
We have $\sqrt{n} \geq k^\prime  \geq \sqrt{n} - 1$.
Thus,
\begin{equation*}
c \cdot \frac{\sqrt{k^\prime} }{\log k^\prime}
\geq c \cdot \frac{\sqrt{\sqrt{n} - 1} }{\log \sqrt{n}}
\geq c \cdot \frac{\sqrt{\sqrt{n}/2 }}{\log \sqrt{n}}
\geq \frac{c}{\sqrt{2}} \cdot \frac{n^{1/4}}{\log n}.
\end{equation*}
This concludes the proof.
\end{proof}

\section{Reducing to Hop-ATSP}\label{sec:reducing_hop_atsp}

The goal of this section is to prove \Cref{thm:reducing_hop_atsp}, i.e., to reduce \APrioriATSP to Hop-ATSP.
Along the way, we give a simple polynomial time $O(\sqrt n)$-approximation algorithm for Asymmetric A Priori TSP and an upper bound of $O(\sqrt n)$ on the adaptivity gap for Asymmetric A Priori TSP, as claimed in \Cref{thm:simple_approx}.
We start by recalling the definition of Hop-ATSP and explaining the high-level ideas behind our reduction.

\subsection{Outline of our reduction}\label{sec:outline_hop_atsp}

Before reducing to Hop-ATSP, we will first show that we may assume that all vertices have the same activation probability.
Then, when computing the expected cost of an a priori tour $T$ that visits vertices $v_1,v_2,\dots, v_{r+1} = v_1$ in this order, the contribution of an edge $(v_i, v_{i+\Delta})$ decreases as $\Delta$ grows because the probability that at least one more vertex between $v_i$ and $v_{i+\Delta}$ is active increases. 
Instead of gradually decreasing this contribution, we show that we can also let edges fully contribute to the cost if $\Delta \leq k$ and ignore edges where $\Delta > k$ for some suitable $k\in \mathbb{N}$.
To describe this formally, we use the $k$-hop cost, as introduced in \cref{sec:overview}:

\begin{definition}[$k$-hop cost]
	Given a vertex set $V$, a cost function $c\colon V\times V \to \mathbb R_{\geq 0}$ satisfying the triangle inequality, and a hop distance $k\in \mathbb N$ with $k<|V|$, the $k$-hop cost of a tour $T$ in $V$ visiting the vertices $v_1,v_2,\dots, v_r$ in this order, is
	\[c^{(k)}(T) = \sum_{\Delta = 1}^k\sum_{i = 1}^r c(v_i, v_{i + \Delta}),\]
	where we define $v_i \coloneqq v_{i - r}$ for $i > r$.
    For a non-closed walk $W$ visiting vertices $v_1,v_2,\dots, v_r$ in this order, i.e., where $v_1 \neq v_r$, we also define
    \[c^{(k)}(W) = \sum_{\Delta = 1}^k\sum_{i = 1}^r c(v_i, v_{i + \Delta}),\]
    but now define $c(v_i, v_{\ell}) \coloneqq 0$ for $\ell>r$.
\end{definition}

Then we define Hop-ATSP as follows:

\defHopATSP*

Recall that a tour on $V$ is a closed walk that visits every vertex at least once.
If a tour visits a vertex a second time, we can skip the second visit of this vertex and, as we will show next, this does not increase the $k$-hop cost of the tour.
In particular, every Hop-ATSP instance has an optimal solution that is a Hamiltonian cycle on $V$.

\begin{lemma}
	\label{lem:khop_skipping_vertices}
	Let $(V,c,k)$ be an instance of \HopATSP.
    If a tour $W^\prime$ results from a tour $W$ by omitting a (second) visit of a vertex, we have $c^{(k)}(W^\prime) \leq c^{(k)}(W)$. 
    
	Similarly, if a walk $W^\prime$ results from a non-closed walk $W$ by omitting one visit of a vertex, we have $c^{(k)}(W^\prime) \leq c^{(k)}(W)$. 
\end{lemma}
\begin{proof}
Let $w_1, \dots, w_t$ be the vertices visited by  $W$ in this order (where the vertices do not need to be distinct).
Let $1 \leq i \leq t$ such that $W^\prime$ results from $W$ by omitting the visit $w_i$, that is $W^\prime$ visits the vertices $w_1, \dots, w_{i-1}, w_{i+1}, \dots, w_t$ in this order.
	
	Suppose the removal of the visit $w_i$ introduces a new $k$-hop edge $(w_{i-\ell}, w_{i - \ell + k + 1})$ for some $\ell \in \parset{1, \dots, k}$.
	Then, the edges $(w_{i-\ell}, w_i)$ and $(w_i,w_{i - \ell + k + 1})$ no longer contribute to $c^{(k)}(W^\prime)$.
	Thus, if $W$ is a tour, by the triangle inequality, we have
	\begin{equation*}
		\begin{split}
			c^{(k)}(W) - c^{(k)}(W^\prime)\
			&=\  \sum_{\ell = 1}^k \big( c(w_{i-\ell}, w_i) + c(w_i,w_{i - \ell + k + 1}) - c(w_{i-\ell}, w_{i - \ell + k + 1})  \big)\ 
			\geq\  0,
		\end{split}
	\end{equation*}
	where we define $w_j \coloneqq w_{j - t}$ for $j > t$ and $w_j \coloneqq w_{i + t}$ for $j < 0$.
	Similarly, if $W$ is a non-closed walk, we have
	\begin{equation*}
		\begin{split}
			c^{(k)}(W) - c^{(k)}(W^\prime)\ 
			&=\  \sum_{\ell = \max(1, k + i - t + 1)}^{\min(k,i-1)} \big( c(w_{i-\ell}, w_i) + c(w_i,w_{i - \ell + k + 1}) - c(w_{i-\ell}, w_{i - \ell + k + 1})  \big)\
			\geq\  0.
		\end{split}
	\end{equation*}
\end{proof}

The main goal of this section is to prove \Cref{thm:reducing_hop_atsp}, i.e., to reduce \APrioriATSP to Hop-ATSP.
We will first prove that we may assume that all vertices have the same activation probability $p(v)=\delta$, and then show that for $k\coloneqq \lfloor \frac{1}{\delta}\rfloor$, every Hamiltonian cycle $T$ on $V$ has cost $c(T) = \Theta(\delta^2) \cdot c^{(k)}(T)$
(see \Cref{thm:reduction_to_khop}).

To prove that we may assume that all vertices have the same activation probability $\delta$, we proceed similarly as prior work for the Symmetric A Priori TSP \cite{blauth_et_al}.
The high-level idea is to replace every vertex~$v$ by many co-located copies with activation probability $\delta$, where the number of copies is chosen such that the probability that at least one of the copies is active is approximately the original activation probability $p(v)$.
Note however, that this strategy does not work for vertices of very small activation probability because we would then need to choose $\delta$ very small, which would lead to a super-polynomial number of copies of other vertices (and also to a super-polynomial number $k$).
Thus, we handle vertices of very small activation probability separately.

We show that we can assume that one vertex is always active, i.e.\@ there is a \emph{depot} $d$ with $p(d) = 1$
(\Cref{sec:reduce_depot}), which can be shown following the same argument as described in  \cite{blauth_et_al} for the symmetric special case.
We use this to handle vertices with very low activation probability by connecting each of them directly to the depot. 
In \Cref{sec:uniform_prob}, we then prove that we can assume uniform activation probabilities for the remaining vertices, which allows us to reduce to Hop-ATSP (\Cref{sec:complete_hop_atsp_reduction}). 

Finally, we go one step further and show that it suffices to consider \emph{well-scaled} instances of Hop-ATSP, i.e., we can assume that $c$ takes only integer values, the diameter of the complete graph on $V$ is upper bounded by $2n^3$ and the optimum solution costs at least $n^2$ (where $n:= |V|$).
This can be achieved easily by scaling all edge costs appropriately, rounding them down to the next integer, and taking the metric closure to ensure that $c$ still satisfies the triangle inequality (see \Cref{subsec:well_scaled_instances}).

Along the way, we give a simple polynomial time $O(\sqrt n)$-approximation algorithm for Asymmetric A Priori TSP and prove an upper bound of $O(\sqrt n)$ on the adaptivity gap for Asymmetric A Priori TSP (\Cref{sec:simple_approx}).

\subsection{Reducing to instances with a depot}\label{sec:reduce_depot}

As explained in \Cref{sec:outline_hop_atsp}, it will later be helpful to assume that there exists a vertex $d\in V$ (called the \emph{depot}) that is always active, i.e.\@ $p(d) = 1$. 
The goal of this subsection is to show that, as long as we do not care about constant factors, this assumption is safe to make.
This follows by essentially the same argument as the symmetric special case \cite{blauth_et_al}.

More precisely, we prove: if there is an algorithm that finds an $f(n)$-approximation for instances of Asymmetric A Priori TSP with depot, then there is an algorithm that finds a $6\cdot f(n)$-approximation for general instances of Asymmetric A Priori TSP.

This follows from the next lemma if the total activation probability $p(V)\coloneqq \sum_{v\in V}p(v)$ is larger than some positive constant, e.g.\@ $p(V)\geq \frac12$.
The proof is adapted from Lemma~15 in \cite{blauth_et_al}:
\begin{lemma}
	\label{lem:atsp_reduction_to_depot_approx}
		For all instances $(V,c,p)$ of \APrioriATSP and all tours $S$ on $V$, there exists a  vertex $d\in V$ such that the instance $(V, c, p')$ of Asymmetric A Priori TSP with Depot defined by
	\[
	p'(v) = \begin{cases}
		1  & v = d\\
		p(v) & \text{else,}
	\end{cases}
	\]
	  satisfies
	 \begin{equation}
		\label{eq:atsp_reduction_to_depot0_approx}
		\expectationAP{c\left(T[A]\right)} \\
		\leq \expectationcustom{A \sim p^\prime}{c\left(T[A]\right)}
	\end{equation}
	for all tours $T$ on $V$ and
	\begin{equation}
		\label{eq:atsp_reduction_to_depot1_approx}
		\begin{split}
			\expectationcustom{A \sim p^\prime}{c(S[A])}
			\leq 2\left(1 + \frac{1}{p(V)} \right)\expectationAP{c(S[A])}.
		\end{split}
	\end{equation}
\end{lemma}

If we apply \Cref{lem:atsp_reduction_to_depot_approx} for $S$ being an optimal a priori tour and we assume $p(V)$ to be lower bounded by some positive constant, then it indeed implies the desired reduction because we can guess the depot $d$ by enumerating all possibilities.
In order to prove \Cref{thm:simple_approx} and prove an upper bound on the adaptivity gap, we will also need the following variant of \Cref{lem:atsp_reduction_to_depot_approx}:

\begin{lemma}
	\label{lem:atsp_reduction_to_depot_adaptivity}
	For all instances $(V,c,p)$ of \APrioriATSP, there exists a  vertex $d\in V$ such that the instance $(V, c, p')$ of Asymmetric A Priori TSP with Depot defined by
	\[
	p'(v) = \begin{cases}
		1  & v = d\\
		p(v) & \text{else,}
	\end{cases}
	\]
	  satisfies
	\begin{equation}
		\label{eq:atsp_reduction_to_depot0_adaptivity}
		\expectationAP{c\left(T[A]\right)} \\
		\leq \expectationcustom{A \sim p^\prime}{c\left(T[A]\right)}
	\end{equation}
	for all tours $T$ on $V$ and
	\begin{equation}
		\label{eq:atsp_reduction_to_depot1_adaptivity}
		\begin{split}
			\expectationcustom{A \sim p^\prime}{{\rm TSP}(A, c )}
			\leq 2\left(1 + \frac{1}{p(V)} \right)\expectationAP{{\rm TSP}(A,c)}.
		\end{split}
	\end{equation}
\end{lemma}

\begin{proof}[Proof of \Cref{lem:atsp_reduction_to_depot_approx} and \Cref{lem:atsp_reduction_to_depot_adaptivity}]
	For a fixed $v\in V$, we consider the modified activation probabilities that result from declaring $v$ to be a depot:
	\begin{equation*}
		 p^{v=d}(u) \coloneqq \begin{cases}1 &v=u,\\ p(u) &\text{else}, \end{cases}
	\end{equation*}
	for all $u \in V$.
	For any family of tours $(T_A)_{A\subseteq V}$ where $T_A$ is a tour on $A$ for each $A\subseteq V$, we have
	\begin{equation}
		\label{eq:atsp_reduction_to_depot2}
		\expectationAP{c\left(T_{A \cup \parset{v}}\right)} = \expectationcustom{A \sim p^{v=d}}{c\left(T_{A}\right)}.
	\end{equation}
    (See \cite{blauth_et_al} for a detailed proof.)
	This implies \eqref{eq:atsp_reduction_to_depot0_approx} and \eqref{eq:atsp_reduction_to_depot0_adaptivity}, since by the triangle inequality
	\begin{equation*}
		\expectationAP{c\left( T[A] \right)} 
		\leq \expectationAP{c\left( T[A \cup \parset{v}] \right)}
		\overset{\eqref{eq:atsp_reduction_to_depot2}}{=} \expectationcustom{A \sim p^{v=d}}{c\left( T[A] \right)},
	\end{equation*}
	for any tour $T$ on $V$. 
	Note that this holds independently of the choice of $v$. 
	We will now see that there is also a choice of $d$ satisfying \eqref{eq:atsp_reduction_to_depot1_approx},
	and a choice of $d$ satisfying \eqref{eq:atsp_reduction_to_depot1_adaptivity}:
	
	For $A \subseteq V$ we denote by $T^\ast_A$ an optimal ATSP tour on $A$ with respect to $c$. For $v \in V$ and a tour $T$ on some subset $A\subseteq V$ with $v\in A$ define $D_v(T)$ to be the sum of the cost of the edges entering and leaving $v$ in $T$.
	
	Any tour $T$ on a subset $A \subseteq V \setminus \parset{v}$ can be extended to a tour on $A \cup \parset{v}$ for some $v \in V$ by adding the edges $(x,v), (v,y)$ and $E(T_{[y,x]})$, where $x = \argmin_{x\in A\setminus\parset{v}} c(x,v)$, $y = \argmin_{y\in A\setminus\parset{v}} c(v,y)$ and $T_{[y,x]}$ denotes the $y$-$x$-subpath on $T$.
	The resulting edge set is connected and Eulerian and thus can be cut short to an ATSP tour on $A\cup \parset{v}$.
	Since $c(E(T_{[y,x]})) \leq c(T)$, we get for all $A\subseteq V, v\in V$:
	\begin{equation}
		\label{eq:atsp_reduction_to_depot3}
		c\left(T^\ast_{A \cup \parset{v}}\right) \leq 2c\left(T^\ast_{A}\right) +  D_v(T^*_{A\cup \parset{v}}).
	\end{equation}
	Furthermore, we have
	$
	\sum_{v \in V} \mathds{1}_{v \in A} \cdot D_v(T^*_{A\cup \parset v}) = 2c(T^\ast_A).
	$
	This holds for every choice of $A$, hence we can apply expectation to both sides of this inequality yielding
	\begin{equation*}
		\begin{split}
			\sum_{v \in V} p(v) \cdot \expectationcustom{A \sim p}{D_v(T^*_{A \cup \{v\}})} 
			&= \expectationcustom{A \sim p}{\sum_{v \in V} \mathds{1}_{v \in A} \cdot D_v(T^*_{A\cup \parset v})} \\
			&= 2\expectationcustom{A \sim p}{c(T^\ast_A)},
		\end{split}
	\end{equation*}
	because $\mathds{1}_{v \in A}$ and $D_v(T^*_{A\cup\parset v})$ are independent random variables as the definition of $D_v(T^*_{A\cup \parset v})$ does not depend on whether $A$ contains $v$.
	In particular, this means that there exists $\hat{v} \in V$ such that
	\begin{equation}
		\label{eq:atsp_reduction_to_depot4}
		\expectationcustom{A \sim p}{D_{\hat{v}}(T^*_{A\cup \parset {\hat v}})} \leq \frac{2\expectationcustom{A \sim p}{c(T^\ast_A)}}{p(V)}. 
	\end{equation}
	We will choose $\hat{v}$ as the vertex minimizing this expression.
	This proves \eqref{eq:atsp_reduction_to_depot1_adaptivity} since
		\begin{align*}
						\expectationcustom{A \sim p^{\hat{v}=d}}{{\rm TSP}(A,c)} 
						&=\expectationcustom{A \sim p^{\hat{v}=d}}{c\left( T^\ast_{A} \right)} \\
						&\overset{\eqref{eq:atsp_reduction_to_depot2}}= \expectationcustom{A \sim p}{c\left( T^\ast_{A \cup \parset{\hat{v}}} \right)} \\
						&\overset{\eqref{eq:atsp_reduction_to_depot3}}{\leq} 2\expectationcustom{A \sim p}{c\left(T^\ast_{A}\right) } + \expectationcustom{A \sim p}{D_{\hat{v}}(T^*_{A\cup\parset{\hat v}}) } \\
						&\overset{\eqref{eq:atsp_reduction_to_depot4}}{\leq} 2(1+1/p(V))  \expectationcustom{A \sim p}{c\left(T^\ast_{A}\right) } \\
						&= 2(1+1/p(V))\expectationAP{{\rm TSP}(A,c)}.
		\end{align*}
	For any tour $S$ on $V$ we have for all $A\subseteq V, v\in V$ by the triangle inequality:
	\[c(S[A]) \leq c(S[A\cup\parset{v}]) \leq c(S[A]) + D_v(A) \leq 2c(S[A]) + D_v(A)\]
	as an analogous inequality to \eqref{eq:atsp_reduction_to_depot3}. Using the same arguments as for \eqref{eq:atsp_reduction_to_depot1_adaptivity}, only replacing $T^*_A$ by $S[A]$ (and of course $T^*_{A\cup \parset v}$ by $S[A\cup\parset v]$) thus shows \eqref{eq:atsp_reduction_to_depot1_approx}.
\end{proof}

Thus, if we are given an instance of \APrioriATSP where $p(V)\geq \frac12$, we can guess an optimal depot vertex $d$, increase its activation probability to $1$, solve the resulting instance, and view the solution we obtain as a solution to our original instance.
In this reduction we lose at most a factor of $2(1+\frac 1{p(V)}) \leq 6$ in the approximation ratio.
If the total activation probability however is small, we take advantage of the following lemma:

\begin{lemma}
	\label{lem:apx_for_tiny_activations}
	Let $(V,c,p)$ be an instance of \APrioriATSP.
	If $p(V) \leq \frac{1}{2}$, then any tour on $V$ has expected cost at most $2\expectationAP{{\rm TSP}(A,c)}$.
\end{lemma}
\begin{proof}
	Let $T$ be any tour on $V$.
	If less than $2$ vertices are active, the cost of $T$ short-cutted to the active vertices is zero. Otherwise let $v_1,\dots v_k$ be the active vertices in this order on $T$. By the triangle inequality this tour costs at most $\sum_{i = 1}^k (c(s, v_i) + c(v_i,s))$ for any $s\in V$. This gives an upper bound on the expected cost for each $s\in V$. Taking the average over all these bounds weighted by the activation probability of $s$ yields
	\begin{align*}
		\mathbb E_{A\sim p}[c(T[A])]
		& \leq \sum_{s \in V} \frac{p(s)}{p(V)} \cdot \left(\sum_{v \in V} \probabilitycustom{A \sim p}{|A| \geq 2 \text{ and } v \in A}  \cdot (c(s,v) + c(v,s)) \right) \\
		&= \sum_{s \in V} \frac{p(s)}{p(V)} \cdot \left( \sum_{v \in V} p(v) \cdot \probabilitycustom{A \sim p}{|A| \geq 2 \mid v \in A} \cdot (c(s,v) + c(v,s)) \right) \\
		&\leq  \sum_{s \in V} \sum_{v \in V} p(s)p(v)(c(s,v) + c(v,s)),
	\end{align*}
	where the last inequality follows from the fact that 
	\[
	\probabilitycustom{A \sim p}{|A| \geq 2 \mid v \in A} = \probabilitycustom{A \sim p}{|A \setminus \parset{v}| \geq 1} \leq \expectationAP{|A \setminus \parset{v}|} = p(V \setminus \parset{v})
	\]
	by Markov's inequality.
	
	On the other hand, we can bound
	\begin{equation*}
		\begin{split}
			\expectationAP{{\rm TSP}(A,c)} 
			&\geq \sum_{x,y \in V} \probabilitycustom{A \sim p}{A = \parset{x,y}} \cdot \left( c(x,y) + c(y,x) \right) \\
			&= \sum_{x,y \in V}  p(x)p(y)\cdot \probabilitycustom{A \sim p}{A \subseteq \parset{x,y} } \cdot \left( c(x,y) + c(y,x) \right) \\
			&\geq \sum_{x,y \in V}  p(x)p(y)\cdot \probabilitycustom{A \sim p}{A = \emptyset } \cdot \left( c(x,y) + c(y,x) \right).
		\end{split}
	\end{equation*}
	Now, by Markov's inequality,
	\begin{equation*}
		\probabilitycustom{A \sim p}{A = \emptyset } = 1 - \probabilitycustom{A \sim p}{|A| \geq 1} \geq 1- \expectationAP{|A|} = 1 - p(V) \geq \frac{1}{2}.
	\end{equation*}
	Therefore we conclude
	\begin{equation*}
		\begin{split} 
			\expectationAP{{\rm TSP}(A,c)} 
			&\geq \sum_{x,y \in V}  \frac{1}{2} p(x)p(y) \left( c(x,y) + c(y,x) \right) 
			\geq \frac{1}{2} \cdot
			\mathbb E_{A\sim p}[c(T[A])] .
		\end{split}
	\end{equation*}
\end{proof}

This leads us to the following conclusion:

\begin{corollary}\label{cor:depot}
Let $\alpha: \bN \to \bR_{\geq 1}$.
Suppose there is a polynomial-time algorithm that, given an instance $(V^\prime,c^\prime,p^\prime)$ of Asymmetric A Priori TSP with depot, returns an a priori tour $T^\prime$ on $V^\prime$ that satisfies
\begin{equation}
\label{eq:cor_depot}
\expectationcustom{A^\prime \sim p^\prime}{c^\prime(T^\prime[A^\prime])} \leq \alpha(n^\prime) \cdot \expectationcustom{A^\prime \sim p^\prime}{{\rm TSP}(A^\prime, c^\prime)},
\end{equation}
then there is a polynomial-time algorithm that, given a general instance $(V,c,p)$ of Asymmetric A Priori TSP, returns a tour $T$ on $V$ that satisfies
\begin{equation}
\expectationcustom{A \sim p}{c(T[A])} \leq 6 \cdot \alpha(n) \cdot \expectationcustom{A \sim p}{{\rm TSP}(A, c)},
\end{equation}
where $n^\prime \coloneqq |V^\prime|$ and $n \coloneqq |V|$.
\end{corollary}
\begin{proof}
Let $(V,c,p)$ be a general instance of Asymmetric A Priori TSP.
If $p(V) = \sum_{v \in V} p(v) \leq 1/2$, then, by \Cref{lem:apx_for_tiny_activations}, any tour tour $T$ on $V$ satisfies $\expectationcustom{A \sim p}{c(T[A])} \leq 2\expectationAP{{\rm TSP}(A,c)}$.

Otherwise, for every $v \in V$, we obtain an instance of Asymmetric A Priori TSP $(V,c,p^{v})$ with depot $v$ by
\begin{equation*}
		 p^{v}(u) \coloneqq \begin{cases}1 &u=v,\\ p(u) &\text{else}. \end{cases}
\end{equation*}
For each such instance, we compute an priori tour $T_v$ that satisfies~\eqref{eq:cor_depot}, and return the one minimizing $\expectationcustom{A \sim p}{c(T_v[A])}$.

Applying \Cref{lem:atsp_reduction_to_depot_adaptivity} to $(V,c,p)$ implies that there exists a vertex $d \in V$ such that~\eqref{eq:atsp_reduction_to_depot0_adaptivity} and~\eqref{eq:atsp_reduction_to_depot1_adaptivity} of \Cref{lem:atsp_reduction_to_depot_adaptivity} are satisfied.
Since $p(V) \geq 1/2$, this yields
\begin{equation*}
\begin{split}
\min_{v \in V} \expectationcustom{A \sim p}{c(T_v[A])}
&\leq \expectationcustom{A \sim p}{c(T_d[A])} \\
\overset{\eqref{eq:atsp_reduction_to_depot0_adaptivity}}&{\leq} \expectationcustom{A \sim p^d}{c(T_d[A])} \\
\overset{\eqref{eq:cor_depot}}&{\leq} \alpha(n) \cdot \expectationcustom{A \sim p^{d}}{{\rm TSP}(A,c)} \\
\overset{\eqref{eq:atsp_reduction_to_depot1_adaptivity}}&{\leq} \alpha(n) \cdot 2\left(1+ \frac{1}{p(V)} \right) \cdot \expectationcustom{A \sim p}{{\rm TSP}(A,c)} \\
&\leq 6 \cdot \alpha(n) \cdot \expectationcustom{A \sim p}{{\rm TSP}(A,c)}.
\end{split}
\end{equation*}
\end{proof}

Using \cref{lem:atsp_reduction_to_depot_approx} instead of \Cref{lem:atsp_reduction_to_depot_adaptivity}, one could prove a similar reduction statement for approximation algorithms comparing to $\OPT(V,c,p)$ instead of $\expectationAP{{\rm TSP}(A,c)}$.

\subsection{A simple $O(\sqrt n)$-approximation algorithm}
\label{sec:simple_approx}

Being able to reduce to instances with a depot allows us to give a simple polynomial-time approximation algorithm with approximation ratio $O(\sqrt n)$ and an upper bound on the adaptivity gap, as claimed in \Cref{sec:our_contribution}:

\SimpleApprox*

Our approach is to partition the instance  into two instances of \APrioriATSP with vertex sets $V_1$ and $V_2$, which we consider separately.
The vertex sets $V_1$ and $V_2$ will satisfy $V = V_1 \cup V_2$ with $V_1 \cap V_2 = \parset{d}$ (where $d\in V$ is a depot).
We will choose a probability threshold $p^\ast \in [0,1]$ and define
	\begin{equation*}
		V_1 \coloneqq \parset{v \in V : p(v) \geq p^\ast} \quad\text{and}\quad V_2 \coloneqq \parset{v \in V : p(v) < p^\ast} \cup \parset{d}.
	\end{equation*}
	Ultimately, we will choose $p^\ast \coloneqq \frac{1}{\sqrt{n}}$.
	Intuitively the instance with vertex set $V_1$, containing only vertices with ``large'' activation probabilities, resembles a problem that is close to standard \ATSP.
	
\begin{proposition}
	\label{prop:large_activations}
	Let $(V,c,p)$ be an instance of \APrioriATSP with a depot $d\in V$ and $p_{\min} \coloneqq \min_{v \in V}p(v)$.
	Let $T$ be an $\alpha$-approximate tour for the instance $(V,c)$ of \ATSP.
	Then,
	\begin{equation*}
		c\left( T \right) \leq \alpha \cdot \ceil*{\frac{1}{p_{\mathrm{min}}}} \cdot \expectationAP{{\rm TSP}(A,c)}.
	\end{equation*}
	\begin{proof}
		Let $k \coloneqq \ceil*{\frac{1}{p_{\text{min}}}}$.
		For the sake of this proof, we conduct the following random experiment: 
		Start with $S_1, \dots, S_k = \parset{d}$.
		Now, for every $v \in V \setminus \parset{d}$ choose an index $i \in \parset{1, \dots, k}$ uniformly and independently at random and add $v$ to $S_i$.
		We have
		\begin{equation*}
			\probability{v \in S_i} = \frac{1}{k} \leq p_{\mathrm{min}} \leq p(v),
		\end{equation*}
		for all $v \in V \setminus \parset{d}$ and fixed $i \in \parset{1, \dots, k}$.
		Hence, $\expectation{{\rm TSP}(S_i, c)} \leq \expectationAP{{\rm TSP}(A, c)}$, using that the index choices were independently made.
		Hence,
		\begin{equation*}
			c\left( T \right) \leq \alpha \cdot \expectation{\sum_{i=1}^k {\rm TSP}(S_i, c)} \leq \alpha \cdot k \cdot \expectationAP{{\rm TSP}(A, c)},
		\end{equation*}
		since for any realization of $S_1, \dots, S_k$, their optimal tours can be concatenated into a single tour on $V$.
	\end{proof}
\end{proposition}

	Observe that the length of any tour on $k$ vertices can be at most a factor of $k$ larger than the length of the optimal tour.
	This leads to an algorithm with an approximation guarantee that depends on the expected number of active vertices, which is good when all vertices (except for the depot) have ``small'' activation probabilities, as it is the case for the vertices in $V_2$ .

\begin{proposition}
	\label{prop:small_activations}
	Let $(V,c,p)$ be an instance of \APrioriATSP with depot $d \in V$.
	Then any Hamiltonian cycle $T$ on $V$ satisfies
	\begin{equation*}
		\expectationAP{c(T[A])} \ \leq\ 6p(V) \cdot \expectationAP{{\rm TSP}(A,c)}.
	\end{equation*}
\end{proposition}
\begin{proof}
	By the triangle inequality, we can bound
	\begin{equation*}
		\expectationAP{{\rm TSP}(A,c)} \ \geq\  \expectationAP{ \frac{\sum_{v \in A} c(v,d) + c(d,v)}{|A|} }.
	\end{equation*}
	We can rewrite the right hand side as
	\begin{equation*}
		\begin{split} 
			\expectationAP{ \frac{\sum_{v \in A}c(v,d) + c(d,v)}{|A|} } \
			&= \  \sum_{v \in V} p(v)\cdot (c(v,d) + c(d,v)) \cdot \condexp{A \sim p}{\frac{1}{|A|}}{v \in A}.
		\end{split}
	\end{equation*}
	By Markov's inequality,
	\begin{equation*}
		\begin{split}
			\condexp{A \sim p}{\frac{1}{|A|}}{v \in A} & \geq \expectationAP{\frac{1}{|A| + 1}} \\ 
			&> \mathbb{P}_{A \sim p}\Big[ |A| < 2 \cdot \mathbb{E}_{A\sim p}\big[|A|\big] \Big] \cdot \frac{1}{ 2 \cdot \expectationAP{|A|} + 1 } \\
			&\geq \frac{1}{2} \cdot \frac{1}{2 \cdot \expectationAP{|A|} + 1} \\
			&\geq \frac{1}{6p(V)},
		\end{split}
	\end{equation*}
	because $\expectationAP{|A|} = p(V) \geq 1$ as $V$ contains a depot.
	Therefore, putting everything together yields
	\begin{equation*}
			\expectationAP{{\rm TSP}(A,c)} \geq \sum_{v \in V} p(v)\cdot (c(v,d) + c(d,v)) \cdot \frac{1}{6p(V)},
	\end{equation*}
	which concludes the proof as every Hamiltonian cycle $T$ on $V$ satisfies $\expectationAP{c(T[A])}  \leq \sum_{v \in V} p(v)\cdot (c(v,d) + c(d,v))$ by the triangle inequality.
\end{proof}

We can now combine \Cref{prop:large_activations,prop:small_activations} into a single algorithm (see also \Cref{fig:split_and_connect}):

\begin{figure}
\begin{center}
\begin{tikzpicture}[
vertex/.style={circle,draw=black, fill,inner sep=1.5pt, outer sep=2.0pt},
edge/.style={-latex, very thick},
scale=0.9
]
\definecolor{color1}{named}{Peach} 
\definecolor{color2}{named}{Violet}

\node[rectangle,draw=black, fill,inner sep=4pt] (d) at (3,2.9) {};

\begin{scope}[every node/.style={vertex, color1}, edge/.append style={draw=color1, densely dashed}]

\node[] (v1) at (0.6,3.24) {};
\node[] (v2) at (4.5,2.1) {};
\node[] (v3) at (2.7,5.6) {};
\node[] (v4) at (6.0,3.8) {};
\node[] (v5) at (3.9,1.2) {};
\node[] (v6) at (7.0,6.5) {};

\draw[edge] (v1) to[bend right=20] (d);
\draw[edge] (d) to[bend right=20] (v1);

\draw[edge] (v2) to[bend right=20] (d);
\draw[edge] (d) to[bend right=20] (v2);

\draw[edge] (v3) to[bend right=20] (d);
\draw[edge] (d) to[bend right=20] (v3);

\draw[edge] (v4) to[bend right=12] (d);
\draw[edge] (d) to[bend right=12] (v4);

\draw[edge] (v5) to[bend right=20] (d);
\draw[edge] (d) to[bend right=20] (v5);

\draw[edge] (v6) to[bend right=10] (d);
\draw[edge] (d) to[bend right=10] (v6);
\end{scope}

\begin{scope}[every node/.style={vertex, color2}, edge/.append style={draw=color2}]

\node[] (v7) at (5.1,7.3) {};
\node[] (v8) at (0.8,4.2) {};
\node[] (v9) at (3.3,7.1) {};
\node[] (v10) at (4.7,3.8) {};
\node[] (v11) at (2.5,1.7) {};
\node[] (v12) at (6.6,2.9) {};

\draw[edge] (d) -- (v11);
\draw[edge] (v11) -- (v8);
\draw[edge] (v8) -- (v9);
\draw[edge] (v9) -- (v7);
\draw[edge] (v7) -- (v12);
\draw[edge] (v12) -- (v10);
\draw[edge] (v10) -- (d);
\end{scope}

\node at (6,1.5) [] {\textcolor{color1}{${V_{2} \setminus \parset{d}}$}};
\node at (1.45,6.3) [] {\textcolor{color2}{${V_{1} \setminus \parset{d}}$}};

\node at (d) [below=4pt] {\textcolor{black}{{$\boldsymbol{d}$}}};
        
\end{tikzpicture} \end{center}
\caption{
Illustration of how we combine \Cref{prop:large_activations,prop:small_activations}.
We concatenate an $\alpha$-approximate ATSP tour on $V_1$ (the vertices with ``large'' activation probabilities) and any Hamiltonian cycle on $V_2$ (those with ``small'' activation probabilities).
In our analysis, the expected cost of the latter is upper bounded by the cost of visiting each (active) vertex $v \in V_2$ individually and returning to the depot $d$ between every visit.
Note that the depot is important to bound the expected cost of the concatenated tour by the sum of the expected costs of the two individual tours given by \Cref{prop:large_activations,prop:small_activations}.
}
\label{fig:split_and_connect}
\end{figure}
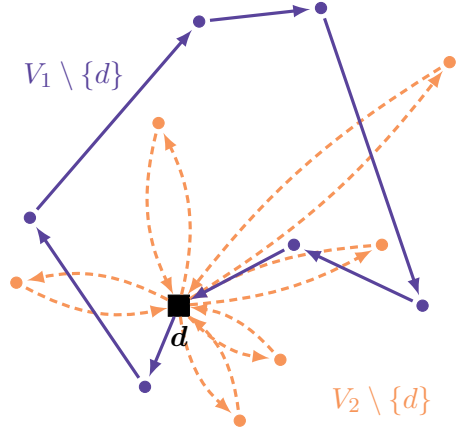

\begin{corollary}
	\label{cor:apx_for_large_and_tiny_activations}
	Let $(V,c,p)$ be an instance of \APrioriATSP with a depot $d\in V$.
	Let  $n\coloneqq |V|$ and
	 \begin{equation*}
	 	V_1 \coloneqq \parset{v \in V : p(v) \geq \frac 1 {\sqrt n}} \quad\text{and}\quad V_2 \coloneqq \parset{v \in V : p(v) < \frac 1{\sqrt n}} \cup \parset{d}.
	 \end{equation*}
	 Concatenating an $\alpha$-approximate tour $T_1$ for the instance $(V_1, c)$ of \ATSP and any Hamiltonian cycle $T_2$ on $V_2$ results in an a priori tour $T$ for $(V,c,p)$ with expected cost at most
	\begin{equation}
		\label{eq:apx_for_large_and_tiny_activations}
		(2\alpha + 12)\cdot \sqrt{n}\cdot \expectationAP{{\rm TSP}(A,c)}.
	\end{equation}
\end{corollary}
\begin{proof}
	On the one hand, we can bound $c(T_1)$ by \Cref{prop:large_activations} by
	\begin{equation*}
		\begin{split}
			c(T_1) 
			&\leq \alpha \cdot \ceil*{\sqrt{n}} \cdot \expectationAP{{\rm TSP}(A \cap V_1,c)} \\
			&\leq \alpha \cdot \ceil*{\sqrt{n}} \cdot \expectationAP{{\rm TSP}(A,c)}\\
			&\leq2 \alpha \cdot \sqrt{n} \cdot \expectationAP{{\rm TSP}(A,c)}.
		\end{split}
	\end{equation*}
	
	On the other hand, by \Cref{prop:small_activations} and the triangle inequality we have
	\begin{equation*}
		\begin{split}
		    \expectationAP{c(T_2[A \cap V_2])} 
		    &\leq 6p(V_2) \cdot \expectationAP{{\rm TSP}(A \cap V_2,c)} \\
			&\leq 12\sqrt{n} \cdot \expectationAP{{\rm TSP}(A,c)},
		\end{split}
	\end{equation*}
	because $p(V_2) < 1 + \sum_{v \in V_2} \frac{1}{\sqrt{n}} \leq 2\sqrt{n}$.
\end{proof}

Finally, combining \cite{svensson2020constant, traub25} with \cref{cor:apx_for_large_and_tiny_activations} and \cref{cor:depot} proves \cref{thm:simple_approx}.

\subsection{Reducing to uniform activation probabilities}\label{sec:uniform_prob}

For our reduction to Hop-ATSP, we will need that the activation probabilities of different vertices do not differ too much. 
We will even show that we can assume that all vertices have the same activation probability.
Using a similar approach as in \Cref{sec:simple_approx} to handle vertices with small activation probability (although with a different threshold for what ``small'' means), we can restrict our attention to instances in which all vertices have large activation probability.
The idea is that we then replace each vertex by a set of copies that then each have a smaller activation probability, similar to an approach used in \cite{blauth_et_al}.
The number of vertices will stay polynomial in the original number of vertices.

We will replace each vertex $v$ by so many co-located copies with the new uniform activation probability that the probability that one of these copies is active is approximately the same as the original activation probability $p(v)$, however it will in general not be exactly equal to $p(v)$.
To prove that this is not a problem, we show that changing the activation probabilities of all vertices by at most a small factor, does not have a large impact on the expected cost.
First, we show that decreasing activation probabilities cannot lead to a larger expected cost.

\begin{lemma}\label{lem:smaller_probabilities_decrease_cost}
	Let $(V,c,p)$ and $(V,c,p')$ be two instances of \APrioriATSP on the same vertex set $V$ and the same cost function $c$ and $p'(v) \leq p(v)$ for all $v\in V$ (where $n:= |V|$). Let $T$ be a tour on $V$. Then the expected cost of $T$ with respect to $p'$ is at most as large as the cost of $T$ with respect to $p$, i.e.
	\[\expectationcustom{A\sim p'}{c(T[A])} \leq \expectationAP{c(T[A])}.\]
\end{lemma}
\begin{proof}
	We prove this for the case where there is exactly one vertex $\hat{v}$ with $p'(v) < p(v)$. The lemma then follows by induction. Let $V = V_1 \cup V_2$ with $V_1 = \{A\subseteq V: \hat v \in A\}$ and $V_2 = \{A\subseteq V: \hat v \notin A\}$. Then $r \colon V_1 \to V_2, A \mapsto A\setminus \parset{\hat v}$ is a bijection. Note that for any $U\in V_1$, we have
	\begin{align*}
	\probabilitycustom{A \sim p}{A = U \text{ or } A = r(U)} &= \prod_{v \in r(U)} p(v) \cdot \prod_{v\notin U}(1-p(v))\\
	 &=  \prod_{v \in r(U)} p'(v) \cdot \prod_{v\notin U}(1-p'(v))\\
	 &= \probabilitycustom{A \sim p'}{A = U \text{ or } A = r(U)}.
	\end{align*}
	The triangle inequality gives us $c(T[r(U)]) \leq c(T[U])$ and therefore
	\begin{align*}
		\expectationcustom{A\sim p'}{c(T[A])}
		&= \sum_{U\subseteq V} \probabilitycustom{A \sim p'}{A = U}\cdot c(T[U]) \\
		&= \sum_{U\in V_1} \probabilitycustom{A \sim p'}{A = U\text{ or } A = r(U)}\cdot \left(p'(\hat{v})\cdot c(T[U]) + (1-p'(\hat{v})) \cdot c(T[r(U)])\right)\\
		&= \sum_{U\in V_1} \probabilitycustom{A \sim p'}{A = U\text{ or } A = r(U)}\cdot \left(c(T[r(U)]) + p'(\hat{v})\cdot (c(T[U]) - c(T[r(U)]))\right)\\
		&\leq \sum_{U\in V_1} \probabilitycustom{A \sim p}{A = U\text{ or } A = r(U)}\cdot \big(c(T[r(U)]) + p(\hat{v})\cdot (c(T[U]) - c(T[r(U)]))\big)\\
		&= \sum_{U\subseteq V} \probabilitycustom{A \sim p}{A = U}\cdot c(T[U]) \\
		&= \expectationcustom{A\sim p}{c(T[A])}.
	\end{align*}
\end{proof}

Next, we show that when decreasing each activation probability by a factor of at most $\beta$, the expected tour cost can decrease at most by a factor $\beta^2$.

\begin{lemma}\label{lem:bounded_cost_decrease}
	Let $(V,c,p)$ and $(V,c,p')$ be two instances of \APrioriATSP on the same vertex set $V$ and the same cost function $c$ and $\beta\cdot p(v) \leq p'(v) \leq p(v)$ for all $v\in V$ (where $0\leq \beta\leq 1$). Let $T$ be a tour on $V$. Then
	the expected cost of $T$ with respect to $p'$ is not much smaller than the cost of $T$ with respect to $p$, in particular
	\[\expectationcustom{A\sim p'}{c(T[A])} \geq \beta^2\expectationAP{c(T[A])}.\]
\end{lemma}
\begin{proof}
	First, consider the case where $T$ is a cycle.
	Given two vertices $s,t\in V$, let $V_{(s,t)}$ denote the set of vertices that lie (strictly) between $s$ and $t$ on $T$. The probability that we have to pay for the edge $(s,t)$ after cutting $T$ short to the active vertices drawn randomly w.r.t.\@ $p'$ is
	\[p'(s)\cdot p'(t) \cdot\prod_{v \in V_{(s,t)}} (1-p'(v)) \geq \beta^2 \cdot p(s)\cdot p(t)\cdot\prod_{v \in V_{(s,t)}}(1-p(v)).\]
	Thus
	\begin{align*}
		\expectationcustom{A\sim p'}{c(T[A])} &= \sum_{s,t\in V} c(s,t) \cdot p'(s)\cdot p'(t) \cdot\prod_{v \in V_{(s,t)}}(1-p'(v)) \\
		&\geq \beta^2\cdot \sum_{s,t\in V} c(s,t) \cdot p(s)\cdot p(t)\cdot\prod_{v \in V_{(s,t)}}(1-p(v))\\
		&= \beta^2 \cdot \expectationAP{c(T[A])}.
	\end{align*}
	The case where $T$ visits vertices multiple times can be treated analogously.
\end{proof}

Finally, we show that when all vertices have large activation probability, we may assume uniform activation probabilities.
Later we will apply the following proposition for $\epsilon = \frac{1}{n}$.

\begin{proposition}\label{prop:uniform_activation_probabilities}
	Given an instance $(V,c,p)$ of \APrioriATSP where $p_{\min} := \min_{v\in V}p(v) \geq \epsilon > 0$, we can transform it in polynomial time (in $|V|$ and $1/\epsilon$) into an instance $(V',c',p')$ of \APrioriATSP such that 
	\begin{enumerate}
		\item $p'(v) = \epsilon/2$ for any $v \in V'$,
		\item $\OPT(V',c',p') \leq \OPT(V,c,p)$,
		\item any tour $T'$ in $(V',c',p')$ can be converted into a tour $T$ in $(V,c,p)$ in polynomial time such that $\expectationcustom{A\sim p'}{c'(T'[A])} \geq \frac 14\expectationAP{c(T[A])}$, and 
	    \item $|V| \leq |V'|\leq 4|V|/\epsilon$.
	\end{enumerate}
\end{proposition}
\begin{proof}
	We construct $(V',c',p')$ from $(V,c,p)$ by replacing each $v\in V$ by a set of vertices $C(v)$ (which we will also call \emph{cluster}), i.e.\@ we set $V' = \bigcup_{v\in V}C(v)$. 
	Each $C(v)$ will contain at least one vertex.	
	Let $\pi\colon V'\to V$ be the projection mapping a vertex $v'\in V'$ to the original vertex $v\in V$ such that $v'\in C(v)$.
	Define $c'$ by setting $c'(v',w') = c(\pi(v'), \pi(w'))$ for $v',w'\in V'$.
	Set $p'(v') = \epsilon/2$ for all $v'\in V'$.
	For each $v \in V$ we choose the number of vertices in $C(v)$ as
	\begin{equation*}
	k_v \coloneqq \ceil*{\log_{1-\epsilon/2}\left(1-\frac{p(v)}2\right)}.
	\end{equation*}		
	We claim that this satisfies
	\begin{equation}\label{eq:cluster_probabilities}
		p(v) \geq \mathbb P_{A\sim p'}[A\cap C(v)\neq \emptyset] \geq \frac12\cdot p(v).
	\end{equation}
	Since we can write $\probabilitycustom{A\sim p'}{A\cap C(v) \neq \emptyset} = 1 - (1-\epsilon / 2)^{k_v}$, the second inequality is clear by our choice of $k_v$. 
	For the first inequality, we may assume that $p(v) < 1$.
	Since $p(v) \geq \epsilon$, we have $(1-\epsilon/2)\cdot (1-\frac{p(v)}{2})\geq (1-p(v))$.
	This implies
	\[k_v \leq \log_{1-\epsilon/2}\left(1-\frac{p(v)}2\right) + 1 \leq \log_{1-\epsilon/2}(1-p(v)),\]
	and hence $1 - (1-\epsilon/2)^{k_v} \leq p(v)$, proving the first inequality of~\eqref{eq:cluster_probabilities}.
	
	Note that $\log_{1-\epsilon/2}\left(1-\frac{p(v)}2\right) \geq 1$.
	Moreover, since $(1- \epsilon/2)^{2/\epsilon} \leq \frac1e \leq \frac 12 \leq 1- \frac{p(v)}2$, we have $\log_{1-\epsilon/2}(1-\frac{p(v)}2) \leq 2/\epsilon$.
	We conclude that $1 \leq k_v \leq 4/\epsilon$ and therefore
	\[|V| \leq |V'| \leq \frac{4|V|}{\epsilon}. \]
	Note that (i) is satisfied by construction.

	For (ii), consider an optimum a priori tour $T^*$ for $(V,c,p)$. Let the tour $T'$ on $V'$ arise from $T^*$ by visiting the clusters $C(v)$ in the same order as $T^*$ visits the vertices $v\in V$, and visit vertices inside each cluster consecutively, but in an arbitrary order.
	Edges inside $C(v)$ have cost 0 and we have $\mathbb P_{A\sim p'}[C(v) \cap A \neq \emptyset] \leq p(v)$
	by \eqref{eq:cluster_probabilities}. Thus \cref{lem:smaller_probabilities_decrease_cost} implies that $\expectationcustom{A\sim p'}{c'(T'[A])} \leq \expectationAP{c(T^*[A])}$ which yields (ii).
	
	Similarly, if a tour $T'$ in $(V',c',p')$ visits all vertices of each cluster consecutively, it corresponds to a tour $T$ in $(V,c,p)$ with $\expectationcustom{A\sim p'}{c'(T'[A])}\geq \frac14\expectationAP{c(T[A])}$ by \cref{lem:bounded_cost_decrease}.
	In order to prove (iii) it thus only remains to show that any tour $T'$ in $(V',c',p')$ can be transformed into a tour $T''$ in $(V',c',p')$ that visits all vertices of the same cluster $C(v)$ consecutively in polynomial time without increasing the cost.
	To this end let $T'$ visit a vertex $v_1 \in  C(v)$, then follow a path $P$, then visit $v_2\in C(v)$, and follow a path $Q$ before returning to $v_1$.
	Construct $T_1$ out of $T'$ by visiting $v_2$ immediately after $v_1$ thus concatenating $P$ and $Q$ without visiting $v_2$ in between, and $T_2$ out of $T'$ by visiting $v_1$ immediately after $v_2$ (thus concatenating $Q$ and $P$ without visiting $v_1$ in between).
	Let $A\subseteq V'$ be the set of active vertices. If $v_1\in A$ then $c'(T_1[A]) \leq c'(T'[A])$ by the triangle inequality because we get rid off the detour to $v_2$ between $P$ and $Q$.
	If $v_2\in A$, but $v_1\notin A$, then $T_1[A]$ visits $v_2$ followed by the concatenation of $P[A]$ and $Q[A]$ before returning to $v_2$ while $T'[A]$ visits $v_2$ followed by the concatenation of $Q[A]$ and $P[A]$ in this order.
	If the expected cost of the path resulting from visiting $Q$ after $P$ is larger than the one of the path resulting from visiting first $Q$, then $P$, $T_1$ is thus more expensive in expectation than $T'$. But then $T_2$ is cheaper than $T'$: If $v_2 \in A$ then we easily get $c'(T_2[A]) \leq c'(T'[A])$.
	If $v_1\in A$, but $v_2\notin A$, then $T_2[A]$ visits first $Q[A]$, then $P[A]$ while $T'[A]$ visits $Q[A]$ after $P[A]$. Therefore at least one of $T_1$ and $T_2$ has expected cost at most the expected cost of $T'$.
	By induction we can hence rearrange $T'$ such that it visits all vertices of each cluster consecutively without increasing the cost. Note that this rearrangement can be done in polynomial time because in each step we can compute the expected cost of attaching $Q$ after $P$ resp.\@ $P$ after $Q$ and decide for the better in polynomial time.
\end{proof}

\subsection{Reducing to Hop-ATSP}\label{sec:complete_hop_atsp_reduction}

In this subsection we reduce Asymmetric A Priori TSP with uniform activation probabilities to Hop-ATSP.
Suppose we have a Hamiltonian cycle $T$ on $V$ and uniform activation probabilities $p(v) = \delta$ for each $v \in V$.
We label the vertices of $V$ as $v_1, \dots, v_n$ in the order in which they are visited by $T$.
We follow our usual convention that we interpret vertices $v_{j+n} \coloneqq v_j$.
When sampling a set $A$ according to $p$, the probability that an edge $e=(v_j, v_{j+\Delta})$ is contained in $T[A]$ is $\delta^2 \cdot (1-\delta)^{\Delta - 1}$, because this happens if and only if $v_j,v_{j+\Delta} \in A$ and $\{v_{j+1}, \dots, v_{j + \Delta - 1}\} \subseteq V \setminus A$.
Hence, we can write the expected cost of $T$ as
\begin{equation}
\label{eq:a_priori_cost_fct}
\expectationcustom{A \sim p}{ c\left(T[A] \right)} = \sum_{\Delta=1}^{n - 1} \sum_{j=1}^n \delta^2(1-\delta)^{\Delta - 1} \cdot c(v_j,v_{j+\Delta}),
\end{equation}

We assume for the sake of this explanation that $\delta > \frac{1}{|V|}$ and define $k \coloneqq \floor*{1/\delta} < |V|$.
In order to reduce Asymmetric A Priori TSP with uniform activation probabilities to Hop-ATSP, we will show that~\eqref{eq:a_priori_cost_fct} is up to constant factors the same as $\delta^2$ times the $k$-hop cost of $T$. 
To this end, we first observe that $(1-\delta)^{\Delta - 1} \geq \beta$ for some suitable constant $\beta \in \bR_{\geq 0}$ and any $\Delta \in [k]$, implying that~\eqref{eq:a_priori_cost_fct} is lower-bounded by $ \delta^2 \cdot \beta \cdot c^{(k)}(T)$, 

In order to obtain an upper bound of $ \delta^2 \cdot \gamma \cdot c^{(k)}(T)$ for~\eqref{eq:a_priori_cost_fct}, where $\gamma \in \bR_{\geq 0}$ is another constant, we show that \begin{equation*}
\sum_{\Delta=k+1}^{n-1} \sum_{j=1}^n (1-\delta)^{\Delta - 1} \cdot c(v_j,v_{j+\Delta}) 
\leq (\gamma - 1) \cdot \sum_{\Delta=1}^{k} \sum_{j=1}^n c(v_j,v_{j+\Delta})
= (\gamma - 1)\cdot c^{(k)}(T).
\end{equation*}
By the triangle inequality, for an edge $e=(v_j,v_{j+\Delta})$ with $\Delta \in \{k+1, \dots, n-1\}$, we can bound the cost $c(e)$ by the cost of any $v_j$-$v_{j+\Delta}$ path $c(P^{(e)})$.
We will choose $P^{(e)}$ such that each of its edges contributes to the $k$-hop cost of $T$.
The following lemma will be useful for this:

\begin{lemma}
	Let $t$ and $k$ be integers with $t\geq k\geq 16$. Then there are positive integers $a_1,\dots a_m$ such that $\sum_{j = 1}^m a_j = t$ and $2\lceil \frac k{16}\rceil \leq a_j \leq \frac k2$ for all $j\in\parset{1,\dots,m}$.
\end{lemma}
\begin{proof}
	Let $m \in \mathbb Z$ be the largest integer such that $m\cdot 2\lceil \frac k{16}\rceil \leq t$. Define $a_j := 2\lceil \frac k{16}\rceil$ for $1\leq j \leq m-1$ and $a_{m} = t - (m-1)\cdot2\lceil\frac k{16}\rceil$. By maximality of $m$, we have $a_m < 4\lceil\frac k{16}\rceil \leq 4(\frac k{16} +1) \leq \frac k2$ since $k\geq 16$.
\end{proof}
This immediately implies the following lemma:
\begin{lemma}
	\label{lem:block_decomp}
	Let $P$ be a path on at least $k$ vertices.
	If $k \geq 16$, then we can subdivide $P$ into vertex disjoint subpaths $P_1, \dots, P_m$ of $P$ such that $\dot \bigcup_{j=1}^m V(P_j) = V(P)$ and $2\ceil*{\frac{k}{16}} \leq |V(P_j)| \leq \frac{k}{2}$ for all $j \in \parset{1, \dots, m}$. \hfill \qed
\end{lemma}

Recall that $e = (v_j,v_{j+\Delta})$.
Applying the previous lemma to the $v_j$-$v_{j+\Delta}$ subpath of $T$ allows us to choose the path $P^{(e)}$ such that it contains approximately $\frac{\Delta}{k}$ edges.
This gives us an upper bound on the number of edges contributing to $c^{(k)}(T)$ that additionally need to pay for the cost of $(1-\delta)^{\Delta - 1} \cdot c(v_j, v_{j+\Delta})$.
Note that
\begin{equation*}
\sum_{j=1}^n \sum_{\Delta = k+1}^{n-1} \frac{\Delta}{k} \cdot  (1 - \delta)^{\Delta - 1} = \frac{n}{k} \cdot O(1/\delta^2) = O(n\cdot k).
\end{equation*}
Moreover, there are exactly $n\cdot k$ edges that contribute to $c^{(k)}(T)$.
By carefully selecting these paths $P^{(v_j,v_{j+\Delta})}$ for any $j \in [n]$ and $\Delta \in \{k+1, \dots, n-1\}$, we will be able to ensure that every edge contributing to $c^{(k)}(T)$ only needs to additionally pay a constant multiple of its own cost.

In fact, we will not only choose a single path $P^{(e)}$ with $e = (v_j,v_{j+\Delta})$, but we choose many such paths $P^{(e)}_1, \dots, P^{(e)}_r$ together with weights $\xi^{(e)}_1, \dots, \xi^{(e)}_{r}$ such that $\sum_{i=1}^r \xi^{(e)}_i = 1$.
This allows us to distribute the cost $(1-\delta)^{\Delta - 1} \cdot c(v_j, v_{j+\Delta})$ evenly onto all edges contributing to $c^{(k)}(T)$.
The following lemma formalizes this:

\begin{lemma}
	\label{lem:path_distribution}
	Let $P$ be an $s$-$t$-path on at least $k+2$ vertices in a complete directed graph $G$ for some $k \in \bZ_{\geq 16}$. Enumerate the vertices on $P$ in the order given by $P$ as $s = x_0, x_1,\dots, x_{\ell + 1} = t$. 
	Then there exists a collection of $s$-$t$-paths $P_1, \dots, P_r$ together with weights $\xi_1, \dots, \xi_r \in [0,1]$ with $\sum_{i=1}^r \xi_i = 1$ such that
	\begin{equation}
		\label{eq:path_distribution1}
		\begin{cases}
			\chi_e \leq 2^4/k &\text{if }e = (x_i, x_j)\in \left(\delta^+(s) \cup \delta^-(t) \right) \text{ and }j-i\leq k, \\
			\chi_e \leq 2^8/k^2 &\text{if }e = (x_i, x_j)\notin \delta^+(s) \cup \delta^-(t)\text{ and }j-i\leq k, \\
			\chi_e = 0 &\text{if } e= (x_i, x_j)\text{ with }j-i>k,
		\end{cases}
	\end{equation}
	where $\chi \in [0,1]^{E(G)}$ is defined as  $\chi \coloneqq \sum_{i = 1}^r\xi_i \cdot \chi(P_i)$ (writing $\chi(P_i)\in \{0,1\}^{E(G)}$ as the characteristic vector of $P_i$, i.e.\@ $\chi(P_i)_e = 1$ if and only if $e\in P_i$).
\end{lemma}
\begin{proof}
	Apply \Cref{lem:block_decomp} to the $x_1$-$x_{\ell}$ subpath $P_{[x_1, x_{\ell}]}$ of $P$ and obtain vertex disjoint subpaths $Q_1, \dots, Q_m$.
	We denote $X_j \coloneqq V(Q_j)$ for all $j \in \parset{1, \dots, m}$.
	Every $X_j$ contains at least $k/16$ and at most $k/2$ vertices.
	
	Next, we define an $s$-$t$-flow $f$ in $G$ by
	\begin{equation*}
		f_e = \begin{cases}
			\frac{1}{|X_1|} &\text{if }e=(s,x) \text{ with } x \in X_1, \\
			\frac{1}{|X_j|\cdot|X_{j+1}|} &\text{if }e=(x,y) \text{ with } x \in X_j, y \in X_{j+1}, j \in \parset{1, \dots, m-1}, \\
			\frac{1}{|X_m|} &\text{if }e=(x,t) \text{ with } x \in X_m, \\
			0 &\text{else},
		\end{cases}
	\end{equation*}
	for all $e\in E(G)$.
	One easily sees that $f$ indeed satisfies the flow conservation rule.
	We observe that its value is $\sum_{x\in X_1} \frac{1}{|X_1|} = 1$.
	The support of $f$, i.e., $\operatorname{supp}(f) \coloneqq \parset{e \in E(G) : f_e > 0}$, is acyclic by construction.
	Therefore we can apply flow decomposition to $f$ and $(V, \operatorname{supp}(f))$ yielding paths $P_1, \dots, P_r$ and $\xi_1, \dots, \xi_r \in [0,1]$ such that $\sum_{i=1}^r \xi_i = 1$ and
	\begin{equation*}
	f = \sum_{i=1}^r \xi_i \cdot \chi(P_i)
	\end{equation*}
	for all $e\in E(G)$.
	Since each $X_j$ contains at most $k/2$ vertices, $\operatorname{supp}(f)$ only contains edges $(x_i, x_j)$ with $j-i <k$.
	The remaining properties \eqref{eq:path_distribution1} now follow from the definition of $f$ using that all $X_j$ contain at least $k/16$ vertices.
\end{proof}

We are now finally able to prove that the $k$-hop cost and the expected cost of an a priori tour only differ by some constant factor (and a scaling factor of $\delta^2$):
\begin{theorem}
	\label{thm:reduction_to_khop}
	Let $(V,c,p)$ be an instance of \APrioriATSP, such that
	$n \geq 2^6$ and $p(v)  = \delta$ for all $v \in V$, where $n \coloneqq |V|$ and $\delta \in \left[\frac{2}{n}, \frac{1}{2^5}\right]$.
	Let $T$ be a Hamiltonian cycle on $V$.
	Then, we have
	\begin{equation}
		\label{eq:reduction_to_khop_ub}
		\expectationcustom{A \sim p}{c\left( T[A] \right)} \leq  2^{13}\delta^2 \cdot c^{(k)}(T)
	\end{equation}
	and
	\begin{equation}
		\label{eq:reduction_to_khop_lb}
		\delta^2 \cdot c^{(k)}(T) \leq 4 \cdot \expectationcustom{A \sim p}{c\left( T[A] \right)},
	\end{equation}
	where $k \coloneqq \floor*{\frac{1}{\delta}}$.
\end{theorem}
\begin{proof}
	First note that $k \leq 1/\delta \leq n/2 \leq n - 2$ and $k \geq 1/\delta - 1 \geq 1/2\delta \geq 2^4$.
	Enumerate the vertices of $V$ as $v_1,\dots, v_n$ in the order that they appear on $T$. An edge $e = (v_j, v_{j + i})$ (where we set $v_{j}  =v_{j - n}$ for $j > n$) is part of $T$ after cutting it short to the active vertices $A$ if and only if $v_j$ and $v_{j+i}$ are active, but the vertices in between are not, i.e.
	\begin{equation}
		\label{eq:reduction_to_khop_edge_act_prob}
		\delta^2 \cdot (1 - \delta)^{i-1}
		= \probabilitycustom{A \sim p}{e \in E(T[A])}.
	\end{equation}
	Therefore,
	\begin{equation*}
		\begin{split}
			c^{(k)}(T)
			&= \sum_{i=1}^{k} \sum_{j = 1}^n c(v_j, v_{j + i}) \\
			&\leq \sum_{i=1}^{k} \sum_{j = 1}^n c(v_j, v_{j + i}) \cdot \left( 1 - \frac{1}{k} \right)^i \cdot 4 \\
			&\leq 4 \cdot \sum_{i=1}^{n-1} \sum_{j = 1}^n c(v_j, v_{j + i}) \cdot \left( 1 - \delta \right)^{i-1} \\
			\overset{\eqref{eq:reduction_to_khop_edge_act_prob}}&{\leq}\frac{4}{\delta^2} \cdot \expectationcustom{A \sim p}{c\left( T[A] \right)}.
		\end{split}
	\end{equation*}
	The first inequality follows from the fact that $k\geq 2$, $i\leq k$ and thus $\left(1-\frac1k\right)^i \geq \left(1-\frac1k\right)^k \geq \left(1-\frac12\right)^2 = \frac 14$.
	This shows \eqref{eq:reduction_to_khop_lb}.
	
	In order to show \eqref{eq:reduction_to_khop_ub}, we can bound the expected cost after shortcutting inactive vertices by
	\begin{equation}
		\label{eq:reduction_to_khop_initial_ub}
		\begin{split}
			\expectationcustom{A \sim p}{c\left( T[A] \right)}
			\overset{\eqref{eq:reduction_to_khop_edge_act_prob}}&{=} \sum_{i = 1}^{n-1} \sum_{j = 1}^n c(v_j, v_{j+i}) \cdot \delta^2 \cdot (1 - \delta)^{i-1} \\
			&\leq \delta^2 \cdot \left( \sum_{i=1}^{k} \sum_{j = 1}^n c(v_{j}, v_{j+i})  + \sum_{i = k+1}^{n-1}  \sum_{j = 1}^n  c(v_j, v_{j+i}) \cdot \left( 1- \delta \right)^{i-1}  \right).
		\end{split}
	\end{equation}
	The first part of this sum is just $c^{(k)}(T)$, but we still have to deal with the second part of the sum, i.e.\@ the cost of the edges that skip at least $k$ vertices on the tour (and are thus not counted in $c^{(k)}(T)$).
	Let $i \in \parset{k+1, \dots, n-1}$ and $e = (v_j, v_{j + i})$ be fixed.
	We denote by $P_e \coloneqq T_{[v_j,v_{j+i}]}$ the subpath from $v_j$ to $v_{j+i}$ on $T$. Since $P_e$ is a path on at least $k+2$ vertices,
	we can apply \Cref{lem:path_distribution} to $P_e$ and obtain paths $Q^{(e)}_1, \dots, Q^{(e)}_{r_e}$ and $\xi^{(e)}_1, \dots , \xi^{(e)}_{r_e}$ with a vector $\chi^{(e)} \coloneqq \sum_{j=1}^{r_e} \xi^{(e)}_j \cdot \chi(Q_j^{e})$.
	By the triangle inequality, 
	\begin{equation}
		\label{eq:reduction_to_khop_path_bounds}
		\begin{split}
			c(e) = \sum_{j=1}^{r_e} \xi^{(e)}_{j} c(e) &\leq \sum_{j=1}^{r_e} \xi^{(e)}_{j} c\left(Q^{(e)}_{j}\right) 
			= \sum_{j=1}^{r_e} \xi^{(e)}_{j} c^{\intercal}\chi \left(Q^{(e)}_{j}\right)
			= c^{\intercal} \chi^{(e)},
\end{split}
	\end{equation}
	where we interpret $c \in \bR_{\geq 0}^{E}$ by $c_e = c(e)$ for each $e\in V \times V$. 
	
	Now, let $i\in \parset{k+1, \dots, n-1}$ and $f = (v_j, v_{j + \Delta})$ for some $\Delta \leq k$ be fixed.
	We are aiming to find an upper bound on $\sum_{\ell = 1}^n \chi_f^{(v_\ell, v_{\ell + i})} $.
	To this end, consider edges $e = (v_{\ell}, v_{\ell + i})$. There is exactly on such edge $e$ such that $e$ and $f$ end in the same vertex (namely for $\ell = j + \Delta -i$). In this case we have $\chi_f^{(e)} \leq 2^4/k$ by \eqref{eq:path_distribution1}. By symmetry, the same argument applies in the case that $e$ and $f$ start in the same vertex. Furthermore there are at most $i-2$ edges $e = (v_\ell, v_{\ell +i})$ such that $v_j, v_{j + \Delta}\in V(P_e)$, but $f = (v_j, v_{j + \Delta})$ and $e$ do not share any common vertices. 
	In these cases, we can bound $\chi^{(e)}_f \leq 2^8/k^2$ by \eqref{eq:path_distribution1}.
	For all other choices of $e$ we have $\chi^{(e)}_f = 0$.
	Hence,
	\begin{equation}
		\label{eq:reduction_to_khop_edge_bound}
		\begin{split}
			\sum_{\ell = 1}^n \chi_f^{(v_{\ell}, v_{\ell + i})}
			&\leq  2\cdot\frac{2^4}{k} + (i-2)\cdot\frac{2^8}{k^2} \\
			&\leq 2^6\delta + (i-2) \cdot 2^{10}\delta^2
		\end{split}
	\end{equation}
	We can use this to bound
	\begin{equation*}
		\begin{split}
			&\frac{1}{2} \cdot \sum_{i = k+1}^{n-1}  \sum_{\ell = 1}^n  c(v_{\ell}, v_{\ell + i}) \cdot \left( 1- \delta \right)^{i-1} \\
			&\leq \sum_{i = k+1}^{n-1}  \sum_{\ell = 1}^n  c(v_{\ell}, v_{\ell + i}) \cdot \left( 1- \delta \right)^{i} \\
			\overset{\eqref{eq:reduction_to_khop_path_bounds}}&{\leq}\sum_{i = k+1}^{n-1}  \sum_{\ell = 1}^n  c^{\intercal} \chi^{(v_{\ell}, v_{\ell + i})} \cdot \left( 1- \delta \right)^i  \\
			&= \sum_{\Delta = 1}^k \sum_{j=1}^n c(v_j,v_{j+\Delta})\cdot \left( \sum_{i = k+1}^{n-1}  \sum_{\ell = 1}^n  \chi^{(v_{\ell}, v_{\ell + i})}_{(v_j,v_{j+\Delta})} \left( 1- \delta \right)^{i} \right) \\
			\overset{\eqref{eq:reduction_to_khop_edge_bound}}&{\leq} \sum_{\Delta = 1}^k \sum_{j=1}^n c(v_j,v_{j+\Delta})\cdot \left( \sum_{i = k+1}^{n-1}  \left(2^6 \delta + (i-2) \cdot 2^{10}\delta^2 \right) \left( 1- \delta \right)^{i} \right) \\
			&\leq \sum_{\Delta = 1}^k \sum_{j=1}^n c(v_j,v_{j+\Delta})\cdot 2^{10} \cdot \left( 
			\delta \cdot \sum_{i = 0}^{\infty} \left( 1- \delta \right)^i + 
			\delta^2 \cdot \sum_{i = 1}^{\infty} i\left( 1- \delta \right)^i 
			\right) \\
			&= \sum_{\Delta = 1}^k \sum_{j=1}^n c(v_j, v_{j+\Delta})\cdot 2^{10} \cdot \left( 
			\delta \cdot \frac{1}{1-\left(1 - \delta \right)} + 
			\delta^2 \cdot \frac{\left( 1- \delta \right)}{\left( 1- \left( 1- \delta \right) \right)^2}
			\right) \\
			&\leq \sum_{\Delta = 1}^k \sum_{j=1}^n c(v_j,v_{j+\Delta})\cdot 2^{11} \\
			&= 2^{11} \cdot c^{(k)}(T),
		\end{split}
	\end{equation*}
	where the equality in the penultimate line is a well known identity about geometric resp.\@ arithmetico-geometric series. 
	Plugging this into \eqref{eq:reduction_to_khop_initial_ub} yields \eqref{eq:reduction_to_khop_ub}.
\end{proof}

\subsection{Reducing to well-scaled instances of Hop-ATSP}
\label{subsec:well_scaled_instances}

We write $\diam(G,c)$ to denote the diameter, i.e., the maximum length of a shortest path in a graph $G$ with respect to edge costs $c$.
If $G$ is complete and $c$ satisfies the triangle inequality, this is the maximum cost of an edge.

Our next and final goal is to round the edge costs $c(e)$ to integral and polynomially bounded values.
We refer to this property as \emph{well-scaled}:

\begin{definition}[Well-scaled instances]
	We say that a \HopATSP instance $(V,c, k)$ is \emph{well-scaled} if $c$ takes only integer values, $\diam(G,c) \leq 2n^3$, and $\OPT(V,c,k) \geq n^2$, where $G$ is the complete directed graph on $V$ and $n \coloneqq |V|$.
\end{definition}

Reducing to such instances is rather straightforward and induces only a constant-factor loss (in addition to a scaling factor~$K$):

\begin{proposition}
	\label{prop:well_scaled_instances}
	Let $(V,c,k)$ be an instance of \HopATSP with $\diam(G,c) > 0$ and $k \leq n - 1$, where $G$ is the complete directed graph on $V$ and $n \coloneqq |V|$.
	Then, one can compute a well-scaled instance $(V,\hat{c}, k)$ of \HopATSP in $O(n^3)$, such that
	\begin{equation}
		\label{eq:well_scaled_instances}
		\hat{c}^{(k)}(T) \leq \frac{c^{(k)}(T)}{K} \leq 2 \hat{c}^{(k)}(T)
	\end{equation}
	for any Hamiltonian cycle $T$ on $V$, where $K \coloneqq \frac{\diam(G,c)}{2n^3}$.
\end{proposition}
\begin{proof}
	We set
	\begin{equation*}
		\tilde{c}(e) \coloneqq \floor*{\frac{c(e)}{K}}
	\end{equation*}
	for all $e \in E(G)$.
	Then, set $\hat{c}$ to be the shortest path metric with respect to $(G, \tilde{c})$.
We denote by $P_{x,y}$ the shortest path in $G$ with respect to $\tilde{c}$.
	Note that
	\begin{equation*}
		\diam(G,\hat{c}) \leq \diam(G,\tilde{c}) \leq \frac{\diam(G,c)}{K} = 2n^3.
	\end{equation*}
	Let $T$ be an arbitrary but fixed Hamiltonian cycle on $V$ and enumerate the vertices of $V$ as $v_1,\dots, v_n$ in the order of $T$. Define $F := \{(v_i, v_{\Delta + i}):  \Delta \in [k], i \in [n] \}$ where $v_j = v_{j -n}$ for $j > n$.
	In particular, $c^{(k)}(T) = c(F)$.
	Then, we can bound
	\begin{equation*}
		\begin{split}
			K \cdot \hat{c}(F) 
			= \sum_{(x,y) \in F} K \cdot \tilde{c}(P_{x,y})
			&\leq \sum_{(x,y) \in F} K \cdot \tilde{c} (x,y)  \\
			&\leq \sum_{(x,y) \in F} c (x,y)  \\
			&= c(F).
		\end{split}
	\end{equation*}
	Moreover, $|F| = nk \leq n^2$ because $T$ is a Hamiltonian cycle.
	Hence,
	\begin{equation*}
		\begin{split}
			K \cdot \hat{c}(F) 
			= \sum_{(x,y) \in F} \sum_{e \in E(P_{x,y})} K \cdot \tilde{c}(e)
			&\geq \sum_{(x,y) \in F} \sum_{e \in E(P_{x,y})}  \left(c (e) - K \right) \\
			&\geq \sum_{(x,y) \in F} \left( c (x,y) - nK \right) \\
			&\geq c(F) - n^3K \\
			&\geq c(F)/2,
		\end{split}
	\end{equation*}
	because $c(F) = c^{(k)}(T) \geq \diam(G,c) = 2Kn^3$ by the triangle inequality.
	In particular, we have $\hat{c}(F) \geq \frac{c(F)}{K} - n^3 \geq n^3 \geq n^2$.
	By \Cref{lem:khop_skipping_vertices}, this implies that $\OPT(V,\hat{c}, k) \geq n^2$ and $(V,\hat{c}, k)$ is indeed a well-scaled instance.
\end{proof}

Putting everything together, we can now prove \cref{thm:reducing_hop_atsp}, which we restate here for convenience:

\thmHopATSP*
\begin{proof}
	If we can bound the  total activation probability by  $\sum_{v\in V}p(v) \leq \frac12$, then by \cref{lem:apx_for_tiny_activations} any tour on $V$ is a 2-approximation. Otherwise, by \cref{lem:atsp_reduction_to_depot_approx}, there exists a vertex $v\in V$ such that we lose only a factor $6$ in the approximation guarantee by changing the activation probability of $v$ to 1. Let $\mathcal J_v = (V,c,p_v)$ be the instance that arises from $\mathcal J$ by setting $p_v(v) = 1$ and $p_v(w) = p(w)$ for $w\neq v$. We will perform the following steps for each possible choice of $v$ and in the end choose the best. This costs us a factor $n:= |V|$ in the running time.
	
	Next, we need to get rid of vertices with activation probability less than $\frac 1n$. Split $\mathcal J_v$ into two subinstances: $\mathcal J_v'$ with vertex set $V_v' := \{v\}\cup\{w\in V:p_v(w) < \frac1n\}$ and $\hat{\mathcal J_v}$ with vertex set $\hat V_v := \{w\in V: p_v(w) \geq \frac1n\}$. By \cref{prop:small_activations} any tour on $V'$ has expected cost at most $6p_v(V')\cdot \OPT(\mathcal J_v') \leq 6\cdot (n \cdot \frac1n + 1) \cdot \OPT(\mathcal J_v') \leq 12\cdot \OPT(\mathcal J_v)$.
	
	If $|\hat V_v| \leq 2^5$, we find the best a priori tour in constant time by enumeration.
	Otherwise, we then transform $\hat{\mathcal J_v}$ into an instance $\overline{\mathcal J_v} = (\overline {V_v}, \overline c, \overline p)$ with uniform activation probabilities $\overline{p}(w) = \frac1{2n}$ for all $w\in V$, as described in \cref{prop:uniform_activation_probabilities}, using $\varepsilon = \frac1n$. While doing so, we make sure to replace $v$ (i.e.\@ a vertex with $p(v) = 1$) by at least $4n$ vertices. Together with the bound given in \cref{prop:uniform_activation_probabilities}, we get $\overline n := |\overline {V_v}| \leq 4n+ 4|\hat V_v|\cdot n \leq 5n^2$.
	  
	We define $k := 2n$ and interpret $(\overline{V_v}, \overline c, k)$ as an Hop-ATSP instance. Since we made sure that $\overline n \geq 4n$, we get $k\leq \frac{\overline n}2$. Moreover we know that $n\geq |\hat V_v| > 2^5$, and thus $\overline n \geq 2^6$. Hence for $\delta := \frac 1k$ we get $\frac 2{\overline n}\leq \delta \leq \frac1{2^5}$ and therefore we can later apply \cref{thm:reduction_to_khop} to $(\overline{V_v}, \overline c, k)$ with $\delta = \frac 1k$.

	Using \cref{prop:well_scaled_instances}, we find a cost function $\tilde{c}$ such that $(\overline{V_v}, \tilde{c}, k)$ is well-scaled and $\tilde{c}^{(k)}(T) \leq \frac{\overline c^{(k)}(T)}{K} \leq 2 \tilde{c}^{(k)}(T)$ for any Hamiltonian cycle $T$ on $\hat V_v$ (where $K:= \operatorname{diam}(G,\overline c)/2\overline n^3$ and $G$ is the complete directed graph on $\overline {V_v}$).
	
	Now let $T^{\operatorname{hop}}$ be an optimum Hop-ATSP-tour in the Hop-ATSP instance  $(\overline{V_v}, \tilde{c}, k)$ and $T^*$ be an optimum a priori tour in $\overline{\mathcal J_v}$.
	Suppose we can compute a tour $T$ in time $t(\overline n)$ with $\tilde c^{(k)}(T) \leq \alpha(\overline{n})\cdot\tilde c^{(k)}( T^{\operatorname{hop}})$. Then by \cref{thm:reduction_to_khop} we have
	\begin{align*}
		\expectationcustom{A\sim \overline p}{\overline c(T[A])}
		 & \leq 2^{13} \delta^2 \cdot \overline c^{(k)}(T) \\
		 & \leq 2^{13}\delta ^2 \cdot 2K\cdot  \tilde c^{(k)}(T)\\
		 & \leq 2^{13}\delta ^2 \cdot 2K\cdot \alpha(\overline n) \cdot  \tilde c^{(k)}( T^{\operatorname{hop}})\\
		 & \leq 2^{13}\delta ^2 \cdot 2K\cdot \alpha(\overline n) \cdot  \tilde c^{(k)}(T^*)\\
		 & \leq 2^{13}\delta ^2 \cdot 2\cdot \alpha(\overline n) \cdot  \overline c^{(k)}(T^*)\\
		 & \leq 2^{13} \cdot 2\cdot \alpha(\overline n) \cdot 4 \cdot  \expectationcustom{A\sim \overline p}{\overline c(T^*[A])}.
	\end{align*}
	Hence we can compute a $(2^{16}\cdot \alpha(\overline n))$-approximate a priori tour in $\overline{\mathcal J_v}$. By \cref{prop:uniform_activation_probabilities} we lose a factor of $4$ when transforming this back to a tour in $\hat{\mathcal J_v}$, i.e.\@ we get a $(2^{18} \cdot \alpha(\overline n))$-approximate a priori tour in $\hat{\mathcal J_v}$. Together with the tour of cost at most $12\cdot \OPT(\mathcal J_v)$ for $\mathcal J_v'$, we therefore get a tour of cost at most $(2^{18} \cdot \alpha(\overline n) + 12)\cdot \OPT(\mathcal J_v)$. Additionally respecting the factor 6 that we lost when declaring one vertex as a depot by \cref{lem:atsp_reduction_to_depot_approx}, we end up with a $6\cdot (2^{18}\cdot \alpha(\overline n) + 12)$-approximation for the Asymmetric A Priori TSP.
	We can bound $6\cdot (2^{18})\cdot \alpha(\overline n) + 12 \leq 6\cdot (2^{18} + 12)\cdot \alpha(\overline n) \leq 2^{21}\cdot\alpha(\overline n)$, and we already established $\overline n \leq 5n^2$. This proves the statement about the approximation guarantee of the algorithm for Asymmetric A Priori TSP.
	
	For the running time we need $t(\overline n)$ time to find $T$, and time that is polynomial in $n$ for the transformation steps where we introduce uniform activation probabilities and a cost-function for a well-scaled Hop-ATSP instance. Since we might need to perform these steps $n$ times to find the best vertex $v$ to declare as a depot, and because $\overline n \leq 5n^2$, we end up with the claimed running time.
	
\end{proof}

\section{Reducing to Hierarchically Ordered Instances}\label{sec:hierarchically_ordered}

In this section we present our reduction of well-scaled instances of Hop-ATSP to hierarchically ordered instances.
We first recall the definition of hierarchically ordered instances and introduce some useful notation in this context (\Cref{sec:definition_hierarchically_ordered}).
Then we describe the known results on directed low-diameter decompositions from \cite{bernstein2022negative,bringmann2025near,li2025simpler} that we will use (\Cref{sec:ldd}), and finally
we provide our reduction, i.e., the proof of \Cref{thm:hierarchically_ordered} (\Cref{sec:proof_hierarchically_ordered}).

\subsection{Hierarchically ordered instances}\label{sec:definition_hierarchically_ordered}

First, recall that in a hierarchically ordered instance we have a hierarchical partition $\cH=(\cH_{\ell})_{\ell\in L}$, consisting of partitions $\cH_1,\dots, \cH_L$, where each partition $\cH_{\ell}$ is a refinement of the partition $\cH_{\ell-1}$.

\begin{definition}[Hierarchical partition]
A \emph{hierarchical partition} of a finite set $V$ is a family $(\cH_{\ell})_{\ell=1}^L$ of partitions of $V$ such that
\begin{enumerate}
\item
$\cH_1 = \{V\}$ is the trivial partition into a single set, and
\item
$\cH_L = \{\{v\} : v \in V\}$ is the trivial partition into singletons,
\item
for all $1 < \ell \leq L$ and $H \in \cH_{\ell}$, there exists  as set $H^\prime \in \cH_{\ell - 1}$ with $H \subseteq H^\prime$.
\end{enumerate}
We say that $A \subseteq V$ is \emph{$\ell$-confined} if there exists $H \in \cH_{\ell}$ with $A \subseteq H$.
We define $\level(A)$ as the maximum $\ell \in [L]$ such that $A$ is $\ell$-confined.
For edges $e=(x,y)$ we simply abbreviate $\level(e) \coloneqq  \level(\{x,y\})$.
For a vertex $v \in V$ and a level $\ell \in [L]$, we denote by $H_{\ell}(v)$ the unique set $H \in \cH_{\ell}$ such that $v \in H$.
\end{definition}

In a hierarchically ordered instance we also have a total order $\prec$ on the vertex set $V$.

\begin{definition}[Compatible partition]
We say that a partition $\cH_{\ell}$ of $V$ is \emph{compatible} with a total order $\prec$ on $V$ if every set $H\in \cH_{\ell}$ (for $\ell\in [L]$) is a set of vertices appearing consecutively in the order $\prec$.
A hierarchical partition $\cH = (\cH_{\ell})_{\ell=1}^L$ is compatible with $\prec$ if the partition $\cH_{\ell}$ is compatible with $\prec$ for all $\ell \in [L]$.
\end{definition}

Then a hierarchically ordered instance of Hop-ATSP is defined as follows.

\begin{definition}[Hierarchically ordered instance]\label{def:hierarchically_ordered}
A \emph{hierarchically ordered instance} of Hop-ATSP is a tuple $(V, L, \cH, D, \prec, c, k)$, where
\begin{itemize}
\item $(V,c,k)$ is an instance of Hop-ATSP,
\item $L\in \bN$ is the \emph{depth},
\item $(\cH_{\ell})_{\ell=1}^L$ is a hierarchical partition of $V$,
\item $D \in \bZ^{L}$ with $D_{\ell} \geq 2\cdot D_{\ell +1}$ for each $\ell\in[L-1]$,
\item $\prec$ is a total order on $V$ such that $\cH$ is compatible with $\prec$, and
\item for every edge $e=(x,y)\in V\times V$, we have $c(e) \leq D_{\level(e)}$ if $e$ is a forward edge, i.e., $x\prec y$, and we have $c(e)=D_{\level(e)}$, if $e$ is a backward edge, i.e., $y\prec x$.
\end{itemize}
\end{definition}

The goal of this section is to provide a reduction of well-scaled instances of Hop-ATSP to instances of hierarchically ordered instances with depth $L=O(\log n)$, where $n\coloneqq |V|$.
The reduction will have a poly-logarithmic overhead in the approximation ratio and a polynomial overhead in the running time.

\subsection{Directed low-diameter decompositions}\label{sec:ldd}

We will prove \Cref{thm:hierarchically_ordered} by repeatedly applying directed low-diameter decompositions.
Directed low-diameter decompositions are defined for general graphs $G$ with nonnegative cost function on the edges.
In this general setting, it is important to distinguish between the weak diameter and the strong diameter of a subset $S$ of vertices of $G$, where the strong diameter is the diameter of the subgraph $G[S]$ induced by $S$ and the weak diameter is the maximum distance of two vertices from $S$ in the original graph $G$.
(The weak diameter can be smaller than the strong diameter if every shortest path between two vertices in $S$ uses vertices outside of $S$.)
In complete graphs where the edge costs satisfy the triangle inequality, there is no difference between the two diameter notions and the diameter of a graph is the same as the maximum cost of an edge.
Because we will only use low-diameter decompositions in this setting, we won't have to worry about the difference between weak and strong diameter in what follows.

Directed low-diameter decompositions allow us to determine a random set $F$ of edges whose removal ensures that every strongly connected component has bounded (weak) diameter, where the probability of an edge $e$ being included in $F$ depends on $c(e)$.

\begin{definition}[Directed low-diameter decomposition]
\label{def:directed_low-diameter_decomposition}
Let $G=(V,E)$ be a directed graph and $c: E \to \bZ_{\geq 0}$.
Given positive integers $\alpha$ and $D$, a \emph{directed low-diameter decomposition with loss $\alpha$} is a collection of edge sets $\cF \subseteq 2^{E}$ such that
\begin{enumerate}
\item
for all $F \in \cF$, every strongly connected components of $(V, E \setminus F)$ has weak diameter at most  $D$, 
\item
for all $e \in E$ we have $\bP[e \in F] \leq \alpha \cdot \frac{c(e)}{D}$, where $F$ is drawn uniformly at random from $\cF$.
\end{enumerate}
\end{definition}

In our reduction we will compute such a directed low-diameter decomposition for the complete directed graph $(V,E)$ on vertex set $V$ with edge cost $c$, and then sample a set $F\in \cF$.
We will choose the partition $\cH_2$ corresponding to the strongly connected components of $(V, E \setminus F)$ and recurse on each set in $\cH_2$ to refine the partition, leading eventually to the complete hierarchical partition $\cH$.
Edges in $F$ correspond to edges that might become backward edges in our hierarchically ordered instance.

Directed low-diameter decompositions have been developed in \cite{bernstein2022negative} and improved and simplified versions have been given in \cite{bringmann2025near,li2025simpler}.

\begin{theorem}[\cite{bringmann2025near}]
\label{thm:directed_low-diameter_decomposition}
Let $G,c$ and $D$ be given as in \Cref{def:directed_low-diameter_decomposition}.
Suppose that $D \leq \poly (n)$, where $n \coloneqq |V(G)|$.
Then, there is a deterministic algorithm that computes a directed low-diameter decomposition $\cF$ with loss at most $O(\log n \log \log n)$ and $|\cF| \leq \poly(n)$ in polynomial time.
\end{theorem}

\subsection{Proof of \Cref{thm:hierarchically_ordered}}\label{sec:proof_hierarchically_ordered}

We now prove \Cref{thm:hierarchically_ordered}, which we restate here for convenience.

\thmHierarchy*

\noindent
The algorithm we will use to prove \Cref{thm:hierarchically_ordered} starts with the partition $\cH_1=\{V\}$ and iteratively refines it to partitions $\cH_2,\dots, \cH_L =\{\{v\} : v\in V\}$.
Together with these partitions we will construct a total order~$\prec_{\ell}$ on each partition $\cH_{\ell}$, eventually leading to an order$\prec$ on $V$.

The algorithm proceeds a s follows:
\begin{enumerate}[label =(\arabic*)]
\item Let $G$ be the complete directed graph on $V$ and set $\cH_1 \coloneqq \{V\}$ and $D_1 \coloneqq \diam(G,c)$.
 Let $\prec_1$ be the trivial total order on $\cH_1$.
\item Define a sequence of integers by $D_{\ell} \coloneqq \lfloor \frac{D_{\ell -1}}{2} \rfloor$ and let $L\in \bN$ be minimal such that $D_L=0$.
\item For $\ell =2,\dots, L-1$, do the following:
\begin{itemize}
\item Initialize $\cH_{\ell} = \emptyset$.
\item  For every set $H\in \cH_{\ell -1}$, apply \Cref{thm:directed_low-diameter_decomposition} to the induced subgraph $G[H]$ with edge costs $c$ and diameter bound $D_{\ell}$ to obtain a directed low-diameter decomposition $\cF$.
Then sample $F\in \cF$ uniformly at random, let $G_H \coloneqq G[H] - F$ be the graph that results from $G[H]$ by deleting $F$, and fix a topological order of the connected components of $G_H$.
Then add the vertex sets of the connected components of $G_H$ to $\cH_{\ell}$.
\item Define a total order $\prec_{\ell}$ on $\cH_{\ell}$ by defining $X \prec_{\ell} Y$ as follows:
\begin{itemize}
\item If $X$ and $Y$ are contained in different parts $\cH_{\ell -1}$, say $X\subseteq X'\in \cH_{\ell-1}$ and $Y\subseteq Y'\in \cH_{\ell-1}$, then $X \prec_{\ell} Y$ if and only if $X' \prec_{\ell-1} Y'$.
\item If $X$ and $Y$ are contained in the same part $H$ of $\cH_{\ell -1}$, then $X\prec_{\ell}Y$ if and only if and $X$ is before $Y$ in the fixed topological order of the strongly connected components of $G[H]$.
\end{itemize}
\end{itemize}
\item Define $\cH_L = \{ \{v\} : v\in V\}$ to be the partition of $V$ into singletons and define $\prec$ to be any total order on $V$ such that for any distinct sets $X,Y\in \cH_{L-1}$ with $X\prec_{L-1} Y$, we have $x\prec y$ for all $x\in X$ and $y\in Y$.
\end{enumerate}
In this algorithm we defined a total order $\prec_{\ell}$ on the parts of the partition $\cH_{\ell}$ for $\ell=1,\dots,L-1$.
For our analysis it will be useful to also define an order $\prec_L$ on the singleton partition $\cH_L$, which we define as $\{x\} \prec_L \{y\}$ if and only if $x\prec y$.

Having constructed the total order $\prec$ on $V$, the hierarchical partition $\cH=(\cH_{\ell})_{\ell=1}^L$ and the vector $D\in \bZ_{\geq 0}^L$, we define the cost function $\tilde{c} : V \times V \to \mathbb{Z}_{\geq 0}$ by
\[
\tilde{c}(e) = 
\begin{cases}
c(e) &\text{ if $e$ is a forward edge}\\
D_{\level(e)} &\text{ if $e$ is a backward edge}.
\end{cases}
\]

A key invariant of our algorithm (following directly from the definition of directed low-diameter decompositions) is that the graph $G[H]$ with edge cost $c$ has (weak) diameter at most $D_{\ell}$ for every set $H\in \cH_{\ell}$.
Because $G$ is a complete graph satisfying the triangle inequality, this is equivalent to $c(e)\leq D_{\level(e)}$ for each edge $e$ of $G$.
In particular, this also implies $c(e) \leq \tilde{c}(e)$ for every edge $e\in V\times V$.
We also note that since we consider a well-scaled instance, the maximum cost $c(e)$ of an edge $e\in V\times V$ is at most $2n^3$, implying $L\leq \log_2(\diam(G,c)) + 1 = O(\log n)$.

In order to prove that $\cJ =(V,L,\cH,D,\prec,\tilde{c},k)$ is a hierarchically ordered instance, it remains to prove that $\tilde{c}$ satisfies the triangle inequality. 
All other properties of hierarchically ordered instances follow directly from the construction.

\begin{lemma}
The function $\tilde c: V\times V \to \bZ_{\geq 0}$ satisfies the triangle inequality.
\end{lemma}
\begin{proof}
Let $x,y,z \in V$ be three distinct vertices.
If $x \prec z$, then
\begin{equation*}
\tilde{c}(x,z) = c(x,z) \leq c(x,y) + c(y,z) \leq \tilde{c}(x,y) + \tilde{c}(y,z).
\end{equation*}

Otherwise we have $x \succ z$ and thus $\tilde{c}(x,z)=  D_{\ell}$ for $\ell\coloneqq \level(x,y) < L$.
In this case, $H_{\ell +1}(z) \prec_{\ell+1} H_{\ell +1}(x)$, implying that we have $H_{\ell +1}(y) \prec_{\ell+1} H_{\ell+1}(x)$ or $H_{\ell+1}(z) \prec_{\ell+1} H_{\ell+1}(y)$, implying $\tilde{c}(x,y) \geq D_{\ell}$ or $\tilde{c}(y,z) \geq D_{\ell}$.
We conclude $\tilde{c}(x,z) = D_{\ell} \leq \tilde{c}(x,y) + \tilde{c}(y,z)$.
\end{proof}

In order to complete the proof of  \Cref{thm:hierarchically_ordered}, the only thing left to show is 
\begin{equation}\label{eq:expected_opt_increase_hierarchy}
\mathbb{E}[\OPT(\cJ)]\ \leq\ L \cdot O(\log n \cdot \log \log n) \cdot \OPT(\cI),
\end{equation}
where $\cI=(V,c,k)$ is the given well-scaled instance of Hop-ATSP.
The following is the key lemma we use to show this property.

\begin{lemma}\label{lem:expected_increase_edge_length}
For every edge $e \in V\times V$, we have
\[
\bE \big[ \tilde{c}(e) \big] \leq L \cdot O(\log n \cdot \log\log n) \cdot c(e) + 1.
\]
\end{lemma}
\begin{proof}
Let $B_{\ell}$ be the event that $e$ is a backward edge and $\level(e)=\ell$.
Using $D_L=0$ and $D_{L-1} =1$, we obtain
\begin{align*}
\bE\big[ \tilde{c}(e) \big] \ =&\ \bP\big[ e\text{ forward edge}\big]\cdot c(e) + \sum_{\ell=1}^{L} \bP[B_{\ell}\big]  \cdot D_{\ell}
\ \leq\ c(e) + 1 + \sum_{\ell=1}^{L-2} \bP\big[B_{\ell}\big]  \cdot D_{\ell}.
\end{align*}
Consider $\ell \leq L-2$ and suppose that $\level(e) \geq \ell$.
Then the two endpoints of $e$ are contained in the same set $H\in \cH_{\ell}$.
In our algorithm we apply \Cref{thm:directed_low-diameter_decomposition} to the graph $G[H]$ with edge costs $c$ and diameter bound $D_{\ell +1}$ to obtain a directed low-diameter decomposition $\cF$, and sample $F\in \cF$ uniformly at random.
By our definition of $\prec_{\ell + 1}$ and the fixed topological order on the connected components of $G[H] - F$, the event $B_{\ell}$ happens only if $e\in F$.
By \Cref{thm:directed_low-diameter_decomposition}, this happens with probability at most $\frac{c(e)}{D_{\ell + 1}} \cdot O(\log n \log \log n)$, and therefore $\bP\big[B_{\ell} \mid \level(e) \geq \ell \big] \leq \frac{c(e)}{D_{\ell + 1}} \cdot O(\log n \log \log n)$.
Using that $B_{\ell}$ can only happen if $\level(e) \geq \ell$, we conclude that for all $\ell \leq L-2$, 
\begin{equation}
\bP\big[B_{\ell}\big] \ \leq\ \bP\big[B_{\ell} \mid \level(e) \geq \ell \big] 
\ \leq\ \frac{c(e)}{D_{\ell + 1}} \cdot O(\log n \log \log n).
\end{equation}
Since $\frac{D_{\ell}}{D_{\ell + 1}} = \frac{D_{\ell}}{\lfloor \frac{D_{\ell}}{2} \rfloor} \leq 4$ for all $\ell \in [L-2]$, we obtain
\[
\bE\big[ \tilde{c}(e) \big] \ \leq\ c(e) + 1 + \sum_{\ell=1}^{L-2} c(e) \cdot O(\log n \log \log n)\ \le\ 1 + L \cdot O(\log n \log \log n) \cdot c(e).
\qedhere
\]
\end{proof}

Applying \Cref{lem:expected_increase_edge_length} to each of the $kn$ edges contributing to the $k$-hop cost of a fixed optimum solution to the given Hop-ATSP instance~$\cI$, linearity of expectation implies
\[
\bE \big[ \OPT(\cJ) \big]\  \leq \ L \cdot O(\log n \cdot \log\log n) \cdot \OPT(\cI)+ kn,
\]
where $n = |V|$ is the number of vertices.
Because the instance $\cI$ is well-scaled, we have $\OPT(\cI) \geq n^2$, which implies \eqref{eq:expected_opt_increase_hierarchy}.
This completes the proof of \Cref{thm:hierarchically_ordered}.

Hence, in the following, we will focus on hierarchically ordered instances.
In this context, we will use the following terminology:

\begin{definition}[Left/right]\label{def:left_right}
We say that a vertex $a$ is left of a vertex $b$ if $a\prec b$, and $a$ is right of $b$ if $b\prec a$.
A vertex $a$ is left/right of a set $B\subseteq V$ if $a$ is left/right of every element of $B$.
If $a$ is left or right of $B$, we also write $a \prec B$ or $B\prec a$, respectively.
Moreover, a vertex set $A$ is left/right of a vertex set $B$ if every element of $A$ is left/right of every element of $B$.
If $A$ is left or right of $B$, we also write $A \prec B$ or $B\prec A$, respectively.
\end{definition}

\section{Reducing to a Covering Problem}\label{sec:reducing_to_covering}

In this section we provide a reduction from well-scaled instances of Hop-ATSP to a covering problem (\Cref{problem:covering}).
We first introduce the covering problem (\Cref{sec:our_covering_problem}) and state the overall result of the section (\Cref{thm:covering_reduction}).
Next, we show that every solution to the covering problem can be turned into a Hop-ATSP solution without significant increase in the objective value (\Cref{sec:covers_to_tours}).
The reverse direction is unfortunately not true in general, but we show that it is true for non-degenerate instances (\Cref{sec:useful_tools_covering,sec:tours_to_covers,sec:splitting,sec:nonsplitting,sec:complte_cover_reduction}) and we can reduce general instances to non-degenerate ones (\Cref{sec:degenerate}).

\subsection{Our Covering Problem}\label{sec:our_covering_problem}

Recall that a path is called monotone if it contains only forward edges.
Moreover, a \emph{path-level pair} is a pair $(P,\ell)$ where $P$ is a monotone path,  $\ell\in [L]$, and $V(P) \subseteq H$ for some $H\in \cH_{\ell}$.
Then the covering problem to which we reduce is defined as follows:

\CoveringProblem*

If a pair $(P,j)$ with $j<\ell$ satisfies \ref{item:responsibility} for a set $H\in \cH_{\ell}$, we also say that the pair \emph{takes responsibility} for the set $H$.
We remark that, technically speaking, this is not well-defined in case the same set $H$ appears in multiple partitions $\cH_{\ell}$ for different levels $\ell$.
To guarantee uniqueness, one could extend this definition by representing elements of $\cH_{\ell}$ explicitly as pairs $(H,\ell)$.
However, to avoid cumbersome notation, we will assume that $H \subseteq V $ implicitly carries its level information (which will always be clear from the context).

The goal of this section is to provide a reduction from Hop-ATSP to \Cref{problem:covering}.
Let $\gamma > 0$ be a constant such that $L \leq \gamma \cdot \log_2 n$ for every instance $\cI =(V, L, \cH, D, \prec, c, k)$ resulting from \Cref{thm:hierarchically_ordered}.
We say that a hierarchically ordered instance $\cI =(V, L, \cH, D, \prec, c, k)$ of Hop-ATSP has \emph{low depth} if  $L \leq \gamma \cdot \log_2 n$.
Then the reduction we provide is the following:

\begin{restatable}{theorem}{CoveringReduction}\label{thm:covering_reduction}
Let $\alpha : \bN \to \bR$ and $t: \bN \to \bN$ be monotonically increasing.
	Suppose we have an algorithm that computes for every instance $\cI =(V, L, \cH, D, \prec, c, k)$ of Path Covering with low depth, in time $t(n)$ a solution $\cP$
 with $w(\cP) \leq \alpha(n) \cdot \OPT(\cI)$, where $n$ denotes the number of vertices in $\cI$.
 
	Then there is a randomized algorithm that computes for every well-scaled instance $\cJ=(V,c,k)$ of Hop-ATSP in time $O(t(n)) + \poly(n)$ a tour $T$ such that in expectation $c^{(k)}(T)\leq O(\log^6 n) \cdot \alpha(n) \cdot \OPT(\cJ)$, where $n$ denotes the number of vertices in $\cJ$ and $\poly(n)$ denotes some fixed polynomial.
\end{restatable}

It will sometimes be convenient to view \Cref{problem:covering} as an instance of Set Cover.
For a hierarchically ordered instance $\cI =(V, L, \cH, D, \prec, c, k)$ of Hop-ATSP, we define the \emph{corresponding Set Cover instance} $(\cU, \cS, w)$ by defining
\begin{itemize}
\item $\cU \coloneqq V\times [L]$,
\item $\cS$ to be the set of all path-level pairs, 
\item $w(P,\ell) \coloneqq c^{(k)}(P) + k^2 \cdot D_{\ell}$ for all $(P,\ell)\in S$,
\end{itemize}
and saying that a path-level pair $(P,j)$ covers an element $(v,\ell)$ if
\begin{itemize}
\item $v\in V(P)$ and $j \leq \ell$, or
\item $(P,j)$ takes responsibility for $H_{\ell+1}(v)$, i.e, we have $|V(P)\cap H_{\ell+1}(v)| \geq k$ and $ j \leq \ell < L$.
\end{itemize}

Then our Path Covering problem is indeed equivalent to the corresponding Set Cover problem, as the following lemma shows:

\begin{lemma}\label{eq:equivalence_to_set_cover}
Let $\cI =(V, L, \cH, D, \prec, c, k)$ be a hierarchically ordered instance of Hop-ATSP and $(\cU, \cS, w)$ the corresponding Set Cover instance.
Then $\cP \subseteq \cS$ is a feasible solution to the Path Covering instance $\cI$ if and only if every element of $\cU$ is covered by some pair $P\in \cP$.
\end{lemma}
\begin{proof}
First consider a set $\cP \subseteq \cS$ covering each element of $\cU$.
We show that $\cP$ is a feasible solution to the Path Covering problem.
To show \ref{item:pseudo_cover_condition}, let $v\in V$ and let $(P,j)$ be a pair covering the element $(v,L)\in \cU$.
Then by the definition of the Set Cover instance $(\cU, \cS, w)$, we must have $v\in V(P)$.
To show \ref{item:overview_extra_covering_condition}, let $\ell \in \{2,\dots, L\}$ and $H\in \cH_{\ell}$.
Suppose \ref{item:direct_covering} is not satisfied, i.e., there exists a vertex $v\in H$ such that $j \geq \ell$ for every pair $(P,j)\in \cP$ with $v\in V(P)$.
Let $(P,j)\in \cP$ be the pair covering the element $(v,\ell-1)\in \cU$.
Then $(P,j)$ takes responsibility for $H$.

Now consider a feasible solution $\cP$ to the Path Covering problem.
We show that $\cP$ covers $\cU$.
To this end, let $(v,\ell) \in \cU$.
 If $\ell = L$, then any pair $(P,j)\in \cP$ with $v\in V(P)$ covers $(v,\ell)$ and such a pair $(P,j)$ exists by \ref{item:pseudo_cover_condition}.
Otherwise, since \ref{item:responsibility} applies to $H_{\ell +1}(v)$, there must be either a pair $(P,j)\in \cP$ with $j\leq \ell$ and $v\in V(P)$, or a pair $(P,j)\in \cP$ that takes responsibility for $H_{\ell +1}(v)$. 
In both cases the pair $(P,j)$ covers the element $(v,\ell)$ of $\cU$.
\end{proof}

\subsection{Turning Covering Solutions into Tours}\label{sec:covers_to_tours}

In order to prove \Cref{thm:covering_reduction}, we will first apply \Cref{thm:hierarchically_ordered} to obtain a hierarchically ordered instance.
Having a hierarchically ordered instance $\cI=(V, L, \cH, D, \prec, c, k)$, we show that any solution $\cP$ to our covering problem can be turned into a tour $T$ of cost $c^{(k)}(T) \leq 2 w(\cP)$.
The following lemma states a key argument for this first part of our proof.
Condition~\ref{item:overview_extra_covering_condition} in our definition of the Path Covering problem plays a crucial role here.
(Recall the example from \Cref{fig:example_covering_needs_complicated_condition} to see why this is necessary.)

\begin{lemma}
\label{lem:walk_insertion}
Let $W_{\rm{main}} = w_1, \dots, w_{t}$ be a walk on $V$ and $W_{\rm{sub}} = w_{s+1}, \dots, w_{s+k}$ a subwalk of $W_{\rm{main}}$. 
Suppose there is a set $H \in \cH_{\ell}$ such that $V(W_{\rm{sub}}) \subseteq H$ and let $Q=q_1, \dots, q_m$ be another walk with vertex set $V(Q) \subseteq H$.
Define the walk 
\begin{equation*}
W^\prime_{\rm{main}} \coloneqq w_1, \dots, w_{s}, q_1, \dots, q_m, w_{s+1}, \dots, w_{s+k}, \dots, w_t
\end{equation*}
that results from $W_{\rm{main}}$ by inserting $Q$ right before the start of $W_{\rm{sub}}$.
Then we have
\begin{equation*}
c^{(k)}(W^\prime_{\rm{main}}) \leq c^{(k)}(W_{\rm{main}}) + c^{(k)}(Q) + 2k^2D_{\ell}.
\end{equation*} 
\end{lemma}
\begin{proof}
For a walk $W$ with vertices $u_1,\dots, u_r$ we write $C(W,u_i) \coloneqq \sum_{\Delta =1}^k c(u_i, u_{i+\Delta})$ to denote the total cost of the outgoing edges of $v_i$ that contribute to the $k$-hop cost of the walk $W$. (Here, $c(u_i, u_{i+\Delta})$ is defined as zero in case $i+\Delta > r$.)
Then $c^{(k)}(W) = \sum_{i=1}^r  C(W,u_i)$.

To prove the lemma, we therefore need to prove an upper bound on
\begin{align*}
& c^{(k)}(W^\prime_{\rm{main}}) - \Big(c^{(k)}(W_{\rm{main}}) + c^{(k)}(Q)\Big) \\
& = \sum_{i=1}^t \Big( C(W^\prime_{\rm main},w_i) -C(W_{\rm main},w_i)\Big) + \sum_{j=1}^m \Big( C(W^\prime_{\rm main},q_j) - C(Q,q_j) \Big).
\end{align*}
We first observe that $C(W^\prime_{\rm main},w_i) - C(W_{\rm main},w_i) = 0$ for all  $i\in \{1,\dots, s-k\}\cup \{s+1,\dots, t\}$ because for each such vertex $w_i$, the $k$ vertices following on $w_i$ in the walk did not change when inserting $Q$.
The same argument shows that $C(W^\prime_{\rm main},q_j) - C(Q,q_j) = 0$ for all $j\in \{1, \dots, m - k\}$.
To complete the proof, we will show that each of the at most $2k$ remaining summands is at most $k\cdot D_{\ell}$.
\addparskip

Let $j\in \{m-k+1, \dots, m\}$. To prove that $C(W^\prime_{\rm main},q_j) - C(Q,q_j) \leq k\cdot D_{\ell}$, we observe that the $k$ vertices following on $q_j$ in the walk $W^\prime_{\rm main}$ are all contained in the set $H\in \cH_{\ell}$ that contains $q_j$.
Thus, each of the $k$ edges contributing to $C(W^\prime_{\rm main},q_j)$ has cost at most $D_{\ell}$, implying $C(W^\prime_{\rm main},q_j) \leq k\cdot D_{\ell}$.
Because $C(Q,q_j) \geq 0$, we can conclude that indeed $C(W^\prime_{\rm main},q_j) - C(Q,q_j) \leq k\cdot D_{\ell}$.
\addparskip

Now, let $i\in \{s-k+1,\dots, s\}$.
We will prove $C(W^\prime_{\rm main},w_{i}) -C(W_{\rm main},w_i) \leq k\cdot D_{\ell}$.
Let $v_1,\dots, v_k$ be the $k$ vertices following on $w_i$ in the walk $W_{\rm main}$ (i.e., before insertion of $Q$) and let $v^\prime_1,\dots, v^\prime_k$ be the $k$ vertices following on $w_i$ in the walk $W^\prime_{\rm main}$ (i.e., after insertion of $Q$).
It suffices to show $c(w_i, v_j^\prime) - c(w_i, v_j) \leq D_{\ell}$ for all $j\in [k]$.
First, we observe
\[
 v^{\prime}_1 = v_1 = w_{i+1}, \quad  v^{\prime}_2 = v_2 = w_{i+2}, \quad  \dots \quad v^{\prime}_{s-i} = v_{s-i} =  w_s
\]
and thus  $c(w_i, v_j^\prime) - c(w_i, v_j) = 0 \leq D_{\ell}$ for all $j\in \{1,\dots, s-i\}$.

Next, we observe that the vertices $v_{s-i+1}, \dots, v_k$ are all contained in $\{w_{s+1},\dots,w_{s+k} \} \subseteq H$ and the vertices $v_{s-i+1}^\prime, \dots, v_k^\prime$ are all contained in $\{q_1,\dots, q_m, w_{s+1},\dots,w_{s+k} \} \subseteq H$.
Thus, for all $j \in \{s-i+1, \dots k\}$, we have $v_j, v^\prime_j \in H\in \cH_{\ell}$, implying 
\[
c(w_i, v_j^\prime) \leq c(w_i, v_j)  + c(v_j, v^\prime_j) \leq  c(w_i, v_j)  + D_{\ell},
\]
and therefore $c(w_i, v_j^\prime) - c(w_i, v_j) \leq D_{\ell}$.
\end{proof}

Using \Cref{lem:walk_insertion}, we can now show how we can construct a sufficiently cheap tour from a given covering solution:
\begin{lemma}
\label{lem:tour_from_set_cover}
Given a Path Covering solution $\cP$, we can construct a tour $T$ on $V$ with $c^{(k)} (T) \leq 2w(\cP)$ in polynomial time (in $|V|$ and $|\cP|$).
\end{lemma}
\begin{proof}
We may assume that $\cP$ is inclusionwise minimal.
For every pair $(P,\ell)\in \cP$ with $\ell > 1$, there is a unique set $H\in \cH_{\ell}$ containing the vertex set $V(P)$ of the path $P$.
Because $\cP$ is inclusionwise minimal, property~\ref{item:direct_covering} from \Cref{problem:covering} cannot apply to $H$. (Otherwise, $\cP\setminus \{ (P,\ell)\}$ would also be a feasible Path Covering solution.)
Therefore, property~\ref{item:responsibility} must apply, i.e., there exists a pair $(P',j)\in \cP$ that takes responsibility for $H$.
We then define ${\rm parent}(P,\ell) \coloneqq (P',j)$.
There might be several pairs taking responsibility for $(P,\ell)$, in which case we choose an arbitrary one as the parent of $(P,\ell)$.
This defines a branching $B$ on the node set $\cP$ with arcs $\big({\rm parent}(P,\ell), (P,\ell)\big)$ for all pairs $(P,\ell)\in \cP$ with $\ell > 1$.
Every connected component of this branching is an arborescence with a root $(P,1)\in \cP$.
\addparskip

Throughout our algorithm, we will maintain a branching with node set $\cW$, where initially $\cW=\cP$.
We will maintain the following invariants:
\begin{enumerate}
\item\label{item:merging_pair_invariant} Every node is a pair $(W,\ell)$, where $\ell\in[L]$ and $W$ is a walk completely contained in some set $H\in \cH_{\ell}$.
\item\label{item:merging_root_invariant} A pair $(W,\ell)\in \cW$ has an incoming arc if and only if $\ell > 1$.
\item\label{item:branching_arc_invariant} For every pair $(Q,\ell)\in \cW$ with parent $(W,j)$, we have $j < \ell$ and the walk $W$ visits at least $k$ vertices from $H$ consecutively, where $H$ is the unique set $H\in \cH_{\ell}$ containing the vertices of $Q$.
\item\label{item:invariant_visit_all_vertices} Every vertex $v\in V$ is contained in at least one walk $W$ with $(W,\ell)\in \cW$.
\end{enumerate}
We start with the branching $B$ defined above, which satisfies all these invariants.
As long as our branching has at least one arc, we consider a pair $(Q,\ell)$ with maximum level $\ell$.
Then $(Q,\ell)$ is a leaf of our branching and it has an incoming arc.
Let $(W,j)$ be the parent of $(Q,\ell)$.
By invariant~\ref{item:branching_arc_invariant}, we can insert $Q$ into $W$ by applying \Cref{lem:walk_insertion}, obtaining a walk $W'$.
We delete the leaf $(Q,\ell)$ and its incoming arc from our branching and replace the node $(W,j)$ by $(W',j)$, maintaining its incident arcs.
\addparskip

We now show that this operation maintains our invariants.
For invariants~\ref{item:merging_root_invariant} and~\ref{item:invariant_visit_all_vertices} this is obvious from the construction.
For invariant~\ref{item:merging_pair_invariant}, let $H_{j}\in \cH_j$ such that all vertices of $W$ are contained in $H_j$, and let $H_{\ell}\in \cH_{\ell}$ such that all vertices of $Q$ are contained in $H_{\ell}$.
Because $W$ contains vertices from $H_{\ell}$ and $j < \ell$, the hierarchical structure of our instance implies $H_{\ell} \subseteq H_j$.
In particular, all vertices of the new walk $W'$ are contained in~$H_{j}$.

Invariant~\ref{item:branching_arc_invariant} could be possibly violated only for outgoing arcs of the new node $(W',j)$.
Consider a child $(\tilde{Q}, \tilde{\ell})$ of $(W',j)$.
Then our choice of the leaf $(Q,\ell)$ implies $\tilde{\ell} \leq \ell$.
Therefore, the unique set $\tilde{H} \in \cH_{\tilde{\ell}}$ containing all vertices from $\tilde{Q}$ is either disjoint from the set $H\in \cH_{\ell}$ containing all vertices of $Q$, or it is a superset of $H$.
If $\tilde{H}$ is disjoint from $H$, any $k$ vertices from $\tilde{H}$ that were previously consecutive in $W$ are still consecutive in $W'$.
If $\tilde{H}$ is a superset of $H$, we observe that $W'$ still contains $k$ consecutive vertices from $H$, and thus also from $\tilde{H}$.
\addparskip

\Cref{lem:walk_insertion} implies that in every iteration $c^{(k)}(W') \leq c^{(k)}(W) + c^{(k)}(Q) + 2 k^2 \cdot D_{\ell}$,
and hence the following quantity is nonincreasing during the algorithm:
\[
\sum_{(W,\ell) \in \cW} \Big( c^{(k)}(W)  + 2k^2 \cdot D_{\ell} \Big)
\]
Initially, this value is bounded by $2 w(\cP)$ and hence this will remain the case throughout the algorithm.
\addparskip

After less than $|\cP|$ iterations, our branching has no edges anymore.
Then the branching has vertices $(W_1,1),(W_2,1),\dots,(W_r,1)$ for some $r\in \bN$, and we have
$
\sum_{i=1}^r \big( c^{(k)}(W_i)  + k^2 \cdot D_1 \big)\ \leq 2 \cdot w(\cP)
$.
We concatenate the walks $W_1,\dots, W_r$ to form a closed walk $T$.
By invariant~\ref{item:invariant_visit_all_vertices}, this is indeed a tour.
Because $c(x,y) \leq D_1$ for all $x,y\in V$, we have
\[
c^{(k)}(T) \ \leq\ \sum_{i=1}^r \Big( c^{(k)}(W_i)  +  k^2 \cdot D_1 \Big)\ \leq 2 \cdot w(\cP).\qedhere
\]
\end{proof}

\subsection{Degenerate Instances}\label{sec:degenerate}

Ideally, we would like to show that every tour $T$ can be transformed into a Covering Solution $\cP$ so that the weight $w(\cP)$ is not much larger than the $k$-hop cost $c^{(k)}(T)$.
Note however, that $\cP$ will contain at least one pair $(P,\ell)$ with level $\ell=1$ and we therefore have $w(\cP)\geq k^2D_1$.
In general the total cost $c^{(k)}(T)$ may be much smaller than $k^2D_1$.
To address this issue, we observe that it can only occur for particular instances, which we call \emph{degenerate}.
We will then provide a reduction from general instances of Path Covering to non-degenerate ones.

\begin{definition}[Degenerate instance]
We say that a hierarchically ordered instance $\cI =(V, L, \cH, D, \prec, \allowbreak c, k)$ of Hop-ATSP is \emph{degenerate} if there exists a set $H \in \cH_2$ with $|V \setminus H| \leq \frac{k}{2}$.
Otherwise, it is called \emph{non-degenerate}.
\end{definition}

We now show that for non-degenerate instances, we indeed have $c^{(k)}(T) = \Omega(k^2 D_1)$.

\begin{lemma}
\label{lem:first_level_fixed_cost_lb}
Let $\cI$ be a non-degenerate hierarchically ordered instance of Hop-ATSP.
Then  we have 
\[
c^{(k)}(T) \geq 2^{-8}\cdot k^2 D_1
\]
 for every tour $T$ on $V.$
\end{lemma}
\begin{proof}
First, we claim that there exists a partition of $V$ into sets $A$ and $B$ such that 
\begin{itemize}
\item $|A|\geq \frac{k}{4}$ and $|B| \geq \frac{k}{4}$, and
\item $c(y,x) = D_1$ for all $x\in A$ and $y\in B$.
\end{itemize}
To prove this claim, we distinguish two cases.
First, consider the case where we have $|H| \geq \frac{k}{4}$ for some $H \in \cH_2$.
We define 
\begin{align*}
V_{\rm left} \coloneqq&\  \{ v\in V: v\text{ is left of }H\} \\
V_{\rm right} \coloneqq&\ \{ v\in V: v\text{ is right of }H\} .
\end{align*}
(Recall the notion of left/right from \Cref{def:left_right}.)
Because $H\in \cH_2$, every edge from $V_{\rm right}$ to $H$ and every edge from $H$ to $V_{\rm left}$ is a backward edge of level~$1$ and has therefore cost $D_1$.
Because $V = V_{\rm left} \cup H \cup V_{\rm right}$ and the instance $\cI$ is non-degenerate, we have $|V_{\rm left}| \geq \frac{k}{4}$ or $|V_{\rm right}| \geq \frac{k}{4}$.
In the former case, we can choose $A \coloneqq V_{\rm left}$ and $B\coloneqq H \cup V_{\rm right}$, and in the latter case, we can choose $A \coloneqq V_{\rm left} \cup H$ and $B\coloneqq V_{\rm right}$.

Now consider the remaining case, where $|H| < \frac{k}{4}$ for all $H \in \cH_2$.
We number the elements of $\cH_2 = \parset{H_1, \dots, H_m}$ so that $H_j$ is left of $H_{j+1}$ for all $j \in[m-1]$, and we define 
\begin{equation*}
j^\ast \coloneqq \min \left\{j \in [m] : \sum_{i = 1}^j |H_{i}| \geq \tfrac{k}{2} \right\}.
\end{equation*}
Then we have $\sum_{j=1}^{j^\ast} |H_j| < \frac{3}{4}k $ because $|H_{j^\ast}| < \frac{k}{4}$.
Therefore, choosing $A \coloneqq \bigcup_{i=1}^{j^\ast} H_i$ and $B\coloneqq V\setminus A$, we have $|A|\geq \frac{k}{4}$ and $|B| \geq k - \frac{3}{4}k = \frac{k}{4}$.
Moreover, every edge starting in $B$ and ending in $A$ is a backward edge of level $1$.
This completes the proof of our claim.
\addparskip

Now let $A, B$ be a partition of $V$ as in our claim and let $T$ be an arbitrary tour on $V$.
By \Cref{lem:khop_skipping_vertices}, we can assume that $T$ is a Hamiltonian cycle.
Let $P$ be a path obtained from $T$ by removing a single edge.
We may assume $k\geq 16$ because otherwise $2^{-8}\cdot k^2 \leq 1$, in which case the statement of the lemma is trivially satisfied (because every tour contains at least one backward edge of level~$1$).
Then we can find  vertex disjoint subpaths $P_1, \dots, P_m$ of $P$ with $\bigcup_{j=1}^m V(P_j) = V(P)$ and $\frac{k}{8} \leq |V(P_j)| \leq \frac{k}{2}$ for all $j \in [m]$ (cf.~\Cref{lem:block_decomp}).
Without loss of generality, we can assume that the subpaths are numbered according to their order along the path $P$.
Then, because each subpath $P_j$ contains at most $\frac{k}{2}$ vertices, the cost of every edge $(x,y)$ with $x\in P_j$ and $y\in P_{j+1}$ contributes to the $k$-hop cost of the tour $T$. (Here, we define $P_{m+1} \coloneqq P_1$.)
In particular, we have
\begin{equation*}
c^{(k)}(T) \ \geq\ \sum_{i=1}^m \ \sum_{x \in V(P_i) } \ \sum_{y \in V(P_{i+1})} c(x,y),
\end{equation*}
where $P_{m+1} \coloneqq P_1$.

Let $j \in [m]$.
We say that the subpath $P_j$ is \emph{favored} if $P_j$ contains at least $\frac{k}{16}$ vertices from $A$.
Otherwise we call $P_j$ \emph{unfavored}.
If $P_j$ is unfavored, we have $|B \cap V(P_j)| \geq \frac{k}{16}$,  because $|V(P_j)| \geq \frac{k}{8}$.
\addparskip

We now distinguish several cases.
If each subpath $P_j$ is favored, then we claim that among the $k$ vertices following on $y\in B$ in the tour $T$, there are at least $\frac{k}{16}$ vertices from $A$.
Indeed, the vertex $y$ is contained in some subpath $P_j$ and all vertices from $P_{j+1}$ are among the $k$ vertices following on $y$ in the tour $T$.
Because $P_{j+1}$ is favored, at least $\frac{k}{16}$ of its vertices belong to $A$.
This shows that the $k$ outgoing edges of $y$ that contribute to the $k$-hop cost of $T$ have total cost at least $\frac{k}{16}\cdot D_1$.
Summing this over all vertices $y\in B$, we get $c^{(k)}(T) \geq |B| \cdot \frac{k}{16} D_1 \geq 2^{-6} \cdot k^2 D_1$.

If each subpath $P_j$ is unfavored, we can apply a symmetric argument to see that among the $k$ vertices preceding a vertex $x\in A$ in the tour $T$, there are at least $\frac{k}{16}$ vertices from $B$ and deduce $c^{(k)}(T) \geq 2^{-6} \cdot k^2 D_1$ also in this case.

It remains to consider the case where at least one subpath is favored and one is unfavored.
Then, there exists an index $j \in [m]$ such that $P_j$ is unfavored and $P_{j+1}$ is favored.
We obtain
\begin{equation*}
c^{(k)}(T)\ \geq\ \sum_{i=1}^m \ \sum_{x \in V(P_i) } \ \sum_{y \in V(P_{i+1})} c(x,y)
\ \geq\
 \sum_{\substack{ x \in V(P_j), \\ x \in B}} \ \sum_{\substack{ y \in V(P_{j+1}), \\ y \in A}} c(x,y) \geq \left(\tfrac{k}{16} \right)^2 D_1.
\end{equation*}
\end{proof}

Next, we reduce general instances to non-degenerate ones.
In case we have a degenerate instance, we will iteratively remove vertices and shrink our hierarchical partition, until either the resulting instance is non-degenerate, or at most $k$ vertices are remaining.

\begin{lemma}\label{lem:reduction_to_non_degenerate}
Let $\alpha : \bN \to \bR$ and $t: \bN \to \bN$ be monotonically increasing.
Suppose we have an algorithm that computes for every non-degenerate hierarchically ordered instance $\cJ$ of Hop-ATSP with low depth, in time $t(n)$ a tour $T$
 with $c^{(k)}(T) \leq \alpha(n) \cdot \OPT(\cI)$, where $n$ denotes the number of vertices in $\cI$.
 
Then there is an algorithm that computes for  every hierarchically ordered  instance $\cI$ of Hop-ATSP with low depth, in time $t(n)+O(nL)$ a tour $T$ with $c^{(k)}(T) \leq (\alpha(n) + 4L ) \cdot \OPT(\cI)$, where $n$ denotes the number of vertices in $\cI$.
\end{lemma}
\begin{proof}
We prove the statement by induction on the depth $L\geq 2$ of the instance $\cI = (V, L, \cH, D, \prec, c, k)$.
For $L=2$, no non-degenerate instances exist because $\cH_1 =\{V\}$, $\cH_2 = \cH_L =\{ \{v\} : v\in V\}$ and $1 \leq k < |V|$.
Hence, the statement is trivially satisfied for $L=2$.
Now, suppose $L>2$.
If we are given a degenerate hierarchically ordered instance $\cI=(V, L, \cH, D, \prec, c, k)$ of Hop-ATSP, let $H^\ast \in \cH_2$ with $|V \setminus H^\ast| \leq \frac{k}{2}$.
We distinguish two cases.
\addparskip

First, consider the case $|H^\ast| > k$.
Then we construct an  instance $\cI^\prime=(V^\prime, L^\prime, \cH^\prime, D^\prime, \prec, c, k)$ with $L^\prime=L-1$ by setting
$V^\prime \coloneqq H^\ast$, and
\begin{align*}
\cH^\prime_{\ell} \coloneqq&\ \big\{ H \in \cH_{\ell+1} : H \subseteq H^\ast \big\}\\
D^\prime_{\ell} \coloneqq&\ D_{\ell+1}
\end{align*}
for all $\ell \in [L^\prime]$.
By the induction hypothesis applied to this instance $\cI^\prime$, we obtain a tour $T'$ on $V'=H^\ast$ of $k$-hop cost
$c^{(k)}(T^\prime) \leq \big((\alpha(|V^\prime|) + 4L^\prime\big)\cdot \OPT(\cI^\prime) $.
Inserting each vertex from $V\setminus H^\ast$ into $T'$ at an arbitrary position, yields a tour $T$ on $V$ with
\begin{align*}
c^{(k)}(T)\ \leq&\ c^{(k)}(T^\prime) + |V\setminus H^\ast| \cdot 2k D_1\\
 \leq&\ \big(\alpha(|V^\prime|) + 4L^\prime \big) \cdot \OPT(\cI^\prime) +   |V\setminus H^\ast| \cdot 2k D_1\\
 \leq&\ \big(\alpha(n) + 4(L-1)\big) \cdot \OPT(\cI) +   |V\setminus H^\ast| \cdot 2k D_1,
\end{align*}
where we used $\OPT(\cI^\prime) \leq \OPT(\cI)$ by \Cref{lem:khop_skipping_vertices}.
To complete the proof in the first case, it remains to show $|V\setminus H^\ast| \cdot k D_1\leq 2 \cdot \OPT(\cI)$.

Consider a vertex $w \in V \setminus H^\ast $ and assume that $w$ is left of $H^\ast$.
(The case with $w$ right of $H^\ast$ is symmetric.)
Because there are at most $\frac{k}{2}$ vertices in $V\setminus H^\ast$, at least $\frac{k}{2}$ of the incoming edges of $w$ that contribute to the $k$-hop costs of an optimal tour on $V$ are edges of the form $(v,w)$ with $v\in H^\ast$.
All these edges are backward edges of level $1$ and therefore have cost $D_1$.
Because all these edges have $w$ as one endpoint and the other endpoint in $H^\ast$, these $\frac{k}{2}$ edges are different ones for each vertex $w \in V \setminus H^\ast $.
This shows $|V\setminus H^\ast| \cdot k D_1\leq 2 \cdot \OPT(\cI)$ and completes the proof in the case  $|H^\ast| > k$.

\addparskip
Now consider the remaining case $|H^\ast| \leq k$.
We choose a set $R \subseteq V\setminus H^\ast$ such that  $|V\setminus R|= k+1$. 
By \Cref{lem:khop_skipping_vertices}, the optimal value $\OPT(V\setminus R, c, k)$ of the Hop-ATSP instance $(V\setminus R,c, k)$ is at most $\OPT(\cI)$ and is attained by a Hamiltonian cycle on $V\setminus R$.
Because $|V\setminus R|= k+1$, any such cycle has the same $k$-hop cost (namely the sum over all costs $c(x,y)$ with $x,y\in V\setminus R$) and hence we can compute an optimal solution $T^\prime$ to the Hop-ATSP instance $(V\setminus R,c, k)$.
Inserting each vertex from $R$ into the tour $T^\prime$, we obtain a tour $T$ on $V$ with 
\[
c^{(k)}(T)\ \leq\ c^{(k)}(T^\prime) + |R| \cdot 2 k D_1
\]
and by the same argument as in the first case,
we have $|R| \cdot k D_1\leq 2 \cdot \OPT(\cI)$.
\end{proof}

For non-degenerate instances, we will be able to show that a Path Covering solution cannot be much more expensive than the $k$-hop cost of a tour:

\begin{restatable}{lemma}{TourToCover}\label{lem:tours_to_covers}
Let $\cI= (V, L, \cH, D, \prec, c, k)$ be a non-degenerate instance of Hop-ATSP and $T$ a tour on $V$.
Then there exists a Path Covering  $\cP$ solution for $\cI$ with
\[
w(\cP) \ \leq\ O(L^3) \cdot c^{(k)}(T).
\]
\end{restatable}

Using \Cref{lem:tours_to_covers}, we can then prove \Cref{thm:covering_reduction}, which we restate here for convenience:

\CoveringReduction*
\begin{proof}
Given a well-scaled instance $\cJ$ of Hop-ATSP, we first apply \Cref{thm:hierarchically_ordered} to sample a hierarchically ordered instance  $\cI$ of Hop-ATSP with low depth such that in expectation $\mathbb{E}[\OPT(\cI)] \leq L \cdot O(\log n \cdot \log \log n) \cdot \OPT(\cJ)$.

By \Cref{lem:reduction_to_non_degenerate}, it suffices to show that there is an algorithm that computes for every non-degenerate hierarchically ordered instance $\cI$ of Hop-ATSP with low depth and $n$ vertices, in time $t(n)$ a tour $T$  with $c^{(k)}(T) \leq O(\log^3 n)\cdot \alpha(n) \cdot \OPT(\cI)$.
\addparskip

Let $\cI =(V, L, \cH, D, \prec, c, k)$ be a non-degenerate hierarchically ordered instance of Hop-ATSP with $n$ vertices and low depth.
By \Cref{lem:tours_to_covers}, the weight $w(\cP^\ast)$ of a minimum-weight Path Covering solution for $\cI$ satisfies $w(\cP^\ast) \leq O(L^3)\cdot  c^{(k)}(T^\ast)$, where $T^\ast$ is a tour on $V$ that minimizes the $k$-hop cost $c^{(k)}(T^\ast)$.
Therefore, applying the given algorithm for the Path Covering problem to $\cI$ yields a Path Covering solution $\cP$ with weight 
\[
w(\cP) \ \leq\ \alpha(n) \cdot O(L^3)\cdot  c^{(k)}(T^\ast)
\]
in running time $t(n)+\poly(n)$.
Because any inclusion-wise minimal Path Covering solution contains a polynomial number of path-level pairs, we may assume that the $|\cP|$ is polynomially bounded in $n$.
Then by \Cref{lem:tour_from_set_cover}, we can turn this Path Covering solution $\cP$ into a tour $T$ with
\[
c^{(k)}(T) \ \leq\ 2 \cdot w(\cP) \ \leq\ \alpha(n) \cdot O(L^3)\cdot  c^{(k)}(T^\ast) \ = \ O(\log^3 n)\cdot \alpha(n) \cdot \OPT(\cI).
\]
\end{proof}

In the remaining part of this section we prove  \Cref{lem:tours_to_covers}, i.e, we focus on non-degenerate instances and show how to construct a good Path Covering solution from a given tour on $V$.

\subsection{Useful Observations on Hop-Costs}\label{sec:useful_tools_covering}

Before explaining how our construction works, we make a couple of useful observations about how the $k$-hop cost of a path changes when we apply certain useful operations such as changing the order of the vertices in a particular way.
Sometimes we will be interested not only in the overall $k$-hop cost of our path, but also in the contribution of edges of particular levels to the $k$-hop cost.
For a path $P$ visiting the vertices $v_1,\dots, v_r$ in this order, we define
\[
c^{(k)}_{= \ell}(P) := \sum_{i = 1}^r \sum_{\substack{\Delta \in[k]: \\ \level(\{v_i, v_{i+\Delta}\}) = \ell} }^k c(v_i, v_{i + \Delta}).
\]
to be the contribution of edges with level exactly $\ell$, and we define
\[
c^{(k)}_{> \ell}(P) := \sum_{j=\ell + 1}^L c^{(k)}_{= j}(P)
\]
to be the contribution of edges with level $> \ell$.

We will often reorder some vertices of a path.
In the simplest case, we reorder all vertices of a path according to the total order $\prec$ on $V$:

\begin{lemma}
\label{lem:sorted_path}
Let $P$ be a path with $|V(P)| \leq k + 1$ and let $P_{\rm{sorted}}$ be the path on the same vertex set that consists only of forward edges.
Then, we have $c^{(k)} (P_{\rm{sorted}}) \leq c^{(k)} (P)$.
\end{lemma}
\begin{proof}
We have 
\begin{equation*}
c^{(k)} (P_{\rm{sorted}}) = \sum_{\substack{x, y \in V(P), \\ x \prec y}} c(x,y),
\end{equation*}
because $|V(P_{\rm{sorted}})|\leq k + 1$.
For any $x,y\in V(P)$ with $x \prec y$, either $c(x,y)$ or $c(y,x)$ contributes to $c^{(k)}(P)$ because $|V(P)| \leq k + 1$.
Since $(x,y)$ is a forward edge, $c(x,y) \leq c(y,x)$.
Hence, we have $c^{(k)} (P_{\rm{sorted}}) \leq c^{(k)} (P)$.
\end{proof}

Sometimes we do not want to reorder the whole path, but only interchange two consecutive vertices on the path:

\begin{lemma}
\label{lem:interchange_vertices}
Let $P$ be a path and let $e = (x,y)$ be an edge of $P$.
Let $P^\prime$ be the path obtained from $P$ by exchanging the visits of $x$ and $y$.
Then, we have
\begin{equation*}
c^{(k)}(P^\prime) \leq c^{(k)}(P) + c(x,y) + c(y,x).
\end{equation*}
Furthermore, if $\level(e) = \ell \in[L]$ and $x \succ y$, i.e., $e$ is a backward edge, then
\begin{equation*}
c^{(k)}_{> \ell}(P^\prime) \leq c^{(k)}_{> \ell}(P) + c(x,y).
\end{equation*}
\end{lemma}
\begin{figure}
\begin{center}
\begin{tikzpicture}[
xscale=1,
vertex/.style={circle,draw=black, fill,inner sep=1.5pt, outer sep=1pt},
edge/.style={-latex,very thick},
]

\definecolor{color1}{named}{Peach}  
\definecolor{color2}{named}{ForestGreen}  
\definecolor{color3}{named}{Violet}  
\definecolor{color4}{named}{ProcessBlue} 

\begin{scope}[every node/.style={vertex}]
\foreach \i in {1,...,12} {
\node (v\i) at ({(\i)*1},0) {};
}
\end{scope}
\node[below=4pt] (xlabel) at (v6) {$x$};
\draw[line width=0.4pt] ([xshift=4pt,yshift=4pt]xlabel.center) -- ([xshift=-4pt,yshift=-4pt]xlabel.center);
\node[below=1.3em] at (v6) {$y$};

\node[below=4pt] (ylabel) at (v7) {$y$};
\draw[line width=0.4pt] ([xshift=4pt,yshift=4pt]ylabel.center) -- ([xshift=-4pt,yshift=-4pt]ylabel.center);
\node[below=1.3em] at (v7) {$x$};

\node[below=4pt] at (v3) {$a$};
\node[below=4pt] at (v10) {$b$};

\foreach \i in {1,...,5} {
\pgfmathtruncatemacro{\next}{\i+1}
\draw[edge] (v\i) to (v\next);
}

\foreach \i in {7,...,11} {
\pgfmathtruncatemacro{\next}{\i+1}
\draw[edge] (v\i) to (v\next);
}

\draw[edge, dotted] (v6) to (v7);
\draw[edge, dashed, color1] (v3) to[bend left=45] (v6);
\draw[edge, dashed, color1] (v7) to[bend left=45] (v10);
\end{tikzpicture}
 \end{center}
\caption{
Illustration of the interchange of $x$ and $y$ on a path $P$ for $k=4$.
The only edges that contribute to $c^{(k)}(P)$ and are affected by this change are the $k$-hop edges leaving $a$ and entering $b$, as well as the edge $(x,y) \in E(P)$ itself.
}
\label{fig:swapping_vertices}
\end{figure}
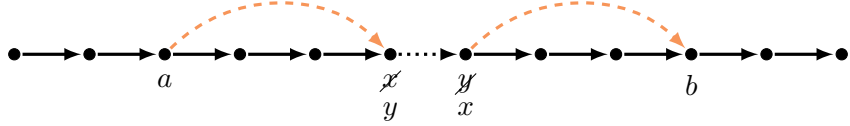
\begin{proof}
Suppose there are vertices $a,b \in V(P)$ such that $x$ is the $k$-th vertex after $a$ in $P$, and similarly, $b$ is the $k$-th vertex after $y$ in $P$.
See \Cref{fig:swapping_vertices} for an illustration.
Then, by the triangle inequality, we have
\begin{equation*}
\begin{split}
c^{(k)}(P^\prime) - c^{(k)}(P) &= c(a,y) + c(x,b) + c(y,x) - \left( c(a,x)+c(y,b)+ c(x,y)  \right) \\
&= (c(a,y) - c(a,x)) + (c(x,b) - c(y,b)) + c(y,x) - c(x,y)  \\
&\leq 2c(x,y)  + c(y,x) - c(x,y) \\
&= c(x,y) + c(y,x).
\end{split}
\end{equation*}
If no such $a$ or $b$ exists, the corresponding terms can be replaced by zero in the above inequality.
This proves the first claim.

Now, consider the case $\level(x,y) = \ell$ and $(x,y)$ is a backward edge.
Again, let $a$ and $b$ be as above.
The edges $(a,y)$ and $(x,b)$ that newly contribute to the $k$-hop costs of the path can only contribute to $c^{(k)}_{> \ell}(P^\prime)$ if $\level(\{a,y\}) > \ell$ and $\level(\{x,b\}) > \ell$, respectively.
In this case, we have $c(a,y) + c(x,b) \leq 2D_{\ell+1} \leq D_{\ell} = c(x,y)$, where the last inequality uses a property of the numbers $D_j$ from the definition of hierarchically ordered instances (cf.~\Cref{def:hierarchically_ordered}).
This proves the second claim.
\end{proof}

Many arguments in our reduction are based on \emph{replacing} a subpath $P_{\mathrm{sub}}$ of a given path $P$ with a different path $Q$.
In most cases we have $V(Q) \subseteq V(P_{\mathrm{sub}})$, so it is clear that the result is indeed a path.

\begin{lemma}
\label{lem:path_replacement}
Let $P_{\mathrm{main}}$ be a path and let $P_{\mathrm{sub}}$ be a subpath of $P_{\mathrm{main}}$.
Let $Q$ be another path with $V(Q) \subseteq V(P_{\mathrm{sub}})$, and denote by $P^\prime_{\mathrm{main}}$ the path that results from $P_{\mathrm{main}}$ by replacing $P_{\mathrm{sub}}$ with $Q$.  
Then, we have
\begin{equation*}
c^{(k)}(P^\prime_{\mathrm{main}}) \leq c^{(k)}(P_{\mathrm{main}}) + c^{(k)} (Q) + 2^{12} \cdot c^{(k)} (P_{\mathrm{sub}}).
\end{equation*}
\end{lemma}
\begin{proof}
Without loss of generality, we assume that $V(Q) = V(P_{\mathrm{sub}})$.
If $V(Q) \subsetneq V(P_{\mathrm{sub}})$, we simply remove all vertices $v \in V(P_{\mathrm{sub}}) \setminus V(Q)$ from $P^\prime_{\mathrm{main}}$, which does not increase the cost by \Cref{lem:khop_skipping_vertices}.

\addparskip

We will construct a sequence of paths $P^{1}, \dots, P^{r}$ with $P^{1} = P_{\mathrm{main}}$ and $P^{r} = P^\prime_{\mathrm{main}}$, for some $r \in \bN$.
Each path $P^{j}$ will arise from $P_{\mathrm{main}}$ by changing the order of the vertices in $V(Q) = V(P_{\mathrm{sub}})$.
For $j \in [r]$ we define
\begin{equation*}
\begin{split}
\hop_{\mathrm{in}}(j)&\coloneqq \Big\{ (v,w) : v\in V(P_{\rm main})\setminus V(Q),\ w \in V(Q), \text{ and $(v,w)$ contributes to $c^{(k)}\big(P^{j}\big)$}  \Big\}\\
\hop_{\mathrm{out}}(j)&\coloneqq \Big\{ (v,w) : v \in V(Q),\ w\in V(P_{\rm main})\setminus V(Q), \text{ and $(v,w)$ contributes to $c^{(k)}\big(P^{j}\big)$}  \Big\},
\end{split}
\end{equation*}
where an edge $(v,w)$ contributes to $c^{(k)}\big(P^{j}\big)$ if the path $P^{j}$ visits $v$ before $w$ and visits less than $k$ other vertices in between.
See \Cref{fig:hop_in_and_hop_out} for an illustration.
\begin{figure}
\begin{center}
\begin{tikzpicture}[
xscale=1,
vertex/.style={circle,draw=black, fill,inner sep=1.5pt, outer sep=1pt},
edge/.style={-latex,very thick},
]

\definecolor{color1}{named}{Peach}  
\definecolor{color2}{named}{ForestGreen}  
\definecolor{color3}{named}{Violet}  
\definecolor{color4}{named}{ProcessBlue} 

\begin{scope}[every node/.style={vertex}]
\foreach \i in {1,...,4} {
\node (v\i) at ({(\i)*1},0) {};
}
\end{scope}

\begin{scope}[every node/.style={vertex, draw=color4, fill=color4}]
\foreach \i in {5,...,11} {
\node (v\i) at ({(\i)*1},0) {};
}
\end{scope}

\begin{scope}[every node/.style={vertex}]
\foreach \i in {12,...,15} {
\node (v\i) at ({(\i)*1},0) {};
}
\end{scope}

\foreach \i in {1,...,14} {
\pgfmathtruncatemacro{\next}{\i+1}
\draw[edge] (v\i) to (v\next);
}

\foreach \i/\j in {2/5,3/5,3/6,4/5,4/6,4/7} {
\draw[edge, bend left = 20+10*(\j-\i), red!90!black] (v\i) to (v\j);
}

\foreach \i/\j in {9/12, 10/12, 10/13, 11/12,11/13,11/14} {
\draw[edge, bend left = 20+10*(\j-\i), red!90!black] (v\i) to (v\j);
}

\node[below=10pt] at (2.5,0) {\textcolor{black}{$P^{j}$}};
\node[below=10pt] at (v8) {\textcolor{color4}{$V(Q)$}};

\node[above=2.5em] at (4.5,0) {\textcolor{red!90!black}{$\hop_{\mathrm{in}}\left(j\right)$}};
\node[above=2.5em] at (11.5,0) {\textcolor{red!90!black}{$\hop_{\mathrm{out}}\left(j\right)$}};
\end{tikzpicture}
 \end{center}
\caption{
Illustration of the sets $\hop_{\mathrm{in}}(j)$ and $\hop_{\mathrm{out}}(j)$.
The edges contained in these sets are shown in red, where we have $k = 3$ in this example.
}
\label{fig:hop_in_and_hop_out}
\end{figure}
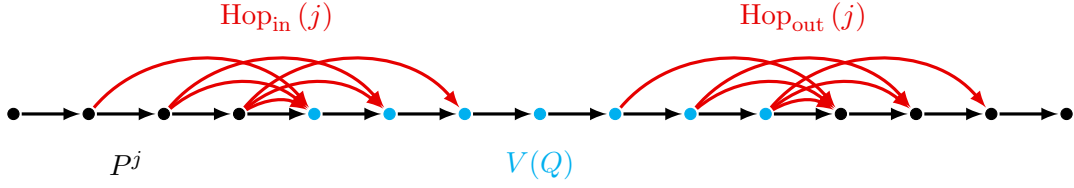

Using this notation, we get
\begin{equation*}
\begin{split}
c^{(k)}(P^\prime_{\mathrm{main}}) - c^{(k)}(P_{\mathrm{main}})
&= c^{(k)}\left(P^{r}\right) - c^{(k)}\left(P^{1}\right) \\
&\leq c^{(k)}(Q) + c\left(\hop_{\mathrm{in}}\left(r\right)\right) - c\left(\hop_{\mathrm{in}}\left(1\right)\right) 
+ c\left(\hop_{\mathrm{out}}\left(r\right)\right) - c\left(\hop_{\mathrm{out}}\left(1\right)\right).
\end{split}
\end{equation*}

Observe that $c\left(\hop_{\mathrm{in}}\left(r\right)\right) - c\left(\hop_{\mathrm{in}}\left(1\right)\right)$ and $c\left(\hop_{\mathrm{out}}\left(r\right)\right) - c\left(\hop_{\mathrm{out}}\left(1\right)\right)$ do not depend on the sequence $P^1, \dots, P^r$, as long as we have $P^{1} = P_{\mathrm{main}}$ and $P^{r} = P^\prime_{\mathrm{main}}$.
In particular, we can prove an upper bound on each of these expressions while each time working with a different sequence $P^1, \dots, P^r$.
We will show 
\begin{equation}\label{eq:bound_in_edges}
c\left(\hop_{\mathrm{in}}\left(r\right)\right) - c\left(\hop_{\mathrm{in}}\left(1\right)\right) \leq 2^{11} \cdot c^{(k)} (P_{\mathrm{sub}}),
\end{equation}
By symmetry we then also get 
\begin{equation}\label{eq:bound_out_edges}
c\left(\hop_{\mathrm{out}}\left(r\right)\right) - c\left(\hop_{\mathrm{out}}\left(1\right)\right) \leq 2^{11} \cdot c^{(k)} (P_{\mathrm{sub}}),
\end{equation}
 implying
\begin{equation*}
c^{(k)}(P^\prime_{\mathrm{main}}) - c^{(k)}(P{\mathrm{main}})
\leq c^{(k)}(Q) + 2^{11} \cdot  c^{(k)} (P_{\mathrm{sub}}) + 2^{11} \cdot  c^{(k)} (P_{\mathrm{sub}}),
\end{equation*}
as claimed.
It remains to prove \eqref{eq:bound_in_edges} (where \eqref{eq:bound_out_edges} follows by a symmetric argument).
\addparskip

To show \eqref{eq:bound_in_edges}, we construct a sequence $P^{1}, \dots, P^{r}$ with $P^{1} = P_{\mathrm{main}}$ and $P^{r} = P^\prime_{\mathrm{main}}$.
During our construction we will maintain a set $F$ of edges, which we initialize as $F \coloneqq \emptyset$.
We number the vertices of the path $Q$ as $q_1, \dots, q_m$, in the order in which $Q$ visits them.
Starting with $j=1$ and $P^j= P_{\mathrm{main}}$, we proceed as follows.
For $a=1,\dots,m$ we do the following:
While there exists a vertex $q_{b}\in V(Q)$ with $a <b$ such that the edge $(q_{b}, q_a)$ is contained in the current path $P^{j}$, we interchange the two vertices and obtain $P^{j+1}$.
If $\hop_{\mathrm{in}}\left(j\right) \neq \hop_{\mathrm{in}}\left(j+1\right)$, we add the edge $(q_b, q_a)$ to $F$.
In other words, the sequence $P^{(1)}, \dots, P^{(r)}$ is obtained by applying the Bubblesort algorithm with a particular pivot rule (prioritizing vertices that are visited earlier by $Q$).

Due to the order in which we interchange vertices, $F$ satisfies the following properties:
\begin{enumerate}
\item\label{item:property_F1}
$|F| \leq k^2$.

\item\label{item:property_F2}
For every $x \in V$, there are at most $k$ edges in $F$ that are incident to $x$.

\item\label{item:property_F3}
Every edge $(x,y) \in V \times V$ is added at most once to $F$.

\item\label{item:property_F4}
If $(x,y) \in F$, then $x$ precedes $y$ on $P_{\mathrm{sub}}$.
\end{enumerate}

Let $(x,y) \in F$ be an edge that was added to $F$ when obtaining $P^{j}$ from $P^{j-1}$.
Then, we have
\begin{equation*}
c\left(\hop_{\mathrm{in}}\left(j\right)\right) - c\left(\hop_{\mathrm{in}}\left(j-1\right)\right) \leq c(x,y),
\end{equation*}
where we used that $c$ satisfies the triangle inequality.
In particular, this shows
\begin{equation*}
c\left(\hop_{\mathrm{in}}\left(r\right)\right) - c\left(\hop_{\mathrm{in}}\left(1\right)\right) \leq c(F).
\end{equation*}

\addparskip

First, suppose that $k\leq 16$.
By property~\ref{item:property_F4} and the triangle inequality, we can bound $c(x,y) \leq c(P_{\rm sub}) \leq c^{(k)}(P_{\rm sub})$ for every $(x,y) \in F$.
In particular, property~\ref{item:property_F1} then implies that $c(F) \leq 2^8 \cdot c^{(k)}(P_{\rm sub})$.
Hence, we may assume that $k \geq 16$ in the following.

Let $\widetilde{F} \subseteq F$ be the subset of those edges $(x,y)$ for which $P_{\mathrm{sub}}$ visits less than $k$ other vertices between the visits of $x$ and $y$, i.e., $c(x,y)$ contributes to $c^{(k)}(P_{\mathrm{sub}})$.
In particular, we have
\begin{equation*}
c(\widetilde{F}) \leq c^{(k)}(P_{\mathrm{sub}}).
\end{equation*}

Next, we upper bound the cost of the edges in $F \setminus \widetilde{F}$.
Fix an edge $(x,y) \in F \setminus \widetilde{F}$.
By \ref{item:property_F4}, the vertex $x$ is visited before $y$ by the path $P_{\mathrm{sub}}$.
Because $k\geq 16$, we can apply \Cref{lem:path_distribution} to the $x$-$y$ subpath of  $P_{\mathrm{sub}}$ and let $\chi^{(x,y)} \in \bR_{\geq 0}^{V \times V}$ be the resulting vector.
The vector $\chi^{(x,y)}$ is a convex combination of incidence vectors of $x$-$y$ paths and thus, by the triangle inequality, $c(x,y) \leq c^\intercal \chi^{(x,y)}$ (where we view the cost function $c$ as a vector in $\bR_{\geq 0}^{V \times V}$).

Now, let $e = (v,w)$ be an edge contributing to the $k$-hop cost of the $x$-$y$ subpath of  $P_{\mathrm{sub}}$.
By the guarantees~\eqref{eq:path_distribution1} of \Cref{lem:path_distribution}, and the properties \ref{item:property_F1}, \ref{item:property_F2}, and \ref{item:property_F3} of $F$, we have
\begin{equation*}
\sum_{(x,y) \in F \setminus \widetilde{F}} \chi^{(x,y)}_e 
\quad \leq\ \sum_{\substack{(x,y) \in F \setminus \widetilde{F}, \\ \parset{x,y} \cap \parset{v,w} \neq \emptyset}} \frac{2^4}{k} + \sum_{\substack{(x,y) \in F \setminus \widetilde{F}, \\ \parset{x,y} \cap \parset{v,w} = \emptyset}} \frac{2^8}{k^2} 
\quad \leq\quad 2k \cdot \frac{2^4}{k} + k^2 \cdot \frac{2^8}{k^2} 
\quad \leq\quad 2^9.
\end{equation*}
This implies
\begin{equation*}
c(F \setminus \widetilde{F}) \ = \sum_{(x,y) \in F \setminus \widetilde{F}} c(x,y)
\ \leq \sum_{(x,y) \in F \setminus \widetilde{F}} c^\intercal \chi^{(x,y)}
\ \leq 2^9 \cdot c^{(k)} (P_{\mathrm{sub}}).
\end{equation*}
Hence, we conclude
\begin{equation*}
c\left(\hop_{\mathrm{in}}\left(r\right)\right) - c\left(\hop_{\mathrm{in}}\left(1\right)\right) \leq  c(F) \leq (2^9 + 1) \cdot c^{(k)} (P_{\mathrm{sub}}),
\end{equation*}
implying~\eqref{eq:bound_in_edges}.
\end{proof}

Sometimes we carefully reorder entire parts of a path $P$ in order to ensure that the resulting path has no backward edges of level $\ell$ anymore:

\begin{definition}
[$\ell$-refinement]
\label{def:i-refinement}
Let $P$ be a path and $\ell \in [L]$.
We define a total order $\prec^\prime$ on $V(P)$ as follows:
Let $x \neq y \in V(P)$ be two distinct vertices. 
\begin{itemize}
\item
If $\level(\parset{x,y}) \leq \ell$, then $x \prec^\prime y$ if and only if $x \prec y$.

\item
If $\level(\parset{x,y}) > \ell$, then $x \prec^\prime y$ if and only if $x$ appears before $y$ on $P$.
\end{itemize}
Let $P^\prime$ with $V(P^\prime) = V(P)$ be the path that visits all vertices in the order given by $\prec^\prime$. 
We say that $P^\prime$ is the \emph{$\ell$-refinement} of $P$.
\end{definition}

See \Cref{fig:i_refinement} for an example.
\begin{figure}
\begin{center}
\begin{tikzpicture}[
vertex/.style={circle,draw=black, fill,inner sep=1.5pt, outer sep=1pt},
edge/.style={-latex, very thick},
box/.style={very thick, fill=Peach, fill opacity=0.05, rounded corners=8pt},
xscale=1.5,
yscale=1.3
]
\clip (0,-2.5) rectangle (8,2.5);

\begin{scope}\def\slackx{0.1}
\def\slacky{0.4}

\def\height{2.2}

\definecolor{color1}{named}{Peach}  
\definecolor{color2}{named}{ForestGreen}  
\definecolor{color3}{named}{Violet}  
\definecolor{color4}{named}{ProcessBlue} 

\clip (0,-1) rectangle (8,6);

\begin{scope}[box/.append style={draw=color1}]
\draw[box] (-2,0) rectangle (2-\slackx,\height);
\draw[box] (2,0) rectangle (4-\slackx,\height);
\draw[box] (4,0) rectangle (6-\slackx,\height);
\draw[box] (6,0) rectangle (10,\height);
\end{scope}

\begin{scope}[every node/.style={vertex}, yshift=2.5pt]
\node (v11) at (0.75,\height-1*\slacky) {};
\node (v12) at (1.5,\height-1*\slacky) {};
\node (v13) at (2.75,\height-1*\slacky) {};
\node (v14) at (4.25,\height-1*\slacky) {};
\node (v15) at (5,\height-1*\slacky) {};
\node (v16) at (5.5,\height-1*\slacky) {};
\node (v17) at (6.5,\height-1*\slacky) {};
\node (v18) at (7.75,\height-1*\slacky) {};

\node (v21) at (2.25,\height-2*\slacky) {};
\node (v22) at (3,\height-2*\slacky) {};
\node (v23) at (3.75,\height-2*\slacky) {};

\node (v31) at (2.5,\height-3*\slacky) {};
\node (v32) at (3.5,\height-3*\slacky) {};
\node (v33) at (4.5,\height-3*\slacky) {};
\node (v34) at (7,\height-3*\slacky) {};

\node (v41) at (5.25,\height-4*\slacky) {};

\node (v51) at (0.25,\height-5*\slacky) {};
\node (v52) at (1.25,\height-5*\slacky) {};
\node (v53) at (3.25,\height-5*\slacky) {};
\node (v54) at (4.75,\height-5*\slacky) {};
\end{scope}

\foreach \i in {1,...,7}{
\pgfmathtruncatemacro{\next}{\i+1}                
\draw[edge] (v1\i) -- (v1\next);
}

\foreach \i in {1,...,2}{
\pgfmathtruncatemacro{\next}{\i+1}                
\draw[edge] (v2\i) -- (v2\next);
}

\foreach \i in {1,...,3}{
\pgfmathtruncatemacro{\next}{\i+1}                
\draw[edge] (v3\i) -- (v3\next);
}

\foreach \i in {1,...,3}{
\pgfmathtruncatemacro{\next}{\i+1}                
\draw[edge] (v5\i) -- (v5\next);
}

\draw[edge,out=-150,in=30, looseness=0.3] (v18) to (v21);
\draw[edge, out=-120,in=30, looseness=0.7] (v23) to (v31);
\draw[edge, out=-140,in=30, looseness=0.6] (v34) to (v41);
\draw[edge, out=-150,in=30, looseness=0.3] (v41) to (v51);

\node at (4,2.4) {\textcolor{color1}{$\cH_{\ell+1}$}};

\end{scope}
\begin{scope}[shift={(0,-2.3)}, edge/.append style={ForestGreen}]

\def\slackx{0.1}
\def\slacky{0.4}

\def\height{2.2}

\definecolor{color1}{named}{Peach}  
\definecolor{color2}{named}{ForestGreen}  
\definecolor{color3}{named}{Violet}  
\definecolor{color4}{named}{ProcessBlue} 

\clip (0,-1) rectangle (8,6);

\begin{scope}[box/.append style={draw=color1}]
\draw[box] (-2,0) rectangle (2-\slackx,\height);
\draw[box] (2,0) rectangle (4-\slackx,\height);
\draw[box] (4,0) rectangle (6-\slackx,\height);
\draw[box] (6,0) rectangle (10,\height);
\end{scope}

\begin{scope}[every node/.style={vertex}, yshift=2.5pt]
\node (v11) at (0.75,\height-1*\slacky) {};
\node (v12) at (1.5,\height-1*\slacky) {};
\node (v13) at (2.75,\height-1*\slacky) {};
\node (v14) at (4.25,\height-1*\slacky) {};
\node (v15) at (5,\height-1*\slacky) {};
\node (v16) at (5.5,\height-1*\slacky) {};
\node (v17) at (6.5,\height-1*\slacky) {};
\node (v18) at (7.75,\height-1*\slacky) {};

\node (v21) at (2.25,\height-2*\slacky) {};
\node (v22) at (3,\height-2*\slacky) {};
\node (v23) at (3.75,\height-2*\slacky) {};

\node (v31) at (2.5,\height-3*\slacky) {};
\node (v32) at (3.5,\height-3*\slacky) {};
\node (v33) at (4.5,\height-3*\slacky) {};
\node (v34) at (7,\height-3*\slacky) {};

\node (v41) at (5.25,\height-4*\slacky) {};

\node (v51) at (0.25,\height-5*\slacky) {};
\node (v52) at (1.25,\height-5*\slacky) {};
\node (v53) at (3.25,\height-5*\slacky) {};
\node (v54) at (4.75,\height-5*\slacky) {};
\end{scope}

\draw[edge] (v11) to (v12);
\draw[edge, out=-150,in=30, looseness=0.7] (v12) to (v51);
\draw[edge] (v51) to (v52);
\draw[edge, in=-180,out=0, looseness=1] (v52) to (v13);
\draw[edge, out=-120,in=30, looseness=0.9] (v13) to (v21);
\draw[edge] (v21) to (v22);
\draw[edge] (v22) to (v23);
\draw[edge] (v22) to (v23);
\draw[edge, out=-120,in=30, looseness=0.7] (v23) to (v31);
\draw[edge] (v31) to (v32);
\draw[edge, out=-90,in=90, looseness=0.8] (v32) to (v53);
\draw[edge, in=-130,out=0, looseness=1] (v53) to (v14);
\draw[edge] (v14) to (v15);
\draw[edge] (v15) to (v16);
\draw[edge, out=-100,in=45, looseness=0.4] (v16) to (v33);
\draw[edge, in=140,out=-30, looseness=0.7] (v33) to (v41);
\draw[edge, out=-120,in=30, looseness=0.4] (v41) to (v54);
\draw[edge, in=-160,out=0, looseness=1] (v54) to (v17);
\draw[edge] (v17) to (v18);
\draw[edge, out=-100,in=45, looseness=0.4] (v18) to (v34);
\end{scope}
\end{tikzpicture}
 \end{center}
\caption{
An example of a path $P$ (above) and its $\ell$-refinement (below).
All vertices are drawn from left to right according to the total order $\prec$.
}
\label{fig:i_refinement}
\end{figure}
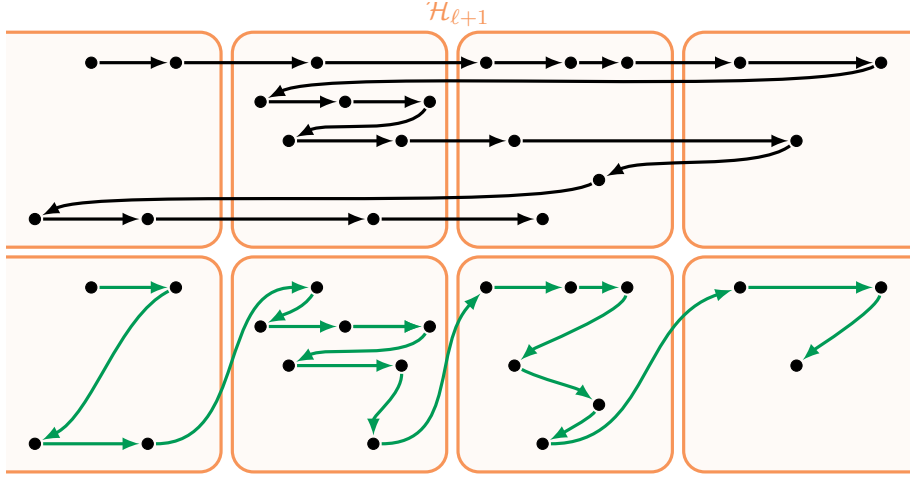
If $\level(P) > \ell$, then $P$ is its own $\ell$-refinement.
We can prove a strong upper bound on the additional costs that arise by replacing a subpath with its $\ell$-refinement:

\begin{lemma}
\label{lem:i-refinement}
Let $\ell \in [L]$.
Let $P_{\mathrm{main}}$ be a path and let $P_{\mathrm{sub}}$ be a subpath of $P_{\mathrm{main}}$ with $|V(P_{\mathrm{sub}})| \leq k + 1$ and $\level(V(P_{\mathrm{sub}})) \geq \ell$.
Let $P^\prime_{\mathrm{sub}}$ be the $\ell$-refinement of $P_{\mathrm{sub}}$ and denote by $P^\prime_{\mathrm{main}}$ the path obtained from $P_{\mathrm{main}}$ by replacing $P_{\mathrm{sub}}$ with $P^\prime_{\mathrm{sub}}$.
Then,
\begin{equation*}
c^{(k)}(P^\prime_{\mathrm{sub}} ) \leq c^{(k)}(P_{\mathrm{sub}}),
\end{equation*}
and
\begin{equation*}
c^{(k)}_{> \ell} (P^\prime_{\mathrm{main}}) \leq c^{(k)}_{> \ell} (P_{\mathrm{main}}) + c^{(k)}_{=\ell} (P_{\mathrm{sub}}).
\end{equation*}
\end{lemma}
\begin{proof}
We start by proving the second inequality.
Let $\prec^\prime$ be the total order on $V(P_{\mathrm{sub}})$ as given in \Cref{def:i-refinement}.
We modify $P_{\mathrm{main}}$ in a similar way as in the proof of \Cref{lem:path_replacement}:
In the beginning, set $F \coloneqq \emptyset$.
Then, iteratively interchange the visits of vertices $x \neq y \in V(P_{\mathrm{sub}})$ if $y \prec^\prime x$ and $(x,y)$ is part of the current path, and add $(x,y)$ to $F$.
Again, in other words, we reorder $P_{\mathrm{sub}}$ according to $\prec^\prime$ using the Bubblesort algorithm, and the resulting path is  $P^\prime_{\mathrm{main}}$.
We have $(x,y) \in F$ if and only if $y \prec x$, $\level(\parset{x,y}) = \ell$ and $x$ appears before $y$ on $P_{\mathrm{sub}}$ because our initial order is given by $P_{\mathrm{sub}}$.
Furthermore, a pair of vertices is interchanged at most once.
Therefore, by the second part of \Cref{lem:interchange_vertices}, we have
\begin{equation*}
c^{(k)}_{> \ell} (P^\prime_{\mathrm{main}}) \leq c^{(k)}_{> \ell} (P_{\mathrm{main}}) + \sum_{(x,y) \in F} c(x,y).
\end{equation*} 
Since $|V(P_{\mathrm{sub}})| \leq k + 1$, each edge $(x,y) \in F$ contributes to $c^{(k)}_{=\ell}(P_{\mathrm{sub}})$.
Thus, $\sum_{(x,y) \in F} c(x,y) \leq c^{(k)}_{=\ell}(P_{\mathrm{sub}})$.
This shows the second inequality.

For the first inequality, simply observe that
\begin{equation*}
c^{(k)}(P^\prime_{\mathrm{sub}}) = c^{(k)}(P_{\mathrm{sub}}) - \sum_{(x,y) \in F} c(x,y)  + \sum_{(x,y) \in F} c(y,x),
\end{equation*}
because $|V(P_{\mathrm{sub}})| \leq k + 1$.
Since $y \prec x$ if $(x,y) \in F$, we have $c(x,y) \geq c(y,x)$.
This proves the first inequality.
\end{proof}

\subsection{Turning Tours into Covering Solutions}\label{sec:tours_to_covers}

Consider a non-degenerate hierarchically ordered instance $\cI =(V, L, \cH, D, \prec, c, k)$ of Hop-ATSP and a tour $T$ on $V$.
Our goal is to prove that there exists a Path Covering solution $\cP$ for the instance $\cI$ with $w(\cP) \leq O(L^3) \cdot c^{(k)}(T)$.
By \Cref{lem:khop_skipping_vertices}, we may assume that $T$ is a cycle.
We will first remove an arbitrary edge from $T$ to obtain a path $P$ with vertex set $V$.
If the path $P$ happens to be monotone, then $\cP=\{(P,1)\}$ would be a Path Covering solution.
Because our instance $\cI$ is non-degenerate, \Cref{lem:first_level_fixed_cost_lb} implies that 
\[
w(\cP) \ =\ c^{(k)}(P) + k^2D_1 \leq (1+ 2^8) \cdot c^{(k)}(P).
\] 
Of course $P$ will in general not be monotone, i.e., it will contain some backward edges.

Our strategy will be to first get rid of all backward edges of level 1, then those of level 2, and so on.
However, removing the backward edges of a given level can become complicated when there are many such jumps close together.
To address this, we first isolate these jumps, which will allows us to treat each jump independently. 
In order to formalize what we mean by ``isolating jumps'', we introduce the notion of a \emph{jump segment}.

\begin{definition}
[Jump segment]
\label{def:jump_segment}
An \emph{$\ell$-jump segment} is a path $P$ with $|V(P)| \leq k + 1$ and $\level(V(P)) = \ell$ that consists of two vertex disjoint monotone paths $P_{\mathrm{entry}}$ and $P_{\mathrm{exit}}$, that are connected by a backward edge~$e$ (where $P_{\mathrm{entry}}$ precedes $P_{\mathrm{exit}}$ on $P$).
\end{definition}
See \Cref{fig:jump_segment} for an illustration.
\begin{figure}
\begin{center}
\begin{tikzpicture}[
xscale=1,
vertex/.style={circle,draw=black, fill,inner sep=1.5pt, outer sep=1pt},
edge/.style={-latex,very thick},
box/.style={very thick, fill=Gray, fill opacity=0.05, rounded corners=8pt},
]

\definecolor{color1}{named}{Peach}  
\definecolor{color2}{named}{ForestGreen}  
\definecolor{color3}{named}{Violet}  
\definecolor{color4}{named}{ProcessBlue} 

\draw[box,draw=color1] (0.5,-0.25) rectangle (3.85,1.25);
\draw[box,draw=color1] (3.95,-0.25) rectangle (7.75,1.25);

\begin{scope}[every node/.style={vertex, draw=color2, fill=color2}]
\node (v1) at (0.75,1) {};
\node (v2) at (1.5,1) {};
\node (v3) at (2.5,1) {};
\node (v4) at (3.25,1) {};
\node (v5) at (4.5,1) {};
\node (v6) at (5.25,1) {};
\node (v7) at (6.5,1) {};
\node (v8) at (7.5,1) {};
\end{scope}
\foreach \i in {1,...,7} {
\pgfmathtruncatemacro{\next}{\i+1}
\draw[edge, color2] (v\i) to (v\next);
}
\node at (2,1.5) {\textcolor{color2}{$P_{\mathrm{entry}}$}};

\begin{scope}[every node/.style={vertex, draw=color3, fill=color3}]
\node (w1) at (2,0) {};
\node (w2) at ($(2.5,0)!0.5!(3.25,0)$) {};
\node (w3) at ($(3.2,0)!0.25!(4.5,0)$) {};
\node (w4) at ($(4.5,0)!0.5!(5.25,0)$) {};
\node (w5) at ($(5.25,0)!0.25!(6.5,0)$) {};
\node (w6) at ($(5.25,0)!0.75!(6.5,0)$) {};
\end{scope}
\foreach \i in {1,...,5} {
\pgfmathtruncatemacro{\next}{\i+1}
\draw[edge, color3] (w\i) to (w\next);
}
\node at (5.5,-0.5) {\textcolor{color3}{$P_{\mathrm{exit}}$}};

\draw[edge, densely dashed, out=-150, in=30, looseness=0.5, red] (v8) to (w1);
\node at (6.5,0.45) {\textcolor{red}{$e$}};

\node at (4,1.5) {\textcolor{color1}{$\cH_{\ell+1}$}};

\end{tikzpicture}
 \end{center}
\caption{
Illustration of an $\ell$-jump segment.
All vertices are drawn from left to right according to the total order~$\prec$.
The path contains a single backward edge $e$ that connects the two monotone subpaths $P_{\mathrm{entry}}$ and $P_{\mathrm{exit}}$.
}
\label{fig:jump_segment}
\end{figure}
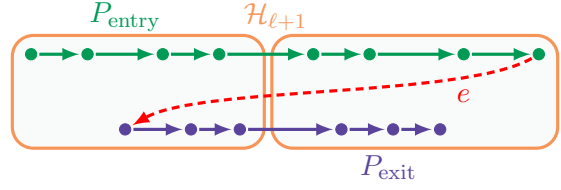
We remark that $\ell \leq \level(e) < L$ because the endpoints of the edge~$e$ are contained in $V(P)$ and $\cH_L$ contains only singletons.
Note however, that $\level(e)$ is not necessarily equal to $\ell = \level(V(P))$, but may be strictly greater.

In a path $P$ with ``isolated jumps'', we will be able to assign a jump-segment $P^e$ to each backward edge~$e$ such that these jump-segments for all backward edges are vertex disjoint subpaths of $P$.
Our strategy will be to change the path only in the local environment of $e$ that is given by the jump-segment $P^e$ when getting rid of the backward edge $e$.
Because the jump segments assigned to different backward edges are vertex disjoint, this will allow us to handle the different jumps independently.

In order to be able to get rid of a backward edge $e$ by only changing $P$ in the local environment given by its assigned jump segment, the jump segment $P^e$ must satisfy certain properties.
We will ensure that each jump segment satisfies one of the following two useful properties:
If all vertices that the path $P$ visits before the jump segment $P^e$ are left of the vertices of $P^e$, then intuitively we can handle the part of $P$ before the jump segment independently of $P^e$.
In this case, we will call $P^e$ \emph{entry-safe}.

\begin{definition}
[Entry-safe, exit-safe]\label{def:safe_jump_segments}
Let $P_{\mathrm{main}}$ be a path, and let $P_{\mathrm{sub}}$ be a subpath of $P_{\mathrm{main}}$.
We say that $P_{\mathrm{sub}}$ is \emph{entry-safe} within $P_{\mathrm{main}}$ if all vertices that appear (strictly) before $P_{\mathrm{sub}}$ on $P_{\mathrm{main}}$ are left of $V(P_{\mathrm{sub}})$.

Analogously, $P_{\mathrm{sub}}$ is \emph{exit-safe} within $P_{\mathrm{main}}$ if all vertices that appear (strictly) after $P_{\mathrm{sub}}$ on $P_{\mathrm{main}}$ are right of $V(P_{\mathrm{sub}})$.
\end{definition}

If a jump-segment $P^e$ is not entry-safe, we will ensure that $P^e_{\rm entry}$ contains many vertices, so the jump is ``far away'' from earlier jumps of $P$.
Analogously, we will ensure that if a jump-segment $P^e$ is not exit-safe, we will ensure that $P^e_{\rm exit}$ contains many vertices, so the jump is ``far away'' from later jumps of $P$.

\begin{definition}[Entry-buffered, exit-buffered]
For $\omega \in \parset{\mathrm{entry}, \mathrm{exit}}$, we say that an $\ell$-jump segment $P$ is \emph{$\omega$-buffered} if $|V(P_{\omega})| \geq \frac{k}{16} - (\ell-1)\cdot \frac{k}{32(L-1)}$.
\end{definition}

\begin{figure}
\begin{center}
\begin{tikzpicture}[
xscale=0.29,
vertex/.style={circle,draw=black, fill,inner sep=1.5pt, outer sep=1pt},
edge/.style={-latex,very thick},
]

\definecolor{color1}{named}{Peach}  
\definecolor{color2}{named}{ForestGreen}  
\definecolor{color4}{named}{purple}  
\definecolor{color3}{named}{ProcessBlue} 

\begin{scope}[every node/.style={vertex}]
\node (v1) at (0,1) {};
\node (v2) at (2,1) {};
\node (v3) at (5,1) {};
\end{scope}

\begin{scope}[every node/.style={vertex, color2}]
\node (v4) at (7,1) {};
\node (v5) at (9,1) {};
\node (v6) at (12,1) {};
\node (v7) at (14,1) {};
\node (v8) at (17,1) {};
\node (v9) at (10,0) {};
\node (v10) at (13,0) {};
\node (v11) at (16,0) {};
\end{scope}

\begin{scope}[every node/.style={vertex}]
\node (v12) at (19,0) {};
\node (v13) at (21,0) {};
\node (v14) at (23,0) {};
\node (v15) at (26,0) {};
\end{scope}

\begin{scope}[every node/.style={vertex, color3}]
\node (v16) at (28,0) {};
\node (v17) at (30,0) {};
\node (v18) at (32,0) {};
\node (v19) at (34,0) {};
\node (v20) at (36,0) {};
\node (v21) at (40,0) {};
\node (v22) at (44,0) {};
\end{scope}

\begin{scope}[every node/.style={vertex, color3}]
\node (v23) at (20,1) {};
\node (v24) at (22,1) {};
\node (v25) at (24,1) {};
\node (v26) at (27,1) {};
\node (v27) at (29,1) {};
\node (v28) at (31,1) {};
\node (v29) at (37,1) {};
\end{scope}

\begin{scope}[every node/.style={vertex, color4}]
\node (v30) at (39,1) {};
\node (v31) at (41,1) {};
\node (v32) at (43,1) {};
\node (v33) at (46,1) {};
\node (v34) at (48,1) {};
\node (v35) at (50,1) {};
\node (v36) at (52,1) {};
\end{scope}

\begin{scope}[every node/.style={vertex, color4}]
\node (v37) at (33,2) {};
\node (v38) at (35,2) {};
\node (v39) at (43,2) {};
\node (v40) at (45,2) {};
\node (v41) at (53,2) {};
\end{scope}

\begin{scope}[every node/.style={vertex}]

\end{scope}

\foreach \i in {1,...,3} {
\pgfmathtruncatemacro{\next}{\i+1}
\draw[edge] (v\i) to (v\next);
}

\foreach \i in {4,...,7} {
\pgfmathtruncatemacro{\next}{\i+1}
\draw[edge,color2] (v\i) to (v\next);
}
\draw[edge, densely dashed, out=-150, in=20, looseness=0.5, color2] (v8) to (v9);

\foreach \i in {9,...,10} {
\pgfmathtruncatemacro{\next}{\i+1}
\draw[edge,color2] (v\i) to (v\next);
}

\foreach \i in {11,...,15} {
\pgfmathtruncatemacro{\next}{\i+1}
\draw[edge] (v\i) to (v\next);
}

\foreach \i in {16,...,21} {
\pgfmathtruncatemacro{\next}{\i+1}
\draw[edge,color3] (v\i) to (v\next);
}

\draw[edge, densely dashed, out=150, in=-30, looseness=0.15, color3] (v22) to (v23);

\foreach \i in {23,...,28} {
\pgfmathtruncatemacro{\next}{\i+1}
\draw[edge,color3] (v\i) to (v\next);
}

\draw[edge] (v29) to (v30);

\foreach \i in {30,...,35} {
\pgfmathtruncatemacro{\next}{\i+1}
\draw[edge, color4] (v\i) to (v\next);
}

\draw[edge, densely dashed, out=150, in=-30, looseness=0.15, color4] (v36) to (v37);

\foreach \i in {37,...,40} {
\pgfmathtruncatemacro{\next}{\i+1}
\draw[edge, color4] (v\i) to (v\next);
}

\node[above=2pt,color2] at (13,1) {$P^{e}$};
\node[above=2pt,color3] at (28,1) {$P^{f}$};
\node[below=2pt,color4] at (48,1) {$P^{h}$};
\end{tikzpicture}
 \end{center}
\caption{\label{fig:good_jump_segments}
Illustration of an $\ell$-buffering consisting of three $\ell$-jump segments, where backward edges are dashed and vertices are drawn from left to right according to the order $\prec$.
All vertices are drawn from left to right according to the total order $\prec$.
The $\ell$-jump segment $P^e$ (green) is is entry-safe and exit-safe within the path $P$.
The $\ell$-jump segment $P^f$ (blue) is ``far away'' from other jumps in the sense that $P_{\rm entry}^f$ and $P_{\rm exit}^f$ contain many vertices.
More precisely, $P^f$ is entry-buffered and exit-buffered (assuming $7 = \frac{k}{16} - (\ell-1)\cdot \frac{k}{32(L-1)}$).
The $\ell$-jump segment $P^h$ (red) is entry-buffered and exit-safe.
}
\end{figure}
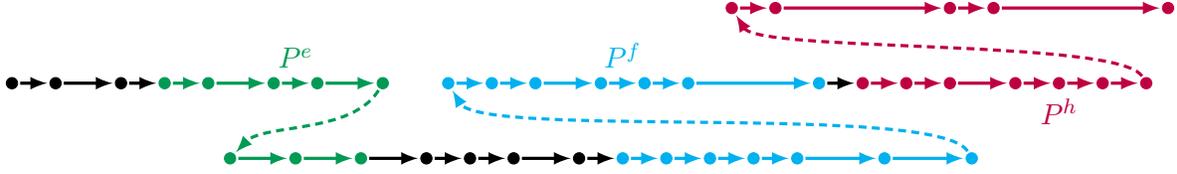

We can now formalize the aforementioned concept of ``isolated jumps'' in a path $P$.
See \Cref{fig:good_jump_segments} for an illustration.

\begin{definition}
[$\ell$-buffering]
An \emph{$\ell$-buffering} of a path $P$ is a collection $\cB$ of  of vertex disjoint subpaths of $P$ such that every backward edge $e$ of $P$ is contained in some path $P^e\in \cB$, and each path $Q\in \cB$ satisfies the following properties:
\begin{enumerate}

\item
$Q$ is an $\ell^\prime$-jump segment with $\ell^\prime \geq \ell$.

\item
$Q$ is entry-buffered or entry-safe within $P$.

\item
$Q$ is exit-buffered or exit-safe within $P$.
\end{enumerate}
We say that $P$ is \emph{$\ell$-buffered} if there exists an $\ell$-buffering of $P$.
\end{definition}

To prove that we can turn our path $P$ into a Path Covering solution $\cP$ of small weight, we will use a recursive argument.
By \Cref{eq:equivalence_to_set_cover}, we want $\cP$ to cover all elements of the Set Cover universe $\cU=V\times [L]$.
To describe our recursive construction, we need the notion of a \emph{main-tour cover}, which will be used to describe what we aim for when recursively applying our argument to a path $P$ that contains only a subset $U\subseteq V$ of the vertices.

\begin{definition}[Main-tour cover]
Let $U \subseteq V$ be an arbitrary set of vertices.
A \emph{main-tour cover} of $U$ is a monotone path $P_{\rm main}$ with $V(P_{\rm main}) \subseteq U$ together with a set $\cP \subseteq \cS$ of path-level pairs satisfying the following two conditions:
\begin{enumerate}
\item
$V(P_{\rm main}) = U$ or $|V(P_{\rm main})| \geq k$.
\item
The path-level pairs $\cP \cup \big\{ \big(P_{\rm main},\level(U)\big)\big\}$ cover the elements $U\times \big\{\level(U),\dots, L\big\} $ of the Set Cover universe $\cU$.
\end{enumerate}
\end{definition}

The key lemma, which we will prove by induction on $L-\ell$, is the following:

\begin{lemma}\label{lem:main_recursion_covering_weak}
Let $P_{\mathrm{old}}$ be an $\ell$-buffered path with $\ell \coloneqq \level(V(P_{\mathrm{old}}))$.
Then there is a main-tour cover $(P_{\mathrm{new}}, \cP)$ of $V(P_{\mathrm{old}})$ such that
\begin{equation}
\label{eq:inductive_reordering}
\begin{split}
c^{(k)}(P_{\mathrm{new}}) + w(\cP)\ \leq &\ \scconstant^\ast \cdot (L - \ell + 1) \cdot c^{(k)}(P_{\mathrm{old}}),
\end{split}
\end{equation}
where $\scconstant^\ast \coloneqq 3\scconstant_1 + \left(\scconstant_2 + 2^{14}\right)$, $\scconstant_1 \coloneqq2^{12}L^2$, and $\scconstant_2 \coloneqq 2^8L $.
\end{lemma}

Note that if we apply  \Cref{lem:main_recursion_covering_weak} to a $1$-buffered path $P_{\rm old}$ with vertex set $V$, we obtain a solution  $\cP\cup \{(P_{\rm new},1)\}$ of our Path Covering problem.
In order to apply \Cref{lem:main_recursion_covering_weak}, we need to start with a path that is $1$-buffered.
Using the tools that we have already established, this is not difficult to obtain:

\begin{lemma}
\label{lem:initial_reordering}
Let $P$ be a path.
Then, there is a $1$-buffered path $\hat P$ with $V(P) = V(\hat P)$ and 
\[
c^{(k)}(\hat P) \leq 2^{13} \cdot c^{(k)}(P).
\]
\end{lemma}
\begin{proof}
If $|V(P)| \leq k$, then we are done by \Cref{lem:sorted_path}.
If $k \leq 16$, any jump segment is entry- and exit-buffered.
Hence, to obtain $P^\prime$, we simply interchange the endpoints of every second edge that is backward.
All remaining backward edges are disjoint and form a $1$-buffering (where each jump segment consists of a single backward edge).
The cost increase can then be bounded by \Cref{lem:sorted_path,lem:path_replacement}.

Thus, we may assume $|V(P)| \geq k > 16$.
We partition $P$ into  vertex disjoint subpaths $P_1, \dots, P_r$  of $P$ with $\dot\bigcup_{j=1}^r V(P_j) = V(P)$ and $2\ceil*{k/16} \leq |V(P_j)| \leq k/2$
(cf.~\Cref{lem:block_decomp}).
For each $j\in [r]$ let $P^\prime_j$ be the path that visits the vertices from $V(P_j)$ according to the order $\prec$.
Replace each subpath $P_j$ of $P$ with $P^\prime_j$ and denote by $P^\prime$ the resulting path.

By \Cref{lem:sorted_path,lem:path_replacement}, we have
\begin{equation*}
c^{(k)}(P^\prime) \leq c^{(k)}(P) + \sum_{j=1}^r 2^{12} \cdot c^{(k)}(P_j) +  \sum_{j=1}^rc^{(k)}(P^{\prime}_j) \leq 2^{13} \cdot c^{(k)}(P).
\end{equation*}
It remains to show that $P^\prime$ is $1$-buffered.
Any backward edge $e$ of $P^\prime$ connects a subpath $P^\prime_{j}$ and $P^\prime_{j+1}$ (if we number the subpaths appropriately).
In that case, we can use $\ceil*{k/16}$ vertices from each of $P^\prime_{j}$ and $P^\prime_{j+1}$ to construct an $\ell$-jump segment around $e$ (for some $\ell \in [L]$) that is both entry- and exit-buffered. 
Doing this for every backward-edge of $P^\prime$ yields a $1$-buffering of $P^\prime$.
\end{proof}

Before proving \Cref{lem:main_recursion_covering_weak}, we first show that its proof will complete our reduction to the Path Covering.
In other words, using \Cref{lem:main_recursion_covering_weak}, we are now ready to prove \Cref{lem:tours_to_covers}, which we restate here for convenience:

\TourToCover*
\begin{proof}
Because $\cI$ is non-degenerate, \Cref{lem:first_level_fixed_cost_lb} implies $k^2 \cdot D_1 \leq 2^8 \cdot c^{(k)}(T)$.
Let $P$ be a path resulting from $T$ by removing an arbitrary edge.
We apply \Cref{lem:initial_reordering} to $P$ and then apply  \Cref{lem:main_recursion_covering_weak} to the resulting $1$-buffered path.
We obtain a main-tour cover $(\hat P_{\rm main},\hat \cP)$ of $V$ such that
\[
c^{(k)}(\hat P_{\rm main}) + w(\hat \cP)\ \leq \ \scconstant^\ast \cdot L \cdot 2 ^{13} \cdot c^{(k)}(P).
\]
Then  $ \cP \coloneqq \hat \cP\cup \big\{ \big(\hat P_{\rm main},1\big)\big\}$ is a Path Covering solution (by \Cref{eq:equivalence_to_set_cover}) and its weight is
\begin{align*}
w(\cP) \ =&\ w(\hat \cP) + c^{(k)}(\hat P) + k^2 \cdot D_1 \\ 
\leq&\ 2 ^{13} \cdot \scconstant^\ast \cdot L\cdot c^{(k)}(P) + k^2 \cdot D_1
\\
 \leq&\ \Big(2^{13} \cdot \scconstant^\ast \cdot L + 2^8 \Big) \cdot  c^{(k)}(T)\\
 =&\ O(L^3)\cdot  c^{(k)}(T),
\end{align*}
where we used $\eta^* = O(L^2)$.
\end{proof}

In the remainder of the section we prove \Cref{lem:main_recursion_covering_weak}.

\subsection{Splitting Jump Segments}\label{sec:splitting}

Recall that we want to consider each jump segment of our given $\ell$-buffered path separately and aim at getting rid of the backward edge contained in it.
One possibility to get rid of a backward edge is to split the path into two paths.
However, this is costly because for every path-level pair $(P,\ell)$ that we add to our main-tour cover, we need to pay the ``fixed cost'' $k^2D_{\ell}$ in addition to the $k$-hop costs $c^{(k)}(P)$.
Because the total weight of the path-level pairs in our main-tour cover should not be too large compared to the $k$-hop cost of our initial path, we can only afford paying these fixed costs when our initial path is sufficiently expensive.

We will classify certain $\ell$-jump segments $P$ as \emph{splitting} $\ell$-jump segments.
Such segments will satisfy the property that they contain many $j$-hop edges ($j \leq k$) that are backward edges of level $\ell$.
In particular, if they contain at least $\lambda = k^2 / (\poly (\log k))$ such $j$-hop edges, we have $c^{(k)}(P) \geq \lambda \cdot D_{\ell}$.
This will enable us to pay the fixed cost of $k^2D_{\ell}$ for a new path-level pair.

In the following, we make this classification more precise.
Let $P$ be an $\ell$-jump segment and let $H_1, \dots, H_r \in \cH_{\ell+1}$ such that 
\begin{enumerate}
\item
$V(P) \cap H_j \neq \emptyset$ for all $j \in [r]$,
\item
$V(P) \subseteq \bigcup_{j=1}^r H_j$, and
\item $H_i$ is left of $H_j$ for all $i,j\in[r]$ with $i<j$.
\end{enumerate}
Note that $r \geq 2$ because $\level(V(P)) = \ell$.
If $P$ is entry-buffered, we define
\begin{equation}
\label{eq:def_j_ast_entry}
j^\ast_{\mathrm{entry}} \coloneqq \max \parset{j \in [r] : \sum_{i = j}^r |H_{i} \cap V(P_{\mathrm{entry}})| \geq \frac{k}{64(L-1)}}.
\end{equation}
This is well-defined, because
\begin{equation*}
|V(P_{\mathrm{entry}})| \geq \frac{k}{16} - (\ell-1)\cdot \frac{k}{32(L-1)} \geq \frac{k}{32} \geq \frac{k}{64(L-1)}.
\end{equation*}
Analogously, if $P$ is exit-buffered, we define
\begin{equation}
\label{eq:def_j_ast_exit}
j^\ast_{\mathrm{exit}} \coloneqq \min \parset{j \in [r] : \sum_{i= 1}^j |H_{i} \cap V(P_{\mathrm{exit}})| \geq \frac{k}{64(L-1)}}.
\end{equation}
See \Cref{fig:choice_of_j_start_entry_and_exit} for an illustration.
\begin{figure}
\begin{center}
\begin{tikzpicture}[
vertex/.style={circle,draw=black, fill,inner sep=1.5pt, outer sep=1.0pt},
edge/.style={-latex, thick},
backedge/.style={-latex, densely dashed},
box/.style={thick, fill=Peach, fill opacity=0.05, rounded corners=8pt},
scale=2
]

\def\slackx{0.03}
\def\slacky{0.4}

\def\height{1.5}

\definecolor{color1}{named}{Peach}  
\definecolor{color2}{named}{ForestGreen}  
\definecolor{color3}{named}{Violet}  
\definecolor{color4}{named}{ProcessBlue}

\begin{scope}[box/.append style={draw=color1}]
\draw[box] (0.1,0) rectangle (1-\slackx,\height);
\draw[box] (6+\slackx,0) rectangle (6.9,\height);

\foreach \i in {1,...,5}{
\pgfmathtruncatemacro{\next}{\i+1}   
\ifnum\i=2             
\draw[box, ultra thick] (\i+\slackx,0) rectangle (\next-\slackx,\height);
\else\ifnum\i=4            
\draw[box, ultra thick] (\i+\slackx,0) rectangle (\next-\slackx,\height);
\else
\draw[box] (\i+\slackx,0) rectangle (\next-\slackx,\height);
\fi\fi
}
\end{scope}

\node at (0.5,\height+0.35) {\textcolor{color1}{$H_{1}$}};
\node at (1.5,\height+0.35) {\textcolor{color1}{$H_{2}$}};
\node at (2.5,\height+0.315) {\textcolor{color1}{$H_{3} = H_{j^\ast_{\mathrm{entry}}}$}};
\node at (3.5,\height+0.35) {\textcolor{color1}{$H_{4}$}};
\node at (3.5,-0.35) {\textcolor{color1}{$H_{4}$}};
\node at (4.5,-0.395) {\textcolor{color1}{$H_{5} = H_{j^\ast_{\mathrm{exit}}}$}};
\node at (5.5,-0.35) {\textcolor{color1}{$H_{6}$}};

\node at (6.5,-0.35) {\textcolor{color1}{$H_{7}$}};

\begin{scope}[every node/.style={vertex, draw=color2, fill=color2}]
\node (v1) at (1.5,1.25) {};
\node (v2) at (2.25,1.25) {};
\node (v3) at (2.75,1.25) {};
\node (v4) at (3.5,1.25) {};
\node (v5) at (4.25,1.25) {};
\node (v6) at (4.75,1.25) {};
\node (v7) at (5.5,1.25) {};
\node (v8) at (6.25,1.25) {};
\node (v9) at (6.7,1.25) {};
\end{scope}

\begin{scope}[every node/.style={vertex, draw=color3, fill=color3}]
\node (w1) at (0.3,0.25) {};
\node (w2) at (0.75,0.25) {};
\node (w3) at (1.25,0.25) {};
\node (w4) at (1.75,0.25) {};
\node (w5) at (2.5,0.25) {};
\node (w6) at (3.25,0.25) {};
\node (w7) at (3.75,0.25) {};
\node (w8) at (4.5,0.25) {};
\node (w9) at (5.25,0.25) {};
\node (w10) at (5.75,0.25) {};
\end{scope}

\foreach \i in {1,...,8}{
\pgfmathtruncatemacro{\next}{\i+1}                
\draw[edge, color2] (v\i) -- (v\next);
}

\foreach \i in {1,...,9}{
\pgfmathtruncatemacro{\next}{\i+1}                
\draw[edge, color3] (w\i) -- (w\next);
}

\node at (3.5,1) {\textcolor{color2}{$P_{\mathrm{entry}}$}};
\node at (3.5,0.5) {\textcolor{color3}{$P_{\mathrm{exit}}$}};

\draw[edge, out=-150,in=30, looseness=0.5, red, densely dashed, very thick] (v9) to (w1);

\draw[decorate, decoration={brace}, thick, color2] (2+\slackx,\height+0.05) -- (7-2*\slackx,\height+0.05) node[midway, above] {$\geq \frac{k}{64(L-1)}$};
\draw[decorate, decoration={brace}, thick, color3] (5-\slackx,-0.05) -- (0+2*\slackx,-0.05) node[midway, below] {$\geq \frac{k}{64(L-1)}$};
\end{tikzpicture}
 \end{center}
\caption{
Illustration of our choice of $j^{\ast}_{\mathrm{entry}}$ and $j^{\ast}_{\mathrm{exit}}$.
All vertices are drawn from left to right according to the total order $\prec$.
The $\ell$-jump segment $P$ consists of $P_{\mathrm{entry}}$ (green), the backward edge $e$ (red), and $P_{\mathrm{exit}}$ (purple).
In this example, we have a \emph{non-splitting} $\ell$-jump segment.
}
\label{fig:choice_of_j_start_entry_and_exit}
\end{figure}
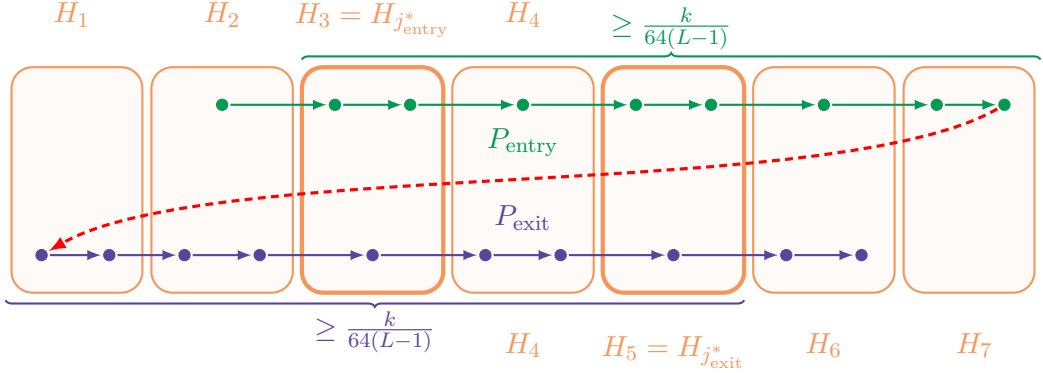
We observe that
\begin{equation}
\label{eq:entry_cutoff_ub}
\sum_{i = j^\ast_{\mathrm{entry}} + 1}^r |H_{i} \cap V(P_{\mathrm{entry}})| < \frac{k}{64(L-1)}
\end{equation}
and
\begin{equation}
\label{eq:exit_cutoff_ub}
\sum_{i = 1}^{j^\ast_{\mathrm{exit}} - 1} |H_{i} \cap V(P_{\mathrm{exit}})| < \frac{k}{64(L-1)},
\end{equation}
because otherwise this would contradict the choice of $j^\ast_{\mathrm{entry}}$ and $j^\ast_{\mathrm{exit}}$ as the maximum and minimum index with their corresponding property, respectively.
The sets $H_1, \dots, H_r$ are uniquely determined by the above properties.
Hence, the following is well-defined:

\begin{definition}
[Splitting $\ell$-jump segment]
Let $P$ be an $\ell$-jump segment that is entry- and exit-buffered.
Let $j^\ast_{\mathrm{entry}}$ and $j^\ast_{\mathrm{exit}}$ be defined as in \eqref{eq:def_j_ast_entry} and \eqref{eq:def_j_ast_exit}.
We say that $P$ is a \emph{splitting $\ell$-jump segment} if $j^\ast_{\mathrm{entry}} > j^\ast_{\mathrm{exit}}$.

An $\ell$-buffering is called $\ell$-\emph{splitting-free} if it contains no splitting $\ell$-jump segment.
\end{definition}
We remark that a $\ell$-splitting-free $\ell$-buffering may contain a splitting $\ell^\prime$-jump segment for some $\ell^\prime > \ell$.
From the definition of an $\ell$-splitting jump segment, we immediately obtain:
\begin{lemma}
\label{lem:splitting_i_jump}
Let $P$ be a splitting $\ell$-jump segment.
Then we have 
\[
\scconstant_1 \cdot c^{(k)} (P) \geq k^2D_{\ell},
\]
where $\scconstant_1 \coloneqq 2^{12} \cdot L^2$.
\end{lemma}
\begin{proof}
We define the $H_1, \dots, H_r \in \cH_{\ell+1}$ as above.
Any edge $e=(x,y)$ with $x \in \bigcup_{i= j^\ast_{\mathrm{entry}}}^r H_i$ and $y \in \bigcup_{i=1}^{j^\ast_{\mathrm{exit}}} H_i$ is a backward edge of level $\ell$, because $j^\ast_{\mathrm{entry}} > j^\ast_{\mathrm{exit}}$.
Thus, each such edge $e$ has cost $c(e) = D_{\ell}$.
For each $x \in V(P_{\mathrm{entry}})$ and each $y \in V(P_{\mathrm{exit}})$, the edge $(x,y)$ contributes to $c^{(k)}(P)$, because $|V(P)| \leq k + 1$.
Therefore, we conclude
\begin{align*}
c^{(k)}(P)\  \geq&\  \textstyle \left| V(P_{\mathrm{entry}}) \cap \Big(\bigcup_{i= j^\ast_{\mathrm{entry}}}^r H_i\Big) \right| \cdot \left| V(P_{\mathrm{exit}}) \cap \Big(\bigcup_{i=1}^{j^\ast_{\mathrm{exit}}} H_i\Big) \right|  \cdot D_{\ell}\\
\geq&\ 
\left( \tfrac{k}{64(L-1)} \right)^2 \cdot D_{\ell} \\
\geq&\  \tfrac{1}{\scconstant_1} \cdot k^2D_{\ell}. 
\end{align*}
\end{proof}

\subsection{Handling Non-Splitting Jump Segments}\label{sec:nonsplitting}

Next, we explain how we handle non-splitting jump segments.
Our goal will be to prove the following:

\begin{restatable}{lemma}{ManyReplacements}
\label{lem:many_replacements}
Let $P_{\mathrm{main}}$ be a path with an $\ell$-splitting-free $\ell$-buffering.
Then, there is an $(\ell+1)$-buffered path $P^\prime_{\mathrm{main}}$ and a subset $R^\ast \subseteq V(P_{\mathrm{main}})$ with $V(P_{\mathrm{main}}) = V(P^\prime_{\mathrm{main}}) \cupp R^\ast $ and
\begin{equation}
\label{eq:many_replacements1}
c^{(k)}_{> \ell} (P^\prime_{\mathrm{main}}) \leq  c^{(k)} (P_{\mathrm{main}}),
\end{equation}
and
\begin{equation}
\label{eq:many_replacements2}
c^{(k)} (P^\prime_{\mathrm{main}}) + |R^\ast| \cdot 2kD_{\ell} \leq \left(\scconstant_2 + 2^{13}\right) \cdot c^{(k)} (P_{\mathrm{main}}),
\end{equation}
where $\scconstant_2 \coloneqq 2^8L$.
\end{restatable}

The bound~\eqref{eq:many_replacements2} will be useful, because we will later be able to re-insert the removed vertices, i.e, those from $R^\ast$, into some path at cost $2kD_{\ell}$ per vertex.
Observe that the overall cost increases by a factor of $(\scconstant_2 + 2^{13})$ and we cannot afford this to happen $L$ times, when iteratively applying the argument to each level $\ell \in [L]$.
This is where the bound~\eqref{eq:many_replacements1} will become important.
Because we will later recursively apply our argument to subpaths inside sets from $\cH_{\ell + 1}$, this bound ensures that we only lose an additive term of $(\scconstant_2 + 2^{13}-1)$ times the $k$-hop-cost of our original path in each application of the lemma.
\addparskip

To prove \Cref{lem:many_replacements}, we will remove all backward edges of level $\ell$ and ensure that the resulting path is $\ell$-buffered.
We then observe that an $\ell$-buffered path with no backward edges of level $\ell$ is $(\ell+1)$-buffered:

\begin{lemma}
\label{lem:shrink_buffering}
Let $P$ be an $\ell$-buffered path.
If $P$ does not contain any backward edge of level $\ell$, then $P$ is $(\ell+1)$-buffered.
\end{lemma}
\begin{proof}
Let $\cB$ be an $\ell$-buffering of $P$.
If $\cB$ contains no $\ell$-jump segment (that is, all elements of $\cB$ are $\ell^\prime$-jump segments for some $\ell^\prime > \ell$), then $\cB$ is an $(\ell+1)$-buffering.
Otherwise, let $P^e$ be an $\ell$-jump segment and let $e$ be the unique backward edge contained in $P^e$.
By our assumption, we have $\level(e) \geq \ell + 1$.
Let $Q^e$ be the maximal subpath of $P^e$ that contains $e$ and satisfies $\level(Q^e) \geq \ell + 1$.
Then $Q^e$ is an $\ell^\prime$-jump segment with $\ell^\prime \geq \ell+ 1 $.
We denote by $s$ and $x$ the start vertices of $P^e$ and $Q^e$, respectively.
See \Cref{fig:shrinking_jump_segments}.

Suppose that $s = x$.
If $P^e$ is entry-safe within $P$, then so is $Q^e$.
If $P^e$ is entry-buffered, then the entry part of $Q^e$ contains at least 
\begin{equation*}
\frac{k}{16} - (\ell-1)\cdot \frac{k}{32(L-1)}
\ \geq\  
\frac{k}{16} - (\ell^\prime -1)\cdot \frac{k}{32(L-1)}
\end{equation*}
vertices.
Hence, $Q^e$ is entry-buffered in this case.

Now, suppose that $s \neq x$.
Let $s^\prime$ be the predecessor of $x$ in the path $P^e$ and let $H \in \cH_{\ell + 1}$ with $V(Q^e) \subseteq H$.
Note that $(s^\prime, x)$ is a forward edge, since $P^e$ contains only one backward edge and this backward edge $e$ is contained in $Q^e$.
Therefore, because we chose $Q^e$ maximal with the property that $V(Q^e) \subseteq H$, the edge $(s^\prime, x)$ is a forward edge of level $\ell$.
This implies that $Q^e$ is entry-safe within $P$, because $P$ does not contain any backward edges of level $\leq \ell$.
See \Cref{fig:shrinking_jump_segments} for a visualization.
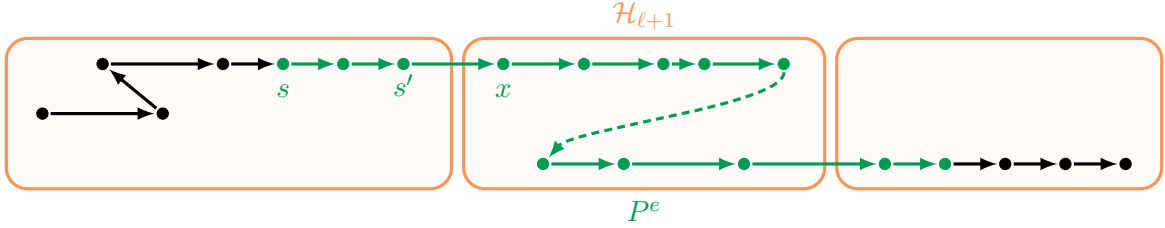
\begin{figure}
\begin{center}
\begin{tikzpicture}[
xscale=1.6,
vertex/.style={circle,draw=black, fill,inner sep=1.5pt, outer sep=1pt},
edge/.style={-latex,very thick},
box/.style={very thick, fill=Peach, fill opacity=0.05, rounded corners=8pt},
]

\definecolor{color1}{named}{Peach}  
\definecolor{color2}{named}{ForestGreen}  
\definecolor{color3}{named}{Violet}  
\definecolor{color4}{named}{ProcessBlue} 

\draw[box,draw=color1] (-0.8,0) rectangle (2.9,2);
\draw[box,draw=color1] (3,0) rectangle (6,2);
\draw[box,draw=color1] (6.1,0) rectangle (8.8,2);

\begin{scope}[every node/.style={vertex}]
\node (v1) at (0,1.66) {};
\node (v2) at (1,1.66) {};

\node (w1) at (-0.5,1) {};
\node (w2) at (0.5,1) {};

\node (v16) at (7.5,0.33) {};
\node (v17) at (8,0.33) {};
\node (v18) at (8.5,0.33) {};
\end{scope}

\begin{scope}[every node/.style={vertex,color2}]
\node (v3) at (1.5,1.66) {};
\node (v4) at (2,1.66) {};
\node (v5) at (2.5,1.66) {};

\node (v6) at (3.33,1.66) {};
\node (v7) at (4,1.66) {};
\node (v8) at (4.66,1.66) {};
\node (v9) at (5,1.66) {};
\node (v10) at (5.66,1.66) {};

\node (v11) at (3.66,0.33) {};
\node (v12) at (4.33,0.33) {};
\node (v13) at (5.33,0.33) {};

\node (v14) at (6.5,0.33) {};
\node (v15) at (7,0.33) {};
\end{scope}

\draw[edge] (w1) to (w2);
\draw[edge]  (w2) to (v1);

\foreach \i in {1,...,2} {
\pgfmathtruncatemacro{\next}{\i+1}
\draw[edge] (v\i) to (v\next);
}

\foreach \i in {3,...,9} {
\pgfmathtruncatemacro{\next}{\i+1}
\draw[edge,color2] (v\i) to (v\next);
}
\draw[edge, out=-90,in=60, looseness=0.5, color2, densely dashed] (v10) to (v11);
\foreach \i in {11,...,14} {
\pgfmathtruncatemacro{\next}{\i+1}
\draw[edge,color2] (v\i) to (v\next);
}
\foreach \i in {15,...,17} {
\pgfmathtruncatemacro{\next}{\i+1}
\draw[edge] (v\i) to (v\next);
}

\node[color1] at (4.5,2.3) {$\cH_{\ell+1}$};
\node[below=4pt, color2] at (v3) {$s$};
\node[below,color2] at (v5) {$s^\prime$};
\node[below=4pt,color2] at (v6) {$x$};
\node[color2] at (4.5,-0.3) {$P^e$};
\end{tikzpicture}
 \end{center}
\caption{
Visualization of the proof of \Cref{lem:shrink_buffering} in the case $s \neq x$.
All vertices are drawn from left to right according to the total order $\prec$.
The $\ell$-jump segment $P_j$ (green) can be shrunk so that the resulting segment is entry-safe (and in this example also exit-safe).
}
\label{fig:shrinking_jump_segments}
\end{figure}
Similarly, one can show that $Q^e$ is exit-buffered or exit-safe within $P$ in any case.

In this way, we can replace every $\ell$-jump segment $P^e$ in the $\ell$-buffering $\cB$ by an $\ell^\prime$-jump segment~$Q^e$ for some $\ell^\prime > \ell$.
\end{proof}

We now discuss how we can get rid of backward edges of level $\ell$.
For an entry-safe and exit-safe $\ell$-jump segment $P$ containing a backward edge $e$ of level $\ell$,
we can get rid of $e$ (without introducing new backward edges of level $\ell$) by replacing $P$ with its $\ell$-refinement.
In particular, the path that arises by replacing $P$ with its $\ell$-refinement contains fewer backward edges of level $\ell$.
If $P$ is not entry-safe or not exit-safe, we will need to be more careful.
The idea is to first choose a removal set $R \subseteq V(P)$ such that it is safe to replace $P$ with the $\ell$-refinement of the path after removing $R$.

To define the removal set $R$, we define $H_1, \dots, H_r \in \cH_{\ell+1}$ with $j^\ast_{\mathrm{entry}}$ (if $P$ is entry-buffered) and $j^\ast_{\mathrm{exit}}$ (if $P$ is exit-buffered) as previously.
If $P$ is entry-buffered, the index $j^\ast_{\mathrm{entry}}$ is defined and we set
\begin{equation*}
R_{\mathrm{exit}} \coloneqq \bigcup_{j = 1}^{j^{\ast}_{\mathrm{entry}} - 1} H_j \cap V(P_{\mathrm{exit}}),
\end{equation*}
otherwise, we define $R_{\mathrm{exit}} \coloneqq \emptyset$.
Analogously, if $P$ is exit-buffered, we set
\begin{equation*}
R_{\mathrm{entry}} \coloneqq \bigcup_{j = j^{\ast}_{\mathrm{exit}} + 1}^{r} H_j \cap V(P_{\mathrm{entry}}),
\end{equation*}
and otherwise we set $R_{\mathrm{entry}} \coloneqq \emptyset$.
We call $R\coloneqq R_{\mathrm{entry}} \cup R_{\mathrm{exit}}$ the \emph{removal set} of the $\ell$-jump segment $P$.
See \Cref{fig:proof_of_valid_replacement} for an illustration.
\begin{figure}
\begin{center}
\begin{tikzpicture}[
vertex/.style={circle,draw=black, fill,inner sep=1.5pt, outer sep=1.0pt},
edge/.style={-latex, thick},
backedge/.style={-latex, densely dashed},
box/.style={thick, fill=Peach, fill opacity=0.05, rounded corners=8pt},
scale=2
]

\def\slackx{0.03}
\def\slacky{0.4}

\def\height{1.5}

\definecolor{color1}{named}{Peach}  
\definecolor{color2}{named}{ForestGreen}  
\definecolor{color3}{named}{Violet}  
\definecolor{color4}{named}{ProcessBlue}

\begin{scope}[box/.append style={draw=color1}]
\draw[box] (0.1,0) rectangle (1-\slackx,\height);
\draw[box] (6+\slackx,0) rectangle (6.9,\height);

\foreach \i in {1,...,5}{
\pgfmathtruncatemacro{\next}{\i+1}   
\ifnum\i=2             
\draw[box, ultra thick] (\i+\slackx,0) rectangle (\next-\slackx,\height);
\else\ifnum\i=4            
\draw[box, ultra thick] (\i+\slackx,0) rectangle (\next-\slackx,\height);
\else
\draw[box] (\i+\slackx,0) rectangle (\next-\slackx,\height);
\fi\fi
}
\end{scope}
\node at (2.5,\height+0.2) {\textcolor{color1}{$H_{j^\ast_{\mathrm{entry}}}$}};
\node at (4.5,-0.2) {\textcolor{color1}{$H_{j^\ast_{\mathrm{exit}}}$}};

\begin{scope}[every node/.style={vertex, draw=color2, fill=color2}]
\node (v1) at (1.5,1.25) {};
\node (v2) at (2.25,1.25) {};
\node (v3) at (2.75,1.25) {};
\node (v4) at (3.5,1.25) {};
\node (v5) at (4.25,1.25) {};
\node (v6) at (4.75,1.25) {};
\node (v7) at (5.5,1.25) {};
\node (v8) at (6.25,1.25) {};
\node (v9) at (6.7,1.25) {};
\end{scope}

\begin{scope}[every node/.style={vertex, draw=color3, fill=color3}]
\node (w1) at (0.3,0.25) {};
\node (w2) at (0.75,0.25) {};
\node (w3) at (1.25,0.25) {};
\node (w4) at (1.75,0.25) {};
\node (w5) at (2.5,0.25) {};
\node (w6) at (3.25,0.25) {};
\node (w7) at (3.75,0.25) {};
\node (w8) at (4.5,0.25) {};
\node (w9) at (5.25,0.25) {};
\node (w10) at (5.75,0.25) {};
\end{scope}

\foreach \i in {1,...,8}{
\pgfmathtruncatemacro{\next}{\i+1}                
\draw[edge, color2] (v\i) -- (v\next);
}

\foreach \i in {1,...,9}{
\pgfmathtruncatemacro{\next}{\i+1}                
\draw[edge, color3] (w\i) -- (w\next);
}

\node at (3.5,1) {\textcolor{color2}{$P_{\mathrm{entry}}$}};
\node at (3.5,0.5) {\textcolor{color3}{$P_{\mathrm{exit}}$}};

\draw[edge, out=-150,in=30, looseness=0.5, red, densely dashed, very thick] (v9) to (w1);

\begin{scope}[box/.append style={draw=color4, fill opacity=0, ultra thick}]
\draw[box] (0.1,0) rectangle (2-\slackx,0.5);
\draw[box] (5+\slackx,1) rectangle (6.9,1.5);
\end{scope}
\node at (6,\height+0.2) {\textcolor{color4}{$R_{\mathrm{entry}}$}};
\node at (1,-0.2) {\textcolor{color4}{$R_{\mathrm{exit}}$}};

\end{tikzpicture}
 \end{center}
\caption{
Illustration of our choice of $R_{\mathrm{exit}}$ and $R_{\mathrm{entry}}$.
All vertices are drawn from left to right according to the total order $\prec$.
Observe that removing $R \coloneqq R_{\mathrm{exit}} \cup R_{\mathrm{entry}}$ from the depicted non-splitting $\ell$-jump segment and replacing it with its $\ell$-refinement changes neither its start vertex nor its end vertex.
}
\label{fig:proof_of_valid_replacement}
\end{figure}
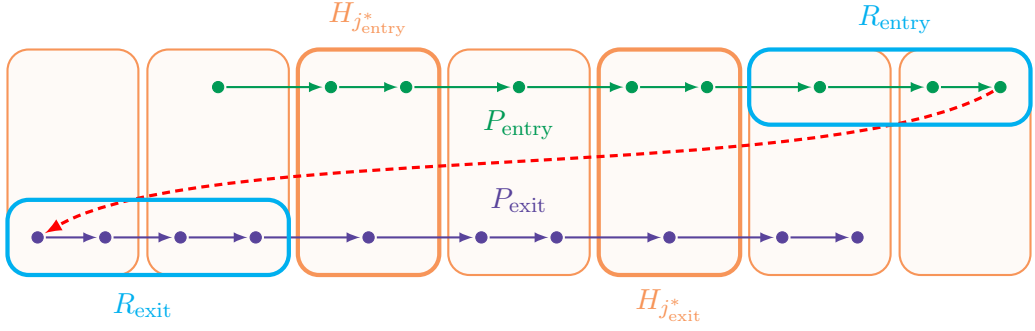

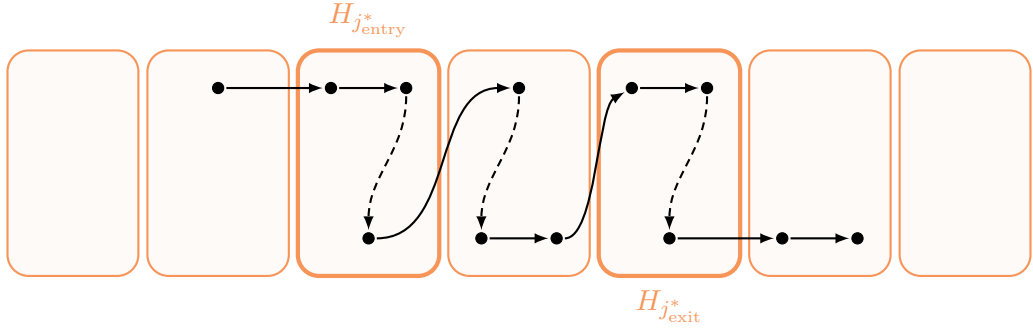
\begin{figure}
\begin{center}
\begin{tikzpicture}[
vertex/.style={circle,draw=black, fill,inner sep=1.5pt, outer sep=1.0pt},
edge/.style={-latex, thick},
backedge/.style={-latex, densely dashed},
box/.style={thick, fill=Peach, fill opacity=0.05, rounded corners=8pt},
scale=2
]

\def\slackx{0.03}
\def\slacky{0.4}

\def\height{1.5}

\definecolor{color1}{named}{Peach}  
\definecolor{color2}{named}{ForestGreen}  
\definecolor{color3}{named}{Violet}  
\definecolor{color4}{named}{ProcessBlue}

\begin{scope}[box/.append style={draw=color1}]
\draw[box] (0.1,0) rectangle (1-\slackx,\height);
\draw[box] (6+\slackx,0) rectangle (6.9,\height);

\foreach \i in {1,...,5}{
\pgfmathtruncatemacro{\next}{\i+1}   
\ifnum\i=2             
\draw[box, ultra thick] (\i+\slackx,0) rectangle (\next-\slackx,\height);
\else\ifnum\i=4            
\draw[box, ultra thick] (\i+\slackx,0) rectangle (\next-\slackx,\height);
\else
\draw[box] (\i+\slackx,0) rectangle (\next-\slackx,\height);
\fi\fi
}
\end{scope}
\node at (2.5,\height+0.2) {\textcolor{color1}{$H_{j^\ast_{\mathrm{entry}}}$}};
\node at (4.5,-0.2) {\textcolor{color1}{$H_{j^\ast_{\mathrm{exit}}}$}};

\begin{scope}[every node/.style={vertex, draw=color2, fill=color2}]

\end{scope}

\begin{scope}[every node/.style={vertex}]
\node (v1) at (1.5,1.25) {};
\node (v2) at (2.25,1.25) {};
\node (v3) at (2.75,1.25) {};
\node (v4) at (3.5,1.25) {};
\node (v5) at (4.25,1.25) {};
\node (v6) at (4.75,1.25) {};

\node (w5) at (2.5,0.25) {};
\node (w6) at (3.25,0.25) {};
\node (w7) at (3.75,0.25) {};
\node (w8) at (4.5,0.25) {};
\node (w9) at (5.25,0.25) {};
\node (w10) at (5.75,0.25) {};
\end{scope}

\draw[edge] (v1) -- (v2);
\draw[edge] (v2) -- (v3);
\draw[edge, out=-90,in=90, looseness=1, densely dashed] (v3) to (w5);
\draw[edge, out=0,in=-180, looseness=1] (w5) to (v4);
\draw[edge, out=-90,in=90, looseness=1, densely dashed] (v4) to (w6);
\draw[edge] (w6) -- (w7);
\draw[edge, out=0,in=-150, looseness=0.5] (w7) to (v5);
\draw[edge] (v5) -- (v6);
\draw[edge, out=-90,in=90, looseness=1, densely dashed] (v6) to (w8);
\draw[edge] (w8) -- (w9);
\draw[edge] (w9) -- (w10);
\end{tikzpicture}
 \end{center}
\caption{
The $\ell$-refinement $Q$ of the path $P[V(P)\setminus R]$ we obtain for the $\ell$- jump segment $P$ from \Cref{fig:proof_of_valid_replacement}.
All three dashed edges are new backward edges whose level is strictly larger than $\ell$.
}
\label{fig:proof_of_valid_replacement_path_Q}
\end{figure}

For a path $P$ and $A \subseteq V(P)$, we write $P[A]$ to denote the path obtained by skipping all vertices $v \in V(P) \setminus A$.
\addparskip

To prove \Cref{lem:many_replacements}, we consider a path $P_{\rm main}$ with an $\ell$-splitting-free $\ell$-buffering $\cB$.
We denote by $\cB_{\ell}$ the set of jump segments from $\cB$ that contain a backward edge of level $\ell$.

Our procedure to get rid of all backward edges of level $\ell$ can then simply be described as follows:
In the beginning, we initialize $R^\ast \coloneqq \emptyset$ as the empty set.
For each $P \in \cB_{\ell}$, we replace $P$ with the $\ell$-refinement of $P[V(P) \setminus R]$, where $R$ is the removal set of $P$, and we add the vertices from $R$ to $R^\ast$.
See also \Cref{fig:proof_of_valid_replacement_path_Q} for an example.
We denote by $P_{\rm main}^\prime$ the path that results from applying this procedure to the path $P_{\rm main}$.

We now argue that $P_{\rm main}^\prime$, together with $R^\ast$, indeed satisfies the properties claimed in \Cref{lem:many_replacements}.
We start by proving that $P_{\rm main}^\prime$ is $(\ell + 1)$-buffered.
In this proof we exploit the fact that $j^\ast_{\rm entry}$ was not chosen smaller and $j^\ast_{\rm exit}$ was not chosen larger than how we chose it in \eqref{eq:def_j_ast_entry} and \eqref{eq:def_j_ast_exit}.

\begin{lemma}
$P_{\rm main}^\prime$ is $(\ell + 1)$-buffered.
\end{lemma}
\begin{proof}
We prove the following four claims in this particular order:
\begin{enumerate}
\item\label{item:claim:non_empty}
Consider an $\ell$-jump segment $P \in \cB_{\ell}$ that is replaced with the $\ell$-refinement $Q$ of $P[V(P) \setminus R]$ in our procedure.
Then we have $V(P) \setminus R \neq \emptyset$, that is, $Q$ contains at least one vertex.

\item\label{item:claim_incoming_and_outgoing_edges}
Let $P$ and $Q$ be as in Claim~\ref{item:claim:non_empty}.
Then the edge $e_1$ of $P_{\rm main}^\prime$ that enters the first vertex of $Q$ and the edge $e_2$ of $P_{\rm main}^\prime$ that leaves the last vertex of $Q$ are forward edges (in case such edges exist).

\item\label{item:claim_no_backward_edge_level_ell}
$P_{\rm main}^\prime$ contains no backward edge of level $\leq \ell$.

\item\label{item:claim_ell_buffered}
$P_{\rm main}^\prime$ is $\ell$-buffered.
\end{enumerate}
Observe that by \Cref{lem:shrink_buffering}, the Claims~\ref{item:claim_no_backward_edge_level_ell} and~\ref{item:claim_ell_buffered} together imply that $P_{\rm main}^\prime$ is $(\ell + 1)$-buffered.
We prove each claim individually:

\begin{enumerate}
\item[\ref{item:claim:non_empty}]
Suppose for the sake of deriving a contradiction that $V(P)\setminus R = \emptyset$.
Then $ R_{\rm entry} = V(P_{\rm entry}) \neq \emptyset$ and  $ R_{\rm exit} = V(P_{\rm exit}) \neq \emptyset$.
In particular, $P$ is entry-buffered and exit-buffered and thus $j^\ast_{\rm entry}$ and $j^\ast_{\rm exit}$ are defined.
By the definition of $j^\ast_{\rm entry}$, there exists a vertex $v\in V(P_{\rm entry})\cap H_{j^\ast_{\rm entry}} \neq \emptyset$.
Because $P \in \cB_{\ell}$ and $\cB$ is $\ell$-splitting-free, we have $j^\ast_{\rm entry} \leq j^\ast_{\rm exit}$, implying $v\notin R_{\rm entry}$.
This contradicts $V(P)\setminus R = \emptyset$.

\item[\ref{item:claim_incoming_and_outgoing_edges}]
To show the second claim, we prove that $Q$ is entry-safe within $P_{\rm main}^\prime$ or the first vertex of $P$ is identical to the first vertex of $Q$.
By a symmetric argument, it then follows that $Q$ is exit-safe within $P_{\rm main}^\prime$ or the last vertex of $P$ is identical to the last vertex of $Q$.
Applying this to all the replacement steps we perform in our procedure proves the second claim because any backward edge of $P_{\rm main}$ must be contained in a jump-segment from $\cB$ and these jump segments are vertex disjoint.

First, we observe that if $P$ is entry-safe within $P_{\rm main}$, this immediately implies that $Q$ is entry-safe within $P_{\rm main}^\prime$.
Now consider the case where $P$ is entry-buffered and we show that in this case the first vertex of the path $Q$ is identical to the first vertex of $P$.
Let $H_1, \dots, H_r$ be the sets as previously defined for $P$.
If $P$ is entry-buffered, the index $j^\ast_{\rm entry}$ is defined.
Let $j\in[r]$ be such that $H_j$ contains the first vertex $v$ of the path $P$.
Because $P_{\rm entry}$ is monotone, we have $j \leq j^\ast_{\rm entry}$ and every vertex $w \in V(P_{\rm entry})$ is contained in a set $H_i$ with $i \geq j$.
We have $j^\ast_{\rm entry} \leq j^\ast_{\rm exit}$ if $j^\ast_{\rm exit}$ is defined because $\cB$ is $\ell$-splitting-free.
This implies that $v \notin R_{\rm entry}$ and therefore $v \in V(Q)$.
Moreover, by the definition of $R_{\rm exit}$, every vertex $w \in V(P_{\rm exit})$ that is contained in a set $H_i$ with $i < j$ is contained in $R$.
We conclude that every vertex $w\in V(P)\setminus R$ is contained in a set $H_i$ with $i \geq j$, which implies (by the definition of the $\ell$-refinement) that indeed $v$ is the first vertex of the path~$Q$.

A symmetric argument shows that $Q$ is exit-safe or the last vertex of $P$ is identical to the last vertex of~$Q$.

\item[\ref{item:claim_no_backward_edge_level_ell}]
This follows immediately from Claims~\ref{item:claim:non_empty} and~\ref{item:claim_incoming_and_outgoing_edges} because the $\ell$-refinement $Q$ of $P[V(P) \setminus R]$ does not contain any backward edge of level~$\ell$ by the definition of the $\ell$-refinement.
Hence, all backward edges of level~$\ell$ are removed by our procedure (and by Claims~\ref{item:claim:non_empty} and~\ref{item:claim_incoming_and_outgoing_edges} no new ones are introduced).

\item[\ref{item:claim_ell_buffered}]
In order to show that $P_{\rm main}^\prime$ is $\ell$-buffered, we extend our procedure described above and explicitly construct an $\ell$-buffering of $P_{\rm main}^\prime$ during the process of obtaining $P_{\rm main}^\prime$ from $P_{\rm main}$.
Initially, we set $\cB^\prime \coloneqq \cB$.
Consider an iteration in which we replace an $\ell$-jump segment $P \in \cB_{\ell}$ with the $\ell$-refinement $Q$ of $P[V(P) \setminus R]$.
As before, let $H_1, \dots, H_r \in \cH_{\ell + 1}$ be the sets defined above for $P$.
Because $P_{\rm entry}$ and $P_{\rm exit}$ are monotone, the $\ell$-refinement $Q$ of $P[V(P)\setminus R]$ contains at most one backward edge within each set $H_i$ (with $i\in [r]$).
See \Cref{fig:proof_of_valid_replacement_path_Q}.
For a backward edge $e$ of $Q$, let $H_i$ be the set containing $e$ and let $Q^e\coloneqq Q[V(Q)\cap H_i]$.
In other words, $Q^e$ is the maximal subpath of $Q$ within the set $H_i$.
Then, we replace $P$ in the current $\cB^\prime$ via
\begin{equation}
\label{eq:replace_buffering}
\cB^\prime \ \coloneqq\ \left(\cB^\prime \setminus \big\{P\big\}\right) \cup \big\{ Q^e : e \text{ is a backward edge of }Q \big\}.
\end{equation}

We now show that the resulting $\cB^\prime$ at the end of our procedure is an $\ell$-buffering of $P_{\rm main}^\prime$.
By construction, $\cB^\prime$ is a collection of vertex disjoint subpaths of $P^\prime_{\rm main}$ and every backward edge of $P_{\rm main}^\prime$ is contained in one of these subpaths (using Claim~\ref{item:claim_incoming_and_outgoing_edges}).
Moreover, every element of $\cB'$ is an $\ell^\prime$-jump segment for some $\ell^\prime \geq \ell$.

We consider a single iteration of our procedure in which we replace $P$ with $Q$ (and modify $\cB^\prime$ as in~\eqref{eq:replace_buffering}).
First, note that for any jump segment $P^\prime \in \cB^\prime \setminus \{P\}$, the properties of being entry-safe, exit-safe, entry-buffered, and exit-buffered are preserved under this replacement step.
Hence, in order to complete the proof of this claim, it suffices to show that each $Q^e$ is entry-buffered or entry-safe within $P_{\rm main}^\prime$, and exit-buffered or exit-safe within $P_{\rm main}^\prime$, where $e$ is a backward edge of $Q$ and $Q^e$ is defined as above.

Let $Q^e$ be such a subpath of $Q$ and we show that $Q^e$ is entry-buffered or entry-safe within $P^\prime_{\rm main}$.
Consider the case that $Q^e$ is not an initial part of $Q$, i.e., the start points of $Q^e$ and $Q$ are not identical.
Since $Q^e$ is the maximal subpath of $Q$ that is contained in $H_j \in \cH_{\ell + 1}$ (for appropriately chosen $j \in [r]$) and $Q$ does not contain any backward edges of level $\ell$, the path $Q$ and thus $P_{\rm main}^\prime$ contains a forward edge of level $\ell$ that enters the first vertex of $Q^e$.
Therefore Claim~\ref{item:claim_no_backward_edge_level_ell} implies that, in this case, $Q^e$ is entry-safe within $P_{\rm main}^\prime$.

Now consider the case that $Q^e$ is an initial part of $Q$.
If $P$ is entry-safe within $P_{\rm main}$, then $Q^e$ is entry-safe within $P_{\rm main}^\prime$.
So we may assume that $P$ is entry-buffered.
Recall from the proof of Claim~\ref{item:claim_incoming_and_outgoing_edges} that in this case the first vertex of $Q$ is identical to the first vertex of $P$, and thus the first vertex of $P_{\rm entry}$.

We have $Q^e = Q[ V(Q)\cap H_j]$ for some $j \in [r]$.
We claim $j = j^\ast_{\rm entry}$.
Because $Q^e$ contains a backward edge, it must contain a vertex from $P_{\rm exit}$ and thus we must have $j \geq j^\ast_{\rm entry}$
(where we used that the vertices in $R_{\rm exit}$ were removed and are not contained in $Q^e$).
If $P$ is exit-buffered, then $ j^\ast_{\rm entry} \leq   j^\ast_{\rm exit}$ because $P \in \cB_{\ell}$ and $\cB$ is $\ell$-splitting-free.
In particular, the vertices from $H_{j^\ast_{\rm entry}}\cap V(P_{\rm entry})$ are not contained in the removal set $R$.
This also applies in case $P$ is not exit-buffered in which case $R_{\rm entry} =\emptyset$.
We conclude that $V(Q) \cap H_{j^\ast_{\rm entry}} \neq \emptyset$ (where we used that $P$ contains at least one vertex from $H_{j^\ast_{\rm entry}}$ by the definition of  $ j^\ast_{\rm entry}$).
Because $Q^e$ contains the first vertex of $Q$, this implies $j \leq j^\ast_{\rm entry}$.
Thus, indeed $j= j^\ast_{\rm entry}$.

Recall that $R_{\rm entry} =\emptyset$ (in case $P$ is not exit-buffered)  or $j= j^\ast_{\rm entry} \leq   j^\ast_{\rm exit}$ because $P$ is non-splitting.
In both cases, $R_{\rm entry} \cap H_j = \emptyset$.
We conclude that the entry part of the jump segment $Q^e$ is $P[V(P_{\rm entry}) \cap H_j]$.
Moreover, the number of vertices in this entry part is 
\begin{equation*}
\big|V(P_{\rm entry})\big| - \sum_{i= j^\ast_{\rm entry}+1}^{r} |V(P_{\rm entry} \cap H_i| ,
\end{equation*}
which by the definition of $j^\ast_{\rm entry}$ is at least
\begin{align*}
\big|V(P_{\rm entry})\big| - \frac{k}{64(L-1)} \ \geq\  \frac{k}{16} - (\ell-1)\cdot \frac{k}{32(L-1)} -  \frac{k}{64(L-1)}\ \geq\ \frac{k}{16} - \ell \cdot \frac{k}{32(L-1)},
\end{align*}
implying that $Q^e$ is entry-buffered (where we used $\level(V(Q^e)) \geq \ell + 1$).
\addparskip

An analogous argument implies that $Q^e$ is exit-buffered or exit-safe within $P^\prime_{\rm main}$.
This completes the proof of Claim~\ref{item:claim_ell_buffered} and thus of our lemma. \qedhere
\end{enumerate}
\end{proof}

Because reinserting removed vertices comes at some cost, we need to bound the number of vertices in the removal set $R$.
In the following lemma we use that $j^\ast_{\rm entry}$ was not chosen larger and $j^\ast_{\rm exit}$ was not chosen smaller than how we chose it in \eqref{eq:def_j_ast_entry} and \eqref{eq:def_j_ast_exit}.

\begin{lemma}\label{lem:size_removal_set}
Let $P$ be an $\ell$-jump segment. 
Then its removal set $R$ satisfies
\begin{equation}
\label{eq:valid_replacement2}
|R|\cdot 2kD_{\ell} \leq \scconstant_2 \cdot c^{(k)} (P),
\end{equation}
where $\scconstant_2 \coloneqq 2^8L $.
\end{lemma}
\begin{proof}
For each vertex $x \in R_{\mathrm{exit}}$ and each vertex $v \in \bigcup_{j= j^{\ast}_{\mathrm{entry}}}^r \big(V(P_{\mathrm{entry}}) \cap H_j \big)$,
the edge $(v,x)$ is a backward edge of level $\ell$ and this edge contributes to $c^{(k)}(P)$, because $|V(P)| \leq k + 1$.
Therefore,
\[
c^{(k)}(P) \ \geq\ \left( |R_{\mathrm{exit}}| \cdot \sum_{j= j^{\ast}_{\mathrm{entry}}}^r |V(P_{\mathrm{entry}}) \cap H_j |  \right) \cdot D_{\ell} 
\ \geq\  \left( |R_{\mathrm{exit}}| \cdot \frac{k}{64(L-1)} \right) \cdot D_{\ell},
\]
where the second inequality follows from the definition of $j^{\ast}_{\mathrm{entry}}$.
Analogously, we also obtain $c^{(k)}(P)  \geq \big( |R_{\mathrm{entry}}| \cdot \frac{k}{64(L-1)} \big) \cdot D_{\ell}$, and thus
\[
2\cdot c^{(k)}(P) \ \geq\ \left( |R_{\mathrm{exit}}| \cdot \tfrac{k}{64(L-1)} + |R_{\mathrm{entry}}| \cdot \tfrac{k}{64(L-1)} \right) \cdot D_{\ell} \\
\ =\  |R| \cdot D_{\ell} \cdot \tfrac{k}{64(L-1)}.
\]
\end{proof}

We are now ready to complete the proof of \Cref{lem:many_replacements} with the following cost bounds:
\begin{lemma}
We have
\begin{equation*}
c^{(k)}_{> \ell} (P^\prime_{\mathrm{main}}) \leq  c^{(k)} (P_{\mathrm{main}}),
\end{equation*}
and
\begin{equation*}
c^{(k)} (P^\prime_{\mathrm{main}}) + |R^\ast| \cdot 2kD_{\ell} \leq \left(\scconstant_2 + 2^{13}\right) \cdot c^{(k)} (P_{\mathrm{main}}).
\end{equation*}
\end{lemma}
\begin{proof}
Suppose our procedure constructing $P^\prime_{\mathrm{main}}$ from $P_{\mathrm{main}}$ as explained above replaces $P_1, \dots, P_m \in \cB_{\ell}$ with $Q_1, \dots, Q_m$, respectively.
We denote by $R_j$ the removal set of each $P_j$, i.e., $R^\ast = \bigcup_{j=1}^m R_j$.

Note that $P_j[V(P_j) \setminus R_j]$ is a subpath of $P_{\rm main}[V(P_{\rm main}) \setminus R^\ast]$, and $P^\prime_{\rm main}$ can be equivalently obtained from $P_{\rm main}[V(P_{\rm main}) \setminus R^\ast]$ by replacing each $P_j[V(P_j) \setminus R_j]$ with $Q_j$.
Thus, by \Cref{lem:i-refinement} and \Cref{lem:khop_skipping_vertices}, we have
\begin{equation*}
\begin{split}
c^{(k)}_{> \ell} (P^\prime_{\rm main}) 
&\leq c^{(k)}_{> \ell} \Big( P_{\rm main}[V(P_{\rm main}) \setminus R^\ast] \Big) + \sum_{j=1}^{m} c^{(k)}_{=\ell}\Big( P_j[V(P_j) \setminus R_j] \Big) \\
&\leq c^{(k)}_{> \ell} \Big( P_{\rm main}[V(P_{\rm main}) \setminus R^\ast] \Big) 
 + c^{(k)}_{= \ell} \Big( P_{\rm main}[V(P_{\rm main}) \setminus R^\ast] \Big) \\
&\leq c^{(k)} \Big( P_{\rm main}[V(P_{\rm main}) \setminus R^\ast] \Big) \\
&\leq c^{(k)} (P_{\rm main}).
\end{split}
\end{equation*}

For $j \in [m]$, \Cref{lem:i-refinement,lem:khop_skipping_vertices} imply that
\begin{equation*}
c^{(k)}(Q_j) \leq c^{(k)}\Big(P_j[V(P_j) \setminus R_j]\Big) \leq c^{(k)}(P_j). 
\end{equation*}
Combining this with \Cref{lem:path_replacement} and \Cref{lem:size_removal_set}, we obtain
\begin{equation*}
\begin{split}
c^{(k)} (P^\prime_{\mathrm{main}}) + |R^\ast| \cdot 2k D_{\ell}
&\leq  c^{(k)} (P_{\mathrm{main}}) + \sum_{j=1}^{m} 2^{12} \cdot c^{(k)}(P_j) + \sum_{j=1}^m c^{(k)}(Q_j) + |R^\ast| \cdot 2k D_{\ell} \\
&\leq c^{(k)} (P_{\mathrm{main}}) + \sum_{j=1}^{m} \Big( \left(2^{12} + 1\right) \cdot c^{(k)}(P_j) + |R_j| \cdot 2k D_{\ell} \Big)  \\
&\leq  c^{(k)} (P_{\mathrm{main}}) + \sum_{j=1}^{m} \left(2^{12} + 1 + \scconstant_2\right) \cdot c^{(k)}(P_j) \\
&\leq \left(\scconstant_2 + 2^{13}\right) \cdot c^{(k)} (P_{\mathrm{main}}).
\end{split}
\end{equation*}
\end{proof}

\subsection{Proof of \Cref{lem:main_recursion_covering_weak}}\label{sec:complte_cover_reduction}

Recall that a path-level pair $(P,j)$ covers an element $(v,\ell)$ if
\begin{itemize}
\item $v\in V(P)$ and $j \leq \ell$, or
\item $(P,j)$ takes responsibility for the set $H_{\ell+1}(v)\in \cH_{\ell +1}$ containing $v$, \\
i.e,  we have $|V(P)\cap H_{\ell+1}(v)| \geq k$ and $ j \leq \ell < L$.
\end{itemize}
Further, recall that a \emph{main-tour cover} of a set $U\subseteq V$ is a monotone path $P_{\rm main}$ with $V(P_{\rm main}) \subseteq U$ together with a set $\cP \subseteq \cS$ of path-level pairs satisfying the following two conditions:
\begin{enumerate}
\item \label{item:property_main_path}
$V(P_{\rm main}) = U$ or $|V(P_{\rm main})| \geq k$.
\item\label{item:cover_property}
The path-level pairs $\cP \cup \big\{ \big(P_{\rm main},\level(U)\big)\big\}$ cover the elements $U\times \big\{\level(U),\dots, L\big\} $ of the Set Cover universe $\cU$.
\end{enumerate}
We remark that the monotone path with vertex set $U$ together with $\cP=\emptyset$ is always a main-tour cover of $U$.

Finally, we prove that main result of this section, \Cref{lem:main_recursion_covering_weak}.
We prove the statement by induction and it will be useful to prove a slightly stronger version:

\begin{lemma}\label{lem:main_recursion_covering_strong}
Let $P_{\mathrm{old}}$ be an $\ell$-buffered path with $\ell \coloneqq \level(V(P_{\mathrm{old}}))$.
Then there is a main-tour cover $(P_{\mathrm{new}}, \cP)$ of $V(P_{\mathrm{old}})$ such that
\begin{equation*}
c^{(k)}(P_{\mathrm{new}}) + w(\cP)\ \leq \ 
\begin{cases}
 \big((L - \ell + 1 ) \scconstant^\ast - 3\scconstant_1\big)\cdot c^{(k)}(P_{\mathrm{old}}) & \text{if $P_{\rm old}$ has an $\ell$-splitting-free $\ell$-buffering}\\
  (L - \ell + 1) \scconstant^\ast \cdot c^{(k)}(P_{\mathrm{old}}) & \text{otherwise}.
\end{cases}
\end{equation*}
where $\scconstant^\ast \coloneqq 3\scconstant_1 + \left(\scconstant_2 + 2^{14}\right)$, $\scconstant_1 \coloneqq2^{12}L^2$, and $\scconstant_2 \coloneqq 2^8L $.
\end{lemma}
\begin{proof}
We prove the statement by induction on $L - \ell$.
First, suppose that $\ell = L$.
Then, $P_{\mathrm{old}}$ consists of a single vertex.
In that case, we can choose $P_{\mathrm{new}} \coloneqq P_{\mathrm{old}} $ and $\cP \coloneqq \emptyset$.
This is a trivial main-tour cover of $V(P_{\mathrm{old}})$ and $c^{(k)}(P_{\mathrm{old}}) = 0$.

For the induction step, suppose that $\ell < L$.
We apply induction on the number of backward edges contained in $P_{\mathrm{old}}$.
We distinguish three cases:

\addparskip

\paragraph*{Case 1:}
Assume that $|V(P_\mathrm{old})| \leq k$.\\
Choose $P_{\mathrm{new}}$ to be the monotone path with vertex set $V(P_{\mathrm{old}})$ and $\cP \coloneqq \emptyset$.
This is a trivial main-tour cover of $V(P_{\mathrm{old}})$.
By \Cref{lem:sorted_path} we have $c^{(k)}(P_{\mathrm{new}}) \leq c^{(k)}(P_{\mathrm{old}})$.

\addparskip

\paragraph{Case 2:} Assume that $P_{\rm old}$ has no $\ell$-splitting-free $\ell$-buffering and $|V(P_\mathrm{old})| \geq k$.\\
Let $\cB$ be an $\ell$-buffering of $P_{\mathrm{old}}$.
Then $\cB$ is not $\ell$-splitting-free.
Let $Q_1, \dots, Q_r$ be the paths that result from $P_{\mathrm{old}}$ by removing all backward edges contained in some splitting $\ell$-jump segment from $\cB$.
We have $\dot \bigcup_{j=1}^r V(Q_j) = V(P_{\mathrm{old}})$ and
\begin{equation}
\label{eq:inductive_reordering:Q_j_partition}
\sum_{j=1}^r c^{(k)}(Q_j) \leq c^{(k)}(P_{\mathrm{old}}).
\end{equation}
Furthermore, $r \geq 2$ and therefore, by \Cref{lem:splitting_i_jump},
\begin{equation}
\label{eq:inductive_reordering:fixed_cost_lb}
3\scconstant_1 \cdot c^{(k)}(P_{\mathrm{old}}) \geq 3(r-1) \cdot k^2D_{\ell} \geq (r+1) \cdot  k^2D_{\ell}.
\end{equation}

Let $j \in [r]$.
There exists an $\ell$-splitting-free $\ell$-buffering of $Q_j$ (given by those elements of $\cB$ from which we did not remove the backward edge and that are subpaths of $Q_j$).
Furthermore, $Q_j$ contains strictly less backward edges than $P_{\mathrm{old}}$.
Thus, we can apply the induction hypothesis and obtain a main-tour cover $(\widetilde{Q}_j, {\cP}_j)$ of $V(Q_j)$.

Choose $P_{\mathrm{new}}$ to be the monotone path on $k$ arbitrarily chosen vertices from $V(P_{\mathrm{old}})$.
This is possible, because $|V(P_{\mathrm{old}})| \geq k$.
Note that $c^{(k)}(P_{\mathrm{new}}) \leq k^2D_{\ell}$.
For $\cP$ we choose
\begin{equation*}
\cP \coloneqq \bigcup_{j=1}^r \Big(\cP_j \cup \big\{(\widetilde{Q}_j, \ell )\big\}\Big).
\end{equation*}
\addparskip

Next, we show that $P_{\mathrm{new}}$ together with $\cP$ is indeed a main-tour cover of $U$.
Property, \ref{item:property_main_path} is satisfied because $|V(P_{\mathrm{new}})|\geq k$.
Because $\bigcup_{j=1}^r V(\widetilde{Q}_j) = V(P_{\rm{old}})$, 
in order to prove \ref{item:cover_property} it suffices to show that for each $j\in [r]$, the elements of $V(Q_j) \times \{\ell,\dots, L\}$ are covered.
By the induction hypothesis, $\cP_j \cup \big\{\big(\widetilde{Q}_j, \level(V(Q_j))\big)\big\}$ covers each element from $V(Q_j) \times \{\level(V(Q_j)),\dots, L\}$.
Now consider $(v, \ell^\prime)$ with $v\in V(Q_j)$ and $\ell \leq \ell^\prime < \level(V(Q_j))$.
If $v \in V(\widetilde{Q}_j)$, then $(v, \ell^\prime)$ is covered by the pair $(\widetilde{Q}_j, \ell )$.
So we may assume that this is not the case, which in particular implies $V(Q_j) \neq V(\widetilde{Q}_j)$.
By \ref{item:property_main_path} of the induction hypothesis, we have $|V(\widetilde{Q}_j)|\geq k$.
The hierarchical structure of our instance implies that we have $H_{\ell^\prime + 1}(v) \supseteq V(Q_j)$ (where we used $v\in V(Q_j)$ and $\ell^\prime +1 \leq \level(V(Q_j))$).
Thus, we have 
\[
|H_{\ell^\prime + 1}(v) \cap V(\widetilde{Q}_j)| = |V(\widetilde{Q}_j)| \geq k,
\]
 which implies that the pair $(\widetilde{Q}_j, \ell )$ takes responsibility for the set $H_{\ell^\prime + 1}(v)$ and thus covers $(v,\ell^\prime)$.
We conclude that $(P_{\mathrm{new}},\cP)$ is a main-tour cover of $V(P_{\mathrm{old}})$.
\addparskip

Now we bound $c^{(k)}(P_{\mathrm{new}}) + w(\cP)$.
Because $P_{\mathrm{new}}$ is a path with $k$ vertices contained in a set of level $\ell$, we have  $c^{(k)}(P_{\mathrm{new}})\leq k^2D_{\ell}$.
For each $j\in[r]$ we have
\[
w(\widetilde{Q}_j,\ell) \ =\ c^{(k)}\big(\widetilde{Q}_j\big) + k^2 \cdot D_{\ell}.
\]
If $\level(\widetilde{Q}_j)) > \ell$, the induction hypothesis implies
\[
c^{(k)}\big(\widetilde{Q}_j\big) + w(\cP_j)\ \leq\ (L - \ell) \scconstant^\ast \cdot c^{(k)}\big(Q_j\big)\ \leq\  \big((L - \ell + 1 ) \scconstant^\ast - 3\scconstant_1\big)\cdot  c^{(k)}\big(Q_j\big),
\]
where we used $\scconstant^\ast \geq 3\scconstant_1$.
If $\level(\widetilde{Q}_j) = \ell$, then the induction hypothesis implies
\[
c^{(k)}\big(\widetilde{Q}_j\big) + w(\cP_j)\ \leq\  \big((L - \ell + 1 ) \scconstant^\ast - 3\scconstant_1\big)\cdot c^{(k)}\big(Q_j\big)
\]
because $Q_j$ has an $\ell$-splitting-free $\ell$-buffering.
We conclude that in both cases
\begin{align*}
w\big(\cP_j \cup \big\{(\widetilde{Q}_j, \ell )\big\}\big)\ =&\ c^{(k)}\big(\widetilde{Q}_j\big) + k^2 \cdot D_{\ell} + w(\cP_j)\\
 \leq&\ k^2 \cdot D_{\ell} +   \big((L - \ell + 1 ) \scconstant^\ast - 3\scconstant_1\big) \cdot c^{(k)}\big(Q_j\big)
\end{align*}
Summing over all $j\in [r]$ and using $c^{(k)}(P_{\mathrm{new}})\leq k^2D_{\ell}$, we conclude
\begin{align*}
c^{(k)}(P_{\mathrm{new}}) + w(\cP)\ \leq&\ (r+1)\cdot k^2D_{\ell} +    \big((L - \ell + 1 ) \scconstant^\ast - 3\scconstant_1\big) \cdot \sum_{j=1}^{r}   c^{(k)}\big(Q_j\big) \\
\overset{\eqref{eq:inductive_reordering:Q_j_partition}}{\leq}&\  (r+1)\cdot k^2D_{\ell} +    \big((L - \ell + 1 ) \scconstant^\ast - 3\scconstant_1\big)  \cdot  c^{(k)}(P_{\mathrm{old}})  \\
\overset{\eqref{eq:inductive_reordering:fixed_cost_lb}}{\leq}&\    3\scconstant_1 \cdot c^{(k)}(P_{\mathrm{old}}) +    \big((L - \ell + 1 ) \scconstant^\ast - 3\scconstant_1\big)  \cdot  c^{(k)}(P_{\mathrm{old}}) \\
=&\ (L - \ell + 1 ) \scconstant^\ast \cdot  c^{(k)}(P_{\mathrm{old}}).
\end{align*}

\addparskip

\paragraph{Case 3:} Assume that $P_{\rm old}$ has an $\ell$-splitting-free $\ell$-buffering $\cB$. \\
Let $P^\prime$ and $R \subseteq V(P_{\mathrm{old}})$ be obtained by applying \Cref{lem:many_replacements} to $P_{\mathrm{old}}$.
The path $P^\prime$ is $(\ell+1)$-buffered, hence it contains no backward edges of level $\ell$.
Therefore, $P^\prime$ enters and leaves each $H \in \cH_{\ell+1}$ at most once.
Let $Q^\prime_1, \dots, Q^\prime_r$ be the inclusion-wise maximal $(\ell+1)$-confined subpaths of $P^\prime$, numbered so that $V(Q^\prime_j)$ is left of $V(Q^\prime_{j+1})$ for all $j \in [r-1]$.
In other words, the paths $Q^\prime_j$ are obtained by removing all (forward) edges of $P^\prime$ of level~$\ell$.
See \Cref{fig:induction_hypothesis} for an illustration.
\begin{figure}
\begin{center}
\begin{tikzpicture}[
xscale=0.8,
box/.style={Peach, very thick, fill=Peach, fill opacity=0.05, rounded corners=8pt},
]

\def\slackx{0.2}

\draw[box] (0,0) rectangle (3,2);
\draw[box] ($(3,0) + ( \slackx,0)$) rectangle (7,2);
\draw[box] ($(7,0) + ( \slackx,0)$) rectangle (10,2);
\draw[box] ($(10,0) + ( \slackx,0)$) rectangle (11.5,2);
\draw[box] ($(11.5,0) + ( \slackx,0)$) rectangle (15.5,2);

\begin{scope}
\clip (0,0) rectangle ({3+\slackx},2);

\draw[thick, color=blue] plot[smooth] coordinates {
    (0.5,1) (1.5,1.2) (1.4,0.8) (0.7,0.5) (1.8,0.2) (2,0.9) (2.5,1.2)
    ({3+\slackx},1.3) (4,0.8) (4.5,1) (3.8,1.4) (5,1.7) (6.5,1.45) (4.8,1) (5.2,0.6) (5.6,0.7) (6.0,0.9) (6.5,0.75) 
    (7,0.55)(7.8,0.65)(8.1,1.1)(8.5,1.2)(9.1,0.9)(9.6,1.5)
    (10,1.45)(10.5,1.4)(11.3,1.1)(11.25,0.9)(10.4,0.6)
    (11.5,0.55)(12.2,0.9)(13.5,1.1)(13.5,1.5)(12.9,1.45)(12.2,1.5)(12.25,1.25)(12,1.2)(11.9,1.7)(12.7,1.85)(13.5,1.75)(14.2,1.1)(12.5,0.6)(12.9,0.5)(13.5,0.65)(14,0.4)(15,0.8)
    };
\end{scope}

\begin{scope}
\clip (7,0) rectangle (15,2);
\draw[thick, color=blue] plot[smooth] coordinates {
    (0.5,1) (1.5,1.2) (1.4,0.8) (0.7,0.5) (1.8,0.2) (2,0.9) (2.5,1.2)
    ({3+\slackx},1.3) (4,0.8) (4.5,1) (3.8,1.4) (5,1.7) (6.5,1.45) (4.8,1) (5.2,0.6) (5.6,0.7) (6.0,0.9) (6.5,0.75) 
    (7,0.55)(7.8,0.65)(8.1,1.1)(8.5,1.2)(8.9,0.9)(9.3,0.9)(9.6,1.4)
    (10.2,1.5)(11.3,1.4)(11.3,1.1)(10.5,0.8)(10.5,0.4)
    (11.5,0.55)(12.2,0.9)(13.5,1.1)(13.5,1.5)(12.9,1.5)(12.2,1.45)(11.9,1.7)(12.7,1.85)(13.5,1.75)(14,1.6)(14,1.)(12.5,0.6)(12.9,0.4)(14.3,0.45)(15,0.8)
    };
\end{scope}
    
\begin{scope}
\clip ({3+\slackx},0) rectangle (7,2);
    \draw[thick, color=ForestGreen, densely dashed] plot[smooth] coordinates {
    (0.5,1) (1.5,1.2) (1.4,0.8) (0.7,0.5) (1.8,0.2) (2,0.9) (2.5,1.2)
    ({3+\slackx},1.3) (4,0.8) (4.5,1) (3.8,1.4) (5,1.7) (6.5,1.45) (4.8,1) (5.2,0.6) (5.6,0.7) (6.0,0.9) (6.5,0.75) 
    (7,0.55)(7.8,0.65)(8.1,1.1)(8.5,1.2)(9.1,0.9)(9.6,1.5)
    (10,1.45)(10.5,1.4)(11.3,1.1)(11.25,0.9)(10.4,0.6)
    (11.5,0.55)(12.2,0.9)(13.5,1.1)(13.5,1.5)(12.9,1.45)(12.2,1.5)(12.25,1.25)(12,1.2)(11.9,1.7)(12.7,1.85)(13.5,1.75)(14.2,1.1)(12.5,0.6)(12.9,0.5)(13.5,0.65)(14,0.4)(15,0.8)
    };
\end{scope}

\node at (7.75,2.3) {\textcolor{Peach}{$\cH_{\ell+1}$}};
\node at (9.3,0.5) {\textcolor{blue}{$P^\prime$}};
\node at (4.55,0.45) {\textcolor{ForestGreen}{$Q^\prime_j$}};

\end{tikzpicture}
 \end{center}
\caption{
Illustration of case~3 in the proof of \Cref{lem:main_recursion_covering_strong}.
The path $P^\prime$ visits each $H \in \cH_{\ell+1}$ at most once and has no backward edges of level $\ell$.
Therefore, we can apply the induction hypothesis to each of the $Q^\prime_j$, which are the maximal $(\ell+1)$-confined subpaths of $P^\prime$.
}
\label{fig:induction_hypothesis}
\end{figure}
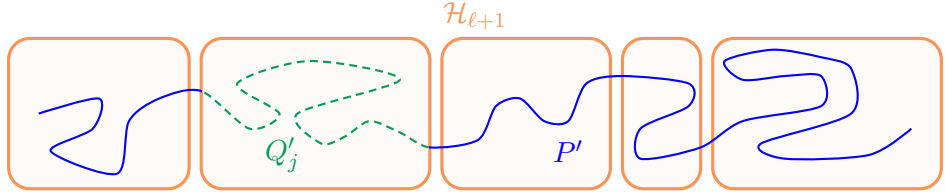
Because $P^\prime$ visits each $H \in \cH_{\ell+1}$ at most once, we have
\begin{equation}
\label{eq:inductive_reordering:greater_i_costs}
\sum_{j=1}^r c^{(k)} (Q^\prime_j)  \ =\ c^{(k)}_{>\ell} (P^\prime) \ \overset{\eqref{eq:many_replacements1}}{\leq}\ c(P_{\rm old}),
\end{equation}
where we used \eqref{eq:many_replacements1} from \Cref{lem:many_replacements}.
Moreover, $\level(Q^\prime_j) \geq \ell + 1$ and $Q^\prime_j$ is $\level(Q^\prime_j)$-buffered for all $j \in [r]$ (by \Cref{lem:shrink_buffering}).
Thus, for each $j \in [r]$ we can apply the induction hypothesis and obtain a main-tour cover $(\widetilde{Q}_j, \widetilde{\cP}_j)$ of $V(Q^\prime_j)$.
\addparskip

Each path $\widetilde{Q}_j$ is monotone and we have we have $V(\widetilde{Q}_j) \subseteq V(Q^\prime_j)$, where  $V(Q^\prime_j)$ is left of $V(Q^\prime_{j+1})$.
Thus, concatenating the paths $\widetilde{Q}_j$ into a single path in an order of increasing $j$, yields a monotone path, which we denote by $\widetilde{P}$.
In other words, $\widetilde{P}$ results from $P^\prime$ by replacing each $Q^\prime_j$ with $\widetilde{Q}_j$.
Thus, by \Cref{lem:path_replacement}, we obtain
\begin{equation}
\label{eq:inductive_reordering:P_2_ub}
c^{(k)}(\widetilde{P}) \ \leq\  c^{(k)}(P^\prime) +  \sum_{j=1}^r 2^{12} \cdot c^{(k)}(Q^\prime_j) + \sum_{j=1}^r c^{(k)}(\widetilde{Q}_j)
\end{equation}

Let $P_{\mathrm{new}}$ be the path that results from $\widetilde{P}$ by inserting each vertex $v \in R$ such that the resulting path remains monotone.
When inserting a vertex into a path, there are at most $2k$ new edges that contribute to the $k$-hop cost of the path.
Because all vertices of $P_{\mathrm{new}}$ are contained in the set $V(P_{\rm old})$, which has level $\ell$, each such new edge has cost at most $D_{\ell}$.
Thus, we obtain
\begin{align*}
c^{(k)}(P_{\mathrm{new}}) \ &\leq\ c^{(k)}(\widetilde{P}) + |R| \cdot 2kD_{\ell}\\
\overset{\eqref{eq:inductive_reordering:P_2_ub}}&{\leq}\  c^{(k)}(P^\prime)  + |R| \cdot 2kD_{\ell}  +  \sum_{j=1}^r 2^{12} \cdot c^{(k)}(Q^\prime_j) + \sum_{j=1}^r c^{(k)}(\widetilde{Q}_j)\\
\overset{\eqref{eq:many_replacements2}}&{\leq} 
\left(\scconstant_2 + 2^{13} \right) \cdot c^{(k)}(P_{\mathrm{old}})  +  \sum_{j=1}^r 2^{12} \cdot c^{(k)}(Q^\prime_j) +  \sum_{j=1}^r c^{(k)}(\widetilde{Q}_j),
\end{align*}
where we used \eqref{eq:many_replacements2} from \Cref{lem:many_replacements}.
\addparskip

We define $\cP \coloneqq \bigcup_{j=1}^r \widetilde{\cP}_j$.
By the induction hypothesis, we have
\begin{equation}\label{eq:induction_hyp_case_3}
c^{(k)}(\widetilde{Q}_j) + w(\widetilde{\cP}_j) \ \leq\  (L - (\ell+1) + 1) \scconstant^\ast \cdot c^{(k)}(Q^\prime_j),
\end{equation}
implying
\begin{equation*}
\begin{split}
c^{(k)}(P_{\mathrm{new}}) + w(\cP)
    \overset{\eqref{eq:inductive_reordering:P_2_ub}}&{\leq}\  
 \left(\scconstant_2 + 2^{13} \right) \cdot c^{(k)}(P_{\mathrm{old}})  +  \sum_{j=1}^r 2^{12} \cdot c^{(k)}(Q^\prime_j) + \sum_{j=1}^r c^{(k)}(\widetilde{Q}_j)  + w(\cP) \\
&=\   
\left(\scconstant_2 + 2^{13} \right) \cdot c^{(k)}(P_{\mathrm{old}}) + \sum_{j=1}^r 2^{12} \cdot c^{(k)}(Q^\prime_j) + \sum_{j=1}^r \left( c^{(k)}(\widetilde{Q}_j) + w(\widetilde{\cP}_j) \right) \\
\overset{\eqref{eq:induction_hyp_case_3}}&{\leq}\ 
\left(\scconstant_2 + 2^{13} \right) \cdot c^{(k)}(P_{\mathrm{old}}) + \left((L - (\ell+1) + 1) \scconstant^\ast + 2^{12} \right) \cdot \sum_{j=1}^r c^{(k)}(Q^\prime_j)\\
\overset{\eqref{eq:inductive_reordering:greater_i_costs}}&{\leq}\ \left(\scconstant_2 + 2^{13} \right) \cdot c^{(k)}(P_{\mathrm{old}}) + \left((L - \ell ) \scconstant^\ast + 2^{12}\right) \cdot c^{(k)} (P_{\mathrm{old}})  \\
&\leq 
\left((L - \ell ) \scconstant^\ast + \scconstant_2 + 2^{14} \right) \cdot c^{(k)} (P_{\mathrm{old}}) \\
&= 
(L - \ell + 1)\scconstant^\ast \cdot c^{(k)}(P_{\mathrm{old}}) - 3\scconstant_1 \cdot c^{(k)}(P_{\mathrm{old}}).
\end{split}
\end{equation*}
\addparskip

It remains to argue that $(P_{\mathrm{new}}, \cP)$ is a main-tour cover of $V(P_{\mathrm{old}})$.
If $V(Q^\prime_j) = V( \widetilde{Q}_j)$ for all $j \in [r]$, then $V(P_{\mathrm{old}}) = V(P_{\mathrm{new}})$.
Otherwise, we have $|V(P_{\mathrm{new}})| \geq |\bigcup_{j=1}^r V( \widetilde{Q}_j)| \geq k$, since each $(\widetilde{Q}_j, \widetilde{\cP}_j)$ is a main-tour cover.
This shows property~\ref{item:property_main_path} of a main-tour cover.

We now prove that $(P_{\mathrm{new}}, \cP)$ also satisfies~\ref{item:cover_property}.
Let $(v,\ell^\prime) \in V(P_{\rm old}) \times \{\ell,\dots, L\}$.
If $v \in V(P_{\rm new})$, then $(v,\ell^\prime)$ is covered by $(P_{\rm new}, \ell)$.
Hence, we may assume that this is not the case.
Because $R\subseteq V(P_{\rm new})$, this implies $v\in V(P^\prime)$ and thus $v\in V(Q_j^\prime)$ for some $j\in [r]$. 

If $\ell^\prime \geq \level(Q_j^\prime)$, then by the induction hypothesis, either a path-level pair from $\cP_j \subseteq \cP$ covers $(v,\ell^\prime)$, or the path-level pair $\big(\widetilde{Q}_j, \level(V(\widetilde{Q}_j)) \big)$ covers $(v,\ell^\prime)$.
In the latter case,  the path-level pair $(P_{\rm new}, \ell)$ covers  $(v,\ell^\prime)$ because $\ell \leq \ell^\prime$ and  $V(\widetilde{Q}_j) \subseteq V(P_{\rm new})$.

It remains to consider the case $\ell^\prime < \level(Q_j^\prime)$, i.e., $\ell^\prime + 1 \leq \level(Q_j^\prime)$.
In particular $V(Q^\prime_j) \subseteq H_{\ell^\prime + 1}(v)$ (by the hierarchical structure of our instance and because $v\in V(Q^\prime_j) $ and $v \in  H_{\ell^\prime + 1}(v)$).
Because $v\notin V(P_{\rm new}) \supseteq V(\widetilde{Q}_j)$, we have  $V(\widetilde{Q}_j) \neq V(Q^\prime_j) $ and thus by property~\ref{item:property_main_path} of the main-tour cover $(\widetilde{Q}_j,\cP_j)$, we must have $|V(\widetilde{Q}_j)| \geq k$.
Therefore,
\begin{equation*}
|V(P_{\mathrm{new}}) \cap H_{\ell^\prime + 1}(v)| 
\geq |V(\widetilde{Q}_j) \cap H_{\ell^\prime + 1}(v)| 
\geq |V(\widetilde{Q}_j)| 
\geq k,
\end{equation*}
so $(P_{\rm new}, \ell)$ covers $(v,\ell^\prime)$ because $\ell \leq \ell^\prime$.
We conclude that $(P_{\mathrm{new}}, \cP)$ is a main-tour cover of $V(P_{\mathrm{old}})$.
\end{proof}

\section{Solving the Covering Problem via an Algorithm for Monotone Paths}\label{sec:monotone_paths}

In this section we provide a quasi-polynomial-time $O(\log^2 n)$-approximation for our covering problem.
More precisely, we prove the following theorem:
\begin{theorem}\label{thm:path_covering_algorithm}
For every instance $\cI =(V, L, \cH, D, \prec, c, k)$ of Path Covering with low depth, we can compute in time $n^{O(\log n)}$ a solution $\cP$ with $w(\cP) \leq O(\log ^2 n) \cdot \OPT(\cI)$, where $n$ denotes the number of vertices in $\cI$.
\end{theorem}
Combining \Cref{thm:reducing_hop_atsp,thm:covering_reduction,thm:path_covering_algorithm} yields our main result, \Cref{thm:main}.
To prove \Cref{thm:path_covering_algorithm}, we first discuss what kind of oracle we need in order to apply the greedy set cover algorithm (\Cref{sec:set_cover}).
In \Cref{sec:weight_bound,sec:block_decomp,sec:block_dp,sec:remaining_proofs_dp} we then develop such an oracle, using an approach based on dynamic programming.

\subsection{Applying the Greedy Set Cover Algorithm}\label{sec:set_cover}

To prove \Cref{thm:path_covering_algorithm}, we view Path Covering as a Set Cover problem as in \Cref{eq:equivalence_to_set_cover}, where task is to cover all element of the universe $\cU = V\times [L]$ by path-level pairs of minimum total weight.
Starting with $U= \cU$ and $\cP=\emptyset$, the greedy algorithm for Set Cover iteratively chooses a path-level pair $(P, \ell)$ (approximately) minimizing
\begin{equation}\label{eq:set_cover_oracle_objective}
\frac{w(P,\ell)}{|\{ (v,j) \in U : (v,j)\text{ covered by }(P,\ell)\}|},
\end{equation}
adds $(P,\ell)$ to $\cP$, and removes the elements covered by $(P,\ell)$ from $U$.
Once the set $U$ of uncovered elements becomes empty, the algorithm returns the Path Covering solution $\cP$.
The well-known analysis of this algorithm (see e.g.~\cite{williamson2011design}) shows that if we choose in every iteration of the algorithm a path-level pair $(P,\ell)$ that minimizes \eqref{eq:set_cover_oracle_objective} up to a factor $\alpha$, then we obtain an
approximation factor of
\[
\alpha \cdot O(\log |\cU|) = \alpha \cdot O(\log n).
\]
for our Path Covering problem.
Thus, in order to prove \Cref{thm:path_covering_algorithm}, it suffices to provide an $O(\log n)$-approximation algorithm with running time $n^{O(\log n)}$ for the following problem:
Given an instance of Path Covering and a nonempty set $U\subseteq \cU$, compute a path-level pair $(P,\ell)$ minimizing \eqref{eq:set_cover_oracle_objective}.
\addparskip

In the remainder of this section we develop such an algorithm.
Given a Path Covering instance $\cI =(V, L, \cH, D, \prec, c, k)$ and a nonempty set $U\subseteq \cU$ of uncovered elements, we write
\[
S_{\ell}\ \coloneqq\ \Big\{ v\in V : (v,\ell) \notin U \Big\}.
\]
to denote the set of vertices, for which $(v,\ell)$ is already covered.

Let $(P^\ast, \ell^\ast)$ denote an optimum solution to our problem, i.e, a path-level pair that minimizes \eqref{eq:set_cover_oracle_objective}.
We define 
\begin{itemize}
\item
$\gamma^\ast \coloneqq |\{ (v,\ell) \in U : (v,\ell)\text{ covered by }(P^\ast,\ell^\ast)\}| $, and
\item
$H^\ast \in \cH_{\ell^\ast}$ to be the unique set in the partition $\cH_{\ell^*}$ that satisfies $V(P^\ast) \subseteq H^\ast$.
\end{itemize}
Note that such a set $H^\ast$ exists by the definition of a path-level pair.
We also observe there are only polynomially many choices for $\ell^\ast, H^\ast$ and $\gamma^\ast$.
Thus, by enumerating all possibilities, we can assume in the following that we guessed these values correctly.

In the following we abbreviate
\begin{equation*}
\begin{split}
w(P) &\coloneqq w(P, \ell^\ast), \quad\text{and}\\
\numnewcovered(P) &\coloneqq  |\{(v,\ell) \in U :   (v,\ell)\text{ covered by }(P, \ell^\ast)\}|  \\
&= \sum_{\ell = \ell^\ast}^{L}|\parset{v \in V : (v,\ell)\text{ covered by }(P, \ell^\ast)\text{ and } v \notin S_{\ell}}|,
\end{split}
\end{equation*}
and we will restrict ourselves to monotone paths $P$ with $V(P) \subseteq H^\ast$.
Thus, our goal is now to find a quasi-polynomial-time $O(\log n)$-approximation for the problem of finding a monotone path $P$ with $V(P) \subseteq H^\ast$ that minimizes $\frac{w(P)}{\numnewcovered(P)}$.

\subsection{Approximating the Path Weights}\label{sec:weight_bound}

As mentioned in \Cref{sec:overview_dp}, we will approximate the weight $w(P)$ of a path by some weight bound $w^q(P)$.
In this section we define these weight bounds and prove that they indeed provide a good approximation.
We fix an integer $k^\prime \in \{k, \dots, 2k\}$ such that $k^\prime = 2^{\Gamma} - 1$ for some integer $\Gamma \in \bZ$.
For each offset $q\in \{0, \dots, k^\prime\}$ and every monotone path $P$, we will define a weight bound $w^q(P)$.
These weight bounds will have the following key properties:
\begin{enumerate}
\item\label{item:weight_bound_upper_bound} for every monotone path $P$ and every offset $q\in \{0, \dots, k^\prime\}$, we have $w^q(P) \geq w(P)$,
\item\label{item:property_weight_bound_good_for_some_offset}  for every monotone path $P$, there exists an offset $q\in\{0, \dots, k^\prime\}$ such that $w^q(P) = O(\log k) \cdot w(P)$, and
\item\label{item:property_weight_bound_efficient}  we can find a monotone path $P$ with $V(P) \subseteq H^\ast$ and an offset $q\in \{0, \dots, k^\prime\}$ minimizing $\frac{w^q(P)}{\numnewcovered(P)}$ in quasi-polynomial time.
\end{enumerate}

To define our weight approximations $w^{q}(P)$ we assign the vertices of $P$ to different heights, as illustrated in \Cref{fig:heights_and_widths}.
More precisely, we define the two functions $\height: \bZ \to \{0,\dots, \Gamma\}$ and $\width:\{0, \dots, \Gamma\} \to \{0, \dots, k^\prime\}$ by
\begin{equation*}
\height(j) \coloneqq \max\left\{ h \in \{0, \dots, \Gamma\} : j \equiv 0 \pmod{2^h} \right\},
\end{equation*}
and
\begin{equation*}
\width(j) \coloneqq 2^{\height(j)} - 1.
\end{equation*}
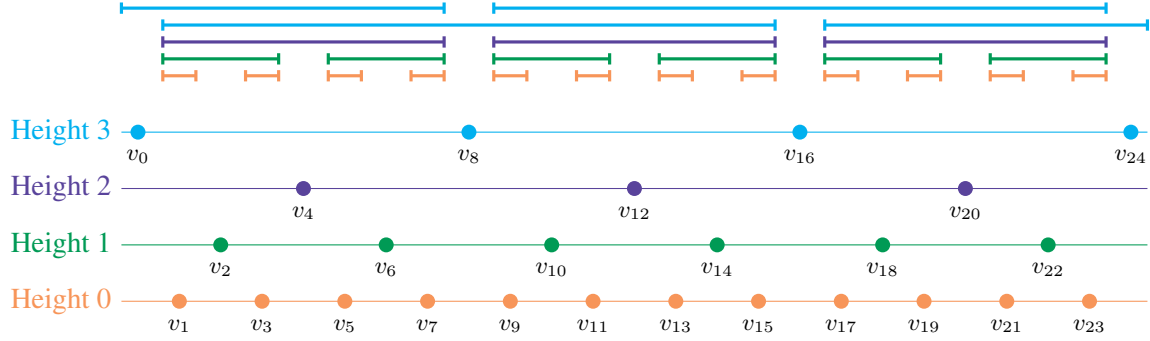
\begin{figure}
\begin{center}

\begin{tikzpicture}[yscale=0.75, xscale=0.55]
\begin{scope}
\definecolor{color1}{named}{Peach}  
\definecolor{color2}{named}{ForestGreen}  
\definecolor{color3}{named}{Violet}  
\definecolor{color4}{named}{ProcessBlue} 

\newcommand{\highlight}[4]{
		\def\len{0.1}
        \draw[#4, very thick] (#1, #3) -- (#2, #3);
        \draw[#4, very thick] ($(#1, #3) + (0, \len)$) -- ($(#1, #3) - (0, \len)$);
        \draw[#4, very thick] ($(#2, #3) + (0, \len)$) -- ($(#2, #3) - (0, \len)$); 
}

\def\slack{0.4}

\foreach \h in {0,1,2,3} {
  \pgfmathtruncatemacro{\hc}{\h}
  \ifcase\hc
    \def\hcolor{color1}
    \def\hlabel{Height 0}
  \or
    \def\hcolor{color2}
    \def\hlabel{Height 1}
  \or
    \def\hcolor{color3}
    \def\hlabel{Height 2}
  \or
    \def\hcolor{color4}
    \def\hlabel{Height 3}
  \fi
  \draw[thin, \hcolor] (0-\slack, \h) -- (24+\slack, \h);
  \node[left, text=\hcolor] at (0-\slack, \h) {{\hlabel}};
}

\foreach \i in {0,...,24} {
\pgfmathsetmacro{\height}{
  ifthenelse(mod(\i,8)==0,3,
    ifthenelse(mod(\i,4)==0,2,
      ifthenelse(mod(\i,2)==0,1,0)
    )
  )
}
\pgfmathtruncatemacro{\h}{\height}
  \ifcase\h
    \def\heightcolor{color1}
  \or
    \def\heightcolor{color2}
  \or
    \def\heightcolor{color3}
  \or
    \def\heightcolor{color4}
  \fi
\pgfmathsetmacro{\width}{
  pow(2, \height) - 1
}

\node[circle,fill=\heightcolor,inner sep=2pt] (v\i) at ({(\i)*1},\height) {};
\node[below=3pt, font=\footnotesize] at (v\i) {$v_{\i}$};

\pgfmathsetmacro{\start}{max(0, \i - \width) - \slack}
\pgfmathsetmacro{\end}{min(24, \i + \width) + \slack}

\pgfmathsetmacro{\yinterval}{
ifthenelse(mod(\i,16)==0,4,
\height
)
}

\highlight{\start}{\end}{4+0.3*\yinterval}{\heightcolor};
}
\end{scope}
\end{tikzpicture}
 \end{center}
\caption{
An illustration of the functions $\height$ and $\width$ for $\Gamma = 3$ and $k^\prime =7$.
Each vertex $v_i$ is vertically positioned according to $\height(i)$.
The intervals shown above indicate for each vertex $v_i$ the set $\{v_j \in V(P) : j = i - \width(i), \dots, i + \width(i)\}$.
}
\label{fig:heights_and_widths}
\end{figure}
Our weight approximations are then defined as follows:
\begin{definition}[Weight-bound]
Let $P$ be a path of length $r$ visiting the vertices $v_0, \dots, v_{r}$ in this order. 
For a given $q \in \{0, \dots, k^\prime\}$, we define the \emph{weight-bound} of $P$ as
\begin{equation*}
w^{q}(P) \coloneqq k^2 \cdot D_{\ell^\ast} + \sum_{i=0}^{r} 2^{\height(i+q)} \cdot \sum_{\Delta = 1}^{\width(i + q)} \big( c(v_{i-\Delta} , v_{i}) + c(v_{i} , v_{i+\Delta}) \big),
\end{equation*}
where any term involving a vertex $v_z$ with $z < 0$ or $z > r$ is interpreted as zero.
\end{definition}

Now property~\ref{item:weight_bound_upper_bound} follows directly from the triangle inequality satisfied by the cost function $c$:
\begin{lemma}
\label{lem:weight_bound_ub}
Let $P$ be a path of length $r$ visiting the vertices $v_0, \dots, v_{r}$ in this order.
Then, for all $q \in \{0, \dots, k^\prime\}$, we have
\begin{equation}
\label{eq:weight_bound_ub}
w(P) \leq w^{q}(P).
\end{equation}
\end{lemma}
\begin{proof}
Fix an arbitrary $q \in \{0, \dots, k^\prime\}$.
First, observe that $k \leq k'$ implies $c^{(k)}(P) \leq c^{(k^{\prime})}(P)$.  
Hence, it suffices to show $w^{q}(P) \geq k^2 \cdot D_{\ell^\ast} + c^{(k')}(P)$.

Let $i < j$ such that the edge $(v_{i}, v_{j})$ contributes to $c^{(k^\prime)}(P)$, that is, $j - i \leq k^\prime$.
We choose an index $i^\ast  \in \{i, \dots, j\}$ maximizing $\height(i^\ast + q)$.
By the triangle inequality, $c(v_{i}, v_{j}) \leq c(v_{i}, v_{i^\ast}) + c(v_{i^\ast}, v_{j}) $.
We will show that summing this upper bound over all edges $(v_{i}, v_{j})$ contributing to $c^{(k^\prime)}(P)$ is at most $w^{q}(P) - k^2 \cdot D_{\ell^\ast}$.
\addparskip

First, we show $ i^\ast  - i \leq \width(i^\ast+ q) $ and $j - i^\ast \leq \width(i^\ast+ q)$, which implies that the costs $c(v_{i}, v_{i^\ast})$ and $c(v_{i^\ast}, v_{j})$ contribute to the weight bound $w^q(P)$.
If $\height(i^\ast + q) = \Gamma$, then we have 
\begin{equation*}
i^\ast - \width(i^\ast+ q) \leq j - k^\prime \leq i.
\end{equation*}
Otherwise, we have $\height(i^\ast + q)  < \Gamma$.
Let $h \coloneqq \height(i^\ast + q)$.
Then by the definition of $\height$, we have $(i^\ast + q)  \equiv 0 \pmod{2^h}$ and we have  $(i^\ast + q)  \not\equiv 0 \pmod{2^{h+1}}$.
This implies
\begin{equation*}
(i^\ast + q) - 2^h \equiv 0 \pmod{2^{h + 1}}.
\end{equation*}
In particular, this implies $\height(i^\ast + q - 2^h) > h = \height(i^\ast + q)$, where we used $h < \Gamma$.
Therefore, by our choice of the index $i^\ast$, we must have $i^\ast - 2^h < i$.
We conclude $ i^\ast - i < 2^h$ and thus $ i^\ast - i \leq 2^h - 1 = \width(i^\ast + q)$.

Hence, in both cases, we have  $ i^\ast  - i \leq \width(i^\ast + q)$.
The proof of $j - i^\ast \leq \width(i^\ast + q)$ is analogous.
This shows that the costs $c(v_{i}, v_{i^\ast})$ and $c(v_{i^\ast}, v_{j})$ contribute to the weight bound $w^q(P)$.
It remains to show that the edges contributing to $w^q(P)$ also contribute with a sufficient multiplicity to account for all edges $(v_i,v_j)$ contributing to $c^{(k'})(P)$.
\addparskip

Finally, for each vertex $v_{i^\ast}$ and each vertex $v_i$ with $i^\ast - i \leq \width(i^\ast + q)$, we have at most $\width(i^\ast+ q) + 1$ indices $j \geq i^\ast$ such that $i^\ast$ satisfies $i^\ast  \in \{i, \dots, j\}$ and maximizes $\height(i^\ast + q)$.
(This is the case because we have shown that any such vertex $j$ satisfies $j - i^\ast \leq \width(i^\ast + q)$.)
Note that the cost $c(v_i, v_{i^\ast})$ indeed contributes with a factor of $2^{\height(i^{\ast}+q)} = \width(i^\ast + q) + 1$.
An analogous argument applies to the cost of the edge $(v_{i^\ast}, v_j)$.
\end{proof}

To prove property~\ref{item:property_weight_bound_good_for_some_offset}  of our weight bound, we use the following observation:

\begin{lemma}\label{lem:cost_k_prime}
For every path $P$ we have $c^{(k')}(P) \leq 4 \cdot c^{(k)}(P)$.
\end{lemma}
\begin{proof}
Let $v_0,\dots, v_r$ be the vertices visited by $P$ in this order.
We upper bound the length of every edge $(v_i,v_j)$ contributing to $c^{(k')}(P)$ by
$ c(v_i,v_j) \leq c(v_i, v_m) + c(v_m, v_j)$ where $m=\lfloor \frac{i+j}{2} \rfloor$.
Observe that $m-i \leq \lfloor\frac{k'}{2}\rfloor \leq k$  and $j-m \leq\lceil\frac{k'}{2}\rceil \leq k$  because $j-i \leq k'$ and  $k' \leq 2k$.

For all $i,m \in\{0,\dots, r\}$ there are at most two numbers $j\in \{0,1,\dots, r\}$ with $m= \lfloor \frac{i+j}{2} \rfloor$.
Similarly, for all $m,j \in\{0,\dots, r\}$ there are at most two numbers $i\in \{0,1,\dots, r\}$ with $m= \lfloor \frac{i+j}{2} \rfloor$.
\end{proof}

Now property~\ref{item:property_weight_bound_good_for_some_offset} follows from a simple averaging argument:

\begin{lemma}
\label{lem:weight_bound_lb}
Let $P$ be a monotone path of length $r$ visiting the vertices $v_0, \dots, v_{r}$ in this order.
Then, there exists an offset $q \in \{0, \dots, k^\prime\}$ such that
\begin{equation*}
\label{eq:weight_bound_lb}
4(\Gamma + 2) \cdot w(P) \geq w^{q}(P).
\end{equation*}
\end{lemma}
\begin{proof}
We choose $q \in \{0, \dots, k^\prime\}$ uniformly at random and prove $\bE\big[ w^q(P) \big] \leq 4(\Gamma + 2) \cdot w(P)$.
Fix an arbitrary index $j \in \{0, \dots, r\}$.
Consider the contribution of $v_j$ to $w^{q}(P)$, which is 
\begin{equation}\label{eq:contribution_to_weight_bound}
2^{\height(j + q)}  \sum_{\Delta = 1}^{\width(j + q)} \Big( c(v_{j - \Delta} , v_{j}) + c(v_{j} , v_{j + \Delta}) \Big),
\end{equation}
 where we again follow the convention that edge costs involving undefined vertices (i.e.\ vertices $v_z$ with $z<0$ or $z > r$) are zero.
Recall that $k^\prime = 2^\Gamma - 1$, and therefore the definition of the function $\height$ implies
\begin{equation*}
\bP\big[\height(j+q) = z\big]
= \begin{cases}
\frac{1}{2^{z+1}} \quad&\text{if $z \in \{0, \dots, \Gamma - 1\}$},\\
\frac{1}{2^\Gamma} \quad&\text{if $z = \Gamma$}.
\end{cases}
\end{equation*}
In particular, we obtain
\begin{equation*}
\begin{split}
\bE\Big[ 2^{\height(j + q)} \Big]
\ &\leq\  \frac{2^\Gamma}{2^\Gamma} + \sum_{z=0}^{\Gamma - 1} \frac{2^z}{2^{z+1}} 
\ =\  \frac{\Gamma}{2} + 1.
\end{split}
\end{equation*}
Using $\width(j + q) \leq k'$, we can bound the expected value of the contribution~\eqref{eq:contribution_to_weight_bound} of $v_j$ to $w^{q}(P)$ by
\begin{equation*}
\begin{split}
 \bE\Big[ 2^{\height(j+q)} \cdot \sum_{\Delta = 1}^{k^\prime} c(v_{j - \Delta} , v_{j}) + c(v_{j} , v_{j + \Delta})\Big] 
\ &\leq\ (\tfrac{\Gamma}{2} + 1) \cdot \sum_{\ell = 1}^{k^\prime} c(v_{j - \ell} , v_{j}) + c(v_{j} , v_{j + \ell}).
\end{split}
\end{equation*}
Because $2c^{(k^\prime)}(P) = \sum_{j=0}^{r} \sum_{\Delta= 1}^{k^\prime} \big( c(v_{j-\Delta} , v_{j}) + c(v_{j} , v_{j + \Delta}) \big)$, we conclude that
\begin{align*}
\bE\big[w^{q}(P) \big] \ \leq&\ k^2D_{\ell^*} + (\Gamma + 2) \cdot c^{(k^\prime)}(P) \\
 \leq&\ k^2D_{\ell^*} + 4(\Gamma + 2) \cdot c^{(k)}(P)\\
  \leq&\  4(\Gamma + 2) \cdot w(P),
\end{align*}
where we used \Cref{lem:cost_k_prime} for the second inequality.
\end{proof}

To prove  that we can find a monotone path $P$ with $V(P) \subseteq H^\ast$ and an offset $q\in \{0, \dots, k^\prime\}$ minimizing $\frac{w^q(P)}{\numnewcovered(P)}$ in quasi-polynomial time, we will use a dynamic programming approach and we will no longer assume that the cost function $c$ satisfies the triangle inequality.
This allows us to argue that we may assume that the optimal path $P^\ast$ and an optimal offset $q^\ast$ satisfy $q^\ast = 0$ and 
\begin{equation}
\label{eq:assumption_cardinality}
|V(P^\ast)| \equiv 1 \pmod{k^\prime + 1}.
\end{equation}
These assumptions will be helpful to simplify notation and avoid special cases to consider in our algorithm.

To see that we can indeed assume \eqref{eq:assumption_cardinality} and $q^\ast =0$, observe that in case these assumptions are not satisfied, we can add dummy vertices to our instance as follows.
We add $2 k^\prime$ dummy vertices to $H^\ast$, out of which $k^\prime $ are located left of all other vertices in $H^\ast$ and the other $k^\prime $ vertices are located right of all other vertices of $H^\ast$.
We extend the hierarchical partition $\cH$ to include these dummy vertices.
(This can be done in any arbitrary way as long as the dummy vertices belong to $H^\ast$ and satisfy the above-mentioned left/right-relation to the existing vertices in $H^\ast$.)
For each dummy vertex, the cost of every incident forward edge is zero.
The cost of incident backward edges does not matter, we could set it e.g.\ to $D_{\ell}$ for  backward edges on level $\ell$.
Moreover, for each dummy vertex $v$, all elements $(v,\ell)$ of the set cover universe are already covered, i.e., none of them is contained in $U$.

Then including dummy vertices in a monotone path $P$ (while maintaining monotonicity) does not change the number $\numnewcovered(P)$ of newly covered elements of the set cover universe.
Moreover, including a dummy vertex at the end of the path $P$ does not change the weight bound $w^q(P)$.
If a path $P^\prime$ arises from inserting a dummy vertex at the beginning of a path $P$ (with $V(P)\subseteq H^\ast$), then we have $w^{q}(P^\prime) = w^{q+1}(P)$.
Thus, on the instance with the additional dummy vertices, there exist an optimal path $P^\ast$ and an optimal offset $q^\ast$ satisfying $q^\ast = 0$ and \eqref{eq:assumption_cardinality}, and from an optimal path $P^\ast$ for the instance with dummy vertices, we can obtain an optimal path $P$ and an optimal offset $q$ by omitting all dummy vertices from $P^\ast$ and setting $q$ to the number of dummy vertices contained in $P^\ast$ that are left of the original vertices of $H^\ast$.

Thus, to prove  that we can find a monotone path $P$ with $V(P) \subseteq H^\ast$ and an offset $q\in \{0, \dots, k^\prime\}$ minimizing $\frac{w^q(P)}{\numnewcovered(P)}$ in quasi-polynomial time, it suffices to prove the following lemma:

\begin{lemma}\label{lem:main_lemma_dp}
In time $n^{O(\log n)}$ we can compute a monotone path $P$  that minimizes $w^0(P)$ among all monotone paths satisfying the following constraints:
\begin{enumerate}
\item\label{item:subset_constraint}
$V(P) \subseteq H^\ast$,

\item\label{item:cardinality_constraint}
$|V(P)| \equiv 1 \pmod{k^\prime + 1}$, and

\item\label{item:covering_constraint}
$\numnewcovered(P) = \gamma^\ast$.
\end{enumerate}
This applies even if the cost function $c$ does not satisfy the triangle inequality.

\end{lemma}

From \Cref{lem:main_lemma_dp} we can conclude the main result of this section, \Cref{thm:path_covering_algorithm}.
As argued in \Cref{sec:set_cover}, it suffices to obtain an $O(\log n)$-approximation for the problem of finding a monotone path $P$ with $V(P) \subseteq H^\ast$ that minimizes $\frac{w(P)}{\numnewcovered(P)}$ in time $n^{O(\log n)}$.
Let $P^\ast$ denote an optimal solution for this problem.
By \Cref{lem:weight_bound_lb}, there exists an offset $q^\ast$ such that with $\frac{w^{q^\ast}(P^\ast)}{\numnewcovered(P^\ast)} \leq 4(\Gamma + 2) \cdot \frac{w(P^\ast)}{\numnewcovered(P^\ast)}$.
As discussed above, \Cref{lem:main_lemma_dp} (applied to the instance with dummy vertices) implies that we can find a monotone path $P$ with $V(P)\subseteq H^\ast$  and an offset $q$ minimizing $\frac{w^q(P)}{\numnewcovered(P)}$.
In particular, we have 
\[
\frac{w^q(P)}{\numnewcovered(P)} \ \leq\ \frac{w^{q^\ast}(P^\ast)}{\numnewcovered(P^\ast)} \ \leq\ 4(\Gamma + 2) \cdot \frac{w(P^\ast)}{\numnewcovered(P^\ast)},
\]
which by \Cref{lem:weight_bound_ub} implies $\frac{w(P)}{\numnewcovered(P)} \leq  4(\Gamma + 2) \cdot \frac{w(P^\ast)}{\numnewcovered(P^\ast)} = O(\log n) \cdot \frac{w(P^\ast)}{\numnewcovered(P^\ast)}$, as desired.

We conclude that in order to complete the proof of \Cref{thm:path_covering_algorithm}, it remains to prove \Cref{lem:main_lemma_dp}.

\subsection{Decomposition into Blocks}\label{sec:block_decomp}

To prove \Cref{lem:main_lemma_dp}, we will first describe how we can decompose every feasible path $P$ into several independent \emph{blocks}.
Our algorithm will then precompute possible blocks of an optimal solution, and we will show how to construct an optimal path by combining precomputed blocks in an optimal way.

First, we explain more precisely what we mean by a \emph{block}.
Given a path monotone path $P$ that satisfies constraint~\ref{item:cardinality_constraint}, we can decompose its vertex set $V(P) = \parset{x_1, \dots, x_r} \cup \bigcup_{i=1}^{r-1} A_i$ with
\begin{equation}
\label{eq:path_decomposition}
x_1 \prec A_1 \prec x_2 \prec A_2 \prec \dots \prec x_{r-1} \prec A_{r-1} \prec x_r,
\end{equation}
and $|A_i| = k^\prime = 2^{\Gamma} - 1$ for all $i \in [r-1]$.
We refer to the sets $A_i$ as \emph{blocks}.
Observe that the vertices $x_1, \dots, x_r$ are exactly those vertices assigned to height $\Gamma$ in the definition of the weight-bound $w^0$.

The following lemma shows that if we aim at minimizing $w^0(P)$ and have already fixed the vertices $x_1, \dots, x_r$, the blocks $A_i$ can be chosen independently :
\begin{lemma}\label{lem:decomposition_in_blocks}
Let $P$ be a monotone path visiting the vertices $v_0, \dots, v_{r\cdot 2^{\Gamma}}$ in this order, for some $r \geq 1$.
For each $i = 0, \dots, r-1$, consider the subpath $Q_i$ of $P$ visiting $v_{i \cdot 2^{\Gamma}}, \dots, v_{(i+1)\cdot 2^{\Gamma}}$ in this order.
Then we have
\begin{equation}
\label{eq:decomposition_in_blocks}
w^0(P) - k^2 \cdot D_{\ell^\ast} = \sum_{i=0}^{r-1} \Big(w^0(Q_i) - k^2 \cdot D_{\ell^\ast} \Big).
\end{equation}
\end{lemma}
\begin{proof}
Recall that the contribution of $v_j$ to $w^0(P)$ is
\begin{equation*}
2^{\height(j)} \sum_{\Delta = 1}^{\width(j)}\Big(c(v_{j - \Delta}, v_j) + c(v_j, v_{j+\Delta}) \Big).
\end{equation*}
We have $\height(j) = \height(j - i\cdot 2^{\Gamma})$ for any $i \in \parset{0, \dots, r-1}$.
Thus, the contribution of any $v_j$ with $j \in \parset{i \cdot 2^\Gamma : i \in \parset{0, \dots, r}}$ is counted equally on both sides of~\eqref{eq:decomposition_in_blocks}.
Therefore, it suffices to show that
\begin{equation*}
\parset{j-\width(j), \dots, j+\width(j)} \subseteq \parset{i \cdot 2^{\Gamma}, \dots, (i+1)\cdot 2^{\Gamma}}
\end{equation*}
for all $i \in \parset{0, \dots, r-1}$ and $j \in \parset{i\cdot 2^{\Gamma} + 1, \dots, (i+1) \cdot 2^{\Gamma} + 1}.$

Fix such indices $i$ and $j$.
For any $\ell \in \parset{1, \dots, 2^{\height(j)} - 1}$ we have $\height(j + \ell) < \height(j) < \Gamma$ by the definition of $\height$.
This implies that $\parset{j, \dots, j + \width(j)} \subseteq \parset{i \cdot 2^{\Gamma}, \dots, (i+1)\cdot 2^{\Gamma}}$.
Using an analogous argument, one shows that $\parset{j, \dots, j - \width(j)} \subseteq \parset{i \cdot 2^{\Gamma}, \dots, (i+1)\cdot 2^{\Gamma}}$.
This completes the proof.
\end{proof}

If we do not simply aim at minimizing $w^0(P)$, but also want $P$ to satisfy the covering constraint \ref{item:covering_constraint} from \Cref{lem:main_lemma_dp}, knowing the vertices $x_1,\dots, x_r$ is not sufficient to make the choice of the different blocks independent.
We will need additional information, which is captured by the notion of a \emph{covering profile}.

The covering profile realized by a path $P$ should contain enough information to determine $\numnewcovered(P)$. 
Moreover, if we fix the covering profiles realized by two monotone paths $Q_1$ and $Q_2$ and we concatenate them to a new monotone path $P$, then the covering profile realized by $P$ should be completely determined by the covering profiles  for $Q_1$ and $Q_2$.
This will allow us to optimize the paths $Q_i$ for the different blocks (as in \Cref{lem:decomposition_in_blocks}) independently once we fixed the vertices $x_i$ and $x_{i+1}$, as well as the covering profile realized by $Q_i$.

In order to define covering profiles, we introduce the following notation:
\begin{definition}[Open and closed intervals]
For $v\preceq w \in V$ we define by
\begin{equation*}
\openint{v,w} \coloneqq \parset{x \in V : v \prec x \prec w}, \quad\text{and}\quad
\closedint{v,w} \coloneqq \parset{x \in V : v \preceq x \preceq w}
\end{equation*}
the \emph{open} and \emph{closed interval} between $v$ and $w$, respectively.
\end{definition}

Recall that $H_{\ell}(v)$ denotes the unique set $H\in \cH_{\ell}$ containing the vertex $v$.
Now we define covering profiles as follows:

\begin{definition}[Covering profile]\label{def:covering_profile}
A \emph{covering profile} $\Phi$ is a tuple $(a_{\min}, a_{\max}, f,g,\gamma)$, where
\begin{enumerate}
\item
$a_{\min} \preceq a_{\max} \in H^\ast$ are (not necessarily distinct) vertices,

\item
$f: \parset{\ell^{\ast}, \dots, L - 1} \times \parset{\min, \max} \to \parset{0, \dots, n}$ is a function,

\item
$g: \parset{\ell^\ast, \dots, L - 1} \times \parset{\min, \max} \to \parset{0, \dots, n}$ is a function, and

\item
$\gamma \in \parset{0, \dots, nL}$ is an integer.
\end{enumerate}
Given a monotone path $P$ with $V(P) \subseteq H^\ast$, we say that $P$ \emph{realizes} $\Phi$ if the following conditions are satisfied:
\begin{enumerate}
\item
$\parset{a_{\min}, a_{\max}} \subseteq V(P) \subseteq \closedint{a_{\min}, a_{\max}}$, that is, $a_{\min}$ and $a_{\max}$ are the leftmost and rightmost vertices of $P$, respectively,

\item
$f(\ell,m) = |V(P) \cap H_{\ell + 1}(a_m)|$ for all $\ell \in \parset{\ell^\ast, \dots, L - 1}$ and $m \in \parset{\min, \max}$,

\item
$g(\ell,m) = | \left(H_{\ell + 1}(a_m) \cap \closedint{a_{\min}, a_{\max}} \right)\setminus (V(P) \cup S_{\ell}) |$ for all $\ell \in \parset{\ell^\ast, \dots, L - 1}$ and $m \in \parset{\min, \max}$, and

\item
we have
\begin{equation*}
\gamma = \sum_{\ell=\ell^\ast}^L |V(P) \setminus S_{\ell}| + \sum_{\ell=\ell^\ast}^{L - 1} \sum_{\substack{H \in \cH_{\ell + 1}, \\ H \subseteq \openint{a_{\min}, a_{\max}}, \\ |V(P) \cap H| \geq k}} |H \setminus (V(P) \cup S_{\ell})|.
\end{equation*}
\end{enumerate}
We say that $\Phi$ is \emph{realizable} if there exists a monotone path $P$ with $V(P) \subseteq H^\ast$ that realizes $\Phi$.
\end{definition}

Observe that there are only quasi-polynomially many covering profiles.
Moreover, note that for any monotone path $P$ with $V(P) \subseteq H^\ast$ there exists a unique covering profile $\Phi$ that is realized by $P$.
For any vertex $v \in H^\ast$, we write $\singletonprofile(v)$ for the unique covering profile of the singleton-path consisting only of $v$.

We say that a covering profile $\Phi=(a_{\min}, a_{\max}, f,g,\gamma)$ is \emph{left}/\emph{right} of a vertex $v$ if the interval $[a_{\min},a_{\max}]_{\prec}$ is left/right of $v$.
Given another covering profile $\Phi^\prime = (a_{\min}^\prime, a_{\max}^\prime, f^\prime,g^\prime,\gamma^\prime)$, we say that $\Phi$ is \emph{left}/\emph{right} of $\Phi^\prime$ if the interval $[a_{\min},a_{\max}]_{\prec}$ is left/right of the interval $[a_{\min}^\prime,a_{\max}]_{\prec}^\prime$.

\addparskip

The following two lemmas state that the covering profile indeed has the desired properties mentioned above.
First, the covering profile realized by a path $P$ contains enough information to recover the number of newly covered elements~$(v,\ell)$ by~$(P,\ell^\ast)$:

\begin{restatable}{lemma}{CoveringProfileNumElemCovered}\label{lem:covering_profile_num_elem_covered}
Let $\Phi$ be a covering profile.
Then there is a number $\xi(\Phi) \in \parset{0, \dots, nL}$ such that $\xi(\Phi) = \numnewcovered(P)$ for every realization $P$ of $\Phi$.

Moreover, given $\Phi$, we can compute $\xi(\Phi)$ in polynomial time.
\end{restatable}

Second, from two covering profiles $\Phi$ and $\Psi$ realized by two monotone paths $P_{\Phi}$ and $P_{\Psi}$ with $V(P_{\phi}) $ left of $V(P_{\Psi})$, we can determine the covering profile of the monotone path $P$ obtained by concatenating $P_{\Phi}$ and $P_{\Psi}$:

\begin{restatable}{lemma}{CombinationOfCoveringProfiles}\label{lem:combination_of_covering_profiles}
Let $\Phi$ and $\Psi$ be two covering profiles such that $\Phi$ is left of $\Psi$.
Then there exists a covering profile, denoted by $\Phi \coveringsum \Psi$, such that for every realization $P_{\Phi}$ of $\Phi$ and every realization $P_{\Psi}$ of $\Psi$, the monotone path obtained by concatenating $P_{\Phi}$ and $P_{\Psi}$ realizes $\Phi \coveringsum \Psi$.

Moreover, given $\Phi$ and $\Psi$, we can compute $\Phi \coveringsum \Psi$ in polynomial time.
\end{restatable}

The proofs of \Cref{lem:covering_profile_num_elem_covered,lem:combination_of_covering_profiles} can be derived in a straightforward way from \Cref{def:covering_profile} and the definition of what it means for a path-level pair $(P,\ell^*)$ to cover an element of $U$.
However, the proofs are quite technical and we thus defer them to \Cref{sec:remaining_proofs_dp}.
\addparskip

Having described the key properties of covering profiles, we can now explain how we use them in our algorithm.
We construct a directed auxiliary graph $\widetilde{G}$ whose vertices are pairs $(v,\Phi)$, where $v \in H^\ast$ is a vertex in our original vertex set, and $\Phi=(a_{\min}, a_{\max}, f,g,\gamma)$ is a covering profile with $a_{\max} = v$.The idea is that our desired path $P$, decomposed into blocks as in \eqref{eq:path_decomposition}, will correspond to a path in the auxiliary graph $\widetilde{G}$ with vertices $(x_i, \Phi_i)$, where $\Phi_i$ is the covering profile realized by the $x_1$-$x_i$ subpath of $P$.
The edge between two vertices $(x_i, \Phi_i)$ and $(x_{i+1},\Phi_{i+1})$ will correspond to the block $A_i$.

We say that $(v,\Phi) \in V(\widetilde{G})$ is a \emph{designated start vertex} if $\Phi = \singletonprofile(v)$.
Moreover, $(v,\Phi) \in V(\widetilde{G})$ is a \emph{designated end vertex} if $\xi(\Phi) = \gamma^\ast$ (see \Cref{lem:covering_profile_num_elem_covered}).

Next, we define the edges of the auxiliary graph $\widetilde{G}$.
Recall that each block $A_i$ contains exactly $k^\prime$ vertices.
We need the following definition:

\begin{definition}[Optimal realization]\label{def:optimal_realization}
Let $x \prec y \in H^\ast$ be two distinct vertices and let $\Phi$ be a covering profile that is right of $x$ and left of $y$.
We say that $\Phi$ is \emph{$(x,y,k')$-realizable} if there exists a realization $Q$ of $\Phi$ with $|V(Q)| = k'$.

An \emph{optimal $(x,y,k')$-realization} is a realization $Q$ of $\Phi$ with $|V(Q)| = k'$ that minimizes 
\begin{equation*}
w^{0}\left( Q^\prime \right) - k^2 \cdot D_{\ell^\ast}
\end{equation*}
among all such realizations, where $Q^\prime$ is the monotone path that visits $x$, then traverses $Q$, and finally visits~$y$.
\end{definition}

To construct $E(\widetilde{G})$, we need to be able to compute optimal realizations in quasi-polynomial time.

\begin{lemma}\label{lem:dynamic_program_single_block}
Let $x \prec y \in H^\ast$ be two distinct vertices.
Then there is an algorithm that computes, for every $(x,y,k^\prime)$-realizable covering profile $\Phi$, an optimal $(x,y,k^\prime)$-realization in time $n^{O(\log n)}$.
\end{lemma}

We will prove \Cref{lem:dynamic_program_single_block} in \Cref{sec:block_dp} via a dynamic program.
We fix an optimal $(x,y,k^\prime)$-realization of $\Phi$ for all $x,y\in H^\ast$ with $x \prec y$ and every $(x,y,k^\prime)$-realizable covering profile $\Phi$.
Because there are only $n^{O(\log n)}$ many covering profiles, \Cref{lem:dynamic_program_single_block} implies that we can do this in time $n^{O(\log n)}$.

Now, to construct the edge set of $\widetilde{G}$, we do the following for every pair of vertices $\left(x, \Phi_x \right), \left(y, \Phi_y \right) \in V(\widetilde{G})$ with $x \prec y$:
For each $(x,y,k^\prime)$-realizable covering profile $\widehat{\Phi}$ that is right of $x$ and left of $y$, we add an edge $e$ from $\left(x, \Phi_x \right)$ to $\left(y, \Phi_y \right)$ to $\widetilde{G}$ if
\begin{equation*}
\Phi_y = \left( \Phi_x \coveringsum \widehat{\Phi} \right) \coveringsum \singletonprofile (y),
\end{equation*}
using the notation of \Cref{lem:combination_of_covering_profiles}.
(Recall that $ \singletonprofile (y)$ is the covering profile realized by the singleton-path with vertex $y$.)
Additionally, we write $P(e)$ to denote the fixed optimal $(x,y,k^\prime)$-realization of the covering profile $\widehat{\Phi}$.
Then we define the weight of the edge $e$ of $\widetilde{G}$ via
\begin{equation*}
w(e) \coloneqq w^{0}(P^\prime) - k^2 \cdot D_{\ell^\ast},
\end{equation*}
where $P^\prime$ is the monotone path that visits $x$, then traverses $P(e)$, and finally visits $y$.
\addparskip

Every path in $\widetilde{G}$ naturally corresponds to a path in our original vertex set:

\begin{definition}[Corresponding path]
Let $\widetilde{Q}$ be a path in $\widetilde{G}$ that visits the vertices $(x_1,\Phi_1), \dots, (x_r, \Phi_r) \in V(\widetilde{G})$ in this order.
We denote its edges by $e_i \coloneqq \left((x_i,\Phi_i),(x_{i+1},\Phi_{i+1})\right)$ for $i \in [r-1]$.

Then we define \emph{corresponding path} of $\widetilde{Q}$ as the monotone path that starts in $x_1$ and, for each $i=1,\dots, r-1$, traverses $P(e_i)$ and then visits $x_{i+1}$.
\end{definition}

We will prove that we can find a monotone path $P$ as desired in  \Cref{lem:main_lemma_dp} by finding shortest-paths in the auxiliary graph $\widetilde{G}$ with edge weights $w$.
To this end, we observe the following:

\begin{lemma}\label{lem:shortest_path_reduction}
Let $\widetilde{Q}$ be path in $\widetilde{G}$ that starts in a designated start vertex and ends in a vertex $(y,\Phi) \in V(\widetilde{G})$, and let $P$ be its corresponding path.
Then, $P$ realizes $\Phi$ and we have $w(\widetilde{Q}) + k^2 \cdot D_{\ell^{\ast}}= w^0(P)$. 

Moreover, for any monotone path $P^\prime$ satisfying conditions \ref{item:subset_constraint} to \ref{item:covering_constraint} of \Cref{lem:main_lemma_dp},
there exists a path $Q^\prime$ in $\widetilde{G}$ starting in a designated start vertex and ending in a designated end vertex such that $w(Q^\prime) + k^2 \cdot D_{\ell^{\ast}}\leq w^0(P^\prime)$.
\end{lemma}
\begin{proof}
It follows immediately from \Cref{lem:combination_of_covering_profiles} and an induction on the number of edges of $\widetilde{Q}$ that $P$ realizes $\Phi$.
Moreover, \Cref{lem:decomposition_in_blocks} implies that $w(\widetilde{Q}) + k^2 \cdot D_{\ell^{\ast}}= w^0(P)$. 

Let $P^\prime$ be a monotone path that satisfies conditions \ref{item:subset_constraint} to \ref{item:covering_constraint} of \Cref{lem:main_lemma_dp}.
In particular, we have $|V(P^\prime)| = 1 + r \cdot 2^{\Gamma}$ for some $r \geq 0$.
Hence, we can decompose its vertex set into vertices $x_1, \dots, x_r$ and blocks $A_1, \dots, A_{r-1}$, as explained in~\eqref{eq:decomposition_in_blocks}.
A straightforward induction on $r$ using \Cref{lem:decomposition_in_blocks,lem:combination_of_covering_profiles} shows the existence of a path $Q^\prime$ in $\widetilde{G}$ such that $w(Q^\prime) + k^2 \cdot D_{\ell^\ast} \leq w^0(P^\prime)$.
\end{proof}

\Cref{lem:shortest_path_reduction} immediately implies \Cref{lem:main_lemma_dp}:

\begin{proof}[Proof of \Cref{lem:main_lemma_dp}]
We start by constructing the graph $\widetilde{G}$.
There are at most $n^{O(\log n)}$ covering profiles, because $L = O(\log n)$ as our instance $\cI$ has low depth.
Hence, by \Cref{lem:dynamic_program_single_block}, we can construct $\widetilde{G}$ in time $n^{O(\log n)}$.

For every pair of designated start and end vertices $(\Phi_x,x), (\Phi_y,y) \in V(\widetilde{G})$, we compute a shortest path in $\widetilde{G}$ from $(\Phi_x,x)$ to $(\Phi_y,y)$ (if there exists one).
In particular, $\xi(\Phi_y) = \gamma^\ast$.
We denote by $\widetilde{Q}$ such a path in $\widetilde{G}$ that minimizes $w(\widetilde{Q})$ among all shortest paths we computed.
Let $P$ be its corresponding path.
By our construction we have $V(P) \subseteq H^\ast$ and $|V(P)| \equiv 1 \pmod{k^\prime + 1}$.
Moreover, \Cref{lem:covering_profile_num_elem_covered,lem:shortest_path_reduction} yield that $\numnewcovered(P) = \xi(\Phi_y) = \gamma^\ast$.

Finally, \Cref{lem:shortest_path_reduction} implies that
\begin{equation*}
w^0(P) = w(\widetilde{Q}) + k^2 \cdot D_{\ell^\ast} \leq w^0(P^\prime)
\end{equation*}
for any monotone path $P^\prime$ satisfying conditions \ref{item:subset_constraint} to \ref{item:covering_constraint} of \Cref{lem:main_lemma_dp}.
\end{proof}

\subsection{Solving a Single Block}\label{sec:block_dp}

We now show that given two distinct vertices $x, y \in H^\ast$ with $x \prec y$ and an $(x,y,k^\prime)$-realizable covering profile $\Phi$, we can compute an optimal $(x,y,k^\prime)$-realization of $\Phi$ in quasi-polynomial time.
That is, we prove \Cref{lem:dynamic_program_single_block}.

\addparskip

To this end, we fix the vertices $x \prec y \in H^\ast$ for the remainder of this section.
We briefly explain the idea of our dynamic programming approach.
Since any $(x,y,k^\prime)$-realization is a monotone path on $k^\prime = 2^{\Gamma} - 1$ vertices, we can view such a realization as a full binary tree $\cB$ of height $\Gamma$, as illustrated in \Cref{fig:arborescence_dp}.
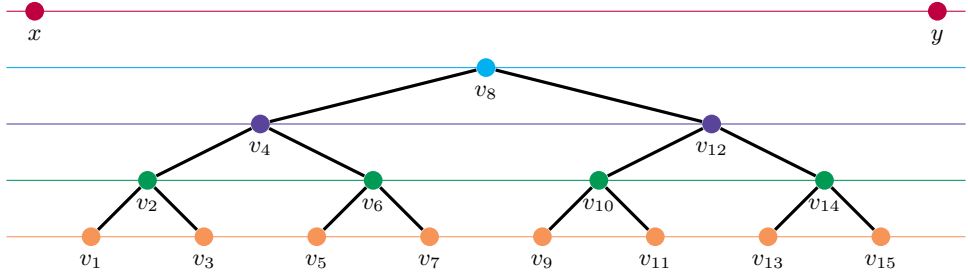
\begin{figure}
\begin{center}

\begin{tikzpicture}[scale=0.75]
\begin{scope}
\definecolor{color1}{named}{Peach}  
\definecolor{color2}{named}{ForestGreen}  
\definecolor{color3}{named}{Violet}  
\definecolor{color4}{named}{ProcessBlue}
\definecolor{color5}{named}{purple}  

\def\slack{1.5}

\foreach \h in {0,1,2,3,4} {
  \pgfmathtruncatemacro{\hc}{\h}
  \ifcase\hc
    \def\hcolor{color1}
  \or
    \def\hcolor{color2}
  \or
    \def\hcolor{color3}
  \or
    \def\hcolor{color4}
   \or
    \def\hcolor{color5}
  \fi
  \draw[thin, \hcolor] (1-\slack, \h) -- (15+\slack, \h);
}

\foreach \i in {1,...,15} {
\pgfmathsetmacro{\height}{
  ifthenelse(mod(\i,8)==0,3,
    ifthenelse(mod(\i,4)==0,2,
      ifthenelse(mod(\i,2)==0,1,0)
    )
  )
}
\pgfmathtruncatemacro{\h}{\height}
  \ifcase\h
    \def\heightcolor{color1}
  \or
    \def\heightcolor{color2}
  \or
    \def\heightcolor{color3}
  \or
    \def\heightcolor{color4}
  \fi
\pgfmathsetmacro{\width}{
  pow(2, \height) - 1
}

\node[circle,fill=\heightcolor,inner sep=2.5pt] (v\i) at ({(\i)*1},\height) {};
\node[below=3pt, font=\footnotesize] at (v\i) {$v_{\i}$};

}

\node[circle,fill=color5,inner sep=2.5pt] (x) at (0,4) {};
\node[below=3pt, font=\footnotesize] at (x) {$x$};
\node[circle,fill=color5,inner sep=2.5pt] (y) at (16,4) {};
\node[below=3pt, font=\footnotesize] at (y) {$y$};

\begin{scope}[very thick]
\draw (v8) to (v4);
\draw (v8) to (v12);
\draw (v4) to (v2);
\draw (v4) to (v6);
\draw (v12) to (v10);
\draw (v12) to (v14);
\draw (v2) to (v1);
\draw (v2) to (v3);
\draw (v6) to (v5);
\draw (v6) to (v7);
\draw (v10) to (v9);
\draw (v10) to (v11);
\draw (v14) to (v13);
\draw (v14) to (v15);
\end{scope}

\end{scope}
\end{tikzpicture} \end{center}
\caption{\label{fig:arborescence_dp}
Illustration of the full binary tree $\cB$ corresponding to an optimal $(x,y,k^\prime)$-realization (for $k^\prime = 15$ and $\Gamma=4$ ) that visits the vertices $v_1,\dots, v_{15}$ in this order.
The color and vertical position of the vertices indicated in the picture corresponds to the height assigned to the vertices in the definition of the weight bound function $w^0(Q^\prime)$, where $Q^\prime$ is the monotone path visiting $x, v_1, \dots, v_{15}, y$ in this order (see \Cref{def:optimal_realization}).
}
\end{figure}

Let $h(v)$ denote the height of $v$ in the binary tree $\cB$ (where leaves have height $1$ and the root has height $\Gamma$), and let $\pi_v(h+1), \dots, \pi_v(\Gamma)$ denote the vertices on the path from $v$ to the root of $\cB$.
The cost function we minimize (see \Cref{def:optimal_realization}) can then be rewritten as
\begin{equation}
\label{eq:block_cost}
2^{\Gamma} \cdot \sum_{v \in V(\cB)}\Big(c(x,v) + c(v,y)\Big) + \sum_{v \in V(\cB)} \sum_{\ell = h(v) + 1}^{\Gamma} 2^{\ell - 1} \cdot \min\Big( c(v,\pi_v(\ell)), c(\pi_v(\ell), v) \Big),
\end{equation}
as we will show in \Cref{lem:tree_solutions}.

The summands involving vertices in the subtree $\cB_w$ rooted at a vertex $w$ are
\[
2^{\Gamma} \cdot \sum_{v \in V(\cB_w)}\Big(c(x,v) + c(v,y)\Big) + \sum_{v \in V(\cB_w)} \sum_{\ell = h(v) + 1}^{\Gamma} 2^{\ell - 1} \cdot \min\Big( c(v,\pi_v(\ell)), c(\pi_v(\ell), v) \Big),
\]
and we will refer to this as the cost of the subtree rooted at $w$.
We highlight that the cost of any subtree depends on the path between$w$ and the root of the entire tree, but not on any other vertices outside of the subtree.
Our dynamic program computes a cheapest subtree of height $h$ for each 
\begin{itemize}
\item
height $h \in [\Gamma]$,

\item
path $\pi(h+1), \dots, \pi(\Gamma)$ from the root of the subtree to the root of the entire tree,

\item
covering profile $\Phi$,
\end{itemize}
such that the monotone path on the vertices of this subtree realizes $\Phi$.
We then construct the entire tree bottom-up by merging such subtree solutions of height less than $\Gamma$ that agree on both their remaining paths to the global root and their covering profiles (see \Cref{lem:combination_of_covering_profiles}).
\addparskip

In the following we make these ideas more precise.

\begin{definition}[Tree solutions]
Let $h \in [\Gamma]$.
A \emph{tree solution of height $h$} is a pair $(\cB,\pi)$, where
\begin{enumerate}
\item
$\cB$ is a full binary tree with $V(\cB) \subseteq \openint{x,y}$ and $|V(\cB)| = 2^h - 1$, such that for every non-leaf $r \in V(\cB)$, we can label the two full binary subtrees rooted at the children of $r$ as $\cB_1$ and $\cB_2$ so that $V(\cB_1)$ is left of $r$ and $V(\cB_2)$ is right of $r$, and

\item
$\pi: \parset{h + 1, \dots, \Gamma} \to \openint{x,y}$ is a function such that $\pi(\ell)$ is left of $V(\cB)$ or right of $V(\cB)$ for every $\ell \in \parset{h + 1, \dots, \Gamma}$.
\end{enumerate}
Given a covering profile $\Phi$, we say that a tree solution $(\cB,\pi)$ \emph{realizes} $\Phi$ if the unique monotone path $P$ whose vertex set is $V(\cB)$ realizes $\Phi$.
Conversely, given a tree solution $(\cB,\pi)$, we define its covering profile to be the covering profile of the unique monotone path $P$ whose vertex set is $V(\cB)$.
\end{definition}

Next, we formalize how we can merge subtree solutions.
See \Cref{fig:tree_solutions.tex} for an illustration.
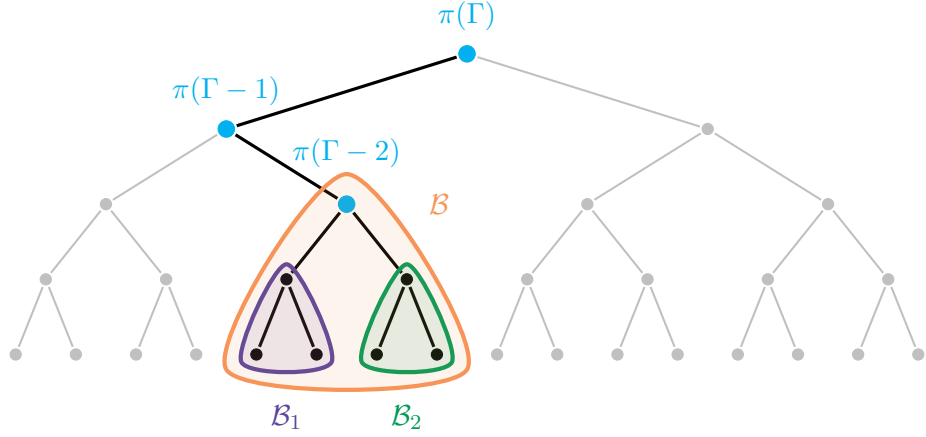
\begin{figure}
\begin{center}
\begin{tikzpicture}[
vertex/.style={circle,draw=black, fill,inner sep=1.5pt, outer sep=1pt},
root/.style={circle,draw=ProcessBlue, fill=ProcessBlue,inner sep=2.2pt, outer sep=1pt,},
edge/.style={thick},
highlight/.style={ultra thick, smooth cycle, fill, fill opacity=0.1},
xscale=0.4
]

\definecolor{color1}{named}{Peach}  
\definecolor{color2}{named}{ForestGreen}  
\definecolor{color3}{named}{Violet}  
\definecolor{color4}{named}{ProcessBlue} 

\def\opac{0.25}

\begin{scope}
\foreach \i in {1,...,31} {
\pgfmathsetmacro{\height}{
ifthenelse(mod(\i,16)==0,4,
ifthenelse(mod(\i,8)==0,3,
ifthenelse(mod(\i,4)==0,2,
ifthenelse(mod(\i,2)==0,1,0
))))
}

\ifnum\i=12
\node[root] (v\i) at ({(\i)*1},\height) {};
\else
\ifnum\i=8
\node[root] (v\i) at ({(\i)*1},\height) {};
\else
\ifnum\i=16
\node[root] (v\i) at ({(\i)*1},\height) {};
\else
\ifnum\i<9
\node[vertex, opacity=\opac] (v\i) at ({(\i)*1},\height) {};
\else
\ifnum\i>15
\node[vertex, opacity=\opac] (v\i) at ({(\i)*1},\height) {};
\else
\node[vertex] (v\i) at ({(\i)*1},\height) {};
\fi\fi\fi\fi\fi
}
\end{scope}

\node[color=ProcessBlue, right=10pt, above=10pt] at (v12) {$\pi(\Gamma - 2)$};
\node[color=ProcessBlue, above=5pt] at (v8) {$\pi(\Gamma - 1)$};
\node[color=ProcessBlue, above=5pt] at (v16) {$\pi(\Gamma)$};

\foreach \i in {1,...,31} {
\pgfmathsetmacro{\childdist}{
ifthenelse(mod(\i,16)==0,8,
ifthenelse(mod(\i,8)==0,4,
ifthenelse(mod(\i,4)==0,2,
ifthenelse(mod(\i,2)==0,1,0
))))
}

\ifnum\childdist>0
\pgfmathtruncatemacro{\childa}{\i + \childdist}
\pgfmathtruncatemacro{\childb}{\i - \childdist}

\ifnum\childa<8
\draw[edge, opacity=\opac] (v\i) -- (v\childa);
\else
\ifnum\childa>16
\draw[edge, opacity=\opac] (v\i) -- (v\childa);
\else
\draw[edge, very thick] (v\i) -- (v\childa);
\fi\fi

\ifnum\childb<8
\draw[edge, opacity=\opac] (v\i) -- (v\childb);
\else
\ifnum\childb>16
\draw[edge, opacity=\opac] (v\i) -- (v\childb);
\else
\draw[edge, very thick] (v\i) -- (v\childb);
\fi\fi

\fi
}

\coordinate (v9slack) at ($(v9) - (0.5,0.1)$);
\coordinate (v10slack) at ($(v10) + (0, 0.2)$);
\coordinate (v11slack) at ($(v11) + (0.5, -0.1)$);
\draw[highlight, color=color3] plot coordinates {(v9slack) (v10slack) (v11slack)};
\node[color=color3, below=4em] at (v10) {$\cB_1$};

\coordinate (v13slack) at ($(v13) - (0.5,0.1)$);
\coordinate (v14slack) at ($(v14) + (0, 0.2)$);
\coordinate (v15slack) at ($(v15) + (0.5, -0.1)$);
\draw[highlight, color=color2] plot coordinates {(v13slack) (v14slack) (v15slack)};
\node[color=color2, below=4em] at (v14) {$\cB_2$};

\coordinate (v9slack2) at ($(v9) - (1,0.2)$);
\coordinate (v15slack2) at ($(v15) + (1, -0.2)$);
\coordinate (v12slack) at ($(v12) + (0, 0.4)$);
\draw[highlight, color=color1] plot coordinates {(v9slack2) (v12slack) (v15slack2)};
\node[color=color1, right=2.5em] at (v12) {$\cB$};

\end{tikzpicture}
 \end{center}
\caption{
Illustration of how two compatible tree solutions $(\cB_1,\pi)$ and $(\cB_2, \pi)$ of height~$h < \Gamma $ are combined.
All vertices are drawn from left to right according to the total order~$\prec$.
The path $\pi(h+1), \dots, \pi(\Gamma)$ from the root of the current subtree to the global root $\pi(\Gamma)$ contains the vertices at which it will later be merged with another subtree.
}
\label{fig:tree_solutions.tex}
\end{figure}

\begin{definition}[Compatible tree solutions and costs]
Given two tree solutions $(\cB_1, \pi_1)$ and $(\cB_2, \pi_2)$ of height~$h \in [\Gamma - 1]$, we say that they are \emph{compatible} if $\pi_1 = \pi_2$ and $V(\cB_1)$ is left of $\pi(h+1)$ and $V(\cB_2)$ is right of $\pi(h+1)$, where  $\pi \coloneqq \pi_1 = \pi_2$.
We denote by $(\cB_1, \pi_1) \coveringsum (\cB_2, \pi_2)$
the tree solution $(\cB,\pi\mid_{\parset{h+2, \dots, \Gamma}})$ of height $h + 1$,
where $\cB$ is the full binary tree with root $\pi(h)$ and subtrees $\cB_1$ and $\cB_2$.
Note that any tree solution $(\cB, \pi)$ of height $h \in \parset{2, \dots, \Gamma}$ can be uniquely decomposed via
\begin{equation*}
(\cB, \pi) = (\cB_1, \pi_1) \coveringsum (\cB_2, \pi_2)
\end{equation*}
into two compatible tree solutions $(\cB_1, \pi_1)$ and $(\cB_2, \pi_2)$ of height $h - 1$.
Therefore, for $h > 1$, we can recursively define
\begin{equation*}
\begin{split}
\cost(\cB, \pi) 
\coloneqq &\cost(\cB_1, \pi_1) + \cost(\cB_2, \pi_2) + 2^{\Gamma} \cdot \Big(c(x,r) + c(r, y)\Big) \\
&+ \sum_{\ell = h+1}^{\Gamma} 2^{\ell - 1} \cdot \min\Big(c(r, \pi(\ell)), c(\pi(\ell), r)\Big)
\end{split}
\end{equation*}
where $r$ is the root of $\cB$.
For $h = 1$, we replace the term $\cost(\cB_1, \pi_1) + \cost(\cB_2, \pi_2)$ by zero.
\end{definition}

The following lemma formally states that $(x,y,k^\prime)$-realizations of covering profiles can indeed be regarded as tree solutions of height $\Gamma$.
Moreover, in this case, our cost definition coincides with our objective function in \Cref{def:optimal_realization}.

\begin{restatable}{lemma}{TreeSolutions}\label{lem:tree_solutions}
Let $\Phi$ be an $(x,y,k^\prime)$-realizable covering profile.
Then for every $(x,y,k^\prime)$-realization $P$ of $\Phi$, there exists a tree solution $(\cB,\pi)$ of height $\Gamma$ with $V(\cB) = V(P)$ and $\cost(\cB,\pi) = w^0(P^\prime) - k^2 \cdot D_{\ell^\ast}$, where $P^\prime$ denotes the monotone path that starts in $x$, then traverses $P$, and ends in $y$.
\end{restatable}

The proof of \Cref{lem:tree_solutions} follows in a straightforward way from the definitions, but it is rather technical.
Thus, we defer it to \Cref{sec:remaining_proofs_dp}  and first show how we use \Cref{lem:tree_solutions} to prove \Cref{lem:dynamic_program_single_block}:

\begin{proof}[Proof of \Cref{lem:dynamic_program_single_block}]
We give a dynamic program.
Our table entries are indexed by triples $(h,\Phi,\pi)$ where 
\begin{itemize}
\item
$h \in [\Gamma]$ is an integer,

\item
$\Phi = \left(a_{\min}, a_{\max}, f, g, \gamma \right)$ is a covering profile with $a_{\min}, a_{\max} \in \openint{x,y}$, and

\item
$\pi: \parset{h+1, \dots, \Gamma} \to \openint{x,y}$ is a function. 
\end{itemize}
For each such triple, we compute a tree solution $(\cB,\pi)$ of height $h$ that realizes $\Phi$ (if one exists), minimizing $\cost(\cB,\pi)$ among all such tree solutions.
\Cref{lem:tree_solutions} implies that this suffices to find an optimal $(x,y,k^\prime)$-realization of every $(x,y,k^\prime)$-realizable covering profile.

Tree solutions of height $1$ are singletons.
Thus, finding such an optimal tree solution for all labels $(h,\Phi,\pi)$ with $h=1$ can be done by enumeration.

\addparskip

Suppose that $h \geq 2$ and assume that for every table index $(h-1, \Psi, \pi^\prime)$ we have already computed a tree solution $(\cB,\pi)$ of height $h-1$ that realizes $\Phi$ and has minimum cost (if such a tree solution exists).
Let $(h,\Phi,\pi)$ be an arbitrary but fixed table index and we show how to fill the corresponding entry.

We enumerate all vertices $r \in \openint{x,y}$, and all covering profiles $\Psi_1$ and $\Psi_2$ that are right of $x$ and left of $y$, such that $\Psi_1$ is left of $r$ and $\Psi_2$ is right of $r$, and
\begin{equation*}
\Phi = \left(\Psi_1 \coveringsum \singletonprofile(r) \right) \coveringsum  \Psi_2.
\end{equation*}
Given such $r$, $\Psi_1$, and $\Psi_2$, we define
\begin{equation*}
\pi^\prime(\ell) \coloneqq
\begin{cases}
\pi(\ell) \quad&\text{if $\ell \in \parset{h+1, \dots, \Gamma}$}, \\
r \quad&\text{if $\ell = h$}.
\end{cases}
\end{equation*}
Suppose that there exist tree solutions $(\cB_1, \pi^\prime)$ and $(\cB_2, \pi^\prime)$ of height $h - 1$ realizing $\Psi_1$ and $\Psi_2$, respectively, and we consider such tree solutions of minimum cost.
Note that $(\cB_1, \pi^\prime)$ and $(\cB_2, \pi^\prime)$ are compatible, and, by \Cref{lem:combination_of_covering_profiles}, $(\cB, \pi^\prime) \coloneqq (\cB_1, \pi^\prime) \coveringsum (\cB_2, \pi^\prime)$ is a tree solution of height $h$ that realizes $\Phi$.

After enumerating all $r$, $\Psi_1$, and $\Psi_2$, we choose such a tree solution $(\cB, \pi) \coloneqq (\cB_1, \pi^\prime) \coveringsum (\cB_2, \pi^\prime)$ of minimum cost for the table entry indexed by $(h,\Phi,\pi)$, if we have constructed at least one such tree solution $(\cB, \pi^\prime)$.

\addparskip

We argue that this procedure is correct:
Suppose that there exists a tree solution $(\cB^\ast, \pi)$ of height $h$ that realizes $\Phi$.
Then we can decompose $(\cB^\ast, \pi) = (\cB^\ast_1, \pi^\ast_1) \coveringsum (\cB^\ast_2, \pi^\ast_2)$.
Note that $\pi^\ast_1(h) = \pi^\ast_2(h) = r^\ast$, where $r^\ast$ denotes the root of $\cB^\ast$.
Let $\Psi_1^\ast$ and $\Psi_2^\ast$ be the covering profiles of $(\cB^\ast_1, \pi^\ast_1)$ and $(\cB^\ast_2, \pi^\ast_2)$, respectively.
By \Cref{lem:combination_of_covering_profiles} we have $\Phi = \left(\Psi_1^\ast \coveringsum \singletonprofile(r^\ast)\right) \coveringsum \Psi_2^\ast$.
Hence, $r^\ast$, $\Psi_1^\ast$, and $\Psi_2^\ast$ occurred in our enumeration, and we correctly find a realization.
Moreover, by the recursive definition of $\cost(\cB,\pi)$, the tree solution we find has minimum cost. 
\addparskip

To bound the running time, note that there are at most $n^{O(\log n)}$ covering profiles, since our fixed instance $\cI$ of Path Covering has low depth and therefore $L = O(\log n)$.
Since $\Gamma = O(\log k^\prime) = O(\log n)$, our table contains at most 
\begin{equation*}
\Gamma \cdot n^{O(\log n)} \cdot n^{\Gamma} \leq n^{O(\log n)} 
\end{equation*}
entries, and filling each entry takes additional time of at most $n^{O(\log n)}$ by the procedure we described above.
\end{proof}

\subsection{Remaining Proofs}\label{sec:remaining_proofs_dp}
It remains to prove \Cref{lem:covering_profile_num_elem_covered,lem:combination_of_covering_profiles,lem:tree_solutions}, which we restate here for convenience:

\CoveringProfileNumElemCovered*
\begin{proof}
Without loss of generality, assume that $\Phi = (a_{\min}, a_{\max}, f, g, \gamma)$ is realizable and fix an arbitrary realization $P$ of $\Phi$.
By the definition of what it means for the path-level pair $(P, \ell^\ast)$ to cover an element of $U$, we have
\begin{equation*}
\numnewcovered(P) 
= \sum_{\ell=\ell^\ast}^L |V(P) \setminus S_\ell| + \sum_{\ell = \ell^\ast}^{L-1} \sum_{\substack{H \in \cH_{\ell + 1}, \\ |V(P) \cap H| \geq k}} |H \setminus (V(P) \cup S_{\ell})|.
\end{equation*}
Let $\ell \in \parset{\ell^\ast,\dots, L - 1}$.
Then, because $V(P) \subseteq \closedint{a_{\min}, a_{\max}}$, for every $H \in \cH_{\ell+1}$ with $|V(P) \cap H| \geq k$ exactly one of the following two is true:
\begin{itemize}
\item
$H = H_{\ell + 1}(a_{\min})$ or $H = H_{\ell + 1}(a_{\max})$.

\item
$H \subseteq \openint{a_{\min}, a_{\max}}$.
\end{itemize}

Thus, for 
\begin{equation*}
\begin{split}
I_{\neq} &\coloneqq \parset{\ell \in \parset{\ell^\ast, \dots, L - 1} : H_{\ell + 1}(a_{\min}) \neq H_{\ell + 1}(a_{\max})}, \text{ and } \\
I_{=} &\coloneqq \parset{\ell \in \parset{\ell^\ast, \dots, L - 1} : H_{\ell + 1}(a_{\min}) = H_{\ell + 1}(a_{\max})},
\end{split}
\end{equation*}
we have
\begin{equation*}
\begin{split}
\sum_{i = \ell^\ast}^{L - 1} \sum_{\substack{H \in \cH_{\ell + 1}, \\ |V(P) \cap H| \geq k}} |H \setminus (V(P) \cup S_{\ell})| 
&= \sum_{\ell = \ell^\ast}^{L - 1} \sum_{\substack{H \in \cH_{\ell + 1}, \\ H \subseteq \openint{a_{\min}, a_{\max}}, \\ |V(P) \cap H| \geq k}} |H \setminus (V(P) \cup S_{\ell})| \\
&\quad + \sum_{\ell \in I_{\neq}} \sum_{\substack{ m \in \parset{\min, \max}, \\ |V(P) \cap H_{\ell + 1}(a_m)| \geq k }} |H_{\ell + 1}(a_m) \setminus (V(P) \cup S_{\ell})| \\
&\quad + \sum_{\substack{ \ell \in I_=, \\ |V(P) \cap H_{\ell + 1}(a_{\min})| \geq k}} |H_{\ell + 1}(a_{\min}) \setminus (V(P) \cup S_{\ell})|.
\end{split}
\end{equation*}
Note that we have
\begin{equation*}
H_{\ell + 1}(a_{m}) \setminus V(P) = \Big(\left( H_{\ell + 1} (a_{m}) \cap \closedint{a_{\min}, a_{\max}} \right) \setminus V(P) \Big) \cupp \Big( H_{\ell + 1} (a_{m}) \setminus \closedint{a_{\min}, a_{\max}} \Big)
\end{equation*}
for every $m \in \parset{\min, \max}$ and $\ell \in \parset{\ell^\ast, \dots, L - 1}$.
Hence, combining this with 
\begin{equation*}
\begin{split}
f(\ell,m) &= |V(P) \cap H_{\ell + 1}(a_m)|, \\
g(\ell,m) &= | \left(H_{\ell + 1}(a_m) \cap \closedint{a_{\min}, a_{\max}} \right)\setminus (V(P) \cup S_{\ell}) | , \quad \text{and} \\
\gamma &= \sum_{\ell=\ell^\ast}^L |V(P) \setminus S_{\ell}| + \sum_{\ell=\ell^\ast}^{L - 1} \sum_{\substack{H \in \cH_{\ell + 1}, \\ H \subseteq \openint{a_{\min}, a_{\max}}, \\ |V(P) \cap H| \geq k}} |H \setminus (V(P) \cup S_{\ell})|,
\end{split}
\end{equation*}
we obtain that
\begin{equation*}
\begin{split}
\numnewcovered(P) 
&= \sum_{\ell=\ell^\ast}^L |V(P) \setminus S_\ell| + \sum_{\ell = \ell^\ast}^{L-1} \sum_{\substack{H \in \cH_{\ell + 1}, \\ |V(P) \cap H| \geq k}} |H \setminus (V(P) \cup S_{\ell})| \\
&= \gamma + \sum_{\ell \in I_{\neq}} \sum_{\substack{ m \in \parset{\min, \max}, \\ f(\ell,m) \geq k }} \Big(g(\ell,m) + \big|H_{\ell + 1}(a_m) \setminus(\closedint{a_{\min}, a_{\max}} \cup S_{\ell})\big| \Big) \\
&\quad + \sum_{\substack{\ell \in I_=, \\ f(\ell, \min) \geq k}} \Big( g(\ell, \min) + \big|H_{\ell + 1}(a_{\min}) \setminus (\closedint{a_{\min}, a_{\max}} \cup S_{\ell})\big| \Big).
\end{split}
\end{equation*}
Observe that the right-hand side depends only on $\Phi$ and not on $P$.
Thus, choosing $\xi(\Phi)$ as this value completes the proof.
\end{proof}

\Cref{lem:combination_of_covering_profiles} follows from similar calculations using the definition of covering profiles:

\CombinationOfCoveringProfiles*
\begin{proof}
Without loss of generality, assume that both covering profiles $\Phi = \left(a_{\min}^{\Phi}, a_{\max}^{\Phi}, f^{\Phi}, g^{\Phi}, \gamma^{\Phi} \right)$ and $\Psi = \left(a_{\min}^{\Psi}, a_{\max}^{\Psi}, f^{\Psi}, g^{\Psi}, \gamma^{\Psi} \right)$ are realizable and fix two arbitrary realizations $P_\Phi$ and $P_\Psi$, respectively.
We denote by $P$ the concatenation of both paths.
In particular, $V(P) = V(P_\Phi) \cupp V(P_\Psi)$.

We will construct each component of
\begin{equation*}
\Phi \coveringsum \Psi \coloneqq \left(a_{\min}^{\coveringsum}, a_{\max}^{\coveringsum}, f^{\coveringsum},g^{\coveringsum},\gamma^{\coveringsum}\right)
\end{equation*}
individually.

First, we set $a_{\min}^{\coveringsum} \coloneqq a_{\min}^{\Phi}$ and $a_{\max}^{\coveringsum} \coloneqq a_{\max}^{\Psi}$.
We have
\begin{equation*}
\parset{a_{\min}^{\coveringsum}, a_{\max}^{\coveringsum}} \subseteq V(P_\Phi) \cup V(P_\Psi) \subseteq \closedint{a_{\min}^{\coveringsum}, a_{\max}^{\coveringsum}},
\end{equation*}
since $a_{\max}^{\Phi} \prec a_{\min}^{\Psi}$ and $a_{\min}^{\Phi} \in V(P_{\Phi}) \subseteq \closedint{a_{\min}^{\Phi}, a_{\max}^{\Phi}}$ and $a_{\max}^{\Psi} \in V(P_\Psi) \subseteq \closedint{a_{\min}^{\Psi}, a_{\max}^{\Psi}}$.

\addparskip

Next, let $\ell \in \parset{\ell^\ast, \dots, L - 1}$.
We have either $H_{\ell + 1}\left( a_{\min}^{\Phi} \right) =  H_{\ell + 1}\left( a_{\min}^{\Psi} \right)$, or $H_{\ell + 1}\left( a_{\min}^{\Phi} \right) \prec H_{\ell + 1}\left( a_{\min}^{\Psi} \right)$.
In the former case, since $V(P_\Phi)$ and $V(P_\Psi)$ are disjoint, we have
\begin{equation*}
\begin{split}
\left| H_{\ell + 1}\left(a_{\min}^{\coveringsum} \right) \cap (V(P_{\Phi}) \cup V(P_\Psi)) \right|
&= \left| H_{\ell + 1}\left(a_{\min}^{\Phi} \right) \cap V(P_\Phi) \right| + \left| H_{\ell + 1}\left(a_{\min}^{\Psi} \right) \cap V(P_{\Psi}) \right| \\
&= f^{\Phi}(\ell,\min) + f^{\Psi}(\ell,\min).
\end{split}
\end{equation*}
In the latter case, we have
\begin{equation*}
\begin{split}
\left| \cH_{\ell + 1}\left(a_{\min}^{\coveringsum} \right) \cap (V(P_\Phi) \cup V(P_\Psi)) \right|
&= \left| H_{\ell + 1}\left(a_{\min}^{\Phi} \right) \cap V(P_{\Phi}) \right| + \left| H_{\ell + 1}\left(a_{\min}^{\Phi} \right) \cap V(P_{\Psi}) \right| \\
&= f^{\Phi}(\ell,\min).
\end{split}
\end{equation*}
Hence, we can choose
\begin{equation*}
f^{\coveringsum}(\ell,\min) \coloneqq 
\begin{cases} 
f^{\Phi}(\ell,\min) + f^{\Psi}(\ell,\min) \quad&\text{if }H_{\ell + 1}\left( a_{\min}^{\Phi} \right) =  H_{\ell + 1}\left( a_{\min}^{\Psi} \right),\\
f^{\Phi}(\ell,\min) \quad&\text{if }H_{\ell + 1}\left( a_{\min}^{\Phi} \right) \prec  H_{\ell + 1}\left( a_{\min}^{\Psi} \right).
\end{cases}
\end{equation*}
Observe that the right-hand side depends only on $\Phi$ and $\Psi$, and not on $P_{\Phi}$ or $P_{\Psi}$.
We can define $f^{\coveringsum}(\ell,\max)$ analogously.

\addparskip

We proceed by defining $g^{\coveringsum}(\ell, \min)$ for $\ell \in \parset{\ell^\ast, \dots, L-1}$.
We follow the same case distinction as before.
Again, suppose that $H_{\ell + 1}\left( a_{\min}^{\Phi} \right) =  H_{\ell + 1}\left( a_{\min}^{\Psi} \right)$.
Then,
\begin{equation*}
\begin{split}
&\Big|\Big( H_{\ell + 1}\left( a_{\min}^{\coveringsum} \right) \cap \closedint{a_{\min}^{\coveringsum}, a_{\max}^{\coveringsum}} \Big) \setminus \Big(V(P_{\Phi}) \cup V(P_{\Psi}) \cup S_{\ell} \Big)\Big| \\
&= \Big|\Big( H_{\ell + 1}\left( a_{\min}^{\Phi} \right) \cap \closedint{a_{\min}^{\Phi}, a_{\max}^{\Phi}} \Big) \setminus \Big(V(P_{\Phi}) \cup S_{\ell} \Big)\Big| \\
&\quad + \Big|\Big( H_{\ell + 1}\left( a_{\min}^{\Phi} \right) \cap \openint{a_{\max}^{\Phi}, a_{\min}^{\Psi}} \Big) \setminus S_{\ell} \Big| \\
&\quad +  \Big|\Big( H_{\ell + 1}\left( a_{\min}^{\Psi} \right) \cap \closedint{a_{\min}^{\Psi}, a_{\max}^{\Psi}} \Big) \setminus \Big(V(P_{\Psi}) \cup S_{\ell} \Big)\Big| \\
&= g^{\Phi}(\ell, \min) + g^{\Psi}(\ell, \min) +  \Big|\Big( H_{\ell + 1}\left( a_{\min}^{\Phi} \right) \cap \openint{a_{\max}^{\Phi}, a_{\min}^{\Psi}} \Big) \setminus S_{\ell} \Big|.
\end{split}
\end{equation*}
The right-hand side depends only $\Phi$ and $\Psi$, and not on $P_{\Phi}$ or $P_{\Psi}$.
Thus, in this case, we can define $g^{\coveringsum}(\ell, \min)$ to be this value.
If we have $H_{\ell + 1}\left( a_{\min}^{\Phi} \right) \prec H_{\ell + 1}\left( a_{\min}^{\Psi} \right)$, then
\begin{equation*}
H_{\ell + 1}\left( a_{\min}^{\Phi} \right) \cap \closedint{a_{\min}^{\Psi}, a_{\max}^{\Psi}} = \emptyset,
\end{equation*}
and this case can be treated similarly.
Moreover, we can define $g^{\coveringsum}(\ell, \max)$ analogously.

\addparskip

Finally, we show how to define $\gamma^{\coveringsum}$.
To this end, consider the index sets
\begin{equation*}
\begin{split}
I_{\neq} &\coloneqq \parset{\ell \in \parset{\ell^\ast, \dots, L - 1} : H_{\ell + 1}\left( a^{\Phi}_{\max} \right) \neq H_{\ell + 1}\left( a^{\Psi}_{\min} \right)}, \quad\text{and} \\
I_{=} &\coloneqq \parset{\ell \in \parset{\ell^\ast, \dots, L - 1} : H_{\ell + 1}\left( a^{\Phi}_{\max} \right) = H_{\ell + 1}\left( a^{\Psi}_{\min} \right)}.
\end{split}
\end{equation*}
Additionally, we define
\begin{equation*}
\begin{split}
\widetilde{I}_{\neq}^{\Phi} &\coloneqq \parset{\ell \in I_{\neq} : H_{\ell + 1}\left(a^{\Phi}_{\max}\right) \subseteq \openint{a_{\min}^{\coveringsum}, a_{\max}^{\coveringsum}}, \left|V(P_\Phi) \cap H_{\ell + 1}\left(a^{\Phi}_{\max}\right)\right| \geq k}, \\
\widetilde{I}_{\neq}^{\Psi} &\coloneqq \parset{\ell \in I_{\neq} : H_{\ell + 1}\left(a^{\Psi}_{\min}\right) \subseteq \openint{a_{\min}^{\coveringsum}, a_{\max}^{\coveringsum}}, \left|V(P_\Psi) \cap H_{\ell + 1}\left(a^{\Psi}_{\min}\right)\right| \geq k}, \quad\text{ and } \\
\widetilde{I}_{=} &\coloneqq \parset{\ell \in I_{=} : H_{\ell + 1}\left(a^{\Phi}_{\max}\right) \subseteq \openint{a_{\min}^{\coveringsum}, a_{\max}^{\coveringsum}}, \left|(V(P_\Phi) \cup V(P_{\Psi})) \cap H_{\ell + 1}\left(a^{\Phi}_{\max}\right)\right| \geq k}.
\end{split}
\end{equation*}
Note that all of the index sets above can be constructed from $\Phi$ and $\Psi$ alone, since, for example,
\begin{equation*}
\left|(V(P_\Phi) \cup V(P_{\Psi})) \cap H_{\ell + 1}\left(a^{\Phi}_{\max}\right)\right|
= f^\Phi(\ell, \max) + f^\Psi(\ell, \min) 
\end{equation*}
for all $\ell \in I_{=}$.

For any $\ell \in \parset{\ell^\ast, \dots, L - 1}$ and $H \in \cH_{\ell + 1}$ with $H \subseteq \openint{a_{\min}^{\coveringsum}, a_{\max}^{\coveringsum}}$ and $\left|(V(P_\Phi) \cup V(P_{\Psi})) \cap H \right| \geq k$ exactly one of the following two is true:
\begin{itemize}
\item
$H \subseteq \openint{a_{\min}^{\Phi}, a_{\max}^{\Phi}}$ or $H \subseteq \openint{a_{\min}^{\Psi}, a_{\max}^{\Psi}}$.

\item
$H = H_{\ell + 1}\left( a^{\Phi}_{\max} \right)$ or $H = H_{\ell + 1}\left( a^{\Psi}_{\min} \right)$.
\end{itemize}
This implies
\begin{equation*}
\begin{split}
&\sum_{\ell=\ell^\ast}^{L - 1} \sum_{\substack{H \in \cH_{\ell + 1}, \\ H \subseteq \openint{a_{\min}^{\coveringsum}, a_{\max}^{\coveringsum}}, \\ |V(P) \cap H| \geq k}} |H \setminus (V(P) \cup S_{\ell})| \\
&= \sum_{\ell=\ell^\ast}^{L - 1} \sum_{\substack{H \in \cH_{\ell + 1}, \\ H \subseteq \openint{a_{\min}^{\Phi}, a_{\max}^{\Phi}}, \\ |V(P_\Phi) \cap H| \geq k}} |H \setminus (V(P_\Phi) \cup S_{\ell})|
+ \sum_{\ell=\ell^\ast}^{L - 1} \sum_{\substack{H \in \cH_{\ell + 1}, \\ H \subseteq \openint{a_{\min}^{\Psi}, a_{\max}^{\Psi}}, \\ |V(P_\Psi) \cap H| \geq k}} |H \setminus (V(P_\Psi) \cup S_{\ell})| \\
&\quad + \sum_{\ell \in \widetilde{I}_{\neq}^{\Phi}} \left| H_{\ell + 1}\left(a_{\max}^\Phi \right) \setminus (V(P_{\Phi}) \cup S_{\ell}) \right|
+ \sum_{\ell \in \widetilde{I}_{\neq}^{\Psi}} \left| H_{\ell + 1}\left(a_{\min}^\Psi \right) \setminus (V(P_{\Psi}) \cup S_{\ell}) \right| \\
&\quad + \sum_{\ell \in \widetilde{I}_{=}} \left| H_{\ell + 1}\left(a_{\max}^\Phi \right) \setminus (V(P_{\Phi}) \cup V(P_{\Psi}) \cup S_{\ell}) \right|.
\end{split}
\end{equation*}
Note that, for all $\ell \in \parset{\ell^\ast, \dots, L}$, we have
\begin{equation*}
\begin{split}
\left| H_{\ell + 1}\left(a_{\max}^\Phi \right) \setminus (V(P_{\Phi}) \cup S_{\ell}) \right|
&= \left| \left( H_{\ell + 1}\left(a_{\max}^\Phi \right) \cap \closedint{a_{\min}^\Phi, a_{\max}^\Phi} \right) \setminus (V(P_{\Phi}) \cup S_{\ell}) \right| \\
&\quad + \left| H_{\ell + 1}\left(a_{\max}^\Phi \right) \setminus \left( \closedint{a_{\min}^\Phi, a_{\max}^\Phi} \cup S_{\ell} \right) \right| \\
&= g^\Phi(\ell, \max) + \left| H_{\ell + 1}\left(a_{\max}^\Phi \right) \setminus \left( \closedint{a_{\min}^\Phi, a_{\max}^\Phi} \cup S_{\ell} \right) \right|,
\end{split}
\end{equation*}
and similarly,
\begin{equation*}
\left| H_{\ell + 1}\left(a_{\min}^\Psi \right) \setminus (V(P_{\Psi}) \cup S_{\ell}) \right|
= g^\Psi(\ell, \min) + \left| H_{\ell + 1}\left(a_{\min}^\Psi \right) \setminus \left( \closedint{a_{\min}^\Psi, a_{\max}^\Psi} \cup S_{\ell} \right) \right|.
\end{equation*}
Moreover, for any $\ell \in I_{=}$, we obtain
\begin{equation*}
\begin{split}
\left| H_{\ell + 1}\left(a_{\max}^\Phi \right) \setminus (V(P_{\Phi}) \cup V(P_{\Psi}) \cup S_{\ell}) \right|
&=  \left| \left( H_{\ell + 1}\left(a_{\max}^\Phi \right) \cap \closedint{a_{\min}^\Phi, a_{\max}^\Phi} \right) \setminus (V(P_{\Phi}) \cup S_{\ell}) \right| \\
&\quad + \left| \left( H_{\ell + 1}\left(a_{\min}^\Psi \right) \cap \closedint{a_{\min}^\Psi, a_{\max}^\Psi} \right) \setminus (V(P_{\Psi}) \cup S_{\ell}) \right|  \\
&\quad + \left| H_{\ell + 1}\left(a_{\max}^\Phi \right) \setminus \left( \closedint{a_{\min}^\Phi, a_{\max}^\Phi} \cup \closedint{a_{\min}^\Psi, a_{\max}^\Psi} \cup S_{\ell} \right) \right|  \\
&= g^\Phi(\ell,\max) + g^\Psi(\ell, \min) \\
&\quad + \left| H_{\ell + 1}\left(a_{\max}^\Phi \right) \setminus \left( \closedint{a_{\min}^\Phi, a_{\max}^\Phi} \cup \closedint{a_{\min}^\Psi, a_{\max}^\Psi} \cup S_{\ell} \right) \right|.
\end{split}
\end{equation*}

Finally, note that
\begin{equation*}
\begin{split}
\gamma^\Phi + \gamma^\Psi 
&= \sum_{\ell = \ell^\ast}^{L} |(V(P_\Phi) \cup V(P_\Psi)) \setminus S_{\ell}| + \sum_{\ell=\ell^\ast}^{L - 1} \sum_{\substack{H \in \cH_{\ell + 1}, \\ H \subseteq \openint{a_{\min}^{\Phi}, a_{\max}^{\Phi}}, \\ |V(P_\Phi) \cap H| \geq k}} |H \setminus (V(P_\Phi) \cup S_{\ell})| \\
&\quad + \sum_{\ell=\ell^\ast}^{L - 1} \sum_{\substack{H \in \cH_{\ell + 1}, \\ H \subseteq \openint{a_{\min}^{\Psi}, a_{\max}^{\Psi}}, \\ |V(P_\Psi) \cap H| \geq k}} |H \setminus (V(P_\Psi) \cup S_{\ell})|
\end{split}
\end{equation*}
together with the observations above, implies that
\begin{equation*}
\begin{split}
&\sum_{\ell = \ell^\ast}^{L} |V(P) \setminus S_{\ell}| + \sum_{\ell=\ell^\ast}^{L - 1} \sum_{\substack{H \in \cH_{\ell + 1}, \\ H \subseteq \openint{a_{\min}^{\coveringsum}, a_{\max}^{\coveringsum}}, \\ |V(P) \cap H| \geq k}} |H \setminus (V(P) \cup S_{\ell})| \\
&= \gamma^\Phi + \gamma^\Psi \\
&\quad + \sum_{\ell \in \widetilde{I}_{\neq}^{\Phi}} \Big( g^\Phi(\ell, \max) + \left| H_{\ell + 1}\left(a_{\max}^\Phi \right) \setminus \left( \closedint{a_{\min}^\Phi, a_{\max}^\Phi} \cup S_{\ell} \right) \right| \Big) \\
&\quad + \sum_{\ell \in \widetilde{I}_{\neq}^{\Psi}} \Big( g^\Psi(\ell, \min) + \left| H_{\ell + 1}\left(a_{\min}^\Psi \right) \setminus \left( \closedint{a_{\min}^\Psi, a_{\max}^\Psi} \cup S_{\ell} \right) \right| \Big) \\
&\quad + \sum_{\ell \in \widetilde{I}_{=}} \Big( g^\Phi(\ell,\max) + g^\Psi(\ell, \min) \Big) \\
&\quad + \sum_{\ell \in \widetilde{I}_{=}} \left| H_{\ell + 1}\left(a_{\max}^\Phi \right) \setminus \left( \closedint{a_{\min}^\Phi, a_{\max}^\Phi} \cup \closedint{a_{\min}^\Psi, a_{\max}^\Psi} \cup S_{\ell} \right) \right|.
\end{split}
\end{equation*}
In particular, the right-hand side depends only on $\Phi$ and $\Psi$.
Therefore, choosing $\gamma^{\coveringsum}$ as this value completes the proof.
\end{proof}

Finally, we prove \Cref{lem:tree_solutions}:

\TreeSolutions*
\begin{proof}
Any $(x,y,k^\prime)$-realization $P$ of $\Phi$ satisfies $|V(P)| = k^\prime = 2^{\Gamma} - 1$.
Hence, it is straightforward to construct the unique full binary tree $\cB$ with $V(\cB) = V(P)$ such that $(\cB, \pi)$ is a tree solution of height $\Gamma$, where $\pi$ is the empty function.

\addparskip

It remains to prove that $\cost(\cB, \pi) = w^{0}(P^\prime) - k^2 \cdot D_{\ell^\ast}$, where $P^\prime$ results from $P$ as described above.
Given a tree solution $(\cB,\pi)$ of height $h > 1$, we recursively define its \emph{internal} cost as
\begin{equation*}
\intcost(\cB,\pi) \coloneqq \intcost(\cB_1, \pi_1) + \intcost(\cB_2, \pi_2) + 2^{h-1} \cdot \left(\sum_{a \in V(\cB_1)} c(a,r) + \sum_{a \in V(\cB_2)} c(r,a) \right),
\end{equation*}
where $r$ denotes the root of $\cB$ and $(\cB, \pi) = (\cB_1, \pi_1) \coveringsum (\cB_2, \pi_2)$.
If $\cB$ has height $1$, we define $\intcost(\cB,\pi) \coloneqq 0$.

Furthermore, we define its \emph{external} cost as
\begin{equation*}
\begin{split}
\extcost(\cB, \pi) 
&\coloneqq 2^{\Gamma} \cdot \sum_{a \in V(\cB)} \Big(c(x,a) + c(a,y) \Big) \\
&\quad + \sum_{a \in V(\cB)} \left(\sum_{\substack{\ell \in \parset{h+1, \dots, \Gamma}, \\ \pi(\ell) \prec a}} 2^{\ell - 1} \cdot c(\pi(\ell), a)
+ \sum_{\substack{\ell \in \parset{h+1, \dots, \Gamma}, \\ \pi(\ell) \succ a}} 2^{\ell - 1} \cdot c(a, \pi(\ell)) \right).
\end{split}
\end{equation*}
As an intermediate result, we prove that for any tree solution $(\cB,\pi)$ we have 
\begin{equation*}
\cost(\cB,\pi) = \intcost(\cB, \pi) + \extcost(\cB,\pi).
\end{equation*}
To this end, we proceed via induction on the height $h$ of $(\cB,\pi)$.
Any tree solution of height $h = 1$ consists of a single vertex.
Hence, in this case, the claim follows immediately from the definitions of $\intcost(\cB,\pi)$ and $\extcost(\cB,\pi)$.

Suppose that $h > 1$, and let $(\cB_1,\pi_1)$ and $(\cB_2,\pi_2)$ be the tree solutions of height $h - 1$ with $(\cB,\pi) = (\cB_1,\pi_1) \coveringsum (\cB_2,\pi_2)$.
Thus, $V(\cB_1) \prec \parset{r} \prec V(\cB_2)$ and $\pi_1(h) = \pi_2(h) = r$, where $r$ denotes the root of $\cB$.
By the induction hypothesis, we then have
\begin{equation*}
\begin{split}
\cost(\cB,\pi)
&=\cost(\cB_1, \pi_1) + \cost(\cB_2, \pi_2) + 2^{\Gamma} \cdot \big(c(x,r) + c(r, y)\big) \\
&\quad + \sum_{\substack{\ell \in \parset{h+1, \dots, \Gamma}, \\ \pi(\ell) \prec r}} 2^{\ell - 1} \cdot c(\pi(\ell), r)
+ \sum_{\substack{\ell \in \parset{h+1, \dots, \Gamma}, \\ \pi(\ell) \succ r}} 2^{\ell - 1} \cdot c(r, \pi(\ell)) \\
&= \intcost(\cB_1, \pi_1) + \extcost(\cB_1, \pi_1) + \intcost(\cB_2, \pi_2) + \extcost(\cB_2, \pi_2) \\
&\quad + 2^{\Gamma} \cdot \big(c(x,r) + c(r, y)\big) \\
&\quad + \sum_{\substack{\ell \in \parset{h+1, \dots, \Gamma}, \\ \pi(\ell) \prec r}} 2^{\ell - 1} \cdot c(\pi(\ell), r)
+ \sum_{\substack{\ell \in \parset{h+1, \dots, \Gamma}, \\ \pi(\ell) \succ r}} 2^{\ell - 1} \cdot c(r, \pi(\ell)) \\
&= \intcost(\cB_1, \pi_1) +  \intcost(\cB_2, \pi_2) + 2^{\Gamma} \cdot \sum_{a \in V(\cB)} \big(c(x,a) + c(a,y) \big) \\
&\quad + \sum_{a \in V(\cB_1)} \left(\sum_{\substack{\ell \in \parset{h, \dots, \Gamma}, \\ \pi_1(\ell) \prec a}} 2^{\ell - 1} \cdot c(\pi_1(\ell), a)
+ \sum_{\substack{\ell \in \parset{h, \dots, \Gamma}, \\ \pi_1(\ell) \succ a}} 2^{\ell - 1} \cdot c(a, \pi_1(\ell)) \right) \\
&\quad + \sum_{a \in V(\cB_2)} \left(\sum_{\substack{\ell \in \parset{h, \dots, \Gamma}, \\ \pi_2(\ell) \prec a}} 2^{\ell - 1} \cdot c(\pi_2(\ell), a)
+ \sum_{\substack{\ell \in \parset{h, \dots, \Gamma}, \\ \pi_2(\ell) \succ a}} 2^{\ell - 1} \cdot c(a, \pi_2(\ell)) \right) \\
&\quad + \sum_{\substack{\ell \in \parset{h + 1, \dots, \Gamma}, \\ \pi(\ell) \prec r}} 2^{\ell - 1} \cdot c(\pi(\ell), r)
+ \sum_{\substack{\ell \in \parset{h + 1, \dots, \Gamma}, \\ \pi(\ell) \succ r}} 2^{\ell - 1} \cdot c(r, \pi(\ell)) \\
&= \intcost(\cB_1, \pi_1) +  \intcost(\cB_2, \pi_2) + 2^{\Gamma} \cdot \sum_{a \in V(\cB)} \big(c(x,a) + c(a,y) \big) \\
&\quad + 2^{h-1} \cdot \left(\sum_{a \in V(\cB_1)} c(a,r) + \sum_{a \in V(\cB_2)} c(r,a) \right) \\
&\quad + \sum_{a \in V(\cB)} \left(\sum_{\substack{\ell \in \parset{h+1, \dots, \Gamma}, \\ \pi(\ell) \prec a}} 2^{\ell - 1} \cdot c(\pi(\ell), a)
+ \sum_{\substack{\ell \in \parset{h+1, \dots, \Gamma}, \\ \pi(\ell) \succ a}} 2^{\ell - 1} \cdot c(a, \pi(\ell)) \right) \\
&= \intcost(\cB, \pi) + \extcost(\cB, \pi).
\end{split}
\end{equation*}
This proves our intermediate result.

\addparskip

Now, let $(\cB,\pi)$ be a tree solution of height $\Gamma$.
We denote by $P$ the monotone path on $V(\cB)$ and by $P^\prime$ the monotone path that additionally visits $x$ and $y$.
Note that in order to complete the proof of our lemma, it suffices to prove
\begin{equation}
\label{eq:tree_solution_cost_claim2}
\begin{split}
w^{0} (P^\prime) - k^2 \cdot D_{\ell^\ast}
&= \intcost(\cB, \pi) + 2^{\Gamma} \cdot \sum_{a \in V(\cB)} \Big(c(x,a) + c(a,y) \Big),
\end{split}
\end{equation}
as this implies
\begin{equation*}
\begin{split}
w^{0} (P^\prime) - k^2 \cdot D_{\ell^\ast}
&= \intcost(\cB, \pi) + 2^{\Gamma} \cdot \sum_{a \in V(\cB)} \Big(c(x,a) + c(a,y) \Big) \\
&= \intcost(\cB, \pi) + \extcost(\cB, \pi) \\
&= \cost(\cB, \pi).
\end{split}
\end{equation*}

We prove~\eqref{eq:tree_solution_cost_claim2} by induction on $\Gamma$.
In the following, we write $\height_{\Gamma}$, $\cost_{\Gamma}$, and $w^{0}_{\Gamma}$ to make their dependence on $\Gamma$ explicit.

The case $\Gamma = 1$ is trivial, since the internal cost are zero in this case.
Suppose that $\Gamma > 1$.
Again, let $\cB_1$ and $\cB_2$ be the full binary subtrees of height $\Gamma - 1$ rooted at the children of the root $r$ of $\cB$ with
$V(\cB_1) \prec \parset{r} \prec V(\cB_2)$.
We denote by $P_1$ and $P_2$ the monotone paths whose vertex sets are $V(\cB_1)$ and $V(\cB_2)$, respectively.
We extend these paths to monotone paths $P_1^\prime$ and $P_2^\prime$ by additionally visiting $x$ and $y$.

We denote by $P^\prime_{[x,r]}$ the $x$-$r$ subpath of $P^\prime$, and by $P_{[r,y]}^\prime$ the $r$-$y$ subpath of $P^\prime$.
Then, \Cref{lem:decomposition_in_blocks} implies that
\begin{equation*}
\begin{split}
w^{0}_{\Gamma - 1} (P^\prime) - k^2 \cdot D_{\ell^\ast}
&= w^{0}_{\Gamma - 1} (P^\prime_{[x,r]})
+ w^{0}_{\Gamma - 1} (P^\prime_{[r,y]}) - 2k^2 \cdot D_{\ell^\ast}\\
&= w^{0}_{\Gamma - 1} (P^\prime_{1})
+ w^{0}_{\Gamma - 1} (P^\prime_{2}) - 2k^2 \cdot D_{\ell^\ast} \\
&\quad + 2^{\Gamma - 1}\cdot \left( \sum_{a \in V(\cB_1)} c(a,r) + \sum_{a \in V(\cB_2)} c(r,a) \right) \\
&\quad - 2^{\Gamma - 1} \cdot \left(  \sum_{a \in V(\cB_1)} c(a,y) + \sum_{a \in V(\cB_2)} c(x,a) \right).
\end{split}
\end{equation*}
For any $j \in \parset{1, \dots, 2^{\Gamma} - 1}$, we have $\height_{\Gamma - 1}(j) = \height_{\Gamma}(j)$, and for $j^\prime \in \parset{0, 2^{\Gamma}}$, we have $\height_{(\Gamma - 1)}(j^\prime) + 1 = \height_{\Gamma}(j^\prime)$.

Together with the induction hypothesis, this implies
\begin{equation*}
\begin{split}
&w^{0}_{\Gamma} (P^\prime) - k^2 \cdot D_{\ell^\ast} \\
&= w^{0}_{\Gamma - 1} (P^\prime) - k^2 \cdot D_{\ell^\ast} +  2^{\Gamma}\cdot \Big( c(x,r) + c(r,y) \Big) \\
&\quad + 2^{\Gamma - 1} \cdot \left(\sum_{a \in V(\cB_1)} c(x,a) +  \sum_{a \in V(\cB_2) } c(a,y)\right)
+ 2^{\Gamma} \cdot \left(\sum_{a \in V(\cB_1) } c(a,y) +  \sum_{a \in V(\cB_2) } c(x,a)\right) \\
&= w^{0}_{\Gamma - 1} (P^\prime_{1})
+ w^{0}_{\Gamma - 1} (P^\prime_{2}) - 2k^2 \cdot D_{\ell^\ast} + 2^{\Gamma}\cdot \Big( c(x,r) + c(r,y) \Big) \\
&\quad + 2^{\Gamma - 1} \cdot \left(\sum_{a \in V(\cB_1)} c(x,a) +  \sum_{a \in V(\cB_2) } c(a,y)\right)
+ 2^{\Gamma - 1} \cdot \left(\sum_{a \in V(\cB_1) } c(a,y) +  \sum_{a \in V(\cB_2) } c(x,a)\right) \\
&\quad + 2^{\Gamma - 1}\cdot \left( \sum_{a \in V(\cB_1)} c(a,r) + \sum_{a \in V(\cB_2)} c(r,a) \right) \\
\overset{\eqref{eq:tree_solution_cost_claim2}}&{=} \intcost(\cB_1, \pi_1) + \intcost(\cB_2, \pi_2) +  2^{\Gamma}\cdot \Big( c(x,r) + c(r,y) \Big) \\
& \quad + 2^{\Gamma} \cdot \left(\sum_{a \in V(\cB_1) } \Big(c(x,a) + c(a,y)\Big) +  \sum_{a \in V(\cB_2) } \Big( c(x,a) + c(a,y) \Big) \right) \\
&\quad + 2^{\Gamma - 1}\cdot \left( \sum_{a \in V(\cB_1)} c(a,r) + \sum_{a \in V(\cB_2)} c(r,a) \right) \\
&= \intcost(\cB, \pi) + 2^{\Gamma} \cdot \sum_{a \in V(\cB)} \Big(c(x,a) + c(a,y)\Big),
\end{split}
\end{equation*}
which completes the proof.
\end{proof}

\section{Conclusion}\label{sec:conclusion}

Combining \Cref{thm:reducing_hop_atsp,thm:covering_reduction,thm:path_covering_algorithm} directly implies our main result, \Cref{thm:main}.
We have thus given a randomized algorithm for Asymmetric A Priori TSP with poly-logarithmic approximation factor and quasi-polynomial running time.
Interestingly, by \Cref{thm:lower_bound}, the approximation factor is below the adaptivity gap.

A natural open question is whether one can improve the running time of our algorithm and achieve a poly-logarithmic approximation factor in polynomial time.
For this, the polynomial-time reductions we provided in this paper could be an important ingredient.
We highlight that our proofs not only imply reductions from Asymmetric A Priori TSP to Hop-ATSP, hierarchically-ordered instances, and finally to the Path Covering problem, but they actually show that these problems are equivalent up to poly-logarithmic factors in the approximation guarantee (and polynomial factors in the running time).

All our reductions rely on showing that a solution for one problem can be turned into a solution of the other problem and vice versa, without changing the cost by more than a poly-logarithmic factor (and possible scaling of the objective).
When reducing to a covering problem in \Cref{sec:reducing_to_covering}, we showed how tours (for non-degenerate instances) can be transformed into covering solutions of not much larger cost. 
We did not specify running times for this transformation because we only applied this statement to an optimal tour in order to show that the value of an optimal Path Covering solution cannot be much larger than the $k$-hop cost of an optimal tour.
However, we remark that our proof is constructive and gives rise to a polynomial-time algorithm.

Another interesting direction for future research would be to see if approximation factors beyond the adaptivity gap can also be achieved for other related problems such as the Symmetric A Priori TSP. 
\appendix
\crefalias{section}{appsec}

\section{Tightness of  \Cref{lem:adaptivity_lb_lb}}
\label{appendix:grid}

\begin{lemma}
Let $k \in \bZ_{\geq 2}$ such that $\sqrt{k} \in \bZ$.
Then, there exists a tour $T$ on $V_k$ that satisfies
\begin{equation*}
\expectationcustom{A \sim p_k}{c_k (T[A]) } \leq O({k^{3/2}}).
\end{equation*}
\end{lemma}
\begin{proof}
The tour $T$ we consider is shown in \Cref{fig:asymp_opt_apriori_for_grid_Z}.
It can be formally described by the tour $x_1, \dots , x_{k^2}$ , where
\begin{equation*}
x_i \coloneqq v_{j \cdot \sqrt{k} +l, t} \quad \text{with } \begin{cases} j&=\floor*{\frac{i-1}{k\sqrt{k}}}, \\[1em]
l&= \left(i-1 \mod \sqrt{k}\right) + 1, \\[1em]
t&= \left(\floor*{\frac{i-1}{\sqrt{k}}} \mod k\right) + 1.
\end{cases}
\end{equation*}
Let $i \in \parset{1, \dots, \sqrt{k}}$ be fixed.
We analyze the contribution of $\delta^+(x_i) \coloneqq \parset{x_i} \times V_k$ to $\expectationcustom{A \sim p_k}{c_k(T[A])}$, that is,
\begin{equation*}
\sum_{e \in \delta^+(x_i)} \probabilitycustom{A \sim p_k}{e \in E(T[A])} \cdot c_k(e).
\end{equation*}
By symmetry of the instance, this is sufficient to obtain a bound on all edges $e \in V_k \times V_k$, because the vertices of $V_k$ can be renumbered accordingly.

We define 
\begin{equation*}
Z_j \coloneqq \parset{x_l : l = j\cdot\sqrt{k} + 1, \dots, (j+1)\cdot\sqrt{k}},
\end{equation*}
for $j=0, \dots,k\cdot \sqrt{k} - 1$.
Then, we can simply bound the contribution of $\delta^+(x_i) \cap \delta^-(Z_0)$ by
\begin{equation*}
\begin{split}
\sum_{e \in \delta^+(x_i) \cap \delta^-(Z_0)} \probabilitycustom{A \sim p_k}{e \in E(T[A])} \cdot c_k(e)
&\leq
\sum_{e \in \delta^+(x_i) \cap \delta^-(Z_0)} (1/k)^2 \cdot k \\
&\leq
\sqrt{k} \cdot (1/k) \\
&= \frac{1}{\sqrt{k}}.
\end{split}
\end{equation*}
For $j \in \parset{1,\dots,k\cdot \sqrt{k} - 1}$, the cost of an edge in $\delta^+(x_i) \cap \delta^-(Z_j)$ can be upper bounded by $j$.
Furthermore, between $x_i$ and any vertex $v$ of $Z_j$ lie at least $(j-1)\sqrt{k}$ many other vertices on $T$.
Therefore, the activation probability of such an edge $(x_i,v)$ can be upper bounded by $(1/k)^2\left(1-1/k\right)^{(j-1)\sqrt{k}}$.
Hence,
\begin{equation*}
\begin{split}
\sum_{e \in \delta^+(x_i) \cap \delta^-(Z_j)} \probabilitycustom{A \sim p_k}{e \in E(T[A])} \cdot c_k(e)
&\leq
\sum_{e \in \delta^+(x_i) \cap \delta^-(Z_j)} \frac{1}{k^2} \cdot \left(1 - \frac{1}{k} \right)^{(j-1)\cdot \sqrt{k}} \cdot j \\
&\leq 
\frac{\sqrt{k}}{k^2} \cdot \left(1 - \frac{1}{k} \right)^{(j-1)\cdot \sqrt{k}} \cdot j \\
&\leq
\frac{\sqrt{k}}{k^2} \cdot e^{\frac{-(j-1)}{\sqrt{k}}} \cdot j \\
&\leq
\frac{e}{k^{3/2}} \cdot \left( e^{-1/\sqrt{k}} \right)^j \cdot j.
\end{split}
\end{equation*}
\begin{figure}
\begin{center}
\begin{tikzpicture}[
vertex/.style={circle,draw=black, fill, inner sep=1pt},
edge/.style={-latex},
box/.style={Peach, very thick, fill=Violet, fill opacity=0.05, rounded corners=5.5pt},
scale=0.5
]
\def\k{16}
\def\sqrtk{4}
\pgfmathsetmacro{\kkm}{16*16 - 1}

\begin{scope}[every node/.style={vertex}]

\foreach \i in {1,...,\k} {
\foreach \j in {1,...,\k}{
\node (x\i y\j) at (\j,16 - \i + 1) {};
}
}

\end{scope}

\newcommand{\getindex}[1]{
\pgfmathsetmacro{\num}{#1}

\pgfmathsetmacro{\t}{mod(floor((\num - 1) / \sqrtk), \k) + 1}
\pgfmathsetmacro{\l}{mod(\num - 1, \sqrtk) + 1}
\pgfmathsetmacro{\j}{floor((\num - 1) / (\k * \sqrtk))}

\pgfmathsetmacro{\first}{\j * \sqrtk + \l}
\pgfmathsetmacro{\second}{\t}

\pgfmathtruncatemacro{\firstindex}{\first}
\pgfmathtruncatemacro{\secondindex}{\second}
}

\foreach \i in {1,...,\kkm} {
\ifnum\i=64\else        
\ifnum\i=128\else        
\ifnum\i=192\else        

\getindex{\i}
\pgfmathsetmacro{\prevx}{\firstindex}
\pgfmathsetmacro{\prevy}{\secondindex}

\getindex{\i +1}
\pgfmathsetmacro{\nextx}{\firstindex}
\pgfmathsetmacro{\nexty}{\secondindex}

\ifnum\numexpr\i-4*(\i/4)\relax=0
\draw[edge, out=40, in=-130, looseness=0.4] (x\prevx y\prevy) to (x\nextx y\nexty);
\else
\draw[edge] (x\prevx y\prevy) to (x\nextx y\nexty);
\fi\fi\fi\fi
}

\draw[edge] (x4y16) .. controls (15,12.2) and (2,12.8) .. (x5y1);  
\draw[edge] (x8y16) .. controls (15,8.2) and (2,8.8) .. (x9y1);  
\draw[edge] (x12y16) .. controls (15,4.2) and (2,4.8) .. (x13y1);  

\draw[] (x16y16) .. controls (16.9,2) .. (16.7,14);
\draw[] (16.7,14) .. controls (16.7,16.7) .. (14,16.7);
\draw[edge] (14,16.7) .. controls (2,16.7) .. (x1y1);

\draw[decorate, decoration={brace}, thick] (0.5,9) -- (0.5,12) node[midway,left] {$\sqrt{k}$ \ };

\def\xslack{0.4}	
\def\yslack{0.5}	

\draw[box] ($(x1y1) + (\xslack, \yslack)$) rectangle ($(x4y1) - (\xslack, \yslack)$);
\node at (x1y1) [above=14pt] {\textcolor{Peach}{\textbf{$Z_{0}$}}};
\end{tikzpicture}
 \end{center}
\caption{
An example for $k=16$.
An asymptotically optimal a priori tour $T$ for $(V_k,c_k,p_k)$ is obtained by grouping the rows into blocks of size $\sqrt{k}$ and traversing each block in a column-wise manner.
}
\label{fig:asymp_opt_apriori_for_grid_Z}
\end{figure}
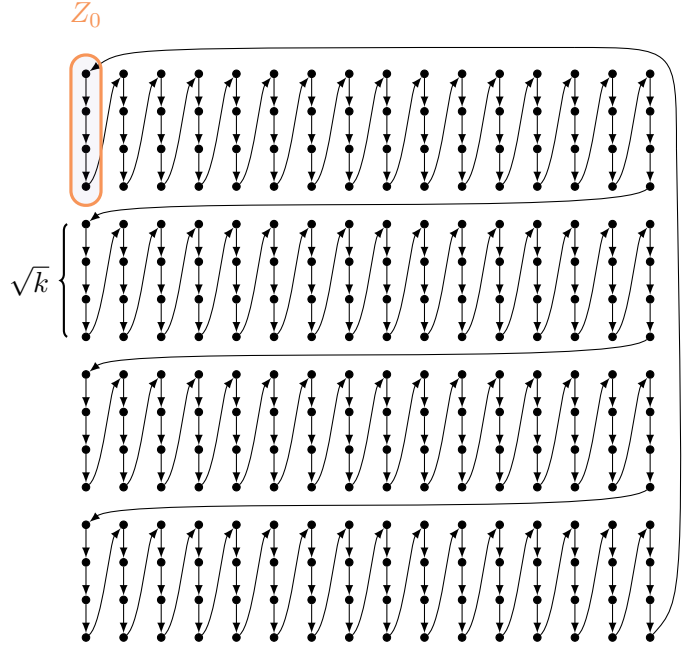
This implies an upper bound on the total contribution of $\delta^+(x_i) \cap \delta^-\left( \bigcup_{j=1}^{k^{3/2}-1} Z_j \right)$, since
\begin{equation*}
\begin{split}
\sum_{j = 1}^{k^{3/2}-1}
\sum_{e \in \delta^+(x_i) \cap \delta^-(Z_j)} \probabilitycustom{A \sim p_k}{e \in E(T[A])} \cdot c_k(e) &\leq \sum_{j = 1}^{k^{3/2}-1} \frac{e}{k^{3/2}} \cdot \left( e^{\frac{-1}{\sqrt{k}} } \right)^j \cdot j \\
&\leq \sum_{i = 0}^{\infty} \frac{e}{k^{3/2}} \cdot \left( e^{\frac{-1}{\sqrt{k}} } \right)^j \cdot j \\
&=
 \frac{e}{k^{3/2}} \cdot \frac{e^{\frac{-1}{\sqrt{k}} }}{\left( 1-e^{\frac{-1}{\sqrt{k}} }\right)^2} \\
&\leq
\frac{e}{k^{3/2}} \cdot \frac{1}{ \left( \frac{1}{2\sqrt{k}} \right)^2 } \\
&= \frac{4e}{\sqrt{k}}.
\end{split}
\end{equation*}

We conclude that
\begin{equation*}
\begin{split}
\sum_{e \in \delta^+(x_i)} \probabilitycustom{A \sim p_k}{e \in E(T[A])} \cdot c_k(e)
&= \sum_{j = 0}^{k^{3/2}-1}
\sum_{e \in \delta^+(x_i) \cap \delta^-(Z_j)} \probabilitycustom{A \sim p_k}{e \in E(T[A])} \cdot c_k(e) \\
&\leq \frac{1}{\sqrt{k}} + \frac{4e}{\sqrt{k}}.
\end{split}
\end{equation*}
As argued above, by symmetry, this bound holds for all $i = 1, \dots, k^2$.
Summing up this bound over all $i = 1, \dots, k^2$ yields a total upper bound of $k^2\cdot \frac{1+4e}{\sqrt{k}} = (1+4e)k^{3/2}$ for $\expectationcustom{A \sim p_k}{c_k(T[A])}$.
\end{proof}

\printbibliography

\end{document}